\documentclass{CSML}

\def\dOi{12(3:2)2016}
\lmcsheading%
{\dOi}
{1--47}
{}
{}
{Nov.~19, 2014}
{Aug.~\phantom09, 2016}
{}

\ACMCCS{[{\bf Theory of computation}]: Logic---Proof theory\,/\,Logic
  and verification\,/\,Type theory; Semantics and Reasoning;  Models
  of computation---Computability---Lambda Calculus} 

\usepackage{amssymb,amstext,amsthm}
\usepackage{mathpartir}
\usepackage{enumitem}

\usepackage{url}
\usepackage{todonotes}
\usepackage{listings}
\usepackage{stmaryrd} 
\usepackage{graphicx}
\usepackage{tikz}
\usepackage{hyperref}
\usetikzlibrary{arrows}
\usetikzlibrary{shapes}

\newcommand{\redto}[0]{\rightsquigarrow}
\newcommand{\redtoby}[1]{\stackrel{ \ifrName{({\tiny #1})}}{\redto}}
\newcommand{\interp}[1]{\llbracket #1 \rrbracket}

\usepackage{supertabular}

\newcommand*{\Cdot}[1][1.25]{%
  \mathpalette{\CdotAux{#1}}\cdot%
}
\newdimen\CdotAxis
\newcommand*{\CdotAux}[3]{%
  {%
    \settoheight\CdotAxis{$#2\vcenter{}$}%
    \sbox0{%
      \raisebox\CdotAxis{%
        \scalebox{#1}{%
          \raisebox{-\CdotAxis}{%
            $\mathsurround=0pt #2#3$%
          }%
        }%
      }%
    }%
    \dp0=0pt %
    \sbox2{$#2\bullet$}%
    \ifdim\ht2<\ht0 %
      \ht0=\ht2 %
    \fi
    \sbox2{$\mathsurround=0pt #2#3$}%
    \hbox to \wd2{\hss\usebox{0}\hss}%
  }%
}

\newcommand{\dttdrule}[4][]{{\displaystyle\frac{\begin{array}{l}#2\end{array}}{#3}\quad\dttdrulename{#4}}}

\newcommand{\dttpremise}[1]{ #1 \\}
\newenvironment{dttdefnblock}[3][]{ \framebox{\mbox{#2}} \quad #3 \\[0pt]}{}
\newenvironment{dttfundefnblock}[3][]{ \framebox{\mbox{#2}} \quad #3 \\[0pt]\begin{displaymath}\begin{array}{l}}{\end{array}\end{displaymath}}

\newcommand{\dttnt}[1]{\mathit{#1}}
\newcommand{\dttmv}[1]{\mathit{#1}}
\newcommand{\dttkw}[1]{\mathbf{#1}}
\newcommand{\dttsym}[1]{#1}

\newcommand{\dttdrulename}[1]{\textsc{#1}}

\newcommand{\dttdrulerelXXax}[1]{\dttdrule[#1]{%
}{
G  \dttsym{,}  \dttnt{n} \,  \preccurlyeq_{ \dttnt{p} }  \, \dttnt{n'}  \dttsym{,}  G'  \vdash  \dttnt{n} \,  \preccurlyeq^*_{ \dttnt{p} }  \, \dttnt{n'}}{%
{\dttdrulename{rel\_ax}}{}%
}}

\newcommand{\dttdrulerelXXrefl}[1]{\dttdrule[#1]{%
}{
G  \vdash  \dttnt{n} \,  \preccurlyeq^*_{ \dttnt{p} }  \, \dttnt{n}}{%
{\dttdrulename{rel\_refl}}{}%
}}

\newcommand{\dttdrulerelXXtrans}[1]{\dttdrule[#1]{%
\dttpremise{ G  \vdash  \dttnt{n} \,  \preccurlyeq^*_{ \dttnt{p} }  \, \dttnt{n'}  \qquad  G  \vdash  \dttnt{n'} \,  \preccurlyeq^*_{ \dttnt{p} }  \, \dttnt{n''} }%
}{
G  \vdash  \dttnt{n} \,  \preccurlyeq^*_{ \dttnt{p} }  \, \dttnt{n''}}{%
{\dttdrulename{rel\_trans}}{}%
}}

\newcommand{\dttdrulerelXXflip}[1]{\dttdrule[#1]{%
\dttpremise{G  \vdash  \dttnt{n'} \,  \preccurlyeq^*_{  \bar{  \dttnt{p}  }  }  \, \dttnt{n}}%
}{
G  \vdash  \dttnt{n} \,  \preccurlyeq^*_{ \dttnt{p} }  \, \dttnt{n'}}{%
{\dttdrulename{rel\_flip}}{}%
}}

\newcommand{\dttdruleax}[1]{\dttdrule[#1]{%
\dttpremise{G  \vdash  \dttnt{n} \,  \preccurlyeq^*_{ \dttnt{p} }  \, \dttnt{n'}}%
}{
G  \dttsym{;}  \Gamma  \dttsym{,}  \dttnt{p} \, \dttnt{A}  \mathbin{@}  \dttnt{n}  \dttsym{,}  \Gamma'  \vdash  \dttnt{p} \, \dttnt{A}  \mathbin{@}  \dttnt{n'}}{%
{\dttdrulename{ax}}{}%
}}

\newcommand{\dttdruleunit}[1]{\dttdrule[#1]{%
}{
G  \dttsym{;}  \Gamma  \vdash  \dttnt{p} \,  \langle  \dttnt{p} \rangle   \mathbin{@}  \dttnt{n}}{%
{\dttdrulename{unit}}{}%
}}

\newcommand{\dttdruleand}[1]{\dttdrule[#1]{%
\dttpremise{ G  \dttsym{;}  \Gamma  \vdash  \dttnt{p} \, \dttnt{A}  \mathbin{@}  \dttnt{n}  \qquad  G  \dttsym{;}  \Gamma  \vdash  \dttnt{p} \, \dttnt{B}  \mathbin{@}  \dttnt{n} }%
}{
G  \dttsym{;}  \Gamma  \vdash  \dttnt{p} \, \dttsym{(}   \dttnt{A}  \ndwedge{ \dttnt{p} }  \dttnt{B}   \dttsym{)}  \mathbin{@}  \dttnt{n}}{%
{\dttdrulename{and}}{}%
}}

\newcommand{\dttdruleandBar}[1]{\dttdrule[#1]{%
\dttpremise{G  \dttsym{;}  \Gamma  \vdash  \dttnt{p} \,  \dttnt{A} _{ \dttnt{d} }   \mathbin{@}  \dttnt{n}}%
}{
G  \dttsym{;}  \Gamma  \vdash  \dttnt{p} \, \dttsym{(}   \dttnt{A_{{\mathrm{1}}}}  \ndwedge{  \bar{  \dttnt{p}  }  }  \dttnt{A_{{\mathrm{2}}}}   \dttsym{)}  \mathbin{@}  \dttnt{n}}{%
{\dttdrulename{andBar}}{}%
}}

\newcommand{\dttdruleimp}[1]{\dttdrule[#1]{%
\dttpremise{\dttnt{n'} \, \not\in \, \dttsym{\mbox{$\mid$}}  G  \dttsym{\mbox{$\mid$}}  \dttsym{,}  \dttsym{\mbox{$\mid$}}  \Gamma  \dttsym{\mbox{$\mid$}}}%
\dttpremise{\dttsym{(}  G  \dttsym{,}  \dttnt{n} \,  \preccurlyeq_{ \dttnt{p} }  \, \dttnt{n'}  \dttsym{)}  \dttsym{;}  \Gamma  \dttsym{,}  \dttnt{p} \, \dttnt{A}  \mathbin{@}  \dttnt{n'}  \vdash  \dttnt{p} \, \dttnt{B}  \mathbin{@}  \dttnt{n'}}%
}{
G  \dttsym{;}  \Gamma  \vdash  \dttnt{p} \, \dttsym{(}   \dttnt{A}  \ndto{ \dttnt{p} }  \dttnt{B}   \dttsym{)}  \mathbin{@}  \dttnt{n}}{%
{\dttdrulename{imp}}{}%
}}

\newcommand{\dttdruleimpBar}[1]{\dttdrule[#1]{%
\dttpremise{G  \vdash  \dttnt{n} \,  \preccurlyeq^*_{  \bar{  \dttnt{p}  }  }  \, \dttnt{n'}}%
\dttpremise{ G  \dttsym{;}  \Gamma  \vdash   \bar{  \dttnt{p}  }  \, \dttnt{A}  \mathbin{@}  \dttnt{n'}  \qquad  G  \dttsym{;}  \Gamma  \vdash  \dttnt{p} \, \dttnt{B}  \mathbin{@}  \dttnt{n'} }%
}{
G  \dttsym{;}  \Gamma  \vdash  \dttnt{p} \, \dttsym{(}   \dttnt{A}  \ndto{  \bar{  \dttnt{p}  }  }  \dttnt{B}   \dttsym{)}  \mathbin{@}  \dttnt{n}}{%
{\dttdrulename{impBar}}{}%
}}

\newcommand{\dttdrulecut}[1]{\dttdrule[#1]{%
\dttpremise{ G  \dttsym{;}  \Gamma  \dttsym{,}   \bar{  \dttnt{p}  }  \, \dttnt{A}  \mathbin{@}  \dttnt{n}  \vdash  \dttsym{+} \, \dttnt{B}  \mathbin{@}  \dttnt{n'}  \qquad  G  \dttsym{;}  \Gamma  \dttsym{,}   \bar{  \dttnt{p}  }  \, \dttnt{A}  \mathbin{@}  \dttnt{n}  \vdash  \dttsym{-} \, \dttnt{B}  \mathbin{@}  \dttnt{n'} }%
}{
G  \dttsym{;}  \Gamma  \vdash  \dttnt{p} \, \dttnt{A}  \mathbin{@}  \dttnt{n}}{%
{\dttdrulename{cut}}{}%
}}

\newcommand{\dttdruleaxCut}[1]{\dttdrule[#1]{%
\dttpremise{  \dttnt{p} \, \dttnt{B}  \mathbin{@}  \dttnt{n'}  \in  \dttsym{(}  \Gamma  \dttsym{,}   \bar{  \dttnt{p}  }  \, \dttnt{A}  \mathbin{@}  \dttnt{n}  \dttsym{)}   \qquad  G  \dttsym{;}  \Gamma  \dttsym{,}   \bar{  \dttnt{p}  }  \, \dttnt{A}  \mathbin{@}  \dttnt{n}  \vdash   \bar{  \dttnt{p}  }  \, \dttnt{B}  \mathbin{@}  \dttnt{n'} }%
}{
G  \dttsym{;}  \Gamma  \vdash  \dttnt{p} \, \dttnt{A}  \mathbin{@}  \dttnt{n}}{%
{\dttdrulename{axCut}}{}%
}}

\newcommand{\dttdruleaxCutBar}[1]{\dttdrule[#1]{%
\dttpremise{   \bar{  \dttnt{p}  }  \, \dttnt{B}  \mathbin{@}  \dttnt{n'}  \in  \dttsym{(}  \Gamma  \dttsym{,}   \bar{  \dttnt{p}  }  \, \dttnt{A}  \mathbin{@}  \dttnt{n}  \dttsym{)}   \qquad  G  \dttsym{;}  \Gamma  \dttsym{,}   \bar{  \dttnt{p}  }  \, \dttnt{A}  \mathbin{@}  \dttnt{n}  \vdash  \dttnt{p} \, \dttnt{B}  \mathbin{@}  \dttnt{n'} }%
}{
G  \dttsym{;}  \Gamma  \vdash  \dttnt{p} \, \dttnt{A}  \mathbin{@}  \dttnt{n}}{%
{\dttdrulename{axCutBar}}{}%
}}

\newcommand{\dttdruleAx}[1]{\dttdrule[#1]{%
\dttpremise{G  \vdash  \dttnt{n} \,  \preccurlyeq^*_{ \dttnt{p} }  \, \dttnt{n'}}%
}{
G  \dttsym{;}  \Gamma  \dttsym{,}  \dttmv{x}  \dttsym{:}  \dttnt{p} \, \dttnt{A}  \mathbin{@}  \dttnt{n}  \dttsym{,}  \Gamma'  \vdash  \dttmv{x}  \dttsym{:}  \dttnt{p} \, \dttnt{A}  \mathbin{@}  \dttnt{n'}}{%
{\dttdrulename{Ax}}{}%
}}

\newcommand{\dttdruleUnit}[1]{\dttdrule[#1]{%
}{
G  \dttsym{;}  \Gamma  \vdash  \dttkw{triv}  \dttsym{:}  \dttnt{p} \,  \langle  \dttnt{p} \rangle   \mathbin{@}  \dttnt{n}}{%
{\dttdrulename{Unit}}{}%
}}

\newcommand{\dttdruleAnd}[1]{\dttdrule[#1]{%
\dttpremise{ G  \dttsym{;}  \Gamma  \vdash  \dttnt{t_{{\mathrm{1}}}}  \dttsym{:}  \dttnt{p} \, \dttnt{A}  \mathbin{@}  \dttnt{n}  \qquad  G  \dttsym{;}  \Gamma  \vdash  \dttnt{t_{{\mathrm{2}}}}  \dttsym{:}  \dttnt{p} \, \dttnt{B}  \mathbin{@}  \dttnt{n} }%
}{
G  \dttsym{;}  \Gamma  \vdash  \dttsym{(}  \dttnt{t_{{\mathrm{1}}}}  \dttsym{,}  \dttnt{t_{{\mathrm{2}}}}  \dttsym{)}  \dttsym{:}  \dttnt{p} \, \dttsym{(}   \dttnt{A}  \ndwedge{ \dttnt{p} }  \dttnt{B}   \dttsym{)}  \mathbin{@}  \dttnt{n}}{%
{\dttdrulename{And}}{}%
}}

\newcommand{\dttdruleAndBar}[1]{\dttdrule[#1]{%
\dttpremise{G  \dttsym{;}  \Gamma  \vdash  \dttnt{t}  \dttsym{:}  \dttnt{p} \,  \dttnt{A} _{ \dttnt{d} }   \mathbin{@}  \dttnt{n}}%
}{
G  \dttsym{;}  \Gamma  \vdash   \mathbf{in}_{ \dttnt{d} }\, \dttnt{t}   \dttsym{:}  \dttnt{p} \, \dttsym{(}   \dttnt{A_{{\mathrm{1}}}}  \ndwedge{  \bar{  \dttnt{p}  }  }  \dttnt{A_{{\mathrm{2}}}}   \dttsym{)}  \mathbin{@}  \dttnt{n}}{%
{\dttdrulename{AndBar}}{}%
}}

\newcommand{\dttdruleImp}[1]{\dttdrule[#1]{%
\dttpremise{\dttnt{n'} \, \not\in \, \dttsym{\mbox{$\mid$}}  G  \dttsym{\mbox{$\mid$}}  \dttsym{,}  \dttsym{\mbox{$\mid$}}  \Gamma  \dttsym{\mbox{$\mid$}}}%
\dttpremise{\dttsym{(}  G  \dttsym{,}  \dttnt{n} \,  \preccurlyeq_{ \dttnt{p} }  \, \dttnt{n'}  \dttsym{)}  \dttsym{;}  \Gamma  \dttsym{,}  \dttmv{x}  \dttsym{:}  \dttnt{p} \, \dttnt{A}  \mathbin{@}  \dttnt{n'}  \vdash  \dttnt{t}  \dttsym{:}  \dttnt{p} \, \dttnt{B}  \mathbin{@}  \dttnt{n'}}%
}{
G  \dttsym{;}  \Gamma  \vdash  \lambda  \dttmv{x}  \dttsym{.}  \dttnt{t}  \dttsym{:}  \dttnt{p} \, \dttsym{(}   \dttnt{A}  \ndto{ \dttnt{p} }  \dttnt{B}   \dttsym{)}  \mathbin{@}  \dttnt{n}}{%
{\dttdrulename{Imp}}{}%
}}

\newcommand{\dttdruleImpBar}[1]{\dttdrule[#1]{%
\dttpremise{G  \vdash  \dttnt{n} \,  \preccurlyeq^*_{  \bar{  \dttnt{p}  }  }  \, \dttnt{n'}}%
\dttpremise{ G  \dttsym{;}  \Gamma  \vdash  \dttnt{t_{{\mathrm{1}}}}  \dttsym{:}   \bar{  \dttnt{p}  }  \, \dttnt{A}  \mathbin{@}  \dttnt{n'}  \qquad  G  \dttsym{;}  \Gamma  \vdash  \dttnt{t_{{\mathrm{2}}}}  \dttsym{:}  \dttnt{p} \, \dttnt{B}  \mathbin{@}  \dttnt{n'} }%
}{
G  \dttsym{;}  \Gamma  \vdash  \langle  \dttnt{t_{{\mathrm{1}}}}  \dttsym{,}  \dttnt{t_{{\mathrm{2}}}}  \rangle  \dttsym{:}  \dttnt{p} \, \dttsym{(}   \dttnt{A}  \ndto{  \bar{  \dttnt{p}  }  }  \dttnt{B}   \dttsym{)}  \mathbin{@}  \dttnt{n}}{%
{\dttdrulename{ImpBar}}{}%
}}

\newcommand{\dttdruleCut}[1]{\dttdrule[#1]{%
\dttpremise{G  \dttsym{;}  \Gamma  \dttsym{,}  \dttmv{x}  \dttsym{:}   \bar{  \dttnt{p}  }  \, \dttnt{A}  \mathbin{@}  \dttnt{n}  \vdash  \dttnt{t_{{\mathrm{1}}}}  \dttsym{:}  \dttsym{+} \, \dttnt{B}  \mathbin{@}  \dttnt{n'}}%
\dttpremise{G  \dttsym{;}  \Gamma  \dttsym{,}  \dttmv{x}  \dttsym{:}   \bar{  \dttnt{p}  }  \, \dttnt{A}  \mathbin{@}  \dttnt{n}  \vdash  \dttnt{t_{{\mathrm{2}}}}  \dttsym{:}  \dttsym{-} \, \dttnt{B}  \mathbin{@}  \dttnt{n'}}%
}{
G  \dttsym{;}  \Gamma  \vdash  \nu \, \dttmv{x}  \dttsym{.}  \dttnt{t_{{\mathrm{1}}}}  \mathbin{\Cdot[2]}  \dttnt{t_{{\mathrm{2}}}}  \dttsym{:}  \dttnt{p} \, \dttnt{A}  \mathbin{@}  \dttnt{n}}{%
{\dttdrulename{Cut}}{}%
}}

\newcommand{\dttdruleRImp}[1]{\dttdrule[#1]{%
}{
\nu \, \dttmv{z}  \dttsym{.}  \lambda  \dttmv{x}  \dttsym{.}  \dttnt{t}  \mathbin{\Cdot[2]}  \langle  \dttnt{t_{{\mathrm{1}}}}  \dttsym{,}  \dttnt{t_{{\mathrm{2}}}}  \rangle  \redto  \nu \, \dttmv{z}  \dttsym{.}  \dttsym{[}  \dttnt{t_{{\mathrm{1}}}}  \dttsym{/}  \dttmv{x}  \dttsym{]}  \dttnt{t}  \mathbin{\Cdot[2]}  \dttnt{t_{{\mathrm{2}}}}}{%
{\dttdrulename{RImp}}{}%
}}

\newcommand{\dttdruleRImpBar}[1]{\dttdrule[#1]{%
}{
\nu \, \dttmv{z}  \dttsym{.}  \langle  \dttnt{t_{{\mathrm{1}}}}  \dttsym{,}  \dttnt{t_{{\mathrm{2}}}}  \rangle  \mathbin{\Cdot[2]}  \lambda  \dttmv{x}  \dttsym{.}  \dttnt{t}  \redto  \nu \, \dttmv{z}  \dttsym{.}  \dttnt{t_{{\mathrm{2}}}}  \mathbin{\Cdot[2]}  \dttsym{[}  \dttnt{t_{{\mathrm{1}}}}  \dttsym{/}  \dttmv{x}  \dttsym{]}  \dttnt{t}}{%
{\dttdrulename{RImpBar}}{}%
}}

\newcommand{\dttdruleRAndOne}[1]{\dttdrule[#1]{%
}{
\nu \, \dttmv{z}  \dttsym{.}  \dttsym{(}  \dttnt{t_{{\mathrm{1}}}}  \dttsym{,}  \dttnt{t_{{\mathrm{2}}}}  \dttsym{)}  \mathbin{\Cdot[2]}   \mathbf{in}_{ \dttsym{1} }\, \dttnt{t}   \redto  \nu \, \dttmv{z}  \dttsym{.}  \dttnt{t_{{\mathrm{1}}}}  \mathbin{\Cdot[2]}  \dttnt{t}}{%
{\dttdrulename{RAnd1}}{}%
}}

\newcommand{\dttdruleRAndTwo}[1]{\dttdrule[#1]{%
}{
\nu \, \dttmv{z}  \dttsym{.}  \dttsym{(}  \dttnt{t_{{\mathrm{1}}}}  \dttsym{,}  \dttnt{t_{{\mathrm{2}}}}  \dttsym{)}  \mathbin{\Cdot[2]}   \mathbf{in}_{ \dttsym{2} }\, \dttnt{t}   \redto  \nu \, \dttmv{z}  \dttsym{.}  \dttnt{t_{{\mathrm{2}}}}  \mathbin{\Cdot[2]}  \dttnt{t}}{%
{\dttdrulename{RAnd2}}{}%
}}

\newcommand{\dttdruleRAndBarOne}[1]{\dttdrule[#1]{%
}{
\nu \, \dttmv{z}  \dttsym{.}   \mathbf{in}_{ \dttsym{1} }\, \dttnt{t}   \mathbin{\Cdot[2]}  \dttsym{(}  \dttnt{t_{{\mathrm{1}}}}  \dttsym{,}  \dttnt{t_{{\mathrm{2}}}}  \dttsym{)}  \redto  \nu \, \dttmv{z}  \dttsym{.}  \dttnt{t}  \mathbin{\Cdot[2]}  \dttnt{t_{{\mathrm{1}}}}}{%
{\dttdrulename{RAndBar1}}{}%
}}

\newcommand{\dttdruleRAndBarTwo}[1]{\dttdrule[#1]{%
}{
\nu \, \dttmv{z}  \dttsym{.}   \mathbf{in}_{ \dttsym{2} }\, \dttnt{t}   \mathbin{\Cdot[2]}  \dttsym{(}  \dttnt{t_{{\mathrm{1}}}}  \dttsym{,}  \dttnt{t_{{\mathrm{2}}}}  \dttsym{)}  \redto  \nu \, \dttmv{z}  \dttsym{.}  \dttnt{t}  \mathbin{\Cdot[2]}  \dttnt{t_{{\mathrm{2}}}}}{%
{\dttdrulename{RAndBar2}}{}%
}}

\newcommand{\dttdruleRRet}[1]{\dttdrule[#1]{%
\dttpremise{\dttmv{x} \, \not\in \, \mathsf{FV} \, \dttsym{(}  \dttnt{t}  \dttsym{)}}%
}{
\nu \, \dttmv{x}  \dttsym{.}  \dttnt{t}  \mathbin{\Cdot[2]}  \dttmv{x}  \redto  \dttnt{t}}{%
{\dttdrulename{RRet}}{}%
}}

\newcommand{\dttdruleRBetaR}[1]{\dttdrule[#1]{%
}{
\nu \, \dttmv{z}  \dttsym{.}  \dttnt{c}  \mathbin{\Cdot[2]}  \dttsym{(}  \nu \, \dttmv{x}  \dttsym{.}  \dttnt{t_{{\mathrm{1}}}}  \mathbin{\Cdot[2]}  \dttnt{t_{{\mathrm{2}}}}  \dttsym{)}  \redto  \nu \, \dttmv{z}  \dttsym{.}  \dttsym{[}  \dttnt{c}  \dttsym{/}  \dttmv{x}  \dttsym{]}  \dttnt{t_{{\mathrm{1}}}}  \mathbin{\Cdot[2]}  \dttsym{[}  \dttnt{c}  \dttsym{/}  \dttmv{x}  \dttsym{]}  \dttnt{t_{{\mathrm{2}}}}}{%
{\dttdrulename{RBetaR}}{}%
}}

\newcommand{\dttdruleRBetaL}[1]{\dttdrule[#1]{%
}{
\nu \, \dttmv{z}  \dttsym{.}  \dttsym{(}  \nu \, \dttmv{x}  \dttsym{.}  \dttnt{t_{{\mathrm{1}}}}  \mathbin{\Cdot[2]}  \dttnt{t_{{\mathrm{2}}}}  \dttsym{)}  \mathbin{\Cdot[2]}  \dttnt{t}  \redto  \nu \, \dttmv{z}  \dttsym{.}  \dttsym{[}  \dttnt{t}  \dttsym{/}  \dttmv{x}  \dttsym{]}  \dttnt{t_{{\mathrm{1}}}}  \mathbin{\Cdot[2]}  \dttsym{[}  \dttnt{t}  \dttsym{/}  \dttmv{x}  \dttsym{]}  \dttnt{t_{{\mathrm{2}}}}}{%
{\dttdrulename{RBetaL}}{}%
}}

\newcommand{\dttdruleClassAx}[1]{\dttdrule[#1]{%
}{
\Gamma  \dttsym{,}  \dttmv{x}  \dttsym{:}  \dttnt{p} \, \dttnt{A}  \dttsym{,}  \Gamma'  \vdash_c  \dttmv{x}  \dttsym{:}  \dttnt{p} \, \dttnt{A}}{%
{\dttdrulename{ClassAx}}{}%
}}

\newcommand{\dttdruleClassUnit}[1]{\dttdrule[#1]{%
}{
\Gamma  \vdash_c  \dttkw{triv}  \dttsym{:}  \dttnt{p} \,  \langle  \dttnt{p} \rangle }{%
{\dttdrulename{ClassUnit}}{}%
}}

\newcommand{\dttdruleClassAnd}[1]{\dttdrule[#1]{%
\dttpremise{ \Gamma  \vdash_c  \dttnt{t_{{\mathrm{1}}}}  \dttsym{:}  \dttnt{p} \, \dttnt{A}  \qquad  \Gamma  \vdash_c  \dttnt{t_{{\mathrm{2}}}}  \dttsym{:}  \dttnt{p} \, \dttnt{B} }%
}{
\Gamma  \vdash_c  \dttsym{(}  \dttnt{t_{{\mathrm{1}}}}  \dttsym{,}  \dttnt{t_{{\mathrm{2}}}}  \dttsym{)}  \dttsym{:}  \dttnt{p} \, \dttsym{(}   \dttnt{A}  \ndwedge{ \dttnt{p} }  \dttnt{B}   \dttsym{)}}{%
{\dttdrulename{ClassAnd}}{}%
}}

\newcommand{\dttdruleClassAndBar}[1]{\dttdrule[#1]{%
\dttpremise{\Gamma  \vdash_c  \dttnt{t}  \dttsym{:}  \dttnt{p} \,  \dttnt{A} _{ \dttnt{d} } }%
}{
\Gamma  \vdash_c   \mathbf{in}_{ \dttnt{d} }\, \dttnt{t}   \dttsym{:}  \dttnt{p} \, \dttsym{(}   \dttnt{A_{{\mathrm{1}}}}  \ndwedge{  \bar{  \dttnt{p}  }  }  \dttnt{A_{{\mathrm{2}}}}   \dttsym{)}}{%
{\dttdrulename{ClassAndBar}}{}%
}}

\newcommand{\dttdruleClassImp}[1]{\dttdrule[#1]{%
\dttpremise{\Gamma  \dttsym{,}  \dttmv{x}  \dttsym{:}  \dttnt{p} \, \dttnt{A}  \vdash_c  \dttnt{t}  \dttsym{:}  \dttnt{p} \, \dttnt{B}}%
}{
\Gamma  \vdash_c  \lambda  \dttmv{x}  \dttsym{.}  \dttnt{t}  \dttsym{:}  \dttnt{p} \, \dttsym{(}   \dttnt{A}  \ndto{ \dttnt{p} }  \dttnt{B}   \dttsym{)}}{%
{\dttdrulename{ClassImp}}{}%
}}

\newcommand{\dttdruleClassImpBar}[1]{\dttdrule[#1]{%
\dttpremise{ \Gamma  \vdash_c  \dttnt{t_{{\mathrm{1}}}}  \dttsym{:}   \bar{  \dttnt{p}  }  \, \dttnt{A}  \qquad  \Gamma  \vdash_c  \dttnt{t_{{\mathrm{2}}}}  \dttsym{:}  \dttnt{p} \, \dttnt{B} }%
}{
\Gamma  \vdash_c  \langle  \dttnt{t_{{\mathrm{1}}}}  \dttsym{,}  \dttnt{t_{{\mathrm{2}}}}  \rangle  \dttsym{:}  \dttnt{p} \, \dttsym{(}   \dttnt{A}  \ndto{  \bar{  \dttnt{p}  }  }  \dttnt{B}   \dttsym{)}}{%
{\dttdrulename{ClassImpBar}}{}%
}}

\newcommand{\dttdruleClassCut}[1]{\dttdrule[#1]{%
\dttpremise{\Gamma  \dttsym{,}  \dttmv{x}  \dttsym{:}   \bar{  \dttnt{p}  }  \, \dttnt{A}  \vdash_c  \dttnt{t_{{\mathrm{1}}}}  \dttsym{:}  \dttsym{+} \, \dttnt{B}}%
\dttpremise{\Gamma  \dttsym{,}  \dttmv{x}  \dttsym{:}   \bar{  \dttnt{p}  }  \, \dttnt{A}  \vdash_c  \dttnt{t_{{\mathrm{2}}}}  \dttsym{:}  \dttsym{-} \, \dttnt{B}}%
}{
\Gamma  \vdash_c  \nu \, \dttmv{x}  \dttsym{.}  \dttnt{t_{{\mathrm{1}}}}  \mathbin{\Cdot[2]}  \dttnt{t_{{\mathrm{2}}}}  \dttsym{:}  \dttnt{p} \, \dttnt{A}}{%
{\dttdrulename{ClassCut}}{}%
}}

\newcommand{\Ldrule}[4][]{{\displaystyle\frac{\begin{array}{l}#2\end{array}}{#3}\quad\Ldrulename{#4}}}

\newcommand{\Lpremise}[1]{ #1 \\}
\newenvironment{Ldefnblock}[3][]{ \framebox{\mbox{#2}} \quad #3 \\[0pt]}{}

\newcommand{\Lnt}[1]{\mathit{#1}}
\newcommand{\Lmv}[1]{\mathit{#1}}

\newcommand{\Lsym}[1]{#1}

\newcommand{\Ldrulename}[1]{\textsc{#1}}

\newcommand{\Ldrulerefl}[1]{\Ldrule[#1]{%
\Lpremise{ \Gamma  \vdash_{   G  \cup  \Lsym{\{}  \Lsym{(}  \Lmv{n}  \Lsym{,}  \Lmv{n}  \Lsym{)}  \Lsym{\}}   }  \Delta }%
}{
 \Gamma  \vdash_{ G }  \Delta }{%
{\Ldrulename{refl}}{}%
}}

\newcommand{\Ldruletrans}[1]{\Ldrule[#1]{%
\Lpremise{ \Lmv{n_{{\mathrm{1}}}}   G   \Lmv{n_{{\mathrm{2}}}} }%
\Lpremise{ \Lmv{n_{{\mathrm{2}}}}   G   \Lmv{n_{{\mathrm{3}}}} }%
\Lpremise{ \Gamma  \vdash_{   G  \cup  \Lsym{\{}  \Lsym{(}  \Lmv{n_{{\mathrm{1}}}}  \Lsym{,}  \Lmv{n_{{\mathrm{3}}}}  \Lsym{)}  \Lsym{\}}   }  \Delta }%
}{
 \Gamma  \vdash_{ G }  \Delta }{%
{\Ldrulename{trans}}{}%
}}

\newcommand{\Ldrulehyp}[1]{\Ldrule[#1]{%
}{
 \Gamma  \Lsym{,}  \Lmv{n}  \Lsym{:}  \Lnt{T}  \vdash_{ G }  \Lmv{n}  \Lsym{:}  \Lnt{T}  \Lsym{,}  \Delta }{%
{\Ldrulename{hyp}}{}%
}}

\newcommand{\LdrulemonL}[1]{\Ldrule[#1]{%
\Lpremise{ \Lmv{n}   G   \Lmv{n'} }%
\Lpremise{ \Gamma  \Lsym{,}  \Lmv{n}  \Lsym{:}  \Lnt{T}  \Lsym{,}  \Lmv{n'}  \Lsym{:}  \Lnt{T}  \vdash_{ G }  \Delta }%
}{
 \Gamma  \Lsym{,}  \Lmv{n}  \Lsym{:}  \Lnt{T}  \vdash_{ G }  \Delta }{%
{\Ldrulename{monL}}{}%
}}

\newcommand{\LdrulemonR}[1]{\Ldrule[#1]{%
\Lpremise{ \Lmv{n'}   G   \Lmv{n} }%
\Lpremise{ \Gamma  \vdash_{ G }  \Lmv{n'}  \Lsym{:}  \Lnt{T}  \Lsym{,}  \Lmv{n}  \Lsym{:}  \Lnt{T}  \Lsym{,}  \Delta }%
}{
 \Gamma  \vdash_{ G }  \Lmv{n}  \Lsym{:}  \Lnt{T}  \Lsym{,}  \Delta }{%
{\Ldrulename{monR}}{}%
}}

\newcommand{\LdruletrueL}[1]{\Ldrule[#1]{%
\Lpremise{ \Gamma  \vdash_{ G }  \Delta }%
}{
 \Gamma  \Lsym{,}  \Lmv{n}  \Lsym{:}   \top   \vdash_{ G }  \Delta }{%
{\Ldrulename{trueL}}{}%
}}

\newcommand{\LdruletrueR}[1]{\Ldrule[#1]{%
}{
 \Gamma  \vdash_{ G }  \Lmv{n}  \Lsym{:}   \top   \Lsym{,}  \Delta }{%
{\Ldrulename{trueR}}{}%
}}

\newcommand{\LdrulefalseL}[1]{\Ldrule[#1]{%
}{
 \Gamma  \Lsym{,}  \Lmv{n}  \Lsym{:}   \perp   \vdash_{ G }  \Delta }{%
{\Ldrulename{falseL}}{}%
}}

\newcommand{\LdrulefalseR}[1]{\Ldrule[#1]{%
\Lpremise{ \Gamma  \vdash_{ G }  \Delta }%
}{
 \Gamma  \vdash_{ G }  \Lmv{n}  \Lsym{:}   \perp   \Lsym{,}  \Delta }{%
{\Ldrulename{falseR}}{}%
}}

\newcommand{\LdruleandL}[1]{\Ldrule[#1]{%
\Lpremise{ \Gamma  \Lsym{,}  \Lmv{n}  \Lsym{:}  \Lnt{T_{{\mathrm{1}}}}  \Lsym{,}  \Lmv{n}  \Lsym{:}  \Lnt{T_{{\mathrm{2}}}}  \vdash_{ G }  \Delta }%
}{
 \Gamma  \Lsym{,}  \Lmv{n}  \Lsym{:}   \Lnt{T_{{\mathrm{1}}}}  \land  \Lnt{T_{{\mathrm{2}}}}   \vdash_{ G }  \Delta }{%
{\Ldrulename{andL}}{}%
}}

\newcommand{\LdruleandR}[1]{\Ldrule[#1]{%
\Lpremise{ \Gamma  \vdash_{ G }  \Lmv{n}  \Lsym{:}  \Lnt{T_{{\mathrm{1}}}}  \Lsym{,}  \Delta }%
\Lpremise{ \Gamma  \vdash_{ G }  \Lmv{n}  \Lsym{:}  \Lnt{T_{{\mathrm{2}}}}  \Lsym{,}  \Delta }%
}{
 \Gamma  \vdash_{ G }  \Lmv{n}  \Lsym{:}   \Lnt{T_{{\mathrm{1}}}}  \land  \Lnt{T_{{\mathrm{2}}}}   \Lsym{,}  \Delta }{%
{\Ldrulename{andR}}{}%
}}

\newcommand{\LdruledisjL}[1]{\Ldrule[#1]{%
\Lpremise{ \Gamma  \Lsym{,}  \Lmv{n}  \Lsym{:}  \Lnt{T_{{\mathrm{1}}}}  \vdash_{ G }  \Delta }%
\Lpremise{ \Gamma  \Lsym{,}  \Lmv{n}  \Lsym{:}  \Lnt{T_{{\mathrm{2}}}}  \vdash_{ G }  \Delta }%
}{
 \Gamma  \Lsym{,}  \Lmv{n}  \Lsym{:}   \Lnt{T_{{\mathrm{1}}}}  \lor  \Lnt{T_{{\mathrm{2}}}}   \vdash_{ G }  \Delta }{%
{\Ldrulename{disjL}}{}%
}}

\newcommand{\LdruledisjR}[1]{\Ldrule[#1]{%
\Lpremise{ \Gamma  \vdash_{ G }  \Lmv{n}  \Lsym{:}  \Lnt{T_{{\mathrm{1}}}}  \Lsym{,}  \Lmv{n}  \Lsym{:}  \Lnt{T_{{\mathrm{2}}}}  \Lsym{,}  \Delta }%
}{
 \Gamma  \vdash_{ G }  \Lmv{n}  \Lsym{:}   \Lnt{T_{{\mathrm{1}}}}  \lor  \Lnt{T_{{\mathrm{2}}}}   \Lsym{,}  \Delta }{%
{\Ldrulename{disjR}}{}%
}}

\newcommand{\LdruleimpL}[1]{\Ldrule[#1]{%
\Lpremise{ \Lmv{n}   G   \Lmv{n'} }%
\Lpremise{ \Gamma  \vdash_{ G }  \Lmv{n'}  \Lsym{:}  \Lnt{T_{{\mathrm{1}}}}  \Lsym{,}  \Delta }%
\Lpremise{ \Gamma  \Lsym{,}  \Lmv{n'}  \Lsym{:}  \Lnt{T_{{\mathrm{2}}}}  \vdash_{ G }  \Delta }%
}{
 \Gamma  \Lsym{,}  \Lmv{n}  \Lsym{:}   \Lnt{T_{{\mathrm{1}}}}  \supset  \Lnt{T_{{\mathrm{2}}}}   \vdash_{ G }  \Delta }{%
{\Ldrulename{impL}}{}%
}}

\newcommand{\LdruleimpR}[1]{\Ldrule[#1]{%
\Lpremise{ \Lmv{n'}  \not\in | G |,| \Gamma |,| \Delta | }%
\Lpremise{ \Gamma  \Lsym{,}  \Lmv{n'}  \Lsym{:}  \Lnt{T_{{\mathrm{1}}}}  \vdash_{   G  \cup  \Lsym{\{}  \Lsym{(}  \Lmv{n}  \Lsym{,}  \Lmv{n'}  \Lsym{)}  \Lsym{\}}   }  \Lmv{n'}  \Lsym{:}  \Lnt{T_{{\mathrm{2}}}}  \Lsym{,}  \Delta }%
}{
 \Gamma  \vdash_{ G }  \Lmv{n}  \Lsym{:}   \Lnt{T_{{\mathrm{1}}}}  \supset  \Lnt{T_{{\mathrm{2}}}}   \Lsym{,}  \Delta }{%
{\Ldrulename{impR}}{}%
}}

\newcommand{\LdrulesubL}[1]{\Ldrule[#1]{%
\Lpremise{ \Lmv{n'}  \not\in | G |,| \Gamma |,| \Delta | }%
\Lpremise{ \Gamma  \Lsym{,}  \Lmv{n'}  \Lsym{:}  \Lnt{T_{{\mathrm{1}}}}  \vdash_{   G  \cup  \Lsym{\{}  \Lsym{(}  \Lmv{n}  \Lsym{,}  \Lmv{n'}  \Lsym{)}  \Lsym{\}}   }  \Lmv{n'}  \Lsym{:}  \Lnt{T_{{\mathrm{2}}}}  \Lsym{,}  \Delta }%
}{
 \Gamma  \Lsym{,}  \Lmv{n'}  \Lsym{:}   \Lnt{T_{{\mathrm{1}}}}  \prec  \Lnt{T_{{\mathrm{2}}}}   \vdash_{ G }  \Delta }{%
{\Ldrulename{subL}}{}%
}}

\newcommand{\LdrulesubR}[1]{\Ldrule[#1]{%
\Lpremise{ \Lmv{n'}   G   \Lmv{n} }%
\Lpremise{ \Gamma  \vdash_{ G }  \Lmv{n'}  \Lsym{:}  \Lnt{T_{{\mathrm{1}}}}  \Lsym{,}  \Delta }%
\Lpremise{ \Gamma  \Lsym{,}  \Lmv{n'}  \Lsym{:}  \Lnt{T_{{\mathrm{2}}}}  \vdash_{ G }  \Delta }%
}{
 \Gamma  \vdash_{ G }  \Lmv{n}  \Lsym{:}   \Lnt{T_{{\mathrm{1}}}}  \prec  \Lnt{T_{{\mathrm{2}}}}   \Lsym{,}  \Delta }{%
{\Ldrulename{subR}}{}%
}}

\newcommand{\ifrName}[1]{\textsc{#1}}

\newcommand{\ndto}[1]{\to_{#1}}
\newcommand{\ndwedge}[1]{\wedge_{#1}}

\newcommand{\SN}[0]{\mathbf{SN}}

\usepackage{url}
\urldef{\mailsa}\path|{harley-eades, aaron-stump, ryan-mcleeary}@uiowa.edu|

\begin{document}
\title{Dualized Simple Type Theory}

\author[H. Eades]{Harley Eades III}
\address{Augusta University\\ 2500 Walton Way\\ Augusta, GA 30904}
\email{heades@augusta.edu}
\author[A. Stump]{Aaron Stump}
\address{University of Iowa\\ 14 Maclean Hall\\ Iowa City, IA 52242-1419}
\email{aaron-stump@uiowa.edu}
\author[R. McCleeary]{Ryan McCleeary}
\address{University of Iowa\\ 14 Maclean Hall\\ Iowa City, IA 52242-1419}
\email{ryan-mcleeary@uiowa.edu}

\maketitle

\begin{abstract}
  We propose a new bi-intuitionistic type theory called Dualized Type
  Theory (DTT).  It is a simple type theory with perfect
  intuitionistic duality, and corresponds to a single-sided polarized
  sequent calculus.  We prove DTT strongly normalizing, and prove type
  preservation. DTT is based on a new propositional bi-intuitionistic
  logic called Dualized Intuitionistic Logic (DIL) that builds on
  Pinto and Uustalu's logic L.  DIL is a simplification of L by
  removing several admissible inference rules while maintaining
  consistency and completeness. Furthermore, DIL is defined using a
  dualized syntax by labeling formulas and logical connectives with
  polarities thus reducing the number of inference rules needed to
  define the logic. We give a direct proof of consistency, but prove
  completeness by reduction to L.
\end{abstract}

\section{Introduction}

The verification of software often requires the mixture of finite and
infinite data types.  The former are used to define tree-based
structures while the latter are used to define infinite stream-based
structures.  An example of a tree-based structure is a list or an AVL
tree. Infinite stream-based structures can be used to verify
properties of a software system over time or to verify liveness
properties of the system; see the introduction to \cite{Jacobs:2012}
for a great discussion of the use of co-induction to study software
systems.  An example of an infinite stream-based structure is an
infinitely branching tree, or an infinite list.

Finite tree-based structures can be modeled by inductive data types
while infinite stream-based structures can be modeled by coinductive
data types.  Thus, tool support for reasoning about the behavior of a
software system must provide both inductive data types as well as
coinductive data types, and allow for their mixture.  However, there
are problems with existing systems that do provide both inductive and
coinductive data types.  For example, Agda restricts how inductive and
coinductive types can be nested (see the discussion in
\cite{abel+13}), while Coq supports general mixed inductive and
coinductive data, but in doing so, sacrifices type preservation.
Therefore, what is the proper logical foundation to study the
relationships between inductive and coinductive data types?  By
studying such a foundation we may determine in what ways inductive and
coinductive data can be mixed without sacrificing expressivity or key
meta-theoretic properties.

One fairly obvious relationship between inductive and coinductive data
types is that they are duals to each other.  We believe that the
proper foundation for studying inductive and coinductive types must be
able to express this symmetry while maintaining constructivity.
It turns out that a constructive logical foundation may lie in an
already known constructive logic known as bi-intuitionistic logic.

Bi-intuitionistic logic (BINT)\footnote{We only consider propositional
  logic in this paper.  Note that first-order BINT is non-conservative
  over first-order intuitionistic logic
  \cite{Rauszer:1980,10.2307/2270260}, but we believe that
  second-order BINT is conservative over second-order intuitionistic
  logic, but we leave this to future work.} is a conservative
extension \cite{Crolard:2004} of intuitionistic logic with prefect
duality.  That is, every logical connective in the logic has a dual.
For example, BINT contains conjunction and disjunction, their units
true and false, but also implication and its dual called
co-implication (also known as subtraction, difference, or exclusion).

Co-implication is fairly unknown in computer science, but an intuition
of its meaning can be seen in its interpretation into Kripke models.
In \cite{Rauszer:1974,Rauszer:1980} Rauszer gives a conservative
extension of the Kripke semantics for intuitionistic logic that models
all of the logical connectives of BINT by introducing a new logical
connective for co-implication. The usual interpretation of implication
in a Kripke model is as follows:
\begin{center}
  \begin{math}
    \interp{A \to B}_w = \forall w'.w \leq w' \to \interp{A}_{w'} \to \interp{B}_{w'}
  \end{math}
\end{center}
Rauszer took the dual of the previous interpretation to obtain the following:
\begin{center}
  \begin{math}
    \interp{A - B}_w = \exists w'.w' \leq w \land \lnot\interp{A}_{w'} \land \interp{B}_{w'}
  \end{math}
\end{center}
The previous interpretation shows that implication considers future
worlds, while \\ co-implication considers past worlds.

We consider BINT logic to be the closest extension of intuitionistic
logic to classical logic while maintaining constructivity.  BINT has
two forms of negation, one defined as usual, $\lnot A
\stackrel{\mathsf{def}}{=} A \to \perp$, and a second defined in terms
of co-implication, ${\sim} A \stackrel{\mathsf{def}}{=} \top -
A$.  The latter we call ``non-$A$''.  Now in BINT it is possible to
prove $A \lor {\sim} A$ for any $A$ \cite{crolard01}.  In
fact, the latter, in a type theoretic setting, corresponds to the type
of a constructive control operator \cite{Crolard:2004}.

BINT is a conservative extension of intuitionistic logic, and hence
maintains constructivity, but contains a rich notion of symmetry
between the logical connectives.  Thus, any extension of a BINT logic
must preserve this symmetry, and hence, if we add inductive data
types, then we must also add co-inductive data types.  However, all of
this is premised on the ability to define a BINT type theory.

The contributions of this paper are a new formulation of Pinto and
Uustalu's BINT labeled sequent calculus L called Dualized
Intuitionistic Logic (DIL) and a corresponding type theory called
Dualized Type Theory (DTT).  DIL is a single-sided polarized
formulation of Pinto and Uustalu's L, thus, DIL is a propositional
bi-intuitionistic logic, and builds on L by removing the following
rules (see Section~\ref{sec:L} for a complete definition of L):
\begin{center}
  \begin{math}
    \begin{array}{lll}
      \Ldrulerefl{} & \Ldruletrans{}\\
      & \\
      \LdrulemonL{} & \LdrulemonR{}   
    \end{array}
  \end{math}
\end{center} 
We show that in the absence of the previous rules DIL still maintains
consistency (Theorem~\ref{thm:consistency}) and completeness
(Theorem~\ref{thm:completeness}).  Furthermore, DIL is defined using a
dualized syntax that reduces the number of inference rules needed to
define the logic.

Since DIL has multiple conclusions, and the active formula is on the
right, DIL must have a means of switching out the active formula with
another conclusion.  This is done in DIL using cuts on hypotheses. We
call these types of cuts ``axiom cuts.''  These axiom cuts show up in
non-trivial proofs like the proof of the axiom $A \lor {\sim} A$ for
any $A$ \cite{crolard01}. Furthermore, when the latter is treated as a
type in DTT, the inhabitant is a continuation without a canonical
form, because the inhabitant contains as a subexpression an axiom cut.
Thus, the presence of these continuations prevents the canonicity
result for a type theory -- like DTT -- from holding.  Thus, if
general cut elimination was a theorem of DIL, then $A \lor {\sim} A$
would not be provable.  So DIL must contain cuts that cannot be
eliminated.  This implies that DIL does not enjoy general cut
elimination, but all cuts other than axiom cuts can be
eliminated. Throughout the sequel we define ``cut elimination'' as the
elimination of all cuts other than axiom cuts, and we call DIL ``cut
free'' with respect to this definition of cut elimination. The latter
point is similar to Wadler's dual calculus \cite{Wadler:2005}.

The general form of a DIL sequent is $G  \dttsym{;}  \Gamma  \vdash  \dttnt{p} \, \dttnt{A}  \mathbin{@}  \dttnt{n}$ where
$\Gamma$ is a context, multiset of hypotheses of the form $\dttnt{p'} \, \dttnt{B}  \mathbin{@}  \dttnt{n'}$, $\dttnt{p}$ is a polarity that can be either $\dttsym{+}$ or
$\dttsym{-}$, and $\dttnt{n}$ is a node of the abstract Kripke graph $G$
which is a list of edges.  Think of $G$ as a list of constraints
on the accessibility relation in the Kripke semantics. The negative
hypotheses in $\Gamma$ are alternate conclusions.  In fact, if we
denote by $\Gamma^{\dttnt{p}}$ the subcontext of $\Gamma$ consisting of all
the hypotheses with polarity $\dttnt{p}$, then we can translate a DIL
sequent, $G  \dttsym{;}  \Gamma  \vdash  \dttnt{p} \, \dttnt{A}  \mathbin{@}  \dttnt{n}$, into the more traditional form, where
if $\dttnt{p} = \dttsym{+}$, then the sequent is equivalent to $G;\Gamma^+
\vdash \dttsym{+} \, \dttnt{A}  \mathbin{@}  \dttnt{n},\Gamma^-$, but if $\dttnt{p} = \dttsym{-}$, then the sequent
is equivalent to $G;\Gamma^+, \dttsym{-} \, \dttnt{A}  \mathbin{@}  \dttnt{n} \vdash \Gamma^-$.

The polarities provide two main properties of DIL and DTT.  The first,
which is more fundamental than the second, is the ability to single
out an active formula providing a single-conclusion perspective of a
multi-conclusion logic.  This is important if we want to obtain a type
theory in the traditional form: a single term on the right.  The
second main property is they provide a means of significantly reducing
the number of inference rules that define the logic.  Above we saw
that in $G  \dttsym{;}  \Gamma  \vdash  \dttsym{+} \, \dttnt{A}  \mathbin{@}  \dttnt{n}$ we think of $\dttnt{A}$ as being on the
right, but in $G  \dttsym{;}  \Gamma  \vdash  \dttsym{-} \, \dttnt{A}  \mathbin{@}  \dttnt{n}$ we think of $\dttnt{A}$ as being on
the left, and thus, if we index the logical operators of DIL with
polarities, for example in $ \dttnt{A}  \ndwedge{ \dttnt{p} }  \dttnt{B} $, we can collapse the left and
right rules into a single rule.  For example, $ \dttnt{A}  \ndwedge{ \dttsym{+} }  \dttnt{B} $ is
conjunction, but $ \dttnt{A}  \ndwedge{ \dttsym{-} }  \dttnt{B} $ is disjunction, and the right-rule for
conjunction mirrors the left-rule for disjunction, but we move from
the right to the left, but in DIL this is just a change in polarity.
The right-rule for conjunction and the left-rule for disjunction can
thus be given by the single rule:
\begin{center}
  $\dttdruleand{}$
\end{center}

A summary of our contributions is as follows:
\begin{itemize}
\item A new formulation of Pinto and Uustalu's BINT labeled sequent
  calculus L called Dualized Intuitionistic Logic (DIL),
\item a corresponding simple type theory called Dualized Type Theory
  (DTT),
\item a computer-checked proof -- in Agda -- of consistency for DIL with
  respect to Rauszer's Kripke semantics for BINT logic,
\item a completeness proof for DIL by reduction to Pinto and Uustalu's
  L, and
\item the basic metatheory for DTT: type preservation and strong normalization
for DTT.  We show the latter using a version of Krivine's classical
realizability by translating DIL into a classical logic.
\end{itemize}

\noindent The rest of this paper is organized as follows.  We first discuss
related work in Section~\ref{sec:related_work}.  Then we introduce
Pinto and Uustalu's L calculus in Section~\ref{sec:L}, and then DIL in
Section~\ref{sec:dualized_intuitionistic_logic_(dil)}. We present the
consistency proof for DIL in Section~\ref{subsec:consistency_of_dil},
and then show DIL is complete (with only axiom cuts) in
Section~\ref{subsec:completeness}.  Following DIL we introduce DTT in
Section~\ref{sec:dualized_type_theory_(dtt)}, and its metatheory in
Section~\ref{sec:metatheory_of_dtt}.  All of the mathematical content
of this paper was typeset with the help of Ott \cite{Sewell:2010}.

\section{Related Work}
\label{sec:related_work}

The main motivation for studying BINT is to use it to study the
mixture of inductive and co-inductive data types, but from a
constructive perspective. However, a natural question to ask is can
classical logic be used?  There has been a lot of work done since
Griffin's seminal paper \cite{Griffin:1990} showing that the type of
Peirce's law corresponds to a control operator, and thus, providing a
means of defining a program from any classical proof; for example see
\cite{Parigot:1992,Rehof:1994,Curien:2000,Wadler:2005}.  Kimura and
Tatsuta extend Wadler's Dual Calculus (DC) with inductive and
coinductive data types in \cite{kimura+09}.  The Dual Calculus was
invented by Wadler \cite{Wadler:2005}, and is a multi-conclusion
classical simple type theory based in sequent calculus instead of
natural deduction. DC only contains the logical operators conjunction,
disjunction, and negation.  Then he defines the other operators in terms
of these.  Thus, co-implication is defined, and not taken as a
primitive operator. Kimura and Tatsuta carry out a very similar
program to what we are proposing here.  They add inductive and
co-inductive types to DC, show that the rich symmetry of classical
logic extends to inductive and co-inductive types, and finally shows
how to embed this extension into the second-order extension of DC.
The starkest difference between their work, and the ultimate goals of
our program is that we wish to be as constructive as possible.  We
choose to do this, because we ultimately wish to extend our work to
dependent types, which we conjecture will be a goal more easily reached
in a constructive setting versus a classical setting.  Extending
control operators to dependent types is currently an open problem; for
example, general $\Sigma$-types cannot be mixed with control operators
\cite{Herbelin:2005}.

As we mentioned above BINT logic is fairly unknown in computer
science.  Crolard introduced a logic and corresponding type theory
called subtractive logic, and showed it can be used to study
constructive coroutines in \cite{crolard01,Crolard:2004}.  He
initially defined subtractive logic in sequent style with the Dragalin
restriction, and then defined the corresponding type theory in natural
deduction style by imposing a restriction on Parigot's
$\lambda\mu$-calculus in the form of complex dependency tracking.
Just as linear logicians have found -- for example in
\cite{Schellinx:1991} -- Pinto and Uustalu were able to show that
imposing the Dragalin restriction in subtractive logic results in a
failure of cut elimination \cite{Pinto:2009}.  They recover cut
elimination by proposing a new BINT logic called L that lifts the
Dragalin restriction by labeling formulas and sequents with nodes and
graphs respectively; this labeling corresponds to placing constraints
on the sequents where the graphs can be seen as abstract Kripke
models. Gor\'e et al. also proposed a new BINT logic that enjoys cut
elimination using nested sequents; however it is currently unclear how
to define a type theory with nested sequents
\cite{DBLP:conf/aiml/GorePT08}.  Bilinear logic in its intuitionistic
form is a linear version of BINT and has been studied by Lambek in
\cite{Lambek:1993,Lambek:1995}.  Biasi and Aschieri propose a term
assignment to polarized bi-intuitionistic logic in
\cite{Biasi:2008:TAP:2365856.2365859}.  One can view the polarities of
their logic as an internalization of the polarities of the logic we
propose in this article. Bellin has studied BINT similar to that of
Biasi and Aschieri from a philosophical perspective in
\cite{Bellin:2004,Bellin:2005,Bellin:2014}, and he defined a linear
version of Crolard's subtractive logic, for which he was able to
construct a categorical model using linear categories in
\cite{Bellin:2012}.

DIL sequents are labeled with an abstract Kripke graph that is
defined as a multiset of edges between abstract nodes -- labels
denoted $n$.  Then all formulas in a sequent are labeled with a node
from the graph, and the inference rules of DIL are restricted using
conditions on the graph and the nodes on formulas that are based on
the interpretation of formulas into the Kripke semantics.  This idea
in BINT logic comes from Pinto and Uustalu's L \cite{Pinto:2009}, but
their work was inspired by Negri's work on contraction and cut-free
modal logics \cite{Negri:2005}.

A system related to both L and DIL is Reed and Pfenning's labeled
intuitionistic logic with a restricted notion of control operators.
Their logic can also be seen as a restriction of classical logic by
labeling the formulas with strings of nodes representing a directed
path in the Kripke semantics.  That is, a formula is of the form
$A[p]$ where $p$ is a string of nodes where if $p = n_1n_2 \cdots
n_{i-1}n_i$ then we can intuitively think of $p$ as a path in the
Kripke semantics, and hence, $p$ represents the path
$R\,n_1\,(R\,n_2\,(\cdots(R\,n_{i-1}\,n_i)\cdots))$, where $R$ is the
accessibility relation.  One very interesting aspect of their natural
deduction formulation -- which has a term assignment -- is that it
contains the terms $\mathsf{throw}$ and $\mathsf{catch}$, which are
used to allow for multiple conclusions.  These give the logic some
control like operators intuitionisticly.  We conjecture that the
propositional fragment of Reed and Pfenning's system should be able to
be embedded into DIL fairly straightforwardly.  In fact,
$\mathsf{throw}$ and $\mathsf{catch}$ correspond to our axiom cuts
mentioned in Section~\ref{subsec:completeness}, which allows DIL to
switch between the multiple conclusions. Both L and DIL have a more
general labeling than Reed and Pfenning's system, because theirs only
speaks about a single path, and future worlds along that path, but L
and DIL allow one to talk about multiple different paths, and consider
both future and past worlds.

Similarly, to L, DIL, DTT, and Reed and Pfenning's logic Murphy et
al. use a labeling system that annotates formulas with worlds, and use
world constraints to restrict the logic \cite{Murphy:2005}.  They even
use the same syntax as DIL and DTT, that is, their formulas are
denoted by $A@w$, which stands for $A$ is true at the world $w$.  It
would be interesting to see if their work provides a means of
extending DIL and DTT to BINT modal logic.

An alternative approach to Pinto and Uustalu's L was given by Galmiche
and M{\'e}ry \cite{Galmiche:2011}.  They give a labeled sequent
calculus for BINT and a counter-model construction similar to L, but
they use a different method for constructing the labeled sequent
calculus called connection-based validity.  Their system uses a
different notion of graph called R-graphs that annotate sequents, but
these graphs are far more complex than the abstract Kripke graphs of
DIL and L.  M{\'e}ry et al. later implement an interactive theorem
prover for their system \cite{Daniel-Mery:2013}.

\vspace{-4.2px}
\section{Pinto and Uustalu's L}
\label{sec:L}
In this section we briefly introduce Pinto and Uustalu's L from
\cite{Pinto:2009}.  The syntax for formulas, graphs, and contexts of L
are defined in Figure~\ref{fig:L-syntax}, while the inference rules
are defined in Figure~\ref{fig:L-ifr}.
\begin{figure}
  \begin{center}
    \begin{math}
      \begin{array}{rrllllllllllllllllllll}
        (\text{formulas})   & \Lnt{A},\Lnt{B},\Lnt{C} & ::= &  \top \,|\, \perp \,|\, \Lnt{A}  \supset  \Lnt{B} \,|\, \Lnt{A}  \prec  \Lnt{B} \,|\, \Lnt{A}  \land  \Lnt{B} \,|\, \Lnt{A}  \lor  \Lnt{B} \\
        (\text{graphs})     & G & ::= &  \cdot \,|\,\Lsym{(}  \Lmv{n}  \Lsym{,}  \Lmv{n'}  \Lsym{)} \,|\,G  \Lsym{,}  G'\\
        (\text{contexts})   & \Gamma,\Delta & ::= &  \cdot \,|\, \Lmv{n}  \Lsym{:}  \Lnt{A} \,|\,\Gamma  \Lsym{,}  \Gamma'\\
        (\text{sequents})   & Q & ::= &  \Gamma  \vdash_{ G }  \Delta 
      \end{array}
    \end{math}
  \end{center}

  \caption{Syntax of L.}
  \label{fig:L-syntax}
\end{figure}
\begin{figure}
  \begin{mathpar}
    \Ldrulerefl{}   \and
    \Ldruletrans{}  \and
    \Ldrulehyp{}    \and
    \LdrulemonL{}   \and
    \LdrulemonR{}   \and
    \LdruletrueL{}  \and
    \LdruletrueR{}  \and
    \LdrulefalseL{} \and
    \LdrulefalseR{} \and
    \LdruleandL{}   \and
    \LdruleandR{}   \and
    \LdruledisjL{}  \and
    \LdruledisjR{}  \and
    \LdruleimpL{}   \and
    \LdruleimpR{}   \and
    \LdrulesubL{}   \and
    \LdrulesubR{}   
  \end{mathpar}
  
  \caption{Inference Rules for L.}
  \label{fig:L-ifr}
\end{figure} 
The formulas include true and false denoted $ \top $ and $ \perp $
respectively, implication and co-implication denoted $ \Lnt{A}  \supset  \Lnt{B} $ and
$ \Lnt{A}  \prec  \Lnt{B} $ respectively, and finally, conjunction and disjunction
denoted $ \Lnt{A}  \land  \Lnt{B} $ and $ \Lnt{A}  \lor  \Lnt{B} $ respectively.  So we can see
that for every logical connective its dual is a logical connective of
the logic.  This is what we meant by BINT containing perfect
intuitionistic duality in the introduction. Sequents have the form
$ \Gamma  \vdash_{ G }  \Delta $, where $\Gamma$ and $\Delta$ are multisets of formulas
$n : A$ labeled by a node $n$, $G$ is the abstract Kripke model
or sometimes referred to as simply the graph of the sequent, and
$\Lmv{n}$ is a node in $G$.

A graph is a multiset of directed edges where each edge is a pair of
nodes.  One should view these edges as constraints on the
accessibility relation in the Kripke semantics; see the interpretation
of graphs in Definition~\ref{def:graph_interp} and the definition
validity for L in Definition~\ref{def:L-counter-model}.  We denote
$\Lsym{(}  \Lmv{n_{{\mathrm{1}}}}  \Lsym{,}  \Lmv{n_{{\mathrm{2}}}}  \Lsym{)} \in G$ by $ \Lmv{n_{{\mathrm{1}}}}   G   \Lmv{n_{{\mathrm{2}}}} $.  Furthermore, we denote
the union of two graphs $G$ and $G'$ as $ G  \cup  G' $. Now
each formula present in a sequent is labeled with a node in the
graph.  This labeling is denoted $\Lmv{n}  \Lsym{:}  \Lnt{A}$ and should be read as the
formula $\Lnt{A}$ is true at the node $\Lmv{n}$.  We denote the operation
of constructing the list of nodes in a graph or context by $|G|$
and $|\Gamma|$ respectively. The reader should note that it is possible
for some nodes in the sequent to not appear in the graph.  For
example, the sequent $ \Lmv{n}  \Lsym{:}  \Lnt{A}  \vdash_{  \cdot  }  \Lmv{n}  \Lsym{:}  \Lnt{A}  \Lsym{,}   \cdot  $ is a derivable sequent.
The complete graph can always be recovered if needed by using the
graph structural rules $\Ldrulename{refl}$, $\Ldrulename{trans}$,
$\Ldrulename{monL}$, and $\Ldrulename{monR}$.

Consistency of L is stated in \cite{Pinto:2009} without a detailed
proof, but is proven complete with respect to Rauszer's Kripke
semantics using a counter model construction.  In
Section~\ref{sec:dualized_intuitionistic_logic_(dil)} we give a
translation of the formulas of L into the formulas of DIL
(Section~\ref{subsec:a_l_to_dil_translation}) and a translation in the
inverse direction (Section~\ref{subsec:a_dil_to_l_translation}), which
are both used to show completeness of DIL in
Section~\ref{subsec:completeness}.

\section{Dualized Intuitionistic Logic (DIL)}
\label{sec:dualized_intuitionistic_logic_(dil)}
The syntax for polarities, formulas, and graphs of DIL are defined in
Figure~\ref{fig:dil-syntax}, where $\dttkw{a}$ ranges over atomic
formulas.  The following definition shows that DIL's formulas are
simply polarized versions of L's formulas.
\begin{defi}
  \label{def:L-form-to-DIL-form}
  The following defines a translation of formulas of L to formulas of DIL:
  \begin{center}
    \begin{tabular}{cccccccccccccc}
      \begin{math}
      \begin{array}{lll}
         \mathsf{D}(  \top  )    & = &  \langle  \dttsym{+} \rangle \\
         \mathsf{D}(  \perp  )  & = &  \langle  \dttsym{-} \rangle \\
      \end{array}
    \end{math}
      & \ \   
      \begin{math}
      \begin{array}{lll}        
         \mathsf{D}(  \Lnt{A}  \land  \Lnt{B}  )  & = &   \mathsf{D}( \dttnt{A} )   \ndwedge{ \dttsym{+} }   \mathsf{D}( \dttnt{B} )  \\
         \mathsf{D}(  \Lnt{A}  \lor  \Lnt{B}  )  & = &   \mathsf{D}( \dttnt{A} )   \ndwedge{ \dttsym{-} }   \mathsf{D}( \dttnt{B} )  \\
      \end{array}
    \end{math}
      & \ \   
      \begin{math}
      \begin{array}{lll}
         \mathsf{D}(  \Lnt{A}  \supset  \Lnt{B}  )  & = &   \mathsf{D}( \dttnt{A} )   \ndto{ \dttsym{+} }   \mathsf{D}( \dttnt{B} )  \\
         \mathsf{D}(  \Lnt{B}  \prec  \Lnt{A}  )  & = &   \mathsf{D}( \dttnt{A} )   \ndto{ \dttsym{-} }   \mathsf{D}( \dttnt{B} )  \\
      \end{array}
    \end{math}
    \end{tabular}
  \end{center}
\end{defi}

We represent graphs as lists of edges denoted $\dttnt{n_{{\mathrm{1}}}} \,  \preccurlyeq_{ \dttnt{p} }  \, \dttnt{n_{{\mathrm{2}}}}$, where
we denote an edge from $\dttnt{n_{{\mathrm{1}}}}$ to $\dttnt{n_{{\mathrm{2}}}}$ by $\dttnt{n_{{\mathrm{1}}}} \,  \preccurlyeq_{ \dttsym{+} }  \, \dttnt{n_{{\mathrm{2}}}}$, and we
denote the edge from $\dttnt{n_{{\mathrm{2}}}}$ to $\dttnt{n_{{\mathrm{1}}}}$ by $\dttnt{n_{{\mathrm{1}}}} \,  \preccurlyeq_{ \dttsym{-} }  \, \dttnt{n_{{\mathrm{2}}}}$.  Lastly,
contexts denoted $\Gamma$ are represented as lists of formulas.
\begin{figure}[t]
  
  \begin{center}
    \begin{math}
      \begin{array}{rrllllllllllllllllllll}
        (\text{polarities}) & \dttnt{p} & ::= & \dttsym{+} \,|\, \dttsym{-}\\
        (\text{formulas})   & \dttnt{A},\dttnt{B},\dttnt{C} & ::= & \dttkw{a}\,|\, \langle  \dttnt{p} \rangle \,|\,  \dttnt{A}  \ndto{ \dttnt{p} }  \dttnt{B} \,|\, \dttnt{A}  \ndwedge{ \dttnt{p} }  \dttnt{B} \\
        (\text{graphs})     & G & ::= &  \cdot  \,|\, \dttnt{n} \,  \preccurlyeq_{ \dttnt{p} }  \, \dttnt{n'} \,|\, G  \dttsym{,}  G'\\
        (\text{contexts})   & \Gamma & ::= &  \cdot  \,|\, \dttnt{p} \, \dttnt{A}  \mathbin{@}  \dttnt{n} \,|\, \Gamma  \dttsym{,}  \Gamma'\\
        (\text{sequents})   & Q & ::= & G  \dttsym{;}  \Gamma  \vdash  \dttnt{p} \, \dttnt{A}  \mathbin{@}  \dttnt{n}
      \end{array}
    \end{math}
  \end{center}

  \caption{Syntax for DIL.}
  \label{fig:dil-syntax}
\end{figure}
Throughout the sequel we denote the opposite of a polarity $\dttnt{p}$ by
$ \bar{  \dttnt{p}  } $.  This is defined by $ \bar{  \dttsym{+}  }  = \dttsym{-}$ and $ \bar{  \dttsym{-}  } 
= +$.  The inference rules for DIL are in Figure~\ref{fig:dil-ifr}.
\begin{figure*}[t]
    \begin{mathpar}
        \dttdruleax{} \and
        \dttdruleunit{} \and
        \dttdruleand{} \and
        \dttdruleandBar{} \and
        \dttdruleimp{} \and
        \dttdruleimpBar{} \and        
        \dttdrulecut{} 
    \end{mathpar}
  
  \caption{Inference Rules for DIL.}
  \label{fig:dil-ifr}
\end{figure*}

The sequent has the form $G  \dttsym{;}  \Gamma  \vdash  \dttnt{p} \, \dttnt{A}  \mathbin{@}  \dttnt{n}$, which when $\dttnt{p}$ is
positive (resp. negative) can be read as the formula $\dttnt{A}$ is true
(resp. false) at node $\dttnt{n}$ in the context $\Gamma$ with respect to
the graph $G$.  Note that the metavariable $\dttnt{d}$ in the premise
of the $\dttdrulename{AndBar}$ rule ranges over the set $\{1,2\}$ and
prevents the need for two rules.  The inference rules depend on a
reachability judgment that provides a means of proving when a node is
reachable from another within some graph $G$.  This judgment is
defined in Figure~\ref{fig:dil-reach}.
\begin{figure*}
    \begin{mathpar}
      \dttdrulerelXXax{} \and
      \dttdrulerelXXrefl{} \and
      \dttdrulerelXXtrans{} \and
      \dttdrulerelXXflip{}
    \end{mathpar}
  
  \caption{Reachability Judgment for DIL.}
  \label{fig:dil-reach}
\end{figure*}
In addition, the $\ifrName{imp}$ rule depends on the operations 
$\dttsym{\mbox{$\mid$}}  G  \dttsym{\mbox{$\mid$}}$ and $\dttsym{\mbox{$\mid$}}  \Gamma  \dttsym{\mbox{$\mid$}}$ that simply compute the list of all 
the nodes in $G$ and $\Gamma$
respectively.  The condition $\dttnt{n'} \, \not\in \, \dttsym{\mbox{$\mid$}}  G  \dttsym{\mbox{$\mid$}}  \dttsym{,}  \dttsym{\mbox{$\mid$}}  \Gamma  \dttsym{\mbox{$\mid$}}$ in the
$\ifrName{imp}$ rule is required for consistency.

The most interesting inference rules of DIL are the rules for
implication and \\ co-implication from Figure~\ref{fig:dil-ifr}.  Let
us consider these two rules in detail. These rules mimic the
definitions of the interpretation of implication and co-implication in
a Kripke model.  The $\ifrName{imp}$ rule states that the formula $\dttnt{p} \, \dttsym{(}   \dttnt{A}  \ndto{ \dttnt{p} }  \dttnt{B}   \dttsym{)}$ holds at node $\dttnt{n}$ if $\dttnt{p} \, \dttnt{A}  \mathbin{@}  \dttnt{n'}$ holds at an
arbitrary node $\dttnt{n'}$ where we add a new edge $\dttnt{n} \,  \preccurlyeq_{ \dttnt{p} }  \, \dttnt{n'}$ to the
graph, then $\dttnt{p} \, \dttnt{B}  \mathbin{@}  \dttnt{n'}$ holds.  Notice that when $\dttnt{p}$ is
positive $\dttnt{n'}$ will be a future node, but when $\dttnt{p}$ is negative
$\dttnt{n'}$ will be a past node.  Thus, universally quantifying over past
and future worlds is modeled here by adding edges to the graph.  Now
the $\ifrName{impBar}$ rule states the formula $\dttnt{p} \, \dttsym{(}   \dttnt{A}  \ndto{  \bar{  \dttnt{p}  }  }  \dttnt{B}   \dttsym{)}$
is derivable if there exists a node $\dttnt{n'}$ that is provably
reachable from $\dttnt{n}$, $ \bar{  \dttnt{p}  }  \, \dttnt{A}$ is derivable at node $\dttnt{n'}$,
and $\dttnt{p} \, \dttnt{B}  \mathbin{@}  \dttnt{n'}$ is derivable at node $\dttnt{n'}$.  When $\dttnt{p}$ is
positive $\dttnt{n'}$ will be a past node, but when $\dttnt{p}$ is negative
$\dttnt{n'}$ will be a future node. This is exactly dual to
implication. Thus, existence of past and future worlds is modeled by
the reachability judgment.

Before moving on to proving consistency and completeness of DIL we
first show that the formula $ \dttnt{A}  \ndwedge{ \dttsym{-} }  \mathop{\sim}  \dttnt{A} $ has a proof in DIL that
contains a cut that cannot be eliminated.  This also serves as an
example of a derivation in DIL. Consider the following where we leave
off the reachability derivations for clarity and $\Gamma' \equiv
 \dttsym{-} \, \dttsym{(}   \dttnt{A}  \ndwedge{ \dttsym{-} }  \mathop{\sim}  \dttnt{A}   \dttsym{)}  \mathbin{@}  \dttnt{n}   \dttsym{,}   \dttsym{-} \, \dttnt{A}  \mathbin{@}  \dttnt{n} $:
\begin{center}  
  \tiny
  \begin{math}
      $$\mprset{flushleft}
      \inferrule* [right={\tiny andBar}] {
        $$\mprset{flushleft}
        \inferrule* [right={\tiny cut}] {
          $$\mprset{flushleft}
          \inferrule* [right={\tiny andBar}] {
            $$\mprset{flushleft}
            \inferrule* [right={\tiny impBar}] {
              $$\mprset{flushleft}
              \inferrule* [right=\tiny ax] {
                \,
              }{G  \dttsym{;}  \Gamma  \dttsym{,}  \Gamma'  \vdash  \dttsym{-} \, \dttnt{A}  \mathbin{@}  \dttnt{n}}
              \\
              $$\mprset{flushleft}
              \inferrule* [right={\tiny unit}] {
                \,
              }{G  \dttsym{;}  \Gamma  \dttsym{,}  \Gamma'  \vdash  \dttsym{+} \,  \langle  \dttsym{+} \rangle   \mathbin{@}  \dttnt{n}}
            }{G  \dttsym{;}  \Gamma  \dttsym{,}  \Gamma'  \vdash  \dttsym{+} \, \mathop{\sim}  \dttnt{A}  \mathbin{@}  \dttnt{n}}
          }{G  \dttsym{;}  \Gamma  \dttsym{,}  \Gamma'  \vdash  \dttsym{+} \, \dttsym{(}   \dttnt{A}  \ndwedge{ \dttsym{-} }  \mathop{\sim}  \dttnt{A}   \dttsym{)}  \mathbin{@}  \dttnt{n}}
          \\
          $$\mprset{flushleft}
          \inferrule* [right={\tiny ax}] {
            \,
          }{G  \dttsym{;}  \Gamma  \dttsym{,}  \Gamma'  \vdash  \dttsym{-} \, \dttsym{(}   \dttnt{A}  \ndwedge{ \dttsym{-} }  \mathop{\sim}  \dttnt{A}   \dttsym{)}  \mathbin{@}  \dttnt{n}}
        }{G  \dttsym{;}  \Gamma  \dttsym{,}  \dttsym{-} \, \dttsym{(}   \dttnt{A}  \ndwedge{ \dttsym{-} }  \mathop{\sim}  \dttnt{A}   \dttsym{)}  \mathbin{@}  \dttnt{n}  \vdash  \dttsym{+} \, \dttnt{A}  \mathbin{@}  \dttnt{n}}
      }{G  \dttsym{;}  \Gamma  \dttsym{,}  \dttsym{-} \, \dttsym{(}   \dttnt{A}  \ndwedge{ \dttsym{-} }  \mathop{\sim}  \dttnt{A}   \dttsym{)}  \mathbin{@}  \dttnt{n}  \vdash  \dttsym{+} \, \dttsym{(}   \dttnt{A}  \ndwedge{ \dttsym{-} }  \mathop{\sim}  \dttnt{A}   \dttsym{)}  \mathbin{@}  \dttnt{n}}
  \end{math}
\end{center}
Now using only an axiom cut we may conclude the following derivation:
\begin{center}
  \footnotesize
  \begin{math}
    $$\mprset{flushleft}
    \inferrule* [right=cut] {
      G  \dttsym{;}  \Gamma  \dttsym{,}  \dttsym{-} \, \dttsym{(}   \dttnt{A}  \ndwedge{ \dttsym{-} }  \mathop{\sim}  \dttnt{A}   \dttsym{)}  \mathbin{@}  \dttnt{n}  \vdash  \dttsym{+} \, \dttsym{(}   \dttnt{A}  \ndwedge{ \dttsym{-} }  \mathop{\sim}  \dttnt{A}   \dttsym{)}  \mathbin{@}  \dttnt{n}
      \\
      $$\mprset{flushleft}
      \inferrule* [right=ax] {
        \,
      }{G  \dttsym{;}  \Gamma  \dttsym{,}  \dttsym{-} \, \dttsym{(}   \dttnt{A}  \ndwedge{ \dttsym{-} }  \mathop{\sim}  \dttnt{A}   \dttsym{)}  \mathbin{@}  \dttnt{n}  \vdash  \dttsym{-} \, \dttsym{(}   \dttnt{A}  \ndwedge{ \dttsym{-} }  \mathop{\sim}  \dttnt{A}   \dttsym{)}  \mathbin{@}  \dttnt{n}}
    }{G  \dttsym{;}  \Gamma  \vdash  \dttsym{+} \, \dttsym{(}   \dttnt{A}  \ndwedge{ \dttsym{-} }  \mathop{\sim}  \dttnt{A}   \dttsym{)}  \mathbin{@}  \dttnt{n}}
  \end{math}
\end{center}
The reader should take notice to the fact that all cuts within the
previous two derivations are axiom cuts -- see the introduction to
Section~\ref{subsec:completeness} for the definition of axiom cuts --
where the inner most cut uses the hypothesis of the outer
cut. Therefore, neither can be eliminated.

\subsection{Consistency of DIL}
\label{subsec:consistency_of_dil}

In this section we prove consistency of DIL with respect to Rauszer's
Kripke semantics for BINT logic.  All of the results in this section
have been formalized in the Agda proof assistant\footnote{Agda source
  code is available at
  \url{https://github.com/heades/DIL-consistency}}.  We begin by first
defining a Kripke frame.

\begin{defi}
  \label{def:kripke_frame}
  A \textbf{Kripke frame} is a pair $(W, R)$ of a set of worlds $W$, and
  a preorder $R$ on $W$.  
\end{defi}
Then we extend the notion of a Kripke frame to include an evaluation for atomic
formulas resulting in a Kripke model.
\begin{defi}
  \label{def:kripke_model}
  A \textbf{Kripke model} is a tuple $(W, R, V)$, such that, $(W, R)$ is
  a Kripke frame, and $V$ is a binary monotone relation on $W$ and the
  set of atomic formulas of DIL.
\end{defi}
Now we can interpret formulas in a Kripke model as follows:
\begin{defi}
  \label{def:interpretation}
  The interpretation of the formulas of DIL in a Kripke model $(W, R,
  V)$ is defined by recursion on the structure of the formula as
  follows:
  \begin{center}
    \begin{tabular}{lll}
      \begin{math}
        \begin{array}{rll}
           \interp{  \langle  \dttsym{+} \rangle  }_{ \dttmv{w} }    & = &  \top \\
           \interp{  \langle  \dttsym{-} \rangle  }_{ \dttmv{w} }    & = &  \perp \\
           \interp{ \dttkw{a} }_{ \dttmv{w} }          & = &  V\, \dttmv{w} \, \dttkw{a} \\  
        \end{array}
      \end{math}
      &
      \begin{math}
        \begin{array}{rll}      
           \interp{  \dttnt{A}  \ndwedge{ \dttsym{+} }  \dttnt{B}  }_{ \dttmv{w} }  & = &  \interp{ \dttnt{A} }_{ \dttmv{w} }   \land  \interp{ \dttnt{B} }_{ \dttmv{w} } \\    
           \interp{  \dttnt{A}  \ndwedge{ \dttsym{-} }  \dttnt{B}  }_{ \dttmv{w} }  & = &  \interp{ \dttnt{A} }_{ \dttmv{w} }   \lor  \interp{ \dttnt{B} }_{ \dttmv{w} } \\    
           \interp{  \dttnt{A}  \ndto{ \dttsym{+} }  \dttnt{B}  }_{ \dttmv{w} }  & = & \forall w' \in W.  R\, \dttmv{w} \, \dttmv{w'}  \to  \interp{ \dttnt{A} }_{ \dttmv{w'} }  \to  \interp{ \dttnt{B} }_{ \dttmv{w'} } \\
           \interp{  \dttnt{A}  \ndto{ \dttsym{-} }  \dttnt{B}  }_{ \dttmv{w} }  & = & \exists w' \in W.  R\, \dttmv{w'} \, \dttmv{w}  \land \lnot  \interp{ \dttnt{A} }_{ \dttmv{w'} }  \land  \interp{ \dttnt{B} }_{ \dttmv{w'} } \\
        \end{array}
      \end{math}
    \end{tabular}
  \end{center}
\end{defi}
The interpretation of formulas really highlights the fact that implication is dual to co-implication.  Monotonicity
holds for this interpretation.

\begin{lem}[Monotonicity]
  \label{lemma:monotonicity}
  Suppose $(W, R, V)$ is a Kripke model, $\dttnt{A}$ is some DIL formula, and $w,w' \in W$.
  Then $ R\, \dttmv{w} \, \dttmv{w'} $ and $ \interp{ \dttnt{A} }_{ \dttmv{w} } $ imply $ \interp{ \dttnt{A} }_{ \dttmv{w'} } $.
\end{lem}

At this point we must set up the mathematical machinery that allows
for the interpretation of sequents in a Kripke model.  This will
require the interpretation of graphs, and hence, nodes.  We interpret
nodes as worlds in the model using a function we call a node
interpreter.
\begin{defi}
  \label{def:node_interpreter}
  Suppose $(W, R, V)$ is a Kripke model and $S$ is a set of nodes of
  an abstract Kripke model $G$.  Then a \textbf{node interpreter} on
  $S$ is a function from $S$ to $W$.
\end{defi}
Now using the node interpreter we can interpret edges as statements
about the reachability relation in the model.  Thus, the
interpretation of a graph is just the conjunction of the
interpretation of its edges.
\begin{defi}
  \label{def:graph_interp}
  Suppose $(W, R, V)$ is a Kripke model, $G$ is an abstract Kripke
  model, and $ N $ is a node interpreter on the set of nodes of $G$.
  Then the interpretation of
  $G$ in the Kripke model is defined by recursion on the structure
  of the graph as follows:
  \[
  \begin{array}{lll}
     \interp{  \cdot  }_{N}             & = &  \top \\
     \interp{ \dttnt{n_{{\mathrm{1}}}} \,  \preccurlyeq_{ \dttsym{+} }  \, \dttnt{n_{{\mathrm{2}}}}  \dttsym{,}  G }_{N}  & = &  R\, \dttsym{(}   N\, \dttnt{n_{{\mathrm{1}}}}   \dttsym{)} \, \dttsym{(}   N\, \dttnt{n_{{\mathrm{2}}}}   \dttsym{)}  \land  \interp{ G }_{N} \\
     \interp{ \dttnt{n_{{\mathrm{1}}}} \,  \preccurlyeq_{ \dttsym{-} }  \, \dttnt{n_{{\mathrm{2}}}}  \dttsym{,}  G }_{N}  & = &  R\, \dttsym{(}   N\, \dttnt{n_{{\mathrm{2}}}}   \dttsym{)} \, \dttsym{(}   N\, \dttnt{n_{{\mathrm{1}}}}   \dttsym{)}  \land  \interp{ G }_{N} \\
  \end{array}
  \]
\end{defi}
The reachability judgment of DIL provides a means to prove that two
particular nodes are reachable in the abstract Kripke graph, but this
proof is really just a syntactic proof of transitivity. The following
lemma makes this precise.
\begin{lem}[Reachability Interpretation]
  \label{lemma:reachability_interpretation}
  Suppose $(W, R, V)$ is a Kripke model, and $ \interp{ G }_{N} $ for some abstract Kripke graph $G$. Then
  \begin{enumerate}[label=\roman*.]
  \item[i.] if $G  \vdash  \dttnt{n_{{\mathrm{1}}}} \,  \preccurlyeq_{ \dttsym{+} }  \, \dttnt{n_{{\mathrm{2}}}}$, then $ R\, \dttsym{(}   N\, \dttnt{n_{{\mathrm{1}}}}   \dttsym{)} \, \dttsym{(}   N\, \dttnt{n_{{\mathrm{2}}}}   \dttsym{)} $, and
  \item[ii.] if $G  \vdash  \dttnt{n_{{\mathrm{1}}}} \,  \preccurlyeq_{ \dttsym{-} }  \, \dttnt{n_{{\mathrm{2}}}}$, then $ R\, \dttsym{(}   N\, \dttnt{n_{{\mathrm{2}}}}   \dttsym{)} \, \dttsym{(}   N\, \dttnt{n_{{\mathrm{1}}}}   \dttsym{)} $.
  \end{enumerate}
\end{lem}

We have everything we need to interpret abstract Kripke models. The
final ingredient to the interpretation of sequents is the
interpretation of contexts.
\begin{defi}
  \label{def:pol_interp}
  If $F$ is some meta-logical formula, we define $p\, F$ as follows:
  \[
  \begin{array}{lll}
    \dttsym{+}\, F = F & \text{ and } &
    \dttsym{-}\, F = \lnot F.\\
  \end{array}
  \]
\end{defi}
\begin{defi}
  \label{def:ctx_interp}
  Suppose $(W, R, V)$ is a Kripke model, $\Gamma$ is a context, 
  and $ N $ is a node interpreter on the set of nodes in $\Gamma$.
  The interpretation of
  $\Gamma$ in the Kripke model is defined by recursion on the structure
  of the context as follows:
  \[
  \begin{array}{lll}
     \interp{  \cdot  }_{N}         & = &  \top \\
     \interp{ \dttnt{p} \, \dttnt{A}  \mathbin{@}  \dttnt{n}  \dttsym{,}  \Gamma }_{N}  & = &  \dttnt{p}  \interp{ \dttnt{A} }_{ \dttsym{(}   N\, \dttnt{n}   \dttsym{)} }  \land  \interp{ \Gamma }_{N} \\
  \end{array}
  \]
\end{defi}
Combining these interpretations results in the following
definition of validity.
\begin{defi}
  \label{def:validity}
  Suppose $(W, R, V)$ is a Kripke model, $\Gamma$ is a context, 
  and $ N $ is a node interpreter on the set of nodes in $\Gamma$.
  The interpretation of sequents is defined as follows:
  \[
  \begin{array}{lll}
     \interp{ G  \dttsym{;}  \Gamma  \vdash  \dttnt{p} \, \dttnt{A}  \mathbin{@}  \dttnt{n} }_{N}  & = & \text{ if }  \interp{ G }_{N}  \text{ and }  \interp{ \Gamma }_{N} , \text{ then }  \dttnt{p}  \interp{ \dttnt{A} }_{ \dttsym{(}   N\, \dttnt{n}   \dttsym{)} } .
  \end{array}
  \]
  Then a sequent $G  \dttsym{;}  \Gamma  \vdash  \dttnt{p} \, \dttnt{A}  \mathbin{@}  \dttnt{n}$ is valid when
  $ \interp{ G  \dttsym{;}  \Gamma  \vdash  \dttnt{p} \, \dttnt{A}  \mathbin{@}  \dttnt{n} }_{N} $ holds for any $N$ and in any Kripke model.
\end{defi}
Notice that in the definition of validity the graph $G$ is
interpreted as a set of constraints imposed on the set of Kripke models, thus
reinforcing the fact that the graphs on sequents really are abstract
Kripke models.  Finally, using the previous definition of validity we
can prove consistency.
\begin{thm}[Consistency]
  \label{thm:consistency}
  Suppose $G  \dttsym{;}  \Gamma  \vdash  \dttnt{p} \, \dttnt{A}  \mathbin{@}  \dttnt{n}$. 
  Then for any Kripke model $(W, R, V)$ and node interpreter $N$ on $|G|$, $ \interp{ G  \dttsym{;}  \Gamma  \vdash  \dttnt{p} \, \dttnt{A}  \mathbin{@}  \dttnt{n} }_{N} $.
\end{thm}

\subsection{Completeness of DIL}
\label{subsec:completeness}

DIL has a tight correspondence with Pinto and Uustalu's L.  In
\cite{Pinto:2009} it is shown that L is complete with respect to
Kripke models using a counter-model construction; see Corollary~1
on p. 13 of ibid.  We will exploit their completeness result to show that
DIL is complete. First, we will give a pair of translations: one from
L to DIL (Definition~\ref{def:L-to-DIL}), and one from DIL to L
(Definition~\ref{def:DIL-form-to-L-form}).  Using these translations
we will show that if a L-sequent\footnote{We will call a sequent in L
  a L-sequent and a sequent in DIL a DIL-sequent.} is derivable, then
its translation to DIL is also derivable
(Lemma~\ref{lemma:containment-l-in-dil}), and vice versa
(Lemma~\ref{lemma:containment_of_dil_in_l}).  Next we will relate
validity of DIL with validity of L, and show that if a DIL-sequent is
valid with respect to the semantics of DIL, then its translation to L
is valid with respect to the semantics of L
(Lemma~\ref{lemma:dil-validity_is_l-validity}).  Finally, we can use
the previous result to show completeness of DIL
(Theorem~\ref{thm:completeness}).

Throughout this section we assume without loss of generality that all
L-sequents have non-empty right-hand sides.  That is, for every
L-sequent, $ \Gamma  \vdash_{ G }  \Delta $, we assume that $\Delta \not=  \cdot $.  We do
not loose generality because it is possible to prove that
$ \Gamma  \vdash_{ G }   \cdot  $ holds if and only if $ \Gamma  \vdash_{ G }  \Lmv{n}  \Lsym{:}   \perp  $ for any node $\Lmv{n}$
(proof omitted).

We proved DIL consistent when DIL contained the general cut rule, but
we prove DIL complete when the cut rule has been replaced with the
following two inference rules, which can be seen as restricted
instances of the cut rule:
\begin{center}
  \footnotesize
  \begin{math}
    \begin{array}{c}
      \dttdruleaxCut{} \\
      \\
      \dttdruleaxCutBar{}
    \end{array}
  \end{math}
\end{center}

\subsubsection{A L to DIL Translation}
\label{subsec:a_l_to_dil_translation}
In this section we show that every derivable L-sequent can be
translated into a derivable DIL-sequent.  Before giving the
translation we will first show several admissibility results for DIL
of inference rules that are similar to the ones we mentioned in
Section~\ref{sec:dualized_intuitionistic_logic_(dil)}.  These two
rules are required for the crucial left-to-right lemma.  This lemma
depends on the following admissible rule:
\begin{lem}[Weakening]
  \label{lemma:weakening}
  If $G  \dttsym{;}  \Gamma  \vdash  \dttnt{p_{{\mathrm{2}}}} \, \dttnt{B}  \mathbin{@}  \dttnt{n}$ is derivable, then $G  \dttsym{;}  \Gamma  \dttsym{,}  \dttnt{p_{{\mathrm{1}}}} \, \dttnt{A}  \mathbin{@}  \dttnt{n_{{\mathrm{1}}}}  \vdash  \dttnt{p_{{\mathrm{2}}}} \, \dttnt{B}  \mathbin{@}  \dttnt{n_{{\mathrm{1}}}}$ is derivable.
\end{lem}
\begin{proof}
  This holds by straightforward induction on the assumed typing
  derivation.    
\end{proof}
Note that we will use admissible rules as if they are inference rules
of the logic throughout the sequel.
\begin{lem}[Left-to-Right]
  \label{lemma:refocus}
  If $G  \dttsym{;}  \Gamma_{{\mathrm{1}}}  \dttsym{,}   \bar{  \dttnt{p}  }  \, \dttnt{A}  \mathbin{@}  \dttnt{n}  \dttsym{,}  \Gamma_{{\mathrm{2}}}  \vdash   \bar{  \dttnt{p'}  }  \, \dttnt{B}  \mathbin{@}  \dttnt{n'}$ is derivable, then \\
  so is $G  \dttsym{;}  \Gamma_{{\mathrm{1}}}  \dttsym{,}  \Gamma_{{\mathrm{2}}}  \dttsym{,}  \dttnt{p'} \, \dttnt{B}  \mathbin{@}  \dttnt{n'}  \vdash  \dttnt{p} \, \dttnt{A}  \mathbin{@}  \dttnt{n}$.  
\end{lem}
\begin{proof}
  Suppose $G  \dttsym{;}  \Gamma_{{\mathrm{1}}}  \dttsym{,}   \bar{  \dttnt{p}  }  \, \dttnt{A}  \mathbin{@}  \dttnt{n}  \dttsym{,}  \Gamma_{{\mathrm{2}}}  \vdash   \bar{  \dttnt{p'}  }  \, \dttnt{B}  \mathbin{@}  \dttnt{n'}$ is derivable and
  $\Gamma_{{\mathrm{3}}} =^{\text{def}} \Gamma_{{\mathrm{1}}}  \dttsym{,}   \bar{  \dttnt{p}  }  \, \dttnt{A}  \mathbin{@}  \dttnt{n}  \dttsym{,}  \Gamma_{{\mathrm{2}}}$.  Then we derive 
  $G  \dttsym{;}  \Gamma_{{\mathrm{1}}}  \dttsym{,}  \Gamma_{{\mathrm{2}}}  \dttsym{,}  \dttnt{p'} \, \dttnt{B}  \mathbin{@}  \dttnt{n'}  \vdash  \dttnt{p} \, \dttnt{A}  \mathbin{@}  \dttnt{n}$ as follows:
  \begin{center}
    \footnotesize
    \begin{math}
      $$\mprset{flushleft}
      \inferrule* [right=\footnotesize axCut] {
         \dttnt{p'} \, \dttnt{B}  \mathbin{@}  \dttnt{n'}  \in  \dttsym{(}  \Gamma_{{\mathrm{3}}}  \dttsym{,}  \dttnt{p'} \, \dttnt{B}  \mathbin{@}  \dttnt{n'}  \dttsym{)} 
        \\
        $$\mprset{flushleft}
        \inferrule* [right=\footnotesize Weakening] {
          G  \dttsym{;}  \Gamma_{{\mathrm{3}}}  \vdash   \bar{  \dttnt{p'}  }  \, \dttnt{B}  \mathbin{@}  \dttnt{n'}
        }{G  \dttsym{;}  \Gamma_{{\mathrm{3}}}  \dttsym{,}  \dttnt{p'} \, \dttnt{B}  \mathbin{@}  \dttnt{n'}  \vdash   \bar{  \dttnt{p'}  }  \, \dttnt{B}  \mathbin{@}  \dttnt{n'}}
      }{G  \dttsym{;}  \Gamma_{{\mathrm{1}}}  \dttsym{,}  \Gamma_{{\mathrm{2}}}  \dttsym{,}  \dttnt{p'} \, \dttnt{B}  \mathbin{@}  \dttnt{n'}  \vdash  \dttnt{p} \, \dttnt{A}  \mathbin{@}  \dttnt{n}}
    \end{math}
  \end{center}
  Thus, we obtain our result.
\end{proof}
We mentioned in the introduction that DIL avoids having to have rules
like the monotonicity rules and other similar rules from L.  To be
able to translate every derivable sequent of L to DIL, we must show
admissibility of those rules in DIL.  The first of these admissible
rules are the rules for reflexivity and transitivity.
\begin{lem}[Reflexivity]
  \label{lemma:reflexivity}
  If $G  \dttsym{,}  \dttnt{m} \,  \preccurlyeq_{ \dttnt{p'} }  \, \dttnt{m}  \dttsym{;}  \Gamma  \vdash  \dttnt{p} \, \dttnt{A}  \mathbin{@}  \dttnt{n}$ is derivable, then so is $G  \dttsym{;}  \Gamma  \vdash  \dttnt{p} \, \dttnt{A}  \mathbin{@}  \dttnt{n}$.
\end{lem}
\begin{proof}
  This holds by a straightforward induction on the form of the assumed derivation.
\end{proof}

\begin{lem}[Transitivity]
  \label{lemma:transitivity}
  If $G  \dttsym{,}  \dttnt{n_{{\mathrm{1}}}} \,  \preccurlyeq_{ \dttnt{p'} }  \, \dttnt{n_{{\mathrm{3}}}}  \dttsym{;}  \Gamma  \vdash  \dttnt{p} \, \dttnt{A}  \mathbin{@}  \dttnt{n}$ is derivable, $ \dttnt{n_{{\mathrm{1}}}}    \preccurlyeq_{ \dttnt{p'} }    \dttnt{n_{{\mathrm{2}}}}  \in  G $ and $ \dttnt{n_{{\mathrm{2}}}}    \preccurlyeq_{ \dttnt{p'} }    \dttnt{n_{{\mathrm{3}}}}  \in  G $, 
  then $G  \dttsym{;}  \Gamma  \vdash  \dttnt{p} \, \dttnt{A}  \mathbin{@}  \dttnt{n}$ is derivable.
\end{lem}
\begin{proof}
  This holds by a straightforward induction on the form of the assumed derivation.
\end{proof}
There is not a trivial correspondence between conjunction in DIL and
conjunction in L, because of the use of polarities in DIL.  Hence,
we must show that L's left rule for conjunction is indeed admissible
in DIL.
\begin{lem}[AndL]
  \label{lemma:andl}
  If $G  \dttsym{;}  \Gamma  \dttsym{,}   \bar{  \dttnt{p}  }  \, \dttnt{A}  \mathbin{@}  \dttnt{n}  \vdash  \dttnt{p} \, \dttnt{B}  \mathbin{@}  \dttnt{n}$ is derivable, then 
  $G  \dttsym{;}  \Gamma  \vdash  \dttnt{p} \, \dttsym{(}   \dttnt{A}  \ndwedge{  \bar{  \dttnt{p}  }  }  \dttnt{B}   \dttsym{)}  \mathbin{@}  \dttnt{n}$ is derivable.
\end{lem}
\begin{proof}
  This proof holds by directly deriving
  $G  \dttsym{;}  \Gamma  \vdash  \dttnt{p} \, \dttsym{(}   \dttnt{A}  \ndwedge{  \bar{  \dttnt{p}  }  }  \dttnt{B}   \dttsym{)}  \mathbin{@}  \dttnt{n}$ in DIL.
  For the complete proof see
  Appendix~\ref{subsec:proof_of_lemma:andl}.
\end{proof}
L has several structural rules.  The following lemmata show that all
of these are admissible in DIL.
\begin{lem}[Exchange]
  \label{lemma:exchange}
  If $G  \dttsym{;}  \Gamma  \vdash  \dttnt{p} \, \dttnt{A}  \mathbin{@}  \dttnt{n}$ is derivable and $ \pi $ is a permutation of $\Gamma$, then
  $G  \dttsym{;}  \pi \, \Gamma  \vdash  \dttnt{p} \, \dttnt{A}  \mathbin{@}  \dttnt{n}$ is derivable.
\end{lem}
\begin{proof}
  This holds by a straightforward induction on the form of the assumed derivation.
\end{proof}
\noindent
Note that we often leave the application of exchange implicit for
readability.

\begin{lem}[Contraction]
  \label{lemma:contract}
  If $G  \dttsym{;}  \Gamma  \dttsym{,}  \dttnt{p} \, \dttnt{A}  \mathbin{@}  \dttnt{n}  \dttsym{,}  \dttnt{p} \, \dttnt{A}  \mathbin{@}  \dttnt{n}  \dttsym{,}  \Gamma'  \vdash  \dttnt{p'} \, \dttnt{B}  \mathbin{@}  \dttnt{n'}$, then
  $G  \dttsym{;}  \Gamma  \dttsym{,}  \dttnt{p} \, \dttnt{A}  \mathbin{@}  \dttnt{n}  \dttsym{,}  \Gamma'  \vdash  \dttnt{p'} \, \dttnt{B}  \mathbin{@}  \dttnt{n'}$.
\end{lem}
\begin{proof}
  This holds by a straightforward induction on the form of the assumed derivation.
\end{proof}

Monotonicity is taken as a primitive in L, but we have decided to
leave monotonicity as an admissible rule in DIL.  To show that it is
admissible in DIL we need to be able to move nodes forward in the
abstract Kripke graph.  This is necessary to be able to satisfy the
graph constraints in the rules $\dttdrulename{imp}$ and
$\dttdrulename{impBar}$ when proving general monotonicity
(Lemma~\ref{lemma:genmono}).  The next result is just weakening for
the reachability judgment.
\begin{lem}[Graph Weakening]
  \label{lemma:graph_weakening}
  If $G  \vdash  \dttnt{n_{{\mathrm{1}}}} \,  \preccurlyeq^*_{ \dttnt{p} }  \, \dttnt{n_{{\mathrm{2}}}}$, then $G  \dttsym{,}  \dttnt{n_{{\mathrm{3}}}} \,  \preccurlyeq_{ \dttnt{p'} }  \, \dttnt{n_{{\mathrm{4}}}}  \vdash  \dttnt{n_{{\mathrm{1}}}} \,  \preccurlyeq^*_{ \dttnt{p} }  \, \dttnt{n_{{\mathrm{2}}}}$.
\end{lem}
\begin{proof}
  This holds by a straightforward induction on the form of the assumed derivation.
\end{proof}
The function $ \mathsf{raise} $ is an operation on abstract Kripke graphs
that takes in two nodes $\dttnt{n_{{\mathrm{1}}}}$ and $\dttnt{n_{{\mathrm{2}}}}$, where $\dttnt{n_{{\mathrm{2}}}}$ is
reachable from $\dttnt{n_{{\mathrm{1}}}}$, and then moves all the edges in an abstract
Kripke graph forward to $\dttnt{n_{{\mathrm{2}}}}$.  This essentially performs
monotonicity on the given edges.  It will be used to show that nodes
in the context of a DIL-sequent can be moved forward using
monotonicity resulting in a lemma called raising the lower bound
logically (Lemma~\ref{lemma:raising_the_lower_bound_logically}).
\begin{defi}
  \label{def:raise}
  We define the function $ \mathsf{raise} $ on abstract graphs as follows:
  \begin{center}
    \begin{math}
      \begin{array}{lll}
        \mathsf{raise} \, \dttsym{(}  \dttnt{n_{{\mathrm{1}}}}  \dttsym{,}  \dttnt{n_{{\mathrm{2}}}}  \dttsym{,}   \cdot   \dttsym{)} & = &  \cdot \\
        \mathsf{raise} \, \dttsym{(}  \dttnt{n_{{\mathrm{1}}}}  \dttsym{,}  \dttnt{n_{{\mathrm{2}}}}  \dttsym{,}  \dttsym{(}  \dttnt{n_{{\mathrm{1}}}} \,  \preccurlyeq_{ \dttnt{p} }  \, \dttnt{m}  \dttsym{,}  G  \dttsym{)}  \dttsym{)} & = & \dttnt{n_{{\mathrm{2}}}} \,  \preccurlyeq_{ \dttnt{p} }  \, \dttnt{m}  \dttsym{,}  \mathsf{raise} \, \dttsym{(}  \dttnt{n_{{\mathrm{1}}}}  \dttsym{,}  \dttnt{n_{{\mathrm{2}}}}  \dttsym{,}  G  \dttsym{)}\\
        \mathsf{raise} \, \dttsym{(}  \dttnt{n_{{\mathrm{1}}}}  \dttsym{,}  \dttnt{n_{{\mathrm{2}}}}  \dttsym{,}  \dttsym{(}  \dttnt{m} \,  \preccurlyeq_{  \bar{  \dttnt{p}  }  }  \, \dttnt{n_{{\mathrm{1}}}}  \dttsym{,}  G  \dttsym{)}  \dttsym{)} & = & \dttnt{m} \,  \preccurlyeq_{  \bar{  \dttnt{p}  }  }  \, \dttnt{n_{{\mathrm{2}}}}  \dttsym{,}  \mathsf{raise} \, \dttsym{(}  \dttnt{n_{{\mathrm{1}}}}  \dttsym{,}  \dttnt{n_{{\mathrm{2}}}}  \dttsym{,}  G  \dttsym{)}\\
        \mathsf{raise} \, \dttsym{(}  \dttnt{n_{{\mathrm{1}}}}  \dttsym{,}  \dttnt{n_{{\mathrm{2}}}}  \dttsym{,}  \dttsym{(}  \dttnt{m} \,  \preccurlyeq_{ \dttnt{p} }  \, \dttnt{m'}  \dttsym{,}  G  \dttsym{)}  \dttsym{)} & = & \dttnt{m} \,  \preccurlyeq_{ \dttnt{p} }  \, \dttnt{m'}  \dttsym{,}  \mathsf{raise} \, \dttsym{(}  \dttnt{n_{{\mathrm{1}}}}  \dttsym{,}  \dttnt{n_{{\mathrm{2}}}}  \dttsym{,}  G  \dttsym{)}, 
                                                      \text{where } \dttnt{m} \not\equiv \dttnt{n_{{\mathrm{1}}}} \text{ and } \dttnt{m'} \not\equiv \dttnt{n_{{\mathrm{1}}}}.\\
        \mathsf{raise} \, \dttsym{(}  \dttnt{n_{{\mathrm{1}}}}  \dttsym{,}  \dttnt{n_{{\mathrm{2}}}}  \dttsym{,}  \dttsym{(}  \dttnt{m} \,  \preccurlyeq_{  \bar{  \dttnt{p}  }  }  \, \dttnt{m'}  \dttsym{,}  G  \dttsym{)}  \dttsym{)} & = &\dttnt{m} \,  \preccurlyeq_{  \bar{  \dttnt{p}  }  }  \, \dttnt{m'}  \dttsym{,}  \mathsf{raise} \, \dttsym{(}  \dttnt{n_{{\mathrm{1}}}}  \dttsym{,}  \dttnt{n_{{\mathrm{2}}}}  \dttsym{,}  G  \dttsym{)},
                                                          \text{where } \dttnt{m} \not\equiv \dttnt{n_{{\mathrm{1}}}} \text{ and } \dttnt{m'} \not\equiv \dttnt{n_{{\mathrm{1}}}}.\\
      \end{array}
    \end{math}
  \end{center}
\end{defi}
\begin{lem}[Raising the Lower Bound]
  \label{lemma:raise_lower}
  If $G  \vdash  \dttnt{n_{{\mathrm{1}}}} \,  \preccurlyeq^*_{ \dttnt{p} }  \, \dttnt{n_{{\mathrm{2}}}}$ and $G  \dttsym{,}  G_{{\mathrm{1}}}  \vdash  \dttnt{m} \,  \preccurlyeq^*_{ \dttnt{p'} }  \, \dttnt{m'}$, then $G  \dttsym{,}  \mathsf{raise} \, \dttsym{(}  \dttnt{n_{{\mathrm{1}}}}  \dttsym{,}  \dttnt{n_{{\mathrm{2}}}}  \dttsym{,}  G_{{\mathrm{1}}}  \dttsym{)}  \vdash  \dttnt{m} \,  \preccurlyeq^*_{ \dttnt{p'} }  \, \dttnt{m'}$.
\end{lem}
\begin{proof}
  This proof holds by induction on the form of $G  \dttsym{,}  G_{{\mathrm{1}}}  \vdash  \dttnt{m} \,  \preccurlyeq^*_{ \dttnt{p'} }  \, \dttnt{m'}$.  For the full proof see
  Appendix~\ref{subsec:proof_of_raising_the_lower_bound}.
\end{proof}

\begin{lem}[Graph Node Containment]
  \label{lemma:graph_node_containment}
  If $G  \vdash  \dttnt{n_{{\mathrm{1}}}} \,  \preccurlyeq^*_{ \dttnt{p} }  \, \dttnt{n_{{\mathrm{2}}}}$ and $\dttnt{n_{{\mathrm{1}}}}$ and $\dttnt{n_{{\mathrm{2}}}}$ are unique, then
  $\dttnt{n_{{\mathrm{1}}}},\dttnt{n_{{\mathrm{2}}}} \in \dttsym{\mbox{$\mid$}}  G  \dttsym{\mbox{$\mid$}}$.
\end{lem}
\begin{proof}
  This holds by straightforward induction on the form of $G  \vdash  \dttnt{n_{{\mathrm{1}}}} \,  \preccurlyeq^*_{ \dttnt{p} }  \, \dttnt{n_{{\mathrm{2}}}}$.
\end{proof}
Finally, we arrive to raising the lower bound logically and general
monotonicity.  The latter depending on the former.  These are the last
of the admissibility results before showing that all translations of
derivable L-sequents are derivable in DIL.
\begin{lem}[Raising the Lower Bound Logically]
  \label{lemma:raising_the_lower_bound_logically}
  If $G  \dttsym{,}  G_{{\mathrm{1}}}  \dttsym{,}  G'  \dttsym{;}  \Gamma  \vdash  \dttnt{p} \, \dttnt{A}  \mathbin{@}  \dttnt{n}$ and $G  \dttsym{,}  G'  \vdash  \dttnt{n_{{\mathrm{1}}}} \,  \preccurlyeq^*_{ \dttnt{p} }  \, \dttnt{n_{{\mathrm{2}}}}$, then
  $G  \dttsym{,}  \mathsf{raise} \, \dttsym{(}  \dttnt{n_{{\mathrm{1}}}}  \dttsym{,}  \dttnt{n_{{\mathrm{2}}}}  \dttsym{,}  G_{{\mathrm{1}}}  \dttsym{)}  \dttsym{,}  G'  \dttsym{;}  \Gamma  \vdash  \dttnt{p} \, \dttnt{A}  \mathbin{@}  \dttnt{n}$.
\end{lem}
\begin{proof}
  This proof holds by induction on the form of $G  \dttsym{,}  G_{{\mathrm{1}}}  \dttsym{,}  G'  \dttsym{;}  \Gamma  \vdash  \dttnt{p} \, \dttnt{A}  \mathbin{@}  \dttnt{n}$.  For the full proof see
  Appendix~\ref{subsec:proof_of_raising_the_lower_bound_logically}.
\end{proof}

\begin{lem}[General Monotonicity]
  \label{lemma:genmono}
  If $G  \vdash  \dttnt{n_{{\mathrm{1}}}} \,  \preccurlyeq^*_{ \dttnt{p_{{\mathrm{1}}}} }  \, \dttnt{n'_{{\mathrm{1}}}}$, \ldots, $G  \vdash  \dttnt{n_{\dttmv{i}}} \,  \preccurlyeq^*_{ \dttnt{p_{\dttmv{i}}} }  \, \dttnt{n'_{\dttmv{i}}}$, $G  \vdash  \dttnt{m} \,  \preccurlyeq^*_{ \dttnt{p} }  \, \dttnt{m'}$, and
  $G  \dttsym{;}   \bar{  \dttnt{p_{{\mathrm{1}}}}  }  \, \dttnt{A_{{\mathrm{1}}}}  \mathbin{@}  \dttnt{n_{{\mathrm{1}}}}  \dttsym{,} \, ... \, \dttsym{,}   \bar{  \dttnt{p_{\dttmv{i}}}  }  \, \dttnt{A_{\dttmv{i}}}  \mathbin{@}  \dttnt{n_{\dttmv{i}}}  \vdash  \dttnt{p} \, \dttnt{B}  \mathbin{@}  \dttnt{m}$, then 
  $G  \dttsym{;}    \bar{  \dttnt{p_{{\mathrm{1}}}}  }    \dttnt{A_{{\mathrm{1}}}}  @  \dttnt{n'_{{\mathrm{1}}}}  , \ldots ,   \bar{  \dttnt{p_{\dttmv{i}}}  }    \dttnt{A_{\dttmv{i}}}  @  \dttnt{n'_{\dttmv{i}}}   \vdash  \dttnt{p} \, \dttnt{B}  \mathbin{@}  \dttnt{m'}$.
\end{lem}
\begin{proof}
  This proof holds by induction on the form of $G  \dttsym{;}   \bar{  \dttnt{p_{{\mathrm{1}}}}  }  \, \dttnt{A_{{\mathrm{1}}}}  \mathbin{@}  \dttnt{n_{{\mathrm{1}}}}  \dttsym{,} \, ... \, \dttsym{,}   \bar{  \dttnt{p_{\dttmv{i}}}  }  \, \dttnt{A_{\dttmv{i}}}  \mathbin{@}  \dttnt{n_{\dttmv{i}}}  \vdash  \dttnt{p} \, \dttnt{B}  \mathbin{@}  \dttnt{m}$.
  For the full proof see Appendix~\ref{subsec:proof_of_general_monotonicity}.
\end{proof}

The following are corollaries of general monotonicity.  The latter
two corollaries show that the monotonicity rules of L are admissible
in DIL.
\begin{cor}[Monotonicity]
  \label{coro:mono}
  Suppose $G  \vdash  \dttnt{n_{{\mathrm{1}}}} \,  \preccurlyeq^*_{ \dttnt{p} }  \, \dttnt{n_{{\mathrm{2}}}}$.  Then
  \begin{enumerate}[label=\roman*.]
  \item[i.]  if $G  \dttsym{;}  \Gamma  \dttsym{,}   \bar{  \dttnt{p}  }  \, \dttnt{A}  \mathbin{@}  \dttnt{n_{{\mathrm{1}}}}  \dttsym{,}  \Gamma'  \vdash  \dttnt{p'} \, \dttnt{B}  \mathbin{@}  \dttnt{n'}$, then 
    $G  \dttsym{;}  \Gamma  \dttsym{,}   \bar{  \dttnt{p}  }  \, \dttnt{A}  \mathbin{@}  \dttnt{n_{{\mathrm{2}}}}  \dttsym{,}  \Gamma'  \vdash  \dttnt{p'} \, \dttnt{B}  \mathbin{@}  \dttnt{n'}$, and
  \item[ii.] if $G  \dttsym{;}  \Gamma  \vdash  \dttnt{p} \, \dttnt{A}  \mathbin{@}  \dttnt{n_{{\mathrm{1}}}}$, then $G  \dttsym{;}  \Gamma  \vdash  \dttnt{p} \, \dttnt{A}  \mathbin{@}  \dttnt{n_{{\mathrm{2}}}}$.    
  \end{enumerate}
\end{cor}

\begin{cor}[MonoL]
  \label{corollary:monol}
  If $G  \dttsym{;}  \Gamma  \dttsym{,}  \dttnt{p} \, \dttnt{A}  \mathbin{@}  \dttnt{n_{{\mathrm{1}}}}  \dttsym{,}  \dttnt{p} \, \dttnt{A}  \mathbin{@}  \dttnt{n_{{\mathrm{2}}}}  \dttsym{,}  \Gamma'  \vdash  \dttnt{p'} \, \dttnt{B}  \mathbin{@}  \dttnt{n'}$ is derivable and $ \dttnt{n_{{\mathrm{1}}}}    \preccurlyeq_{ \dttnt{p} }    \dttnt{n_{{\mathrm{2}}}}  \in  G $, then
  $G  \dttsym{;}  \Gamma  \dttsym{,}  \dttnt{p} \, \dttnt{A}  \mathbin{@}  \dttnt{n_{{\mathrm{1}}}}  \dttsym{,}  \Gamma'  \vdash  \dttnt{p'} \, \dttnt{B}  \mathbin{@}  \dttnt{n'}$ is derivable.
\end{cor}
\begin{proof}
  This result easily follows by part one of Corollary~\ref{coro:mono}, and contraction (Lemma~\ref{lemma:contract}).
\end{proof}

\begin{cor}[MonoR]
  \label{corollary:monor}
  If $G  \dttsym{;}  \Gamma  \dttsym{,}   \bar{  \dttnt{p}  }  \, \dttnt{A}  \mathbin{@}  \dttnt{n_{{\mathrm{1}}}}  \dttsym{,}  \Gamma'  \vdash  \dttnt{p} \, \dttnt{A}  \mathbin{@}  \dttnt{n_{{\mathrm{2}}}}$ and $\dttnt{n_{{\mathrm{1}}}} \,  \preccurlyeq_{ \dttnt{p} }  \, \dttnt{n_{{\mathrm{2}}}} \in G$, then 
  $G  \dttsym{;}  \Gamma  \dttsym{,}  \Gamma'  \vdash  \dttnt{p} \, \dttnt{A}  \mathbin{@}  \dttnt{n_{{\mathrm{2}}}}$ is derivable.
\end{cor}
\begin{proof}
  Suppose $G  \dttsym{;}  \Gamma  \dttsym{,}   \bar{  \dttnt{p}  }  \, \dttnt{A}  \mathbin{@}  \dttnt{n_{{\mathrm{1}}}}  \dttsym{,}  \Gamma'  \vdash  \dttnt{p} \, \dttnt{A}  \mathbin{@}  \dttnt{n_{{\mathrm{2}}}}$ and $\dttnt{n_{{\mathrm{1}}}} \,  \preccurlyeq_{ \dttnt{p} }  \, \dttnt{n_{{\mathrm{2}}}} \in G$.
  Then by part one of monotonicity (Corollary~\ref{coro:mono}) we know 
  $G  \dttsym{;}  \Gamma  \dttsym{,}   \bar{  \dttnt{p}  }  \, \dttnt{A}  \mathbin{@}  \dttnt{n_{{\mathrm{2}}}}  \dttsym{,}  \Gamma'  \vdash  \dttnt{p} \, \dttnt{A}  \mathbin{@}  \dttnt{n_{{\mathrm{2}}}}$.  Finally, we know by the 
  axiom cut rule that $G  \dttsym{;}  \Gamma  \dttsym{,}  \Gamma'  \vdash  \dttnt{p} \, \dttnt{A}  \mathbin{@}  \dttnt{n_{{\mathrm{2}}}}$.    
\end{proof}

We now have everything we need to prove that every derivable sequent
of L can be translated to a derivable sequent in DIL.  The following
definition defines the translation from L into DIL.
\begin{defi}
  \label{def:L-to-DIL}
  The following defines a translation of formulas of L to formulas of DIL:
  \begin{center}
    \begin{tabular}{cccccccccccccc}
      \begin{math}
      \begin{array}{lll}
         \mathsf{D}(  \top  )    & = &  \langle  \dttsym{+} \rangle \\
         \mathsf{D}(  \perp  )  & = &  \langle  \dttsym{-} \rangle \\
      \end{array}
    \end{math}
      & \ \   
      \begin{math}
      \begin{array}{lll}        
         \mathsf{D}(  \Lnt{A}  \land  \Lnt{B}  )  & = &   \mathsf{D}( \dttnt{A} )   \ndwedge{ \dttsym{+} }   \mathsf{D}( \dttnt{B} )  \\
         \mathsf{D}(  \Lnt{A}  \lor  \Lnt{B}  )  & = &   \mathsf{D}( \dttnt{A} )   \ndwedge{ \dttsym{-} }   \mathsf{D}( \dttnt{B} )  \\
      \end{array}
    \end{math}
      & \ \   
      \begin{math}
      \begin{array}{lll}
         \mathsf{D}(  \Lnt{A}  \supset  \Lnt{B}  )  & = &   \mathsf{D}( \dttnt{A} )   \ndto{ \dttsym{+} }   \mathsf{D}( \dttnt{B} )  \\
         \mathsf{D}(  \Lnt{B}  \prec  \Lnt{A}  )  & = &   \mathsf{D}( \dttnt{A} )   \ndto{ \dttsym{-} }   \mathsf{D}( \dttnt{B} )  \\
      \end{array}
    \end{math}
    \end{tabular}
  \end{center}

  \ \\
  \noindent 
  Next we extend the previous definition to contexts:\\
  \begin{center}
    \begin{math}
      \begin{array}{lll}
         \mathsf{D}(  \cdot  )^{ \Lnt{p} }  & = &  \cdot \\
         \mathsf{D}( \Lmv{n}  \Lsym{:}  \Lnt{A}  \Lsym{,}  \Gamma )^{ \Lnt{p} }  & = & \dttnt{p} \,  \mathsf{D}( \dttnt{A} )   \mathbin{@}  \dttnt{n}  \dttsym{,}   \mathsf{D}( \Gamma )^{ \dttnt{p} } \\
      \end{array}
    \end{math}
  \end{center}

  \ \\
  \noindent 
  The following defines the translation of graphs:\\
  \begin{center}
    \begin{math}
      \begin{array}{lll}
         \mathsf{D}(  \cdot  )  & = &  \cdot \\
         \mathsf{D}( \Lsym{(}  \Lmv{n_{{\mathrm{1}}}}  \Lsym{,}  \Lmv{n_{{\mathrm{2}}}}  \Lsym{)}  \Lsym{,}  G )  & = & \dttnt{n_{{\mathrm{1}}}} \,  \preccurlyeq_{ \dttsym{+} }  \, \dttnt{n_{{\mathrm{2}}}}  \dttsym{,}   \mathsf{D}( G ) \\
      \end{array}
    \end{math}
  \end{center}
  The translation of a L-sequent is a DIL-sequent that requires a
  particular formula as the active formula.  The following defines
  such a translation:
  \begin{itemize}
  \item[] An activation of a L-sequent $ \Gamma  \vdash_{ G }  \Delta $ is a
    DIL-sequent\\ $ \mathsf{D}( G )   \dttsym{;}   \mathsf{D}( \Gamma )^{ \dttsym{+} }   \dttsym{,}   \mathsf{D}( \Delta_{{\mathrm{1}}}  \dttsym{,}  \Delta_{{\mathrm{2}}} )^{ \dttsym{-} }   \vdash  \dttsym{+} \,  \mathsf{D}( \dttnt{A} )   \mathbin{@}  \dttnt{n}$, where $\Delta =
    \Delta_{{\mathrm{1}}}  \Lsym{,}  \Lmv{n}  \Lsym{:}  \Lnt{A}  \Lsym{,}  \Delta_{{\mathrm{2}}}$.
  \end{itemize}
\end{defi}
\noindent
The previous definition implies the following result:
\begin{lem}[Reachability]
  \label{lemma:reach}
  If $ \Lmv{n_{{\mathrm{1}}}}   G   \Lmv{n_{{\mathrm{2}}}} $, then $ \mathsf{D}( G )   \vdash  \dttnt{n_{{\mathrm{1}}}} \,  \preccurlyeq^*_{ \dttsym{+} }  \, \dttnt{n_{{\mathrm{2}}}}$.
\end{lem}

The following result shows that every derivable L-sequent can be
translated into a derivable DIL-sequent.  We do this by considering an
arbitrary activation of the L-sequent, and then show that this
arbitrary activation is derivable in DIL, but if it so happens that
this is not the correct activation, then we can always get the correct
one by using the left-to-right lemma (Lemma~\ref{lemma:refocus}) to
switch out the active formula.
\begin{lem}[Containment of L in DIL]
  \label{lemma:containment-l-in-dil}
  If $ \mathsf{D}( G )   \dttsym{;}  \Gamma'  \vdash  \dttsym{+} \, \dttnt{A}  \mathbin{@}  \dttnt{n}$ is an activation of the derivable L-sequent 
  $ \Gamma  \vdash_{ G }  \Delta $, then $ \mathsf{D}( G )   \dttsym{;}  \Gamma'  \vdash  \dttsym{+} \, \dttnt{A}  \mathbin{@}  \dttnt{n}$ is derivable.
\end{lem}
\begin{proof}
  This proof holds by induction on the form of the sequent
  $ \Gamma  \vdash_{ G }  \Delta $.  For the full proof see Appendix~\ref{subsec:proof_of_containment_L_in_DIL}.
\end{proof}

\subsubsection{A DIL to L Translation}
\label{subsec:a_dil_to_l_translation}
This section is similar to the previous one, but we give a translation
of DIL-sequents to L-sequents.  We first have the definition of the
translation from DIL to L.
\begin{defi}
  \label{def:DIL-form-to-L-form}
  The following defines a translation of formulas of DIL to formulas of L:\\
  \begin{center}
    \begin{tabular}{cccccccccccccc}
      \begin{math}
        \begin{array}{lll}
           \mathsf{L}(  \langle  \dttsym{+} \rangle  )  & = &  \top \\
           \mathsf{L}(  \langle  \dttsym{-} \rangle  )  & = &  \perp \\
        \end{array}
      \end{math}
      & \ \  
      \begin{math}
        \begin{array}{lll}
           \mathsf{L}(  \dttnt{A}  \ndwedge{ \dttsym{+} }  \dttnt{B}  )  & = &   \mathsf{L}( \Lnt{A} )   \land   \mathsf{L}( \Lnt{B} )  \\
           \mathsf{L}(  \dttnt{A}  \ndwedge{ \dttsym{-} }  \dttnt{B}  )  & = &   \mathsf{L}( \Lnt{A} )   \lor   \mathsf{L}( \Lnt{B} )  \\          
        \end{array}
      \end{math}
      & \ \  
      \begin{math}
        \begin{array}{lll}
           \mathsf{L}(  \dttnt{A}  \ndto{ \dttsym{+} }  \dttnt{B}  )  & = &   \mathsf{L}( \Lnt{A} )   \supset   \mathsf{L}( \Lnt{B} )  \\
           \mathsf{L}(  \dttnt{B}  \ndto{ \dttsym{-} }  \dttnt{A}  )  & = &   \mathsf{L}( \Lnt{A} )   \prec   \mathsf{L}( \Lnt{B} )  \\
        \end{array}
      \end{math}
    \end{tabular}
  \end{center}

  \ \\
  \noindent 
  Next we extend the previous definition to positive and negative
  contexts:\\
  \begin{center}
    \begin{math}
      \begin{array}{lll}
         \mathsf{L}( \dttsym{+} \, \dttnt{A}  \mathbin{@}  \dttnt{n}  \dttsym{,}  \Gamma )^{ \dttsym{+} }  & = & \Lmv{n}  \Lsym{:}   \mathsf{L}( \Lnt{A} )   \Lsym{,}   \mathsf{L}( \Gamma )^{ \Lsym{+} } \\
         \mathsf{L}( \dttsym{-} \, \dttnt{A}  \mathbin{@}  \dttnt{n}  \dttsym{,}  \Gamma )^{ \dttsym{+} }  & = &  \mathsf{L}( \Gamma )^{ \Lsym{+} } \\\\
         \mathsf{L}( \dttsym{-} \, \dttnt{A}  \mathbin{@}  \dttnt{n}  \dttsym{,}  \Gamma )^{ \dttsym{-} }  & = & \Lmv{n}  \Lsym{:}   \mathsf{L}( \Lnt{A} )   \Lsym{,}   \mathsf{L}( \Gamma )^{ \Lsym{-} } \\
         \mathsf{L}( \dttsym{+} \, \dttnt{A}  \mathbin{@}  \dttnt{n}  \dttsym{,}  \Gamma )^{ \dttsym{-} }  & = &  \mathsf{L}( \Gamma )^{ \Lsym{-} } \\
      \end{array}
    \end{math}
  \end{center}
  
  \ \\
  \noindent 
  The following defines the translation of graphs:\\
  \begin{center}
    \begin{math}
      \begin{array}{lll}
         \mathsf{L}( \dttnt{n_{{\mathrm{1}}}} \,  \preccurlyeq_{ \dttsym{+} }  \, \dttnt{n_{{\mathrm{2}}}}  \dttsym{,}  G )  & = & \Lsym{(}  \Lmv{n_{{\mathrm{1}}}}  \Lsym{,}  \Lmv{n_{{\mathrm{2}}}}  \Lsym{)}  \Lsym{,}   \mathsf{L}( G ) \\
         \mathsf{L}( \dttnt{n_{{\mathrm{2}}}} \,  \preccurlyeq_{ \dttsym{-} }  \, \dttnt{n_{{\mathrm{1}}}}  \dttsym{,}  G )  & = & \Lsym{(}  \Lmv{n_{{\mathrm{1}}}}  \Lsym{,}  \Lmv{n_{{\mathrm{2}}}}  \Lsym{)}  \Lsym{,}   \mathsf{L}( G ) \\
      \end{array}
    \end{math}
  \end{center}

  \ \\
  \noindent 
  Finally, the following defines the translation of DIL sequents:\\
  \begin{center}
    \begin{math}
      \begin{array}{lll}
        \mathsf{L}(G  \dttsym{;}  \Gamma  \vdash  \dttsym{+} \, \dttnt{A}  \mathbin{@}  \dttnt{n}) & = &   \mathsf{L}( \Gamma )^{ \Lsym{+} }   \vdash_{  \mathsf{L}( G )  }  \Lmv{n}  \Lsym{:}  \Lnt{A}  \Lsym{,}   \mathsf{L}( \Gamma )^{ \Lsym{-} }  \\
        \mathsf{L}(G  \dttsym{;}  \Gamma  \vdash  \dttsym{-} \, \dttnt{A}  \mathbin{@}  \dttnt{n}) & = &   \mathsf{L}( \Gamma )^{ \Lsym{+} }   \Lsym{,}  \Lmv{n}  \Lsym{:}  \Lnt{A}  \vdash_{  \mathsf{L}( G )  }   \mathsf{L}( \Gamma )^{ \Lsym{-} }  \\        
      \end{array}
    \end{math}
  \end{center}  
\end{defi}
\noindent
Next we have a few admissible rules that are needed to complete the
proof of containment of DIL in L.
\begin{lem}[Left and Right Weakening in L]
  \label{lemma:right_weakening_in_l}
  \begin{itemize}
    \item[]
    \item[] $\dttdrulename{weakL}$: If $ \Gamma  \Lsym{,}  \Lmv{n}  \Lsym{:}  \Lnt{A}  \vdash_{ G }  \Delta $, then $ \Gamma  \Lsym{,}  \Lmv{n}  \Lsym{:}  \Lnt{A}  \Lsym{,}  \Lmv{n}  \Lsym{:}  \Lnt{B}  \vdash_{ G }  \Delta $.
    \item[] $\dttdrulename{weakR}$: If $ \Gamma  \vdash_{ G }  \Lmv{n}  \Lsym{:}  \Lnt{A}  \Lsym{,}  \Delta $, then $ \Gamma  \vdash_{ G }  \Lmv{n}  \Lsym{:}  \Lnt{A}  \Lsym{,}  \Lmv{n}  \Lsym{:}  \Lnt{B}  \Lsym{,}  \Delta $.
  \end{itemize}
\end{lem}
\begin{proof}
  Both parts of this result hold by straightforward induction on the assumed derivation.
\end{proof}

\begin{lem}[Left and Right Contraction in L]
  \label{lemma:contraction_in_l}
  \begin{itemize}
    \item[]
    \item[] $\dttdrulename{contrL}$: If $ \Gamma_{{\mathrm{1}}}  \Lsym{,}  \Lmv{n}  \Lsym{:}  \Lnt{A}  \Lsym{,}  \Gamma_{{\mathrm{2}}}  \Lsym{,}  \Lmv{n}  \Lsym{:}  \Lnt{A}  \Lsym{,}  \Gamma_{{\mathrm{3}}}  \vdash_{ G }  \Delta $, then $ \Gamma_{{\mathrm{1}}}  \Lsym{,}  \Lmv{n}  \Lsym{:}  \Lnt{A}  \Lsym{,}  \Gamma_{{\mathrm{2}}}  \Lsym{,}  \Gamma_{{\mathrm{3}}}  \vdash_{ G }  \Delta $.
    \item[] $\dttdrulename{contrR}$: If $ \Gamma  \vdash_{ G }  \Delta_{{\mathrm{1}}}  \Lsym{,}  \Lmv{n}  \Lsym{:}  \Lnt{A}  \Lsym{,}  \Delta_{{\mathrm{2}}}  \Lsym{,}  \Lmv{n}  \Lsym{:}  \Lnt{A}  \Lsym{,}  \Delta_{{\mathrm{3}}} $, then $ \Gamma  \vdash_{ G }  \Delta_{{\mathrm{1}}}  \Lsym{,}  \Delta_{{\mathrm{2}}}  \Lsym{,}  \Lmv{n}  \Lsym{:}  \Lnt{A}  \Lsym{,}  \Delta_{{\mathrm{3}}} $.
  \end{itemize}
\end{lem}
\begin{proof}
  Both parts of this result hold by straightforward induction on the assumed derivation.
\end{proof}

\begin{lem}[Reachability Weakening in DIL]
  \label{lemma:reachability_weakening_in_dil}
  For any $\dttnt{n_{{\mathrm{1}}}}  \dttsym{,}  \dttnt{n_{{\mathrm{2}}}} \, \in \, \dttsym{\mbox{$\mid$}}  \dttnt{n} \,  \preccurlyeq_{ \dttnt{p} }  \, \dttnt{n}  \dttsym{,}  G  \dttsym{\mbox{$\mid$}}  \dttsym{,}  \dttsym{\mbox{$\mid$}}  \Gamma  \dttsym{\mbox{$\mid$}}$ if $G  \vdash  \dttnt{n_{{\mathrm{1}}}} \,  \preccurlyeq^*_{ \dttnt{p} }  \, \dttnt{n_{{\mathrm{2}}}}$ and $G  \dttsym{;}  \Gamma  \vdash  \dttnt{p} \, \dttnt{A}  \mathbin{@}  \dttnt{n}$,
  then $G  \dttsym{,}  \dttnt{n_{{\mathrm{1}}}} \,  \preccurlyeq_{ \dttnt{p} }  \, \dttnt{n_{{\mathrm{2}}}}  \dttsym{;}  \Gamma  \vdash  \dttnt{p} \, \dttnt{A}  \mathbin{@}  \dttnt{n}$.
\end{lem}
\begin{proof}
  By straightforward induction on the form of $G  \dttsym{;}  \Gamma  \vdash  \dttnt{p} \, \dttnt{A}  \mathbin{@}  \dttnt{n}$.
\end{proof}

Finally, the next two results show that every derivable DIL-sequent
can be translated into a derivable L-sequent.  One interesting aspect
of these results is that DIL inference rules where the active formula
is positive correspond to the right-inference rules of L, and when the
active formula is negative correspond to left-inference rules of L.
In addition, the use of axiom cuts in DIL correspond to uses of
contraction in L.
\begin{lem}
  \label{lemma:containment_of_dil_in_l_part1}
  Suppose $G  \dttsym{;}  \Gamma  \vdash  \dttnt{p} \, \dttnt{A}  \mathbin{@}  \dttnt{n}$ is a derivable DIL-sequent such that
  for any $\dttnt{n_{{\mathrm{1}}}}  \dttsym{,}  \dttnt{n_{{\mathrm{2}}}} \, \in \, \dttsym{\mbox{$\mid$}}  \dttnt{n} \,  \preccurlyeq_{ \dttnt{p'} }  \, \dttnt{n}  \dttsym{,}  G  \dttsym{\mbox{$\mid$}}  \dttsym{,}  \dttsym{\mbox{$\mid$}}  \Gamma  \dttsym{\mbox{$\mid$}}$ if $G  \vdash  \dttnt{n_{{\mathrm{1}}}} \,  \preccurlyeq^*_{ \dttnt{p'} }  \, \dttnt{n_{{\mathrm{2}}}}$,
  then $\dttnt{n_{{\mathrm{1}}}} \,  \preccurlyeq_{ \dttnt{p'} }  \, \dttnt{n_{{\mathrm{2}}}} \in G$.  Then by using the definition of
  the translation of DIL-sequents we have that $\mathsf{L}(G  \dttsym{;}  \Gamma  \vdash  \dttnt{p} \, \dttnt{A}  \mathbin{@}  \dttnt{n})$ is a derivable L-sequent.
\end{lem}
\begin{proof}
This proof holds by induction on the assumed derivation.  For the full
proof see
Appendix~\ref{subsec:proof_of_lemma:containment_of_dil_in_l_part1}.
\end{proof}

\begin{lem}[Containment of DIL in L]
  \label{lemma:containment_of_dil_in_l}
  Suppose $G  \dttsym{;}  \Gamma  \vdash  \dttnt{p} \, \dttnt{A}  \mathbin{@}  \dttnt{n}$ is a derivable DIL-sequent.  Then
  there exists an abstract Kripke graph $G'$, such that,
  $\mathsf{L(G'  \dttsym{;}  \Gamma  \vdash  \dttnt{p} \, \dttnt{A}  \mathbin{@}  \dttnt{n})}$.
\end{lem}
\begin{proof}
  Suppose $G  \dttsym{;}  \Gamma  \vdash  \dttnt{p} \, \dttnt{A}  \mathbin{@}  \dttnt{n}$ is a derivable DIL-sequent.  Then by
  repeatedly applying Reachability Weakening in DIL
  (Lemma~\ref{lemma:reachability_weakening_in_dil}), which can only be
  applied a finite number of times before reaching a fixed point, we
  will obtain a derivation $G''  \dttsym{;}  \Gamma  \vdash  \dttnt{p} \, \dttnt{A}  \mathbin{@}  \dttnt{n}$ satisfying the
  condition:
  \[
  \text{for any } \dttnt{n_{{\mathrm{1}}}}  \dttsym{,}  \dttnt{n_{{\mathrm{2}}}} \, \in \, \dttsym{\mbox{$\mid$}}  \dttnt{n} \,  \preccurlyeq_{ \dttnt{p'} }  \, \dttnt{n}  \dttsym{,}  G''  \dttsym{\mbox{$\mid$}}  \dttsym{,}  \dttsym{\mbox{$\mid$}}  \Gamma  \dttsym{\mbox{$\mid$}}, \text{ if } G''  \vdash  \dttnt{n_{{\mathrm{1}}}} \,  \preccurlyeq^*_{ \dttnt{p'} }  \, \dttnt{n_{{\mathrm{2}}}}\text{, then } \dttnt{n_{{\mathrm{1}}}} \,  \preccurlyeq_{ \dttnt{p'} }  \, \dttnt{n_{{\mathrm{2}}}} \in G''
  \]
  Choose $G' = G''$.  Then we obtain our result by applying
  Lemma~\ref{lemma:containment_of_dil_in_l_part1} to
  $G'  \dttsym{;}  \Gamma  \vdash  \dttnt{p} \, \dttnt{A}  \mathbin{@}  \dttnt{n}$.
\end{proof}

\subsubsection{Completeness}
\label{subsubsec:completeness}

We now use the previous translations as a means to exploit the
completeness result of L.  The following definition and lemma relate
the two translations that will be needed by our main results of this
section.
\begin{defi}
  \label{def:graph-iso}
  We say two abstract Kripke graphs, $G_{{\mathrm{1}}}$ and $G_{{\mathrm{2}}}$,  are \textbf{isomorphic} iff
  for any $\dttnt{n_{{\mathrm{1}}}} \,  \preccurlyeq_{ \dttnt{p} }  \, \dttnt{n_{{\mathrm{2}}}} \in G_{{\mathrm{1}}}$, $\dttnt{n_{{\mathrm{1}}}} \,  \preccurlyeq_{ \dttnt{p} }  \, \dttnt{n_{{\mathrm{2}}}} \in G_{{\mathrm{2}}}$ or $\dttnt{n_{{\mathrm{2}}}} \,  \preccurlyeq_{  \bar{  \dttnt{p}  }  }  \, \dttnt{n_{{\mathrm{1}}}} \in G_{{\mathrm{2}}}$, and
  for any $\dttnt{n_{{\mathrm{1}}}} \,  \preccurlyeq_{ \dttnt{p} }  \, \dttnt{n_{{\mathrm{2}}}} \in G_{{\mathrm{2}}}$, $\dttnt{n_{{\mathrm{1}}}} \,  \preccurlyeq_{ \dttnt{p} }  \, \dttnt{n_{{\mathrm{2}}}} \in G_{{\mathrm{1}}}$ or $\dttnt{n_{{\mathrm{2}}}} \,  \preccurlyeq_{  \bar{  \dttnt{p}  }  }  \, \dttnt{n_{{\mathrm{1}}}} \in G_{{\mathrm{1}}}$.
\end{defi}

\begin{lem}[L and D Relationships]
  \label{lemma:L-D-relations}\hfill
  \begin{enumerate}[label=\roman*.]
  \item For any abstract Kripke graph $G$, $ \mathsf{D}(  \mathsf{L}( G )  ) $ is isomorphic to $G$.
  \item For any abstract Kripke graph, $ \mathsf{L}(  \mathsf{D}( G )  )  = G$.    
  \item For any DIL-formula $\dttnt{A}$, $ \mathsf{D}(  \mathsf{L}( \dttnt{A} )  )  = \dttnt{A}$.
  \item For any L-formula $\dttnt{A}$, $ \mathsf{L}(  \mathsf{D}( \dttnt{A} )  )  = \dttnt{A}$.
  \end{enumerate}
\end{lem}
\begin{proof}
  Part i and ii follow directly by induction on $G$, and part iii
  and iv follow directly by induction on $\dttnt{A}$.
\end{proof}
\noindent
It is straightforward to extend the previous result to contexts in
both DIL and L.

The interpretation of L-formulas into a Kripke model is identical to
the interpretation of DIL-formulas.  Thus, we use the same syntax to
denote the interpretation of an L-formula.  In fact, we have the
following straightforward result.

\begin{lem}
  \label{lemma:DIL-interp-L-interps}
  Suppose $(W,R,V)$ is a Kripke model and $N$ is a node interpreter.
  Then the following hold:
  \begin{enumerate}[label=\roman*.]
  \item $ \interp{ \dttnt{A} }_{ \dttsym{(}   N\, \dttnt{n}   \dttsym{)} } $ iff $\interp{ \mathsf{L}( \Lnt{A} ) }_{ N\, \dttnt{n} }$.
  \item $ \interp{ G }_{N} $ iff for any $ \Lmv{n_{{\mathrm{1}}}}    \mathsf{L}( G )    \Lmv{n_{{\mathrm{2}}}} $, $ R\, \dttsym{(}   N\, \dttnt{n_{{\mathrm{1}}}}   \dttsym{)} \, \dttsym{(}   N\, \dttnt{n_{{\mathrm{2}}}}   \dttsym{)} $.
  \item $ \interp{ \Gamma }_{N} $ iff for any $\Lmv{n}  \Lsym{:}   \mathsf{L}( \Lnt{A} )  \in  \mathsf{L}( \Gamma )^{ \Lsym{+} } , \interp{ \mathsf{L}( \Lnt{A} ) }_{ N\, \dttnt{n} }$, and
    for any $\Lmv{n}  \Lsym{:}   \mathsf{L}( \Lnt{A} )  \in  \mathsf{L}( \Gamma )^{ \Lsym{-} } ,$ \\ $\lnot \interp{ \mathsf{L}( \Lnt{A} ) }_{ N\, \dttnt{n} }$.
  \end{enumerate}
\end{lem}

\noindent
We recall the definition of validity in L due to Pinto and Uustalu
\cite{Pinto:2009}.
\begin{defi}[Counter Models and L-validity (p. 6, Definition~1, \cite{Pinto:2009})]
  \label{def:L-counter-model}
  A Kripke model $(W,R,V)$ and node interpreter $N$ is a
  \textbf{counter-model} to a L-sequent $ \Gamma  \vdash_{ G }  \Delta $, if
  \begin{enumerate}[label=\roman*.]
  \item[i.] for any $ \Lmv{n_{{\mathrm{1}}}}   G   \Lmv{n_{{\mathrm{2}}}} $, $ R\, \dttsym{(}   N\, \dttnt{n_{{\mathrm{1}}}}   \dttsym{)} \, \dttsym{(}   N\, \dttnt{n_{{\mathrm{2}}}}   \dttsym{)} $;
  \item[ii.] for any $\Lmv{n}  \Lsym{:}  \Lnt{A} \in \Gamma, \interp{\Lnt{A}}_{ N\, \dttnt{n} }$; and
  \item[iii.] for any $\Lmv{n}  \Lsym{:}  \Lnt{B} \in \Delta, \lnot \interp{\Lnt{B}}_{ N\, \dttnt{n} }$.
  \end{enumerate}
  The L-sequent is \textbf{L-valid} if it has no counter-models.
\end{defi}
The following lemma relates validity of DIL to validity of L, and is
the key to proving completeness of DIL.
\begin{lem}[DIL-validity is L-validity]
  \label{lemma:dil-validity_is_l-validity}
  Suppose $ \interp{ G  \dttsym{;}  \Gamma  \vdash  \dttnt{p} \, \dttnt{A}  \mathbin{@}  \dttnt{n} }_{N} $ holds for some Kripke model
  $(W,R,V)$ and node interpreter $N$ on $|G|$.  Then by using the
  translation of DIL-sequents from Definition
  \ref{subsec:a_dil_to_l_translation} we have that $\mathsf{L}(G  \dttsym{;}  \Gamma  \vdash  \dttnt{p} \, \dttnt{A}  \mathbin{@}  \dttnt{n})$ is L-valid.
\end{lem}
\begin{proof}
  This result holds essentially by definition.  For the full proof see
  Appendix~\ref{subsec:proof_of_dil-validity_is_l-validity}.
\end{proof}
\noindent
Finally, we have completeness of DIL by connecting all of the results
of this section.
\begin{thm}[Completeness]
  \label{thm:completeness}
  Suppose $(W,R,V)$ is a Kripke model and $N$ is a node interpreter.  If
  $ \interp{ G  \dttsym{;}  \Gamma  \vdash  \dttnt{p} \, \dttnt{A}  \mathbin{@}  \dttnt{n} }_{N} $ holds, then $G  \dttsym{;}  \Gamma  \vdash  \dttnt{p} \, \dttnt{A}  \mathbin{@}  \dttnt{n}$ is
  derivable.
\end{thm}
\begin{proof}
  Suppose $(W,R,V)$ is a Kripke model and $N$ is a node interpreter. Furthermore,
  suppose $ \interp{ G  \dttsym{;}  \Gamma  \vdash  \dttnt{p} \, \dttnt{A}  \mathbin{@}  \dttnt{n} }_{N} $ holds.  Let $\dttnt{p} = \dttsym{+}$.  By
  Lemma~\ref{lemma:dil-validity_is_l-validity}
  we know $  \mathsf{L}( \Gamma )^{ \Lsym{+} }   \vdash_{  \mathsf{L}( G )  }  \Lmv{n}  \Lsym{:}   \mathsf{L}( \Lnt{A} )   \Lsym{,}   \mathsf{L}( \Gamma )^{ \Lsym{-} }  $ is valid, and by completeness of L
  (Corollary 1, p. 13, \cite{Pinto:2009}) we know $  \mathsf{L}( \Gamma )^{ \Lsym{+} }   \vdash_{  \mathsf{L}( G )  }  \Lmv{n}  \Lsym{:}   \mathsf{L}( \Lnt{A} )   \Lsym{,}   \mathsf{L}( \Gamma )^{ \Lsym{-} }  $
  is derivable.  By containment of L in DIL (Lemma~\ref{lemma:containment-l-in-dil}) we
  know that the activation $ \mathsf{D}(  \mathsf{L}( G )  )   \dttsym{;}   \mathsf{D}(  \mathsf{L}( \Gamma )^{ \dttsym{+} }  )^{ \dttsym{+} }   \dttsym{,}   \mathsf{D}(  \mathsf{L}( \Gamma )^{ \dttsym{-} }  )^{ \dttsym{-} }   \vdash  \dttsym{+} \,  \mathsf{D}(  \mathsf{L}( \dttnt{A} )  )   \mathbin{@}  \dttnt{n}$ is derivable.
  Finally, by Lemma~\ref{lemma:L-D-relations} we can see that the former sequent is
  equivalent to $G  \dttsym{;}  \Gamma  \vdash  \dttnt{p} \, \dttnt{A}  \mathbin{@}  \dttnt{n}$, and thus, we obtain our result.  The case when
  $\dttnt{p} = \dttsym{-}$ is similar, but before using Lemma~\ref{lemma:L-D-relations} one must
  first use the left-to-right admissible rule (Lemma~\ref{lemma:refocus}).
\end{proof}




\section{Dualized Type Theory (DTT)}
\label{sec:dualized_type_theory_(dtt)}

In this section we give DIL a term assignment yielding Dualized Type
Theory (DTT).  First, we introduce DTT, and give several examples
illustrating how to program in DTT.  Then we present the metatheory of
DTT.

The syntax for DTT is defined in Figure~\ref{fig:dtt-syntax}.
\begin{figure}[t]
  
  \begin{center}
    \begin{math}
      \begin{array}{crllllllllllllllllllll}
        (\text{indices})    & \dttnt{d} & ::= & 1\,|\,2\\
        (\text{polarities}) & \dttnt{p} & ::= & \dttsym{+} \,|\, \dttsym{-}\\
        (\text{types})      & \dttnt{A},\dttnt{B},\dttnt{C} & ::= &  \langle  \dttnt{p} \rangle \,|\,  \dttnt{A}  \ndto{ \dttnt{p} }  \dttnt{B} \,|\, \dttnt{A}  \ndwedge{ \dttnt{p} }  \dttnt{B} \\
        (\text{terms})      & \dttnt{t}  & ::= & 
                                 \dttmv{x}\,|\,\dttkw{triv}\,|\,\dttsym{(}  \dttnt{t}  \dttsym{,}  \dttnt{t'}  \dttsym{)}\,|\, \mathbf{in}_{ \dttnt{d} }\, \dttnt{t} \,|\,\lambda  \dttmv{x}  \dttsym{.}  \dttnt{t}\,|\,\langle  \dttnt{t}  \dttsym{,}  \dttnt{t'}  \rangle\,|\,\nu \, \dttmv{x}  \dttsym{.}  \dttnt{t}  \mathbin{\Cdot[2]}  \dttnt{t'}\\
   (\text{canonical terms}) & \dttnt{c} & ::= & \dttmv{x} \mid \dttkw{triv} \mid \dttsym{(}  \dttnt{t}  \dttsym{,}  \dttnt{t'}  \dttsym{)} \mid  \mathbf{in}_{ \dttnt{d} }\, \dttnt{t}  \mid \lambda  \dttmv{x}  \dttsym{.}  \dttnt{t} \mid \langle  \dttnt{t}  \dttsym{,}  \dttnt{t'}  \rangle\\
        (\text{graphs})     & G & ::= &  \cdot  \,|\, \dttnt{n} \,  \preccurlyeq_{ \dttnt{p} }  \, \dttnt{n'} \,|\, G  \dttsym{,}  G'\\
        (\text{contexts})   & \Gamma & ::= &  \cdot  \,|\, \dttmv{x}  \dttsym{:}  \dttnt{p} \, \dttnt{A}  \mathbin{@}  \dttnt{n} \,|\, \Gamma  \dttsym{,}  \Gamma'\\
      \end{array}
    \end{math}
  \end{center}

  \caption{Syntax for DTT.}
  \label{fig:dtt-syntax}
\end{figure}
Polarities, types, and graphs are all the same as they were in DIL.
Contexts differ only by the addition of labeling each hypothesis with
a variable.  Terms, denoted $\dttnt{t}$, consist of introduction forms,
together with cut terms $\nu \, \dttmv{x}  \dttsym{.}  \dttnt{t}  \mathbin{\Cdot[2]}  \dttnt{t'}$\footnote{In classical type
  theories the symbol $\mu$ usually denotes cut, but we have reserved
  that symbol -- indexed by a polarity -- to be used with inductive
  (positive polarity) and coinductive (negative polarity) types in
  future work.}.  We denote variables as $\dttmv{x}$, $\dttmv{y}$, $\dttmv{z}$,
\ldots. The term $\dttkw{triv}$ is the introduction form for units,
$\dttsym{(}  \dttnt{t}  \dttsym{,}  \dttnt{t'}  \dttsym{)}$ is the introduction form for pairs, similarly the terms
$ \mathbf{in}_{ \dttsym{1} }\, \dttnt{t} $ and $ \mathbf{in}_{ \dttsym{2} }\, \dttnt{t} $ introduce disjunctions, $\lambda  \dttmv{x}  \dttsym{.}  \dttnt{t}$
introduces implication, and $\langle  \dttnt{t}  \dttsym{,}  \dttnt{t'}  \rangle$ introduces co-implication.
The type-assignment rules are defined in Figure~\ref{fig:dtt-ifr}, and
result from a simple term assignment to the rules for DIL.
\begin{figure*}[t]
    \begin{mathpar}
      \dttdruleAx{}     \and
      \dttdruleUnit{}   \and
      \dttdruleAnd{}    \and
      \dttdruleAndBar{} \and
      \dttdruleImp{}    \and 
      \dttdruleImpBar{} \and
      \dttdruleCut{}    
    \end{mathpar}
  
  \caption{Type-Assignment Rules for DTT.}
  \label{fig:dtt-ifr}
\end{figure*}
Finally, the reduction rules for DTT are defined in
Figure~\ref{fig:dtt-red}.  
\begin{figure*}[t]
  \begin{mathpar}
    \dttdruleRImp{}        \and
    \dttdruleRImpBar{}     \and
    \dttdruleRAndOne{}     \and
    \dttdruleRAndTwo{}     \and
    \dttdruleRAndBarOne{}  \and
    \dttdruleRAndBarTwo{}  \and
    \dttdruleRRet{}        \and
    \dttdruleRBetaL{}      \and 
    \dttdruleRBetaR{}      
  \end{mathpar}
  
  \caption{Reduction Rules for DTT.}
  \label{fig:dtt-red}
\end{figure*}
The reduction rules should be considered rewrite rules that can be
applied anywhere within a term.  (The congruence rules are omitted.)

Programming in DTT is not functional programming as usual, so we now
give several illustrative examples.  The reader familiar with type
theories based on sequent calculi will find the following very
familiar. The encodings are similar to that of Curien and Herbelin's
$\bar\lambda\mu\tilde\mu$-calculus \cite{Curien:2000}.  The locus of
computation is the cut term, so naturally, function application is
modeled using cuts.  Suppose
\[
\begin{array}{lll}
D_1 & =^{\text{def}} & G  \dttsym{;}  \Gamma  \vdash  \lambda  \dttmv{x}  \dttsym{.}  \dttnt{t}  \dttsym{:}  \dttsym{+} \, \dttsym{(}   \dttnt{A}  \ndto{ \dttsym{+} }  \dttnt{B}   \dttsym{)}  \mathbin{@}  \dttnt{n}\\
D_2 & =^{\text{def}} & G  \dttsym{;}  \Gamma  \vdash  \dttnt{t'}  \dttsym{:}  \dttsym{+} \, \dttnt{A}  \mathbin{@}  \dttnt{n}\\
\Gamma' & =^{\text{def}} & \Gamma  \dttsym{,}  \dttmv{y}  \dttsym{:}  \dttsym{-} \, \dttnt{B}  \mathbin{@}  \dttnt{n}
\end{array}
\]
Then we can construct the following typing derivation:
\begin{center}
  \footnotesize
  \begin{math}
    $$\mprset{flushleft}
    \inferrule* [right=\tiny cut] {
      D_1
      \\
      $$\mprset{flushleft}
      \inferrule* [right=\tiny impBar] {
        D_2
        \\
        $$\mprset{flushleft}
        \inferrule* [right=\tiny ax] {
          \,
        }{G  \dttsym{;}  \Gamma'  \vdash  \dttmv{y}  \dttsym{:}  \dttsym{-} \, \dttnt{B}  \mathbin{@}  \dttnt{n}}
      }{G  \dttsym{;}  \Gamma'  \vdash  \langle  \dttnt{t'}  \dttsym{,}  \dttmv{y}  \rangle  \dttsym{:}  \dttsym{-} \, \dttsym{(}   \dttnt{A}  \ndto{ \dttsym{+} }  \dttnt{B}   \dttsym{)}  \mathbin{@}  \dttnt{n}}
    }{G  \dttsym{;}  \Gamma  \vdash  \nu \, \dttmv{y}  \dttsym{.}  \lambda  \dttmv{x}  \dttsym{.}  \dttnt{t}  \mathbin{\Cdot[2]}  \langle  \dttnt{t'}  \dttsym{,}  \dttmv{y}  \rangle  \dttsym{:}  \dttsym{+} \, \dttnt{B}  \mathbin{@}  \dttnt{n}}
  \end{math}
\end{center}
Implication was indeed eliminated, yielding the conclusion.

There is some intuition one can use while thinking about this style of
programming that is based on the encoding of classical logic --
Parigot's $\lambda\mu$-calculus -- into the $\pi$-calculus.  See for
example \cite{Thielecke:1997,Honda:2014}.  We can think of positive
variables as input ports, and negative variables as output ports.
Clearly, these notions are dual.  Then a cut of the form $\nu \, \dttmv{z}  \dttsym{.}  \dttnt{t}  \mathbin{\Cdot[2]}  \dttnt{t'}$ can be intuitively understood as a device capable of routing
information.  We think of this term as first running the term $\dttnt{t}$,
and then plugging its value into the continuation $\dttnt{t'}$.  Thus,
negative terms are continuations. Now consider the instance of the
previous term $\nu \, \dttmv{z}  \dttsym{.}  \dttnt{t}  \mathbin{\Cdot[2]}  \dttmv{y}$ where $\dttnt{t}$ is a positive term and
$\dttmv{y}$ is a negative variable (an output port).  This can be
intuitively understood as after running $\dttnt{t}$, route its value
through the output port $\dttmv{y}$.  Now consider the instance $\nu \, \dttmv{z}  \dttsym{.}  \dttnt{t}  \mathbin{\Cdot[2]}  \dttmv{z}$.  This term can be understood as after running the term
$\dttnt{t}$, route its value through the output port $\dttmv{z}$, but then
capture this value as the return value.  Thus, the cut term reroutes
output ports into the actual return value of the cut.

%

There is one additional bit of intuition we can use when thinking about
programming in DTT. We can think of cuts of the form
$\nu z.(\lambda x_1\cdots\lambda x_i.t) \mathbin{\Cdot[2]} \langle t_1, \langle t_2, \cdots \langle t_i, z \rangle \cdots \rangle$ 
as an abstract machine, where $\lambda x_1\cdots\lambda x_i.t$ is the 
functional part of the machine, and $\langle t_1, \langle t_2, \cdots \langle t_i, z \rangle \cdots \rangle$ is 
the stack of inputs the abstract machine will apply the function to
ultimately routing the final result of the application through
$\dttmv{z}$, but rerouting this into the return value. 
This intuition is not new, but was first observed by Curien and
Herbelin in \cite{Curien:2000}; see also \cite{Curien:2002}.

Similarly to the eliminator for implication we can define the eliminator for disjunction in the form
of the usual case analysis. Suppose $G  \dttsym{;}  \Gamma  \vdash  \dttnt{t}  \dttsym{:}  \dttsym{+} \, \dttsym{(}   \dttnt{A}  \ndwedge{ \dttsym{-} }  \dttnt{B}   \dttsym{)}  \mathbin{@}  \dttnt{n}$, $G  \dttsym{;}  \Gamma  \dttsym{,}  \dttmv{x}  \dttsym{:}  \dttsym{+} \, \dttnt{A}  \mathbin{@}  \dttnt{n}  \vdash  \dttnt{t_{{\mathrm{1}}}}  \dttsym{:}  \dttsym{+} \, \dttnt{C}  \mathbin{@}  \dttnt{n}$, and
$G  \dttsym{;}  \Gamma  \dttsym{,}  \dttmv{x}  \dttsym{:}  \dttsym{+} \, \dttnt{B}  \mathbin{@}  \dttnt{n}  \vdash  \dttnt{t_{{\mathrm{2}}}}  \dttsym{:}  \dttsym{+} \, \dttnt{C}  \mathbin{@}  \dttnt{n}$ are all admissible.  Then we can
derive the usual eliminator for disjunction.  Define 
$ \mathbf{case}\, \dttnt{t} \,\mathbf{of}\, \dttmv{x} . \dttnt{t_{{\mathrm{1}}}} , \dttmv{x} . \dttnt{t_{{\mathrm{2}}}}  =^{\text{def}} \nu \, \dttmv{z_{{\mathrm{0}}}}  \dttsym{.}  \dttsym{(}  \nu \, \dttmv{z_{{\mathrm{1}}}}  \dttsym{.}  \dttsym{(}  \nu \, \dttmv{z_{{\mathrm{2}}}}  \dttsym{.}  \dttnt{t}  \mathbin{\Cdot[2]}  \dttsym{(}  \dttmv{z_{{\mathrm{1}}}}  \dttsym{,}  \dttmv{z_{{\mathrm{2}}}}  \dttsym{)}  \dttsym{)}  \mathbin{\Cdot[2]}  \dttsym{(}  \nu \, \dttmv{x}  \dttsym{.}  \dttnt{t_{{\mathrm{2}}}}  \mathbin{\Cdot[2]}  \dttmv{z_{{\mathrm{0}}}}  \dttsym{)}  \dttsym{)}  \mathbin{\Cdot[2]}  \dttsym{(}  \nu \, \dttmv{x}  \dttsym{.}  \dttnt{t_{{\mathrm{1}}}}  \mathbin{\Cdot[2]}  \dttmv{z_{{\mathrm{0}}}}  \dttsym{)}$.
Then we have the following result.
\begin{lem}
  \label{lemma:disj-elim-adm}
  The following rule is derivable:
  \begin{center}
    \begin{math}
      $$\mprset{flushleft}
      \inferrule* [right=case] {
        G  \dttsym{;}  \Gamma  \dttsym{,}  \dttmv{x}  \dttsym{:}  \dttnt{p} \, \dttnt{A}  \mathbin{@}  \dttnt{n}  \vdash  \dttnt{t_{{\mathrm{1}}}}  \dttsym{:}  \dttnt{p} \, \dttnt{C}  \mathbin{@}  \dttnt{n}
        \\\\
        G  \dttsym{;}  \Gamma  \dttsym{,}  \dttmv{x}  \dttsym{:}  \dttnt{p} \, \dttnt{B}  \mathbin{@}  \dttnt{n}  \vdash  \dttnt{t_{{\mathrm{2}}}}  \dttsym{:}  \dttnt{p} \, \dttnt{C}  \mathbin{@}  \dttnt{n}
        \\
        G  \dttsym{;}  \Gamma  \vdash  \dttnt{t}  \dttsym{:}  \dttnt{p} \, \dttsym{(}   \dttnt{A}  \ndwedge{  \bar{  \dttnt{p}  }  }  \dttnt{B}   \dttsym{)}  \mathbin{@}  \dttnt{n}
      }{G  \dttsym{;}  \Gamma  \vdash   \mathbf{case}\, \dttnt{t} \,\mathbf{of}\, \dttmv{x} . \dttnt{t_{{\mathrm{1}}}} , \dttmv{x} . \dttnt{t_{{\mathrm{2}}}}   \dttsym{:}  \dttnt{p} \, \dttnt{C}  \mathbin{@}  \dttnt{n}}
    \end{math}
  \end{center}
\end{lem}
\begin{proof}
  A full derivation in DTT can be found in
  Appendix~\ref{subsec:proof_of_lemma:disj-elim-adm}.
\end{proof}

Now consider the term $\nu \, \dttmv{x}  \dttsym{.}   \mathbf{in}_{ \dttsym{1} }\, \dttsym{(}  \nu \, \dttmv{y}  \dttsym{.}   \mathbf{in}_{ \dttsym{2} }\, \langle  \dttmv{y}  \dttsym{,}  \dttkw{triv}  \rangle   \mathbin{\Cdot[2]}  \dttmv{x}  \dttsym{)}   \mathbin{\Cdot[2]}  \dttmv{x}$.  This term is the inhabitant of the type $ \dttnt{A}  \ndwedge{ \dttsym{-} }  \mathop{\sim}  \dttnt{A} $, and its
typing derivation follows from the derivation given in
Section~\ref{sec:dualized_intuitionistic_logic_(dil)}.  We can see by
looking at the syntax that the cuts involved are indeed on the axiom
$\dttmv{x}$, thus this term has no canonical form.  In \cite{Crolard:2004}
Crolard shows that inhabitants such as these amount to a constructive
coroutine.  That is, it is a restricted form of a continuation.

We now consider several example reductions in DTT. In the following
examples we underline non-top-level redexes. The first example simply
$\alpha$-converts the function $\lambda  \dttmv{x}  \dttsym{.}  \dttmv{x}$ into $\lambda  \dttmv{z}  \dttsym{.}  \dttmv{z}$ as follows:
\begin{center}
  \begin{math}
    \begin{array}{lcl}    
      \lambda  \dttmv{z}  \dttsym{.}   \underline{ \nu \, \dttmv{y}  \dttsym{.}  \lambda  \dttmv{x}  \dttsym{.}  \dttmv{x}  \mathbin{\Cdot[2]}  \langle  \dttmv{z}  \dttsym{,}  \dttmv{y}  \rangle }  & \redtoby{RImp}   & \lambda  \dttmv{z}  \dttsym{.}   \underline{ \nu \, \dttmv{y}  \dttsym{.}  \dttmv{z}  \mathbin{\Cdot[2]}  \dttmv{y} } \\
                                    & \redtoby{RRet}   & \lambda  \dttmv{z}  \dttsym{.}  \dttmv{z}\\
    \end{array}
\end{math}
\end{center}
A more involved example is the application of the function
$\lambda  \dttmv{x}  \dttsym{.}  \dttsym{(}  \lambda  \dttmv{y}  \dttsym{.}  \dttmv{y}  \dttsym{)}$ to the arguments $\dttkw{triv}$ and $\dttkw{triv}$. 
\begin{center}
  \begin{math}
    \begin{array}{lcl}
      \nu \, \dttmv{z}  \dttsym{.}  \lambda  \dttmv{x}  \dttsym{.}  \dttsym{(}  \lambda  \dttmv{y}  \dttsym{.}  \dttmv{y}  \dttsym{)}  \mathbin{\Cdot[2]}  \langle  \dttkw{triv}  \dttsym{,}  \langle  \dttkw{triv}  \dttsym{,}  \dttmv{z}  \rangle  \rangle & \redtoby{RImp} & \nu \, \dttmv{z}  \dttsym{.}  \lambda  \dttmv{y}  \dttsym{.}  \dttmv{y}  \mathbin{\Cdot[2]}  \langle  \dttkw{triv}  \dttsym{,}  \dttmv{z}  \rangle\\
      & \redtoby{RImp} & \nu \, \dttmv{z}  \dttsym{.}  \dttkw{triv}  \mathbin{\Cdot[2]}  \dttmv{z}\\
      & \redtoby{RRet} & \dttkw{triv}\\
    \end{array}
  \end{math}
\end{center}

\section{Metatheory of DTT}
\label{sec:metatheory_of_dtt}

We now present the basic metatheory of DTT, starting with type
preservation. We begin with the inversion lemma, which is necessary for
proving type preservation.

\begin{lem}[Inversion]
  \label{lemma:inverstion}
  \begin{itemize}
  \item[] 
  \item[i.] If $G  \dttsym{;}  \Gamma  \vdash  \dttsym{(}  \dttnt{t_{{\mathrm{1}}}}  \dttsym{,}  \dttnt{t_{{\mathrm{2}}}}  \dttsym{)}  \dttsym{:}  \dttnt{p} \, \dttsym{(}   \dttnt{A}  \ndwedge{ \dttnt{p} }  \dttnt{B}   \dttsym{)}  \mathbin{@}  \dttnt{n}$, then 
    $G  \dttsym{;}  \Gamma  \vdash  \dttnt{t_{{\mathrm{1}}}}  \dttsym{:}  \dttnt{p} \, \dttnt{A}  \mathbin{@}  \dttnt{n}$ and $G  \dttsym{;}  \Gamma  \vdash  \dttnt{t_{{\mathrm{2}}}}  \dttsym{:}  \dttnt{p} \, \dttnt{B}  \mathbin{@}  \dttnt{n}$.
  \item[ii.] If $G  \dttsym{;}  \Gamma  \vdash   \mathbf{in}_{ \dttnt{d} }\, \dttnt{t}   \dttsym{:}  \dttnt{p} \, \dttsym{(}   \dttnt{A_{{\mathrm{1}}}}  \ndwedge{  \bar{  \dttnt{p}  }  }  \dttnt{A_{{\mathrm{2}}}}   \dttsym{)}  \mathbin{@}  \dttnt{n}$, then 
    $G  \dttsym{;}  \Gamma  \vdash  \dttnt{t}  \dttsym{:}  \dttnt{p} \,  \dttnt{A} _{ \dttnt{d} }   \mathbin{@}  \dttnt{n}$.
  \item[iii.] If $G  \dttsym{;}  \Gamma  \vdash  \lambda  \dttmv{x}  \dttsym{.}  \dttnt{t}  \dttsym{:}  \dttnt{p} \, \dttsym{(}   \dttnt{A}  \ndto{ \dttnt{p} }  \dttnt{B}   \dttsym{)}  \mathbin{@}  \dttnt{n}$, then 
    $\dttsym{(}  G  \dttsym{,}  \dttnt{n} \,  \preccurlyeq_{ \dttnt{p} }  \, \dttnt{n'}  \dttsym{)}  \dttsym{;}  \Gamma  \dttsym{,}  \dttmv{x}  \dttsym{:}  \dttnt{p} \, \dttnt{A}  \mathbin{@}  \dttnt{n'}  \vdash  \dttnt{t}  \dttsym{:}  \dttnt{p} \, \dttnt{B}  \mathbin{@}  \dttnt{n'}$ for any
    $\dttnt{n'} \, \not\in \, \dttsym{\mbox{$\mid$}}  G  \dttsym{\mbox{$\mid$}}  \dttsym{,}  \dttsym{\mbox{$\mid$}}  \Gamma  \dttsym{\mbox{$\mid$}}$.
  \item[iv.] If $G  \dttsym{;}  \Gamma  \vdash  \langle  \dttnt{t_{{\mathrm{1}}}}  \dttsym{,}  \dttnt{t_{{\mathrm{2}}}}  \rangle  \dttsym{:}  \dttnt{p} \, \dttsym{(}   \dttnt{A}  \ndto{  \bar{  \dttnt{p}  }  }  \dttnt{B}   \dttsym{)}  \mathbin{@}  \dttnt{n}$, then
    $G  \vdash  \dttnt{n} \,  \preccurlyeq^*_{  \bar{  \dttnt{p}  }  }  \, \dttnt{n'}$, $G  \dttsym{;}  \Gamma  \vdash  \dttnt{t_{{\mathrm{1}}}}  \dttsym{:}   \bar{  \dttnt{p}  }  \, \dttnt{A}  \mathbin{@}  \dttnt{n'}$, and
    $G  \dttsym{;}  \Gamma  \vdash  \dttnt{t_{{\mathrm{2}}}}  \dttsym{:}  \dttnt{p} \, \dttnt{B}  \mathbin{@}  \dttnt{n'}$ for some node $\dttnt{n'}$.
  \end{itemize}
\end{lem}
  \begin{proof}
    Each case of the above lemma holds by a trivial proof by induction
    on the assumed typing derivation.
  \end{proof}

\noindent The results node substitution and substitution for
typing are essential for the cases of type preservation that reduce a
top-level redex. Node substitution, denoted $\dttsym{[}  \dttnt{n_{{\mathrm{1}}}}  \dttsym{/}  \dttnt{n_{{\mathrm{2}}}}  \dttsym{]}  \dttnt{n}$, is defined as
follows: \begin{center}
  \begin{math}
    \begin{array}{lll}
      \dttsym{[}  \dttnt{n_{{\mathrm{1}}}}  \dttsym{/}  \dttnt{n_{{\mathrm{2}}}}  \dttsym{]}  \dttnt{n_{{\mathrm{2}}}} & = & \dttnt{n_{{\mathrm{1}}}}\\
      \dttsym{[}  \dttnt{n_{{\mathrm{1}}}}  \dttsym{/}  \dttnt{n_{{\mathrm{2}}}}  \dttsym{]}  \dttnt{n} & = & \dttnt{n} \text{ where } \dttnt{n} \text{ is
        distinct from } \dttnt{n_{{\mathrm{2}}}}
    \end{array}
  \end{math}
\end{center} The following lemmas are necessary in the proof of
node substitution for typing.
\begin{lem}[Node Renaming]
  \label{lemma:renaming_nodes_in_graph}
  If $G_{{\mathrm{1}}}  \dttsym{,}  G_{{\mathrm{2}}}  \vdash  \dttnt{n_{{\mathrm{1}}}} \,  \preccurlyeq^*_{ \dttnt{p} }  \, \dttnt{n_{{\mathrm{3}}}}$, then for any nodes $\dttnt{n_{{\mathrm{4}}}}$ and
  $\dttnt{n_{{\mathrm{5}}}}$, we have $\dttsym{[}  \dttnt{n_{{\mathrm{4}}}}  \dttsym{/}  \dttnt{n_{{\mathrm{5}}}}  \dttsym{]}  G_{{\mathrm{1}}}  \dttsym{,}  \dttsym{[}  \dttnt{n_{{\mathrm{4}}}}  \dttsym{/}  \dttnt{n_{{\mathrm{5}}}}  \dttsym{]}  G_{{\mathrm{2}}}  \vdash  \dttsym{[}  \dttnt{n_{{\mathrm{4}}}}  \dttsym{/}  \dttnt{n_{{\mathrm{5}}}}  \dttsym{]}  \dttnt{n_{{\mathrm{1}}}} \,  \preccurlyeq^*_{ \dttnt{p} }  \, \dttsym{[}  \dttnt{n_{{\mathrm{4}}}}  \dttsym{/}  \dttnt{n_{{\mathrm{5}}}}  \dttsym{]}  \dttnt{n_{{\mathrm{3}}}}$.
\end{lem}
\begin{proof}
  This proof holds by induction on the assumed reachability
  derivation.  For the full proof see
  Appendix~\ref{subsec:proof_of_lemma_node_renaming}.
\end{proof}

\begin{lem}[Node Substitution for Reachability]
  \label{lemma:node_substitution_for_reachability}
  If $G  \dttsym{,}  \dttnt{n_{{\mathrm{1}}}} \,  \preccurlyeq_{ \dttnt{p_{{\mathrm{1}}}} }  \, \dttnt{n_{{\mathrm{2}}}}  \dttsym{,}  G'  \vdash  \dttnt{n_{{\mathrm{4}}}} \,  \preccurlyeq^*_{ \dttnt{p} }  \, \dttnt{n_{{\mathrm{5}}}}$ and $G  \dttsym{,}  G'  \vdash  \dttnt{n_{{\mathrm{1}}}} \,  \preccurlyeq^*_{ \dttnt{p_{{\mathrm{1}}}} }  \, \dttnt{n_{{\mathrm{3}}}}$, then 
  $\dttsym{[}  \dttnt{n_{{\mathrm{3}}}}  \dttsym{/}  \dttnt{n_{{\mathrm{2}}}}  \dttsym{]}  G  \dttsym{,}  \dttsym{[}  \dttnt{n_{{\mathrm{3}}}}  \dttsym{/}  \dttnt{n_{{\mathrm{2}}}}  \dttsym{]}  G'  \vdash  \dttsym{[}  \dttnt{n_{{\mathrm{3}}}}  \dttsym{/}  \dttnt{n_{{\mathrm{2}}}}  \dttsym{]}  \dttnt{n_{{\mathrm{4}}}} \,  \preccurlyeq^*_{ \dttnt{p} }  \, \dttsym{[}  \dttnt{n_{{\mathrm{3}}}}  \dttsym{/}  \dttnt{n_{{\mathrm{2}}}}  \dttsym{]}  \dttnt{n_{{\mathrm{5}}}}$. 
\end{lem}
\begin{proof}
  This proof holds by by induction on the form of the assumed
  reachability derivation. For the full proof see
  Appendix~\ref{subsec:proof_of_lemma_node_substitution_for_reachability}.
\end{proof}

\begin{lem}[Node Substitution for Typing]
  \label{lemma:node_substitution_for_typing}
  If $G  \dttsym{,}  \dttnt{n_{{\mathrm{1}}}} \,  \preccurlyeq_{ \dttnt{p_{{\mathrm{1}}}} }  \, \dttnt{n_{{\mathrm{2}}}}  \dttsym{,}  G'  \dttsym{;}  \Gamma  \vdash  \dttnt{t}  \dttsym{:}  \dttnt{p_{{\mathrm{2}}}} \, \dttnt{A}  \mathbin{@}  \dttnt{n_{{\mathrm{3}}}}$ and $G  \dttsym{,}  G'  \vdash  \dttnt{n_{{\mathrm{1}}}} \,  \preccurlyeq^*_{ \dttnt{p_{{\mathrm{1}}}} }  \, \dttnt{n_{{\mathrm{4}}}}$, then
  $\dttsym{[}  \dttnt{n_{{\mathrm{4}}}}  \dttsym{/}  \dttnt{n_{{\mathrm{2}}}}  \dttsym{]}  G  \dttsym{,}  \dttsym{[}  \dttnt{n_{{\mathrm{4}}}}  \dttsym{/}  \dttnt{n_{{\mathrm{2}}}}  \dttsym{]}  G'  \dttsym{;}  \dttsym{[}  \dttnt{n_{{\mathrm{4}}}}  \dttsym{/}  \dttnt{n_{{\mathrm{2}}}}  \dttsym{]}  \Gamma  \vdash  \dttnt{t}  \dttsym{:}  \dttnt{p_{{\mathrm{2}}}} \, \dttnt{A}  \mathbin{@}  \dttsym{[}  \dttnt{n_{{\mathrm{4}}}}  \dttsym{/}  \dttnt{n_{{\mathrm{2}}}}  \dttsym{]}  \dttnt{n_{{\mathrm{3}}}}$. 
\end{lem}
\begin{proof}
  This holds by induction on the form of the assumed typing
  derivation.  See
  Appendix~\ref{subsec:proof_of_lemma_node_substitution_for_typing}
  for the full proof.
\end{proof}
\noindent
The next lemma is crucial for type preservation.
\begin{lem}[Substitution for Typing]
  \label{lemma:substitution_for_typing}
  If $G  \dttsym{;}  \Gamma  \vdash  \dttnt{t_{{\mathrm{1}}}}  \dttsym{:}  \dttnt{p_{{\mathrm{1}}}} \, \dttnt{A}  \mathbin{@}  \dttnt{n_{{\mathrm{1}}}}$ and $G  \dttsym{;}  \Gamma  \dttsym{,}  \dttmv{x}  \dttsym{:}  \dttnt{p_{{\mathrm{1}}}} \, \dttnt{A}  \mathbin{@}  \dttnt{n_{{\mathrm{1}}}}  \dttsym{,}  \Gamma'  \vdash  \dttnt{t_{{\mathrm{2}}}}  \dttsym{:}  \dttnt{p_{{\mathrm{2}}}} \, \dttnt{B}  \mathbin{@}  \dttnt{n_{{\mathrm{2}}}}$, then
  $G  \dttsym{;}  \Gamma  \dttsym{,}  \Gamma'  \vdash  \dttsym{[}  \dttnt{t_{{\mathrm{1}}}}  \dttsym{/}  \dttmv{x}  \dttsym{]}  \dttnt{t_{{\mathrm{2}}}}  \dttsym{:}  \dttnt{p_{{\mathrm{2}}}} \, \dttnt{B}  \mathbin{@}  \dttnt{n_{{\mathrm{2}}}}$.
\end{lem}
\begin{proof}
This proof holds by induction on the second assumed typing relation.
For the full proof see
Appendix~\ref{subsec:proof_of_lemma_substitution_for_typing}.
\end{proof}
\begin{lem}[Type Preservation]
  \label{lemma:type_preservation}
  If $G  \dttsym{;}  \Gamma  \vdash  \dttnt{t}  \dttsym{:}  \dttnt{p} \, \dttnt{A}  \mathbin{@}  \dttnt{n}$, and $\dttnt{t} \redto \dttnt{t'}$, then $G  \dttsym{;}  \Gamma  \vdash  \dttnt{t'}  \dttsym{:}  \dttnt{p} \, \dttnt{A}  \mathbin{@}  \dttnt{n}$.
\end{lem}
  \begin{proof}
    This proof holds by induction on the form of the assumed typing
    derivation.  For the full proof see
    Appendix~\ref{subsec:proof_of_lemma:type_preservation}.
  \end{proof}
\begin{figure}[t]
    \begin{mathpar}
      \dttdruleClassAx{}     \and
      \dttdruleClassUnit{}   \and
      \dttdruleClassAnd{}    \and
      \dttdruleClassAndBar{} \and
      \dttdruleClassImp{}    \and 
      \dttdruleClassImpBar{} \and
      \dttdruleClassCut{}    
    \end{mathpar}
\caption{Classical typing of DTT terms}
\label{fig:classtp}
\end{figure}
A more substantial property is strong normalization of reduction for
typed terms.  To prove this result, we will prove a stronger property,
namely strong normalization for reduction of terms that are typable
using the system of classical typing rules in Figure~\ref{fig:classtp}
\cite{crolard01}.  This is justified by the following easy result
(proof omitted), where $ \mathsf{DN}( \Gamma ) $ just drops the world annotations
from assumptions in $\Gamma$:
\begin{thm}
\label{thm:inttoclass}
If $G  \dttsym{;}  \Gamma  \vdash  \dttnt{t}  \dttsym{:}  \dttnt{p} \, \dttnt{A}  \mathbin{@}  \dttnt{n}$, then $ \mathsf{DN}( \Gamma )   \vdash_c  \dttnt{t}  \dttsym{:}  \dttnt{p} \, \dttnt{A}$
\end{thm}

Let $\SN$ be the set of terms that are strongly normalizing with
respect to the reduction relation.  Let \textit{Var} be the set of
term variables, and let us use $x$ and $y$ as metavariables for variables.  We
will prove strong normalization for classically typed terms using a
version of Krivine's classical realizability~\cite{krivine09}.  We
define three interpretations of types in Figure~\ref{fig:classreal}.
The definition is by mutual induction, and can easily be seen to
be well-founded, as the definition of $\interp{A}^+$ invokes the
definition of $\interp{A}^-$ with the same type, which in turn invokes
the definition of $\interp{A}^{+c}$ with the same type; and the
definition of $\interp{A}^{+c}$ may invoke either of the other
definitions at a strictly smaller type.  The reader familiar with
such proofs will also recognize the debt owed to Girard~\cite{gtl90}.

\begin{figure}
\small
\[
\begin{array}{lll}
t \in \interp{A}^+ & \Leftrightarrow & \forall x \in \textit{Var}.\ \forall t'\in\interp{A}^-.\ \nu \, \dttmv{x}  \dttsym{.}  \dttnt{t}  \mathbin{\Cdot[2]}  \dttnt{t'} \in \SN\\
t \in \interp{A}^- & \Leftrightarrow & \forall x \in \textit{Var}.\ \forall t'\in\interp{A}^{+c}.\ \nu \, \dttmv{x}  \dttsym{.}  \dttnt{t'}  \mathbin{\Cdot[2]}  \dttnt{t} \in \SN\\
t \in \interp{ \langle  \dttsym{+} \rangle }^{+c} & \Leftrightarrow & t \in \textit{Var}\ \vee\ t \equiv \dttkw{triv} \\
t \in \interp{ \langle  \dttsym{-} \rangle }^{+c} & \Leftrightarrow & t \in \textit{Var}\\
t \in \interp{ \dttnt{A}  \ndto{ \dttsym{+} }  \dttnt{B} }^{+c} & \Leftrightarrow & t \in \textit{Var}\ \vee\ \exists x, t'. t \equiv \lambda x.\, t'\ \wedge\ \forall t''\in\interp{A}^+.\ [t''/x] t'\in\interp{B}^+\\
t \in \interp{ \dttnt{A}  \ndto{ \dttsym{-} }  \dttnt{B} }^{+c} & \Leftrightarrow & t \in \textit{Var}\ \vee\ \exists t_1\in\interp{A}^-, t_2\in\interp{B}^+.\ t \equiv \langle  \dttnt{t_{{\mathrm{1}}}}  \dttsym{,}  \dttnt{t_{{\mathrm{2}}}}  \rangle \\
t \in \interp{ \dttnt{A}  \ndwedge{ \dttsym{+} }  \dttnt{B} }^{+c} & \Leftrightarrow & t \in \textit{Var}\ \vee\ \exists t_1\in\interp{A}^+, t_2\in\interp{B}^+.\ t \equiv \dttsym{(}  \dttnt{t_{{\mathrm{1}}}}  \dttsym{,}  \dttnt{t_{{\mathrm{2}}}}  \dttsym{)} \\
t \in \interp{ \dttnt{A_{{\mathrm{1}}}}  \ndwedge{ \dttsym{-} }  \dttnt{A_{{\mathrm{2}}}} }^{+c} & \Leftrightarrow & t \in \textit{Var}\ \vee\ \exists d. \exists t'\in\interp{A_d}^+.\ t \equiv  \mathbf{in}_{ \dttnt{d} }\, \dttnt{t'}  
\end{array}
\]
\caption{Interpretations of types}
\label{fig:classreal}
\end{figure}

\begin{lem}[Step interpretations]
\label{lem:stepinterp}
If $t\in\interp{A}^+$ and $t\leadsto t'$, then
$t'\in\interp{A}^+$; and similarly if $t\in\interp{A}^-$ or $t\in\interp{A}^{+c}$.
\end{lem}
\begin{proof}
The proof is by a mutual well-founded induction.
Assume $t\in\interp{A}^+$ and $t\leadsto t'$.  We must show $t'\in\interp{A}^+$.
For this, it suffices to assume $y\in\textit{Var}$ and $t''\in\interp{A}^-$,
and show $\nu \, \dttmv{y}  \dttsym{.}  \dttnt{t'}  \mathbin{\Cdot[2]}  \dttnt{t''}\in\SN$.  From the assumption that $t\in\interp{A}^+$,
we have 
\[
\nu \, \dttmv{y}  \dttsym{.}  \dttnt{t}  \mathbin{\Cdot[2]}  \dttnt{t''} \in\SN
\]
which indeed implies that 
\[
\nu \, \dttmv{y}  \dttsym{.}  \dttnt{t'}  \mathbin{\Cdot[2]}  \dttnt{t''} \in\SN
\]
A similar argument applies if $t\in\interp{A}^-$.  

For the last part of the lemma, assume $t\in\interp{A}^{+c}$ with
$t\leadsto t'$, and show $t'\in\interp{A}^{+c}$.  The only possible
cases are the following, where $t\not\in\textit{Vars}$.

If $A \equiv  \dttnt{A_{{\mathrm{1}}}}  \ndto{ \dttsym{+} }  \dttnt{A_{{\mathrm{2}}}} $, then $t$ is of the form $\lambda x.t_a$
for some $x$ and $t_a$, where for all $t_b\in\interp{A_1}^+$, we have
$[t_b/x]t_a\in\interp{A_2}^+$.  Since $t\leadsto t'$, $t'$ must be
$\lambda x.t_a'$ for some $t_a'$ with $t_a\leadsto t_a'$.  It suffices
now to assume an arbitrary $t_b\in\interp{A_1}^+$, and show
$[t_b/x]t_a'\in\interp{A_2}^+$.  But $[t_b/x]t_a\leadsto [t_b/x]t_a'$
follows from $t_a\leadsto t_a'$, so by our IH, we have
$[t_b/x]t_a'\in\interp{A_2}^+$, as required.

If $A\equiv  \dttnt{A_{{\mathrm{1}}}}  \ndto{ \dttsym{-} }  \dttnt{A_{{\mathrm{2}}}} $, then $t$ is of the form $\langle  \dttnt{t_{{\mathrm{1}}}}  \dttsym{,}  \dttnt{t_{{\mathrm{2}}}}  \rangle$
for some $t_1\in\interp{A_1}^-$ and $t_2\in\interp{A_2}^+$; and
$t'\equiv \langle  \dttnt{t'_{{\mathrm{1}}}}  \dttsym{,}  \dttnt{t'_{{\mathrm{2}}}}  \rangle$ where either $t_1'\equiv t_1$ and $t_2\leadsto t_2'$
or else $t_1\leadsto t_1'$ and $t_2'\equiv t_2$.  Either way, we have
$t_1'\in\interp{A_1}^-$ and $t_2'\in\interp{A_2}^+$ by our IH, so
we have $\langle  \dttnt{t'_{{\mathrm{1}}}}  \dttsym{,}  \dttnt{t'_{{\mathrm{2}}}}  \rangle\in\interp{ \dttnt{A_{{\mathrm{1}}}}  \ndto{ \dttsym{-} }  \dttnt{A_{{\mathrm{2}}}} }^{+c}$ as required.

The other cases for $A\equiv  \dttnt{A_{{\mathrm{1}}}}  \ndwedge{ \dttnt{p} }  \dttnt{A_{{\mathrm{2}}}} $ are similar to the previous one.
\end{proof}
\begin{lem}[SN interpretations]
  \label{lem:sninterp}
  
  \ \\
  \begin{tabular}{llllll}
    1. & $\interp{A}^+ \subseteq \SN$ & \hspace{2cm} & 3. & $\interp{A}^- \subseteq \SN$\\
    2. & $\textit{Vars}\subseteq \interp{A}^-$ & \hspace{2cm} & 4. & $\interp{A}^{+c} \subseteq \SN$
  \end{tabular}
\end{lem}
\begin{proof}
  The proof holds by mutual well-founded induction on the pair
  $(A,n)$, where $n$ is the number of the proposition in the statement
  of the lemma; the well-founded ordering in question is the
  lexicographic combination of the structural ordering on types (for
  $A$) and the ordering $1 > 2 > 4 > 3$ (for $n$).  For the full proof
  see Appendix~\ref{subsec:proof_of_lemma_sn_interpretations}.
\end{proof}  

\begin{defi}[Interpretation of contexts]
$\interp{\Gamma}$ is the set of substitutions $\sigma$ such that
for all $\dttmv{x}  \dttsym{:}  \dttnt{p} \, \dttnt{A}\in\Gamma$, $\sigma(x)\in\interp{A}^p$.
\end{defi}

\begin{lem}[Canonical positive is positive]
\label{lem:canonpos}
$\interp{A}^{+c}\subseteq\interp{A}^+$
\end{lem}
\begin{proof}
Assume $t\in\interp{A}^{+c}$ and show $t\in\interp{A}^+$.
For the latter, assume arbitrary $x\in\textit{Vars}$ and $t'\in\interp{A}^-$,
and show $\nu \, \dttmv{x}  \dttsym{.}  \dttnt{t}  \mathbin{\Cdot[2]}  \dttnt{t'}\in\SN$.  This follows immediately
from the assumption that $t'\in\interp{A}^-$.
\end{proof}

\begin{thm}[Soundness]
\label{thm:sndinterp}
If $\Gamma  \vdash_c  \dttnt{t}  \dttsym{:}  \dttnt{p} \, \dttnt{A}$ then for all $\sigma\in\interp{\Gamma}$, $\sigma t\in\interp{A}^p$.
\end{thm}
\begin{proof}
  The proof holds by induction on the derivation of $\Gamma  \vdash_c  \dttnt{t}  \dttsym{:}  \dttnt{p} \, \dttnt{A}$.  For the full proof see
  Appendix~\ref{subsec:proof_of_soundness}.
\end{proof}

\begin{cor}[Strong Normalization]
  \label{thm:strong_normalization}
  If $G  \dttsym{;}  \Gamma  \vdash  \dttnt{t}  \dttsym{:}  \dttnt{p} \, \dttnt{A}  \mathbin{@}  \dttnt{n}$, then $t \in \SN$.
\end{cor}
\begin{proof} This follows easily by putting together Theorems~\ref{thm:inttoclass} and~\ref{thm:sndinterp}, with
Lemma~\ref{lem:sninterp}.
\end{proof}

\begin{cor}[Cut Elimination]
If $G  \dttsym{;}  \Gamma  \vdash  \dttnt{t}  \dttsym{:}  \dttnt{p} \, \dttnt{A}  \mathbin{@}  \dttnt{n}$, then there is normal $t'$ with
$t\leadsto^* t'$ and $t'$ containing only cut terms of the form
$\nu \, \dttmv{x}  \dttsym{.}  \dttmv{y}  \mathbin{\Cdot[2]}  \dttnt{t}$ or $\nu \, \dttmv{x}  \dttsym{.}  \dttnt{t}  \mathbin{\Cdot[2]}  \dttmv{y}$, for $y$ a variable.
\end{cor}

\begin{lem}[Local Confluence]
\label{lem:localconf}
The reduction relation of Figure~\ref{fig:dtt-red} is locally confluent.
\end{lem}
\begin{proof} We may view the reduction rules as higher-order pattern
  rewrite rules.  It is easy to confirm that all critical pairs (e.g.,
  between \dttdrulename{RBetaR} and the rules \dttdrulename{RImp},
  \dttdrulename{RImpBar}, \dttdrulename{RAnd1},
  \dttdrulename{RAndBar1}, \dttdrulename{RAnd2}, and
  \dttdrulename{RAndBar2}) are joinable.  Local confluence then
  follows by the higher-order critical pair lemma~\cite{nipkow91}.
\end{proof}

\begin{thm}[Confluence for Typable Terms]
The reduction relation restricted to terms typable in DTT is confluent.
\end{thm}
\begin{proof} Suppose $G  \dttsym{;}  \Gamma  \vdash  \dttnt{t}  \dttsym{:}  \dttnt{p} \, \dttnt{A}  \mathbin{@}  \dttnt{n}$ for some $G$, $\Gamma$, $\dttnt{p}$, and $\dttnt{A}$.
By Lemma~\ref{lemma:type_preservation}, any reductions in the unrestricted reduction
relation from $t$ are also in the reduction relation restricted to typable terms.
The result now follows from Newman's Lemma, using Lemma~\ref{lem:localconf} and
Theorem~\ref{thm:strong_normalization}.
\end{proof}




\section{Conclusion}

We have presented a new type theory for bi-intuitionistic logic.  We
began with a compact dualized formulation of the logic, Dualized
Intuitionistic Logic (DIL), and showed soundness with respect to a
standard Kripke semantics (in Agda), and completeness with respect to
Pinto and Uustalu's system L.  We then presented Dualized Type Theory
(DTT), and showed type preservation, strong normalization, and
confluence for typable terms.  Future work includes further additions
to DTT, for example with polymorphism and inductive types.  It would
also be interesting to obtain a Canonicity Theorem as in
\cite{Stump:2014:RPD:2541568.2541575}, identifying some set of types
where closed normal forms are guaranteed to be canonical values (as
canonicity fails in general in DIL/DTT, as in other bi-intuitionistic
systems).

\section*{Acknowledgments}

We thank the anonymous reviewers for their detailed reviews that
helped improve this paper.

\bibliographystyle{plain}

\appendix

\section{Proofs from Section~\ref{subsec:completeness}: Completeness of DIL}
\label{sec:proofs_from_section_completeness_of_dil}

\subsection{Proof of Lemma~\ref{lemma:andl}}
\label{subsec:proof_of_lemma:andl}
Suppose $G  \dttsym{;}  \Gamma  \dttsym{,}   \bar{  \dttnt{p}  }  \, \dttnt{A}  \mathbin{@}  \dttnt{n}  \vdash  \dttnt{p} \, \dttnt{B}  \mathbin{@}  \dttnt{n}$ is derivable. By weakening
we know $G  \dttsym{;}  \Gamma  \dttsym{,}   \bar{  \dttnt{p}  }  \, \dttsym{(}   \dttnt{A}  \ndwedge{  \bar{  \dttnt{p}  }  }  \dttnt{B}   \dttsym{)}  \mathbin{@}  \dttnt{n}  \dttsym{,}   \bar{  \dttnt{p}  }  \, \dttnt{B}  \mathbin{@}  \dttnt{n}  \dttsym{,}   \bar{  \dttnt{p}  }  \, \dttnt{A}  \mathbin{@}  \dttnt{n}  \vdash  \dttnt{p} \, \dttnt{B}  \mathbin{@}  \dttnt{n}$. 
Then $G  \dttsym{;}  \Gamma  \vdash  \dttnt{p} \, \dttsym{(}   \dttnt{A}  \ndwedge{  \bar{  \dttnt{p}  }  }  \dttnt{B}   \dttsym{)}  \mathbin{@}  \dttnt{n}$ is derivable as follows:
\begin{center}
  \small
  \begin{math}
    $$\mprset{flushleft}
    \inferrule* [right=\scriptsize axCut] {
       \bar{  \dttnt{p}  }  \, \dttsym{(}   \dttnt{A}  \ndwedge{  \bar{  \dttnt{p}  }  }  \dttnt{B}   \dttsym{)}  \mathbin{@}  \dttnt{n} \in \Gamma  \dttsym{,}   \bar{  \dttnt{p}  }  \, \dttsym{(}   \dttnt{A}  \ndwedge{  \bar{  \dttnt{p}  }  }  \dttnt{B}   \dttsym{)}  \mathbin{@}  \dttnt{n}
      \\
      $$\mprset{flushleft}
      \inferrule* [right=\scriptsize AndBar] {
        $$\mprset{flushleft}
        \inferrule* [right=\scriptsize Cut] {
          D_1
          \\
          D_2
        }{G  \dttsym{;}  \Gamma  \dttsym{,}   \bar{  \dttnt{p}  }  \, \dttsym{(}   \dttnt{A}  \ndwedge{  \bar{  \dttnt{p}  }  }  \dttnt{B}   \dttsym{)}  \mathbin{@}  \dttnt{n}  \vdash  \dttnt{p} \, \dttnt{B}  \mathbin{@}  \dttnt{n}}
      }{G  \dttsym{;}  \Gamma  \dttsym{,}   \bar{  \dttnt{p}  }  \, \dttsym{(}   \dttnt{A}  \ndwedge{  \bar{  \dttnt{p}  }  }  \dttnt{B}   \dttsym{)}  \mathbin{@}  \dttnt{n}  \vdash  \dttnt{p} \, \dttsym{(}   \dttnt{A}  \ndwedge{  \bar{  \dttnt{p}  }  }  \dttnt{B}   \dttsym{)}  \mathbin{@}  \dttnt{n}}      
    }{G  \dttsym{;}  \Gamma  \vdash  \dttnt{p} \, \dttsym{(}   \dttnt{A}  \ndwedge{  \bar{  \dttnt{p}  }  }  \dttnt{B}   \dttsym{)}  \mathbin{@}  \dttnt{n}}
  \end{math}
\end{center}
where we have the following subderivations:  

\begin{math}
  \tiny
  \begin{array}{lll}
    D_0: \\
    & $$\mprset{flushleft}
    \inferrule* [right=\scriptsize AndBar] {
      $$\mprset{flushleft}
      \inferrule* [right=\scriptsize ax] {
        \,
      }{G  \dttsym{;}  \Gamma  \dttsym{,}   \bar{  \dttnt{p}  }  \, \dttsym{(}   \dttnt{A}  \ndwedge{  \bar{  \dttnt{p}  }  }  \dttnt{B}   \dttsym{)}  \mathbin{@}  \dttnt{n}  \dttsym{,}   \bar{  \dttnt{p}  }  \, \dttnt{B}  \mathbin{@}  \dttnt{n}  \dttsym{,}  \dttnt{p} \, \dttnt{A}  \mathbin{@}  \dttnt{n}  \vdash  \dttnt{p} \, \dttnt{A}  \mathbin{@}  \dttnt{n}}
    }{G  \dttsym{;}  \Gamma  \dttsym{,}   \bar{  \dttnt{p}  }  \, \dttsym{(}   \dttnt{A}  \ndwedge{  \bar{  \dttnt{p}  }  }  \dttnt{B}   \dttsym{)}  \mathbin{@}  \dttnt{n}  \dttsym{,}   \bar{  \dttnt{p}  }  \, \dttnt{B}  \mathbin{@}  \dttnt{n}  \dttsym{,}  \dttnt{p} \, \dttnt{A}  \mathbin{@}  \dttnt{n}  \vdash  \dttnt{p} \, \dttsym{(}   \dttnt{A}  \ndwedge{  \bar{  \dttnt{p}  }  }  \dttnt{B}   \dttsym{)}  \mathbin{@}  \dttnt{n}}
  \end{array}
\end{math}

\begin{math}
  \tiny
  \begin{array}{lll}
    D_1: \\
    &
    $$\mprset{flushleft}
    \inferrule* [right=\scriptsize axCut] {
       \bar{  \dttnt{p}  }  \, \dttnt{B}  \mathbin{@}  \dttnt{n} \in \Gamma  \dttsym{,}   \bar{  \dttnt{p}  }  \, \dttsym{(}   \dttnt{A}  \ndwedge{  \bar{  \dttnt{p}  }  }  \dttnt{B}   \dttsym{)}  \mathbin{@}  \dttnt{n}  \dttsym{,}   \bar{  \dttnt{p}  }  \, \dttnt{B}  \mathbin{@}  \dttnt{n}  \dttsym{,}   \bar{  \dttnt{p}  }  \, \dttnt{A}  \mathbin{@}  \dttnt{n}
      \\
      G  \dttsym{;}  \Gamma  \dttsym{,}   \bar{  \dttnt{p}  }  \, \dttsym{(}   \dttnt{A}  \ndwedge{  \bar{  \dttnt{p}  }  }  \dttnt{B}   \dttsym{)}  \mathbin{@}  \dttnt{n}  \dttsym{,}   \bar{  \dttnt{p}  }  \, \dttnt{B}  \mathbin{@}  \dttnt{n}  \dttsym{,}   \bar{  \dttnt{p}  }  \, \dttnt{A}  \mathbin{@}  \dttnt{n}  \vdash  \dttnt{p} \, \dttnt{B}  \mathbin{@}  \dttnt{n}      
    }{G  \dttsym{;}  \Gamma  \dttsym{,}   \bar{  \dttnt{p}  }  \, \dttsym{(}   \dttnt{A}  \ndwedge{  \bar{  \dttnt{p}  }  }  \dttnt{B}   \dttsym{)}  \mathbin{@}  \dttnt{n}  \dttsym{,}   \bar{  \dttnt{p}  }  \, \dttnt{B}  \mathbin{@}  \dttnt{n}  \vdash  \dttnt{p} \, \dttnt{A}  \mathbin{@}  \dttnt{n}}
  \end{array}
\end{math}

\begin{math}
  \tiny
  \begin{array}{lll}
    D_2: \\
    & $$\mprset{flushleft}
    \inferrule* [right=\scriptsize axCut] {
       \bar{  \dttnt{p}  }  \, \dttsym{(}   \dttnt{A}  \ndwedge{  \bar{  \dttnt{p}  }  }  \dttnt{B}   \dttsym{)}  \mathbin{@}  \dttnt{n} \in \Gamma  \dttsym{,}   \bar{  \dttnt{p}  }  \, \dttsym{(}   \dttnt{A}  \ndwedge{  \bar{  \dttnt{p}  }  }  \dttnt{B}   \dttsym{)}  \mathbin{@}  \dttnt{n}  \dttsym{,}   \bar{  \dttnt{p}  }  \, \dttnt{B}  \mathbin{@}  \dttnt{n}  \dttsym{,}  \dttnt{p} \, \dttnt{A}  \mathbin{@}  \dttnt{n}
      \\
      D_0
    }{G  \dttsym{;}  \Gamma  \dttsym{,}   \bar{  \dttnt{p}  }  \, \dttsym{(}   \dttnt{A}  \ndwedge{  \bar{  \dttnt{p}  }  }  \dttnt{B}   \dttsym{)}  \mathbin{@}  \dttnt{n}  \dttsym{,}   \bar{  \dttnt{p}  }  \, \dttnt{B}  \mathbin{@}  \dttnt{n}  \vdash   \bar{  \dttnt{p}  }  \, \dttnt{A}  \mathbin{@}  \dttnt{n}}
  \end{array}
\end{math}

\subsection{Proof of Lemma~\ref{lemma:raise_lower}: Raising the Lower Bound}
\label{subsec:proof_of_raising_the_lower_bound}
This is a proof by induction on the form of $G  \dttsym{,}  G_{{\mathrm{1}}}  \vdash  \dttnt{m} \,  \preccurlyeq^*_{ \dttnt{p'} }  \, \dttnt{m'}$.
\begin{description}
\item[\cW]
  \[
  \mprset{flushleft}
  \inferrule* [right=\ifrName{ax}] {
    \ 
  }{G'  \dttsym{,}  \dttnt{m} \,  \preccurlyeq_{ \dttnt{p'} }  \, \dttnt{m'}  \dttsym{,}  G''  \vdash  \dttnt{m} \,  \preccurlyeq^*_{ \dttnt{p'} }  \, \dttnt{m'}}
  \leqno{\raise 8 pt\hbox{\textbf{Case}}}
  \]
  Note that it is the case that $G'  \dttsym{,}  \dttnt{m} \,  \preccurlyeq_{ \dttnt{p'} }  \, \dttnt{m'}  \dttsym{,}  G'' \equiv G  \dttsym{,}  G_{{\mathrm{1}}}$.  If $\dttnt{m} \,  \preccurlyeq_{ \dttnt{p'} }  \, \dttnt{m'} \in G$,
  then we obtain our result, so suppose $\dttnt{m} \,  \preccurlyeq_{ \dttnt{p'} }  \, \dttnt{m'} \in G_{{\mathrm{1}}}$.  Suppose $\dttnt{p} \equiv \dttnt{p'}$. 
  Now if $\dttnt{m} \not\equiv \dttnt{n_{{\mathrm{1}}}}$, then clearly, we obtain
  our result.  Consider the case where $\dttnt{m} \equiv \dttnt{n_{{\mathrm{1}}}}$.  Then it suffices to show
  $G  \dttsym{,}  \mathsf{raise} \, \dttsym{(}  \dttnt{n_{{\mathrm{1}}}}  \dttsym{,}  \dttnt{n_{{\mathrm{2}}}}  \dttsym{,}  G'_{{\mathrm{1}}}  \dttsym{)}  \dttsym{,}  \dttnt{n_{{\mathrm{2}}}} \,  \preccurlyeq_{ \dttnt{p} }  \, \dttnt{m'}  \dttsym{,}  \mathsf{raise} \, \dttsym{(}  \dttnt{n_{{\mathrm{1}}}}  \dttsym{,}  \dttnt{n_{{\mathrm{2}}}}  \dttsym{,}  G''_{{\mathrm{1}}}  \dttsym{)}  \vdash  \dttnt{n_{{\mathrm{1}}}} \,  \preccurlyeq^*_{ \dttnt{p} }  \, \dttnt{m'}$ where $G_{{\mathrm{1}}} \equiv G'_{{\mathrm{1}}}  \dttsym{,}  \dttnt{n_{{\mathrm{1}}}} \,  \preccurlyeq_{ \dttnt{p} }  \, \dttnt{m'}  \dttsym{,}  G''_{{\mathrm{1}}}$. 
  This holds by the following derivation:
  \begin{center}
    \tiny
    \begin{math}
      $$\mprset{flushleft}
      \inferrule* [right=\ifrName{rel\_trans}] {
        G  \vdash  \dttnt{n_{{\mathrm{1}}}} \,  \preccurlyeq^*_{ \dttnt{p} }  \, \dttnt{n_{{\mathrm{2}}}}
        \\
        $$\mprset{flushleft}
        \inferrule* [right=\ifrName{rel\_ax}] {
          \ 
        }{G  \dttsym{,}  \mathsf{raise} \, \dttsym{(}  \dttnt{n_{{\mathrm{1}}}}  \dttsym{,}  \dttnt{n_{{\mathrm{2}}}}  \dttsym{,}  G'_{{\mathrm{1}}}  \dttsym{)}  \dttsym{,}  \dttnt{n_{{\mathrm{2}}}} \,  \preccurlyeq_{ \dttnt{p} }  \, \dttnt{m'}  \dttsym{,}  \mathsf{raise} \, \dttsym{(}  \dttnt{n_{{\mathrm{1}}}}  \dttsym{,}  \dttnt{n_{{\mathrm{2}}}}  \dttsym{,}  G''_{{\mathrm{1}}}  \dttsym{)}  \vdash  \dttnt{n_{{\mathrm{2}}}} \,  \preccurlyeq^*_{ \dttnt{p} }  \, \dttnt{m'}}
      }{G  \dttsym{,}  \mathsf{raise} \, \dttsym{(}  \dttnt{n_{{\mathrm{1}}}}  \dttsym{,}  \dttnt{n_{{\mathrm{2}}}}  \dttsym{,}  G'_{{\mathrm{1}}}  \dttsym{)}  \dttsym{,}  \dttnt{n_{{\mathrm{2}}}} \,  \preccurlyeq_{ \dttnt{p} }  \, \dttnt{m'}  \dttsym{,}  \mathsf{raise} \, \dttsym{(}  \dttnt{n_{{\mathrm{1}}}}  \dttsym{,}  \dttnt{n_{{\mathrm{2}}}}  \dttsym{,}  G''_{{\mathrm{1}}}  \dttsym{)}  \vdash  \dttnt{n_{{\mathrm{1}}}} \,  \preccurlyeq^*_{ \dttnt{p} }  \, \dttnt{m'}}
    \end{math}
  \end{center}      
  Now suppose $\dttnt{p'} \equiv  \bar{  \dttnt{p}  } $.  if $\dttnt{m'} \not\equiv \dttnt{n_{{\mathrm{1}}}}$, then clearly, we obtain
  our result.  Consider the case where $\dttnt{m'} \equiv \dttnt{n_{{\mathrm{1}}}}$.  Then it suffices to show
  $G  \dttsym{,}  \mathsf{raise} \, \dttsym{(}  \dttnt{n_{{\mathrm{1}}}}  \dttsym{,}  \dttnt{n_{{\mathrm{2}}}}  \dttsym{,}  G'_{{\mathrm{1}}}  \dttsym{)}  \dttsym{,}  \dttnt{m} \,  \preccurlyeq_{  \bar{  \dttnt{p}  }  }  \, \dttnt{n_{{\mathrm{2}}}}  \dttsym{,}  \mathsf{raise} \, \dttsym{(}  \dttnt{n_{{\mathrm{1}}}}  \dttsym{,}  \dttnt{n_{{\mathrm{2}}}}  \dttsym{,}  G''_{{\mathrm{1}}}  \dttsym{)}  \vdash  \dttnt{m} \,  \preccurlyeq^*_{  \bar{  \dttnt{p}  }  }  \, \dttnt{n_{{\mathrm{1}}}}$ 
  where $G_{{\mathrm{1}}} \equiv G'_{{\mathrm{1}}}  \dttsym{,}  \dttnt{m} \,  \preccurlyeq_{  \bar{  \dttnt{p}  }  }  \, \dttnt{n_{{\mathrm{1}}}}  \dttsym{,}  G''_{{\mathrm{1}}}$. This holds by the following derivation:
  \begin{center}
    \tiny
    \begin{math}
      $$\mprset{flushleft}
      \inferrule* [right={\tiny \ifrName{rel\_trans}}] {                       
        $$\mprset{flushleft}
        \inferrule* [right={\tiny \ifrName{rel\_ax}}] {
          \ 
        }{G  \dttsym{,}  \mathsf{raise} \, \dttsym{(}  \dttnt{n_{{\mathrm{1}}}}  \dttsym{,}  \dttnt{n_{{\mathrm{2}}}}  \dttsym{,}  G'_{{\mathrm{1}}}  \dttsym{)}  \dttsym{,}  \dttnt{m} \,  \preccurlyeq_{  \bar{  \dttnt{p}  }  }  \, \dttnt{n_{{\mathrm{2}}}}  \dttsym{,}  \mathsf{raise} \, \dttsym{(}  \dttnt{n_{{\mathrm{1}}}}  \dttsym{,}  \dttnt{n_{{\mathrm{2}}}}  \dttsym{,}  G''_{{\mathrm{1}}}  \dttsym{)}  \vdash  \dttnt{m} \,  \preccurlyeq^*_{  \bar{  \dttnt{p}  }  }  \, \dttnt{n_{{\mathrm{2}}}}}
        \\
        $$\mprset{flushleft}
        \inferrule* [right={\tiny \ifrName{rel\_flip}}] {
          G  \vdash  \dttnt{n_{{\mathrm{1}}}} \,  \preccurlyeq^*_{ \dttnt{p} }  \, \dttnt{n_{{\mathrm{2}}}}
        }{G  \vdash  \dttnt{n_{{\mathrm{2}}}} \,  \preccurlyeq^*_{  \bar{  \dttnt{p}  }  }  \, \dttnt{n_{{\mathrm{1}}}}}
      }{G  \dttsym{,}  \mathsf{raise} \, \dttsym{(}  \dttnt{n_{{\mathrm{1}}}}  \dttsym{,}  \dttnt{n_{{\mathrm{2}}}}  \dttsym{,}  G'_{{\mathrm{1}}}  \dttsym{)}  \dttsym{,}  \dttnt{m} \,  \preccurlyeq_{  \bar{  \dttnt{p}  }  }  \, \dttnt{n_{{\mathrm{2}}}}  \dttsym{,}  \mathsf{raise} \, \dttsym{(}  \dttnt{n_{{\mathrm{1}}}}  \dttsym{,}  \dttnt{n_{{\mathrm{2}}}}  \dttsym{,}  G''_{{\mathrm{1}}}  \dttsym{)}  \vdash  \dttnt{m} \,  \preccurlyeq^*_{  \bar{  \dttnt{p}  }  }  \, \dttnt{n_{{\mathrm{1}}}}}
    \end{math}
  \end{center}      
  
\item[\cW] 
\[
      \mprset{flushleft}
      \inferrule* [right=\ifrName{refl}] {
        \ 
      }{G  \dttsym{,}  G_{{\mathrm{1}}}  \vdash  \dttnt{m} \,  \preccurlyeq^*_{ \dttnt{p'} }  \, \dttnt{m}}
      \leqno{\raise 8 pt\hbox{\textbf{Case}}}
\]
  Note that in this case $\dttnt{m'} \equiv \dttnt{m}$.  Our result follows from simply
  an application of the $\ifrName{rel\_refl}$ rule.

\item[\cW] 
\[
      \mprset{flushleft}
      \inferrule* [right=\ifrName{rel\_trans}] {
         G  \dttsym{,}  G_{{\mathrm{1}}}  \vdash  \dttnt{m} \,  \preccurlyeq^*_{ \dttnt{p'} }  \, \dttnt{m''}  \qquad  G  \dttsym{,}  G_{{\mathrm{1}}}  \vdash  \dttnt{m''} \,  \preccurlyeq^*_{ \dttnt{p'} }  \, \dttnt{m'} 
      }{G  \dttsym{,}  G_{{\mathrm{1}}}  \vdash  \dttnt{m} \,  \preccurlyeq^*_{ \dttnt{p'} }  \, \dttnt{m'}}
      \leqno{\raise 8 pt\hbox{\textbf{Case}}}
\]
  This case holds by two applications of the induction hypothesis followed by
  applying the $\ifrName{rel\_trans}$ rule.
  
\item[\cW] 
\[
      \mprset{flushleft}
      \inferrule* [right=\ifrName{flip}] {
        G  \dttsym{,}  G_{{\mathrm{1}}}  \vdash  \dttnt{m'} \,  \preccurlyeq^*_{  \bar{  \dttnt{p'}  }  }  \, \dttnt{m}
      }{G  \dttsym{,}  G_{{\mathrm{1}}}  \vdash  \dttnt{m} \,  \preccurlyeq^*_{ \dttnt{p'} }  \, \dttnt{m'}}
      \leqno{\raise 8 pt\hbox{\textbf{Case}}}
\]
  This case holds by an application of the induction hypothesis followed by
  applying the $\ifrName{rel\_flip}$ rule.

\end{description}

\subsection{Proof of Lemma~\ref{lemma:raising_the_lower_bound_logically}: Raising the Lower Bound Logically}
\label{subsec:proof_of_raising_the_lower_bound_logically}
This is a proof by induction on the form of $G  \dttsym{,}  G_{{\mathrm{1}}}  \dttsym{,}  G'  \dttsym{;}  \Gamma  \vdash  \dttnt{p} \, \dttnt{A}  \mathbin{@}  \dttnt{n}$.  We assume with out loss of generality that
$\dttnt{n_{{\mathrm{1}}}} \, \in \, \dttsym{\mbox{$\mid$}}  G_{{\mathrm{1}}}  \dttsym{\mbox{$\mid$}}$, and that $\dttnt{n_{{\mathrm{1}}}} \not\equiv \dttnt{n_{{\mathrm{2}}}}$.  
If this is not the case then $\mathsf{raise} \, \dttsym{(}  \dttnt{n_{{\mathrm{1}}}}  \dttsym{,}  \dttnt{n_{{\mathrm{2}}}}  \dttsym{,}  G_{{\mathrm{1}}}  \dttsym{)} = G_{{\mathrm{1}}}$, and the result holds trivially.

\begin{description}
\item[\cW] 
\[
      \mprset{flushleft}
      \inferrule* [right=\ifrName{ax}] {
        G  \dttsym{,}  G_{{\mathrm{1}}}  \dttsym{,}  G'  \vdash  \dttnt{n'} \,  \preccurlyeq^*_{ \dttnt{p} }  \, \dttnt{n}
      }{G  \dttsym{,}  G_{{\mathrm{1}}}  \dttsym{,}  G'  \dttsym{;}  \Gamma  \dttsym{,}  \dttnt{p} \, \dttnt{A}  \mathbin{@}  \dttnt{n'}  \vdash  \dttnt{p} \, \dttnt{A}  \mathbin{@}  \dttnt{n}}
      \leqno{\raise 8 pt\hbox{\textbf{Case}}}
\]
  Clearly, if $G  \dttsym{,}  G_{{\mathrm{1}}}  \dttsym{,}  G'  \vdash  \dttnt{n'} \,  \preccurlyeq^*_{ \dttnt{p} }  \, \dttnt{n}$, then $G  \dttsym{,}  G'  \dttsym{,}  G_{{\mathrm{1}}}  \vdash  \dttnt{n'} \,  \preccurlyeq^*_{ \dttnt{p} }  \, \dttnt{n}$.
  Thus, this case follows by raising the lower bound (Lemma~\ref{lemma:raise_lower}), and applying
  the $\ifrName{ax}$ rule.

\item[\cW] 
\[
      \mprset{flushleft}
      \inferrule* [right=\ifrName{unit}] {
        \ 
      }{G  \dttsym{,}  G_{{\mathrm{1}}}  \dttsym{,}  G'  \dttsym{;}  \Gamma  \vdash  \dttnt{p} \,  \langle  \dttnt{p} \rangle   \mathbin{@}  \dttnt{n}}
      \leqno{\raise 8 pt\hbox{\textbf{Case}}}
\]
  Trivial.

\item[\cW] 
\[
      \mprset{flushleft}
      \inferrule* [right=\ifrName{and}] {
         G  \dttsym{,}  G_{{\mathrm{1}}}  \dttsym{,}  G'  \dttsym{;}  \Gamma  \vdash  \dttnt{p} \, \dttnt{A_{{\mathrm{1}}}}  \mathbin{@}  \dttnt{n}  \qquad  G  \dttsym{,}  G_{{\mathrm{1}}}  \dttsym{,}  G'  \dttsym{;}  \Gamma  \vdash  \dttnt{p} \, \dttnt{A_{{\mathrm{2}}}}  \mathbin{@}  \dttnt{n} 
      }{G  \dttsym{,}  G_{{\mathrm{1}}}  \dttsym{,}  G'  \dttsym{;}  \Gamma  \vdash  \dttnt{p} \, \dttsym{(}   \dttnt{A_{{\mathrm{1}}}}  \ndwedge{ \dttnt{p} }  \dttnt{A_{{\mathrm{2}}}}   \dttsym{)}  \mathbin{@}  \dttnt{n}}
      \leqno{\raise 8 pt\hbox{\textbf{Case}}}
\]
  This case holds by two applications of the induction hypothesis, and then applying
  the $\ifrName{and}$ rule.

\item[\cW] 
\[
      \mprset{flushleft}
      \inferrule* [right=\ifrName{andBar}] {
        G  \dttsym{,}  G_{{\mathrm{1}}}  \dttsym{,}  G'  \dttsym{;}  \Gamma  \vdash  \dttnt{p} \,  \dttnt{A} _{ \dttnt{d} }   \mathbin{@}  \dttnt{n}
      }{G  \dttsym{,}  G_{{\mathrm{1}}}  \dttsym{,}  G'  \dttsym{;}  \Gamma  \vdash  \dttnt{p} \, \dttsym{(}   \dttnt{A_{{\mathrm{1}}}}  \ndwedge{  \bar{  \dttnt{p}  }  }  \dttnt{A_{{\mathrm{2}}}}   \dttsym{)}  \mathbin{@}  \dttnt{n}}
      \leqno{\raise 8 pt\hbox{\textbf{Case}}}
\]
  Similar to the previous case.

\item[\cW] 
\[
      \mprset{flushleft}
      \inferrule* [right=\ifrName{imp}] {
        \dttnt{n'} \, \not\in \, \dttsym{\mbox{$\mid$}}  G  \dttsym{,}  G_{{\mathrm{1}}}  \dttsym{,}  G'  \dttsym{\mbox{$\mid$}}  \dttsym{,}  \dttsym{\mbox{$\mid$}}  \Gamma  \dttsym{\mbox{$\mid$}}
        \\\\
            \dttsym{(}  G  \dttsym{,}  G_{{\mathrm{1}}}  \dttsym{,}  G'  \dttsym{,}  \dttnt{n} \,  \preccurlyeq_{ \dttnt{p} }  \, \dttnt{n'}  \dttsym{)}  \dttsym{;}  \Gamma  \dttsym{,}  \dttnt{p} \, \dttnt{A_{{\mathrm{1}}}}  \mathbin{@}  \dttnt{n'}  \vdash  \dttnt{p} \, \dttnt{A_{{\mathrm{2}}}}  \mathbin{@}  \dttnt{n'}
      }{G  \dttsym{,}  G_{{\mathrm{1}}}  \dttsym{,}  G'  \dttsym{;}  \Gamma  \vdash  \dttnt{p} \, \dttsym{(}   \dttnt{A_{{\mathrm{1}}}}  \ndto{ \dttnt{p} }  \dttnt{A_{{\mathrm{2}}}}   \dttsym{)}  \mathbin{@}  \dttnt{n}}
      \leqno{\raise 8 pt\hbox{\textbf{Case}}}
\]
  Since we know $\dttnt{n_{{\mathrm{1}}}} \not\equiv \dttnt{n_{{\mathrm{2}}}}$, then by Lemma~\ref{lemma:graph_node_containment} we know
  $\dttnt{n_{{\mathrm{1}}}},\dttnt{n_{{\mathrm{2}}}} \in \dttsym{\mbox{$\mid$}}  G  \dttsym{,}  G'  \dttsym{\mbox{$\mid$}}$. Thus, $\dttnt{n'} \not\equiv \dttnt{n_{{\mathrm{1}}}} \not\equiv \dttnt{n_{{\mathrm{2}}}}$.  Now by the
  induction hypothesis we know $\dttsym{(}  G  \dttsym{,}  \mathsf{raise} \, \dttsym{(}  \dttnt{n_{{\mathrm{1}}}}  \dttsym{,}  \dttnt{n_{{\mathrm{2}}}}  \dttsym{,}  G_{{\mathrm{1}}}  \dttsym{)}  \dttsym{,}  G'  \dttsym{,}  \dttnt{n} \,  \preccurlyeq_{ \dttnt{p} }  \, \dttnt{n'}  \dttsym{)}  \dttsym{;}  \Gamma  \dttsym{,}  \dttnt{p} \, \dttnt{A_{{\mathrm{1}}}}  \mathbin{@}  \dttnt{n'}  \vdash  \dttnt{p} \, \dttnt{A_{{\mathrm{2}}}}  \mathbin{@}  \dttnt{n'}$.
  This case then follows by the application of the $\ifrName{imp}$ rule to the former.

\item[\cW] 
\[
      \mprset{flushleft}
      \inferrule* [right=\ifrName{impBar}] {
        G  \dttsym{,}  G_{{\mathrm{1}}}  \dttsym{,}  G'  \vdash  \dttnt{n} \,  \preccurlyeq^*_{  \bar{  \dttnt{p}  }  }  \, \dttnt{n'}
        \\\\
            G  \dttsym{,}  G_{{\mathrm{1}}}  \dttsym{,}  G'  \dttsym{;}  \Gamma  \vdash   \bar{  \dttnt{p}  }  \, \dttnt{A_{{\mathrm{1}}}}  \mathbin{@}  \dttnt{n'}  
            \\
              G  \dttsym{,}  G_{{\mathrm{1}}}  \dttsym{,}  G'  \dttsym{;}  \Gamma  \vdash  \dttnt{p} \, \dttnt{A_{{\mathrm{2}}}}  \mathbin{@}  \dttnt{n'}
      }{G  \dttsym{,}  G_{{\mathrm{1}}}  \dttsym{,}  G'  \dttsym{;}  \Gamma  \vdash  \dttnt{p} \, \dttsym{(}   \dttnt{A_{{\mathrm{1}}}}  \ndto{  \bar{  \dttnt{p}  }  }  \dttnt{A_{{\mathrm{2}}}}   \dttsym{)}  \mathbin{@}  \dttnt{n}}
      \leqno{\raise 8 pt\hbox{\textbf{Case}}}
\]
  Clearly, $G  \dttsym{,}  G_{{\mathrm{1}}}  \dttsym{,}  G'  \vdash  \dttnt{n} \,  \preccurlyeq^*_{  \bar{  \dttnt{p}  }  }  \, \dttnt{n'}$ implies $G  \dttsym{,}  G'  \dttsym{,}  G_{{\mathrm{1}}}  \vdash  \dttnt{n} \,  \preccurlyeq^*_{  \bar{  \dttnt{p}  }  }  \, \dttnt{n'}$, and
  by raising the lower bound (Lemma~\ref{lemma:raise_lower}) we know $G  \dttsym{,}  G'  \dttsym{,}  \mathsf{raise} \, \dttsym{(}  \dttnt{n_{{\mathrm{1}}}}  \dttsym{,}  \dttnt{n_{{\mathrm{2}}}}  \dttsym{,}  G_{{\mathrm{1}}}  \dttsym{)}  \vdash  \dttnt{n} \,  \preccurlyeq^*_{  \bar{  \dttnt{p}  }  }  \, \dttnt{n'}$,
  which then implies $G  \dttsym{,}  \mathsf{raise} \, \dttsym{(}  \dttnt{n_{{\mathrm{1}}}}  \dttsym{,}  \dttnt{n_{{\mathrm{2}}}}  \dttsym{,}  G_{{\mathrm{1}}}  \dttsym{)}  \dttsym{,}  G'  \vdash  \dttnt{n} \,  \preccurlyeq^*_{  \bar{  \dttnt{p}  }  }  \, \dttnt{n'}$.  Thus, this case follows from applying
  \textsc{impBar} to the application of the induction hypothesis to each premise and \\
  $G  \dttsym{,}  \mathsf{raise} \, \dttsym{(}  \dttnt{n_{{\mathrm{1}}}}  \dttsym{,}  \dttnt{n_{{\mathrm{2}}}}  \dttsym{,}  G_{{\mathrm{1}}}  \dttsym{)}  \dttsym{,}  G'  \vdash  \dttnt{n} \,  \preccurlyeq^*_{  \bar{  \dttnt{p}  }  }  \, \dttnt{n'}$.

\item[\cW] 
\[
      \mprset{flushleft}
      \inferrule* [right=\ifrName{axCut}] {
          \dttnt{p} \, \dttnt{T'}  \mathbin{@}  \dttnt{n'}  \in  \Gamma   \qquad  G  \dttsym{,}  G_{{\mathrm{1}}}  \dttsym{,}  G'  \dttsym{;}  \Gamma  \dttsym{,}   \bar{  \dttnt{p}  }  \, \dttnt{T}  \mathbin{@}  \dttnt{n}  \vdash   \bar{  \dttnt{p}  }  \, \dttnt{T'}  \mathbin{@}  \dttnt{n'} 
      }{G  \dttsym{,}  G_{{\mathrm{1}}}  \dttsym{,}  G'  \dttsym{;}  \Gamma  \vdash  \dttnt{p} \, \dttnt{T}  \mathbin{@}  \dttnt{n}}
      \leqno{\raise 8 pt\hbox{\textbf{Case}}}
\]
  This case follows by a simple application of the induction hypothesis, and then reapplying
  the rule.

\item[\cW] 
\[
      \mprset{flushleft}
      \inferrule* [right=\ifrName{axCutBar}] {
           \bar{  \dttnt{p}  }  \, \dttnt{T'}  \mathbin{@}  \dttnt{n'}  \in  \Gamma   \qquad  G  \dttsym{,}  G_{{\mathrm{1}}}  \dttsym{,}  G'  \dttsym{;}  \Gamma  \dttsym{,}   \bar{  \dttnt{p}  }  \, \dttnt{T}  \mathbin{@}  \dttnt{n}  \vdash  \dttnt{p} \, \dttnt{T'}  \mathbin{@}  \dttnt{n'} 
      }{G  \dttsym{,}  G_{{\mathrm{1}}}  \dttsym{,}  G'  \dttsym{;}  \Gamma  \vdash  \dttnt{p} \, \dttnt{T}  \mathbin{@}  \dttnt{n}}
      \leqno{\raise 8 pt\hbox{\textbf{Case}}}
\]
Similar to the previous case.
\end{description}

\subsection{Proof of Lemma~\ref{lemma:genmono}: General Monotonicity}
\label{subsec:proof_of_general_monotonicity}
This is a proof by induction on the form of $G  \dttsym{;}   \bar{  \dttnt{p_{{\mathrm{1}}}}  }  \, \dttnt{A_{{\mathrm{1}}}}  \mathbin{@}  \dttnt{n_{{\mathrm{1}}}}  \dttsym{,} \, ... \, \dttsym{,}   \bar{  \dttnt{p_{\dttmv{i}}}  }  \, \dttnt{A_{\dttmv{i}}}  \mathbin{@}  \dttnt{n_{\dttmv{i}}}  \vdash  \dttnt{p} \, \dttnt{B}  \mathbin{@}  \dttnt{m}$.
We assume without loss of generality that all of $\dttnt{n_{{\mathrm{1}}}}, \dttnt{n'_{{\mathrm{1}}}} \ldots , \dttnt{n_{\dttmv{i}}}, \dttnt{n'_{\dttmv{i}}}$ are unique.  Thus, they
are all members of $\dttsym{\mbox{$\mid$}}  G  \dttsym{\mbox{$\mid$}}$ by Lemma~\ref{lemma:graph_node_containment}.

\begin{description}
\item[\cW] 
  \[
      \mprset{flushleft}
      \inferrule* [right=\ifrName{ax}] {
        G  \vdash  \dttnt{n_{\dttmv{j}}} \,  \preccurlyeq^*_{  \bar{  \dttnt{p_{\dttmv{j}}}  }  }  \, \dttnt{m}
      }{G  \dttsym{;}   \bar{  \dttnt{p_{{\mathrm{1}}}}  }  \, \dttnt{A_{{\mathrm{1}}}}  \mathbin{@}  \dttnt{n_{{\mathrm{1}}}}  \dttsym{,} \, ... \, \dttsym{,}   \bar{  \dttnt{p_{\dttmv{i}}}  }  \, \dttnt{A_{\dttmv{i}}}  \mathbin{@}  \dttnt{n_{\dttmv{i}}}  \vdash   \bar{  \dttnt{p_{\dttmv{j}}}  }  \, \dttnt{A_{\dttmv{j}}}  \mathbin{@}  \dttnt{m}}
      \leqno{\raise 8 pt\hbox{\textbf{Case}}}
\]
  It must be the case that $\dttnt{p} \, \dttnt{B}  \mathbin{@}  \dttnt{m} \equiv  \bar{  \dttnt{p_{\dttmv{j}}}  }  \, \dttnt{A_{\dttmv{j}}}  \mathbin{@}  \dttnt{m}$ for some $1 \leq j \leq i$.
  In addition, we know $G  \vdash  \dttnt{n_{\dttmv{j}}} \,  \preccurlyeq^*_{ \dttnt{p_{\dttmv{j}}} }  \, \dttnt{n'_{\dttmv{j}}}$, $G  \vdash  \dttnt{n_{\dttmv{j}}} \,  \preccurlyeq^*_{  \bar{  \dttnt{p_{\dttmv{j}}}  }  }  \, \dttnt{m}$, and $G  \vdash  \dttnt{m} \,  \preccurlyeq^*_{  \bar{  \dttnt{p_{\dttmv{j}}}  }  }  \, \dttnt{m'}$.
  It suffices to show $G  \dttsym{;}    \bar{  \dttnt{p_{{\mathrm{1}}}}  }    \dttnt{A_{{\mathrm{1}}}}  @  \dttnt{n'_{{\mathrm{1}}}}  , \ldots ,   \bar{  \dttnt{p_{\dttmv{i}}}  }    \dttnt{A_{\dttmv{i}}}  @  \dttnt{n'_{\dttmv{i}}}   \vdash   \bar{  \dttnt{p_{\dttmv{j}}}  }  \, \dttnt{A_{\dttmv{j}}}  \mathbin{@}  \dttnt{m'}$, but to obtain this result it suffices to show
  that $G  \vdash  \dttnt{n'_{\dttmv{j}}} \,  \preccurlyeq_{  \bar{  \dttnt{p_{\dttmv{j}}}  }  }  \, \dttnt{m'}$, but this holds by first using $\textsc{rel\_flip}$ to obtain
  $G  \vdash  \dttnt{n'_{\dttmv{j}}} \,  \preccurlyeq^*_{  \bar{  \dttnt{p_{\dttmv{j}}}  }  }  \, \dttnt{n_{\dttmv{j}}}$ followed by two applications of transitivity. 
  
\item[\cW] 
  \[
      \mprset{flushleft}
      \inferrule* [right=\ifrName{unit}] {
        \ 
      }{G  \dttsym{;}   \bar{  \dttnt{p_{{\mathrm{1}}}}  }  \, \dttnt{A_{{\mathrm{1}}}}  \mathbin{@}  \dttnt{n_{{\mathrm{1}}}}  \dttsym{,} \, ... \, \dttsym{,}   \bar{  \dttnt{p_{\dttmv{i}}}  }  \, \dttnt{A_{\dttmv{i}}}  \mathbin{@}  \dttnt{n_{\dttmv{i}}}  \vdash  \dttnt{p} \,  \langle  \dttnt{p} \rangle   \mathbin{@}  \dttnt{m_{{\mathrm{1}}}}}
      \leqno{\raise 8 pt\hbox{\textbf{Case}}}
\]
  Trivial.

\item[\cW] 
  \[
      \mprset{flushleft}
      \inferrule* [right=\ifrName{and}] {
        G  \dttsym{;}   \bar{  \dttnt{p_{{\mathrm{1}}}}  }  \, \dttnt{A_{{\mathrm{1}}}}  \mathbin{@}  \dttnt{n_{{\mathrm{1}}}}  \dttsym{,} \, ... \, \dttsym{,}   \bar{  \dttnt{p_{\dttmv{i}}}  }  \, \dttnt{A_{\dttmv{i}}}  \mathbin{@}  \dttnt{n_{\dttmv{i}}}  \vdash  \dttnt{p} \, \dttnt{B_{{\mathrm{1}}}}  \mathbin{@}  \dttnt{m} 
        \\\\
            G  \dttsym{;}   \bar{  \dttnt{p_{{\mathrm{1}}}}  }  \, \dttnt{A_{{\mathrm{1}}}}  \mathbin{@}  \dttnt{n_{{\mathrm{1}}}}  \dttsym{,} \, ... \, \dttsym{,}   \bar{  \dttnt{p_{\dttmv{i}}}  }  \, \dttnt{A_{\dttmv{i}}}  \mathbin{@}  \dttnt{n_{\dttmv{i}}}  \vdash  \dttnt{p} \, \dttnt{B_{{\mathrm{2}}}}  \mathbin{@}  \dttnt{m}
      }{G  \dttsym{;}   \bar{  \dttnt{p_{{\mathrm{1}}}}  }  \, \dttnt{A_{{\mathrm{1}}}}  \mathbin{@}  \dttnt{n_{{\mathrm{1}}}}  \dttsym{,} \, ... \, \dttsym{,}   \bar{  \dttnt{p_{\dttmv{i}}}  }  \, \dttnt{A_{\dttmv{i}}}  \mathbin{@}  \dttnt{n_{\dttmv{i}}}  \vdash  \dttnt{p} \, \dttsym{(}   \dttnt{B_{{\mathrm{1}}}}  \ndwedge{ \dttnt{p} }  \dttnt{B_{{\mathrm{2}}}}   \dttsym{)}  \mathbin{@}  \dttnt{m}}
      \leqno{\raise 8 pt\hbox{\textbf{Case}}}
\]
  This case follows easily by applying the induction hypothesis to each premise and then
  applying the $\ifrName{and}$ rule.

\item[\cW] 
  \[
      \mprset{flushleft}
      \inferrule* [right=\ifrName{andBar}] {
        G  \dttsym{;}   \bar{  \dttnt{p_{{\mathrm{1}}}}  }  \, \dttnt{A_{{\mathrm{1}}}}  \mathbin{@}  \dttnt{n_{{\mathrm{1}}}}  \dttsym{,} \, ... \, \dttsym{,}   \bar{  \dttnt{p_{\dttmv{i}}}  }  \, \dttnt{A_{\dttmv{i}}}  \mathbin{@}  \dttnt{n_{\dttmv{i}}}  \vdash  \dttnt{p} \,  \dttnt{B} _{ \dttnt{d} }   \mathbin{@}  \dttnt{m}
      }{G  \dttsym{;}   \bar{  \dttnt{p_{{\mathrm{1}}}}  }  \, \dttnt{A_{{\mathrm{1}}}}  \mathbin{@}  \dttnt{n_{{\mathrm{1}}}}  \dttsym{,} \, ... \, \dttsym{,}   \bar{  \dttnt{p_{\dttmv{i}}}  }  \, \dttnt{A_{\dttmv{i}}}  \mathbin{@}  \dttnt{n_{\dttmv{i}}}  \vdash  \dttnt{p} \, \dttsym{(}   \dttnt{B_{{\mathrm{1}}}}  \ndwedge{  \bar{  \dttnt{p}  }  }  \dttnt{B_{{\mathrm{2}}}}   \dttsym{)}  \mathbin{@}  \dttnt{m}}
      \leqno{\raise 8 pt\hbox{\textbf{Case}}}
\]
  This case follows easily by the induction hypothesis and then applying $\ifrName{andBar}$.      

\item[\cW] 
  \[
      \mprset{flushleft}
      \inferrule* [right=\ifrName{imp}] {
        \dttnt{n'} \, \not\in \, \dttsym{\mbox{$\mid$}}  G  \dttsym{\mbox{$\mid$}}  \dttsym{,}  \dttsym{\mbox{$\mid$}}   \bar{  \dttnt{p_{{\mathrm{1}}}}  }  \, \dttnt{A_{{\mathrm{1}}}}  \mathbin{@}  \dttnt{n_{{\mathrm{1}}}}  \dttsym{,} \, ... \, \dttsym{,}   \bar{  \dttnt{p_{\dttmv{i}}}  }  \, \dttnt{A_{\dttmv{i}}}  \mathbin{@}  \dttnt{n_{\dttmv{i}}}  \dttsym{\mbox{$\mid$}}
        \\\\
            \dttsym{(}  G  \dttsym{,}  \dttnt{m_{{\mathrm{1}}}} \,  \preccurlyeq_{ \dttnt{p} }  \, \dttnt{n'}  \dttsym{)}  \dttsym{;}   \bar{  \dttnt{p_{{\mathrm{1}}}}  }  \, \dttnt{A_{{\mathrm{1}}}}  \mathbin{@}  \dttnt{n_{{\mathrm{1}}}}  \dttsym{,} \, ... \, \dttsym{,}   \bar{  \dttnt{p_{\dttmv{i}}}  }  \, \dttnt{A_{\dttmv{i}}}  \mathbin{@}  \dttnt{n_{\dttmv{i}}}  \dttsym{,}  \dttnt{p} \, \dttnt{B_{{\mathrm{1}}}}  \mathbin{@}  \dttnt{n'}  \vdash  \dttnt{p} \, \dttnt{B_{{\mathrm{2}}}}  \mathbin{@}  \dttnt{n'}
      }{G  \dttsym{;}   \bar{  \dttnt{p_{{\mathrm{1}}}}  }  \, \dttnt{A_{{\mathrm{1}}}}  \mathbin{@}  \dttnt{n_{{\mathrm{1}}}}  \dttsym{,} \, ... \, \dttsym{,}   \bar{  \dttnt{p_{\dttmv{i}}}  }  \, \dttnt{A_{\dttmv{i}}}  \mathbin{@}  \dttnt{n_{\dttmv{i}}}  \vdash  \dttnt{p} \, \dttsym{(}   \dttnt{B_{{\mathrm{1}}}}  \ndto{ \dttnt{p} }  \dttnt{B_{{\mathrm{2}}}}   \dttsym{)}  \mathbin{@}  \dttnt{m}}
      \leqno{\raise 8 pt\hbox{\textbf{Case}}}
\]
  
  We know by assumption $G  \vdash  \dttnt{n_{{\mathrm{1}}}} \,  \preccurlyeq^*_{ \dttnt{p_{{\mathrm{1}}}} }  \, \dttnt{n'_{{\mathrm{1}}}}$, \ldots, $G  \vdash  \dttnt{n_{\dttmv{i}}} \,  \preccurlyeq^*_{ \dttnt{p_{\dttmv{i}}} }  \, \dttnt{n'_{\dttmv{i}}}$, and by graph weakening 
  (Lemma~\ref{lemma:graph_weakening})
  $G  \dttsym{,}  \dttnt{m} \,  \preccurlyeq_{ \dttnt{p} }  \, \dttnt{n'}  \vdash  \dttnt{n_{{\mathrm{1}}}} \,  \preccurlyeq^*_{ \dttnt{p_{{\mathrm{1}}}} }  \, \dttnt{n'_{{\mathrm{1}}}}$, \ldots, $G  \dttsym{,}  \dttnt{m} \,  \preccurlyeq_{ \dttnt{p} }  \, \dttnt{n'}  \vdash  \dttnt{n_{\dttmv{i}}} \,  \preccurlyeq^*_{ \dttnt{p_{\dttmv{i}}} }  \, \dttnt{n'_{\dttmv{i}}}$.  We also know by
  applying the $\ifrName{rel\_refl}$ rule that $G  \dttsym{,}  \dttnt{m} \,  \preccurlyeq_{ \dttnt{p} }  \, \dttnt{n'}  \vdash  \dttnt{n'} \,  \preccurlyeq^*_{  \bar{  \dttnt{p}  }  }  \, \dttnt{n'}$ and 
  $G  \dttsym{,}  \dttnt{m} \,  \preccurlyeq_{ \dttnt{p} }  \, \dttnt{n'}  \vdash  \dttnt{n'} \,  \preccurlyeq^*_{ \dttnt{p} }  \, \dttnt{n'}$.  
  Thus, by the induction hypothesis we know 
  $\dttsym{(}  G  \dttsym{,}  \dttnt{m} \,  \preccurlyeq_{ \dttnt{p} }  \, \dttnt{n'}  \dttsym{)}  \dttsym{;}    \bar{  \dttnt{p_{{\mathrm{1}}}}  }    \dttnt{A_{{\mathrm{1}}}}  @  \dttnt{n'_{{\mathrm{1}}}}  , \ldots ,   \bar{  \dttnt{p_{\dttmv{i}}}  }    \dttnt{A_{\dttmv{i}}}  @  \dttnt{n'_{\dttmv{i}}}   \dttsym{,}  \dttnt{p} \, \dttnt{B_{{\mathrm{1}}}}  \mathbin{@}  \dttnt{n'}  \vdash  \dttnt{p} \, \dttnt{B_{{\mathrm{2}}}}  \mathbin{@}  \dttnt{n'}$.
  Now we can raise the lower bound logically (Lemma~\ref{lemma:raising_the_lower_bound_logically}) with 
  $G_{{\mathrm{1}}} \equiv \dttnt{m} \,  \preccurlyeq_{ \dttnt{p} }  \, \dttnt{n'}$ and the assumption $G  \vdash  \dttnt{m} \,  \preccurlyeq^*_{ \dttnt{p} }  \, \dttnt{m'}$ to obtain \\
  $\dttsym{(}  G  \dttsym{,}  \mathsf{raise} \, \dttsym{(}  \dttnt{m}  \dttsym{,}  \dttnt{m'}  \dttsym{,}  \dttnt{m} \,  \preccurlyeq_{ \dttnt{p} }  \, \dttnt{n'}  \dttsym{)}  \dttsym{)}  \dttsym{;}    \bar{  \dttnt{p_{{\mathrm{1}}}}  }    \dttnt{A_{{\mathrm{1}}}}  @  \dttnt{n'_{{\mathrm{1}}}}  , \ldots ,   \bar{  \dttnt{p_{\dttmv{i}}}  }    \dttnt{A_{\dttmv{i}}}  @  \dttnt{n'_{\dttmv{i}}}   \dttsym{,}  \dttnt{p} \, \dttnt{B_{{\mathrm{1}}}}  \mathbin{@}  \dttnt{n'}  \vdash  \dttnt{p} \, \dttnt{B_{{\mathrm{2}}}}  \mathbin{@}  \dttnt{n'}$, but this is
  equivalent to $\dttsym{(}  G  \dttsym{,}  \dttnt{m} \,  \preccurlyeq_{ \dttnt{p} }  \, \dttnt{n'}  \dttsym{)}  \dttsym{;}    \bar{  \dttnt{p_{{\mathrm{1}}}}  }    \dttnt{A_{{\mathrm{1}}}}  @  \dttnt{n'_{{\mathrm{1}}}}  , \ldots ,   \bar{  \dttnt{p_{\dttmv{i}}}  }    \dttnt{A_{\dttmv{i}}}  @  \dttnt{n'_{\dttmv{i}}}   \dttsym{,}  \dttnt{p} \, \dttnt{B_{{\mathrm{1}}}}  \mathbin{@}  \dttnt{n'}  \vdash  \dttnt{p} \, \dttnt{B_{{\mathrm{2}}}}  \mathbin{@}  \dttnt{n'}$.  
  Finally, using the former, we obtain our result by applying the $\ifrName{imp}$ rule.

\item[\cW] 
  \[
      \mprset{flushleft}
      \inferrule* [right=\ifrName{impBar}] {
        G  \vdash  \dttnt{m} \,  \preccurlyeq^*_{  \bar{  \dttnt{p}  }  }  \, \dttnt{n'}
        \\\\
            G  \dttsym{;}   \bar{  \dttnt{p_{{\mathrm{1}}}}  }  \, \dttnt{A_{{\mathrm{1}}}}  \mathbin{@}  \dttnt{n_{{\mathrm{1}}}}  \dttsym{,} \, ... \, \dttsym{,}   \bar{  \dttnt{p_{\dttmv{i}}}  }  \, \dttnt{A_{\dttmv{i}}}  \mathbin{@}  \dttnt{n_{\dttmv{i}}}  \vdash   \bar{  \dttnt{p}  }  \, \dttnt{B_{{\mathrm{1}}}}  \mathbin{@}  \dttnt{n'} 
            \\
              G  \dttsym{;}   \bar{  \dttnt{p_{{\mathrm{1}}}}  }  \, \dttnt{A_{{\mathrm{1}}}}  \mathbin{@}  \dttnt{n_{{\mathrm{1}}}}  \dttsym{,} \, ... \, \dttsym{,}   \bar{  \dttnt{p_{\dttmv{i}}}  }  \, \dttnt{A_{\dttmv{i}}}  \mathbin{@}  \dttnt{n_{\dttmv{i}}}  \vdash  \dttnt{p} \, \dttnt{B_{{\mathrm{2}}}}  \mathbin{@}  \dttnt{n'}
      }{G  \dttsym{;}   \bar{  \dttnt{p_{{\mathrm{1}}}}  }  \, \dttnt{A_{{\mathrm{1}}}}  \mathbin{@}  \dttnt{n_{{\mathrm{1}}}}  \dttsym{,} \, ... \, \dttsym{,}   \bar{  \dttnt{p_{\dttmv{i}}}  }  \, \dttnt{A_{\dttmv{i}}}  \mathbin{@}  \dttnt{n_{\dttmv{i}}}  \vdash  \dttnt{p} \, \dttsym{(}   \dttnt{B_{{\mathrm{1}}}}  \ndto{  \bar{  \dttnt{p}  }  }  \dttnt{B_{{\mathrm{2}}}}   \dttsym{)}  \mathbin{@}  \dttnt{m}}
      \leqno{\raise 8 pt\hbox{\textbf{Case}}}
\]
  We can easily derive $G  \vdash  \dttnt{m'} \,  \preccurlyeq^*_{  \bar{  \dttnt{p}  }  }  \, \dttnt{n'}$ as follows:
  \[
      \mprset{flushleft}
      \inferrule* [right=\ifrName{rel\_flip}] {
        \mprset{flushleft}
        \inferrule* [right=\ifrName{rel\_trans}] {
          \mprset{flushleft}
          \inferrule* [right=\ifrName{rel\_flip}] {
            G  \vdash  \dttnt{m} \,  \preccurlyeq^*_{  \bar{  \dttnt{p}  }  }  \, \dttnt{n'}
          }{G  \vdash  \dttnt{n'} \,  \preccurlyeq^*_{ \dttnt{p} }  \, \dttnt{m}}
          \\
            G  \vdash  \dttnt{m} \,  \preccurlyeq^*_{ \dttnt{p} }  \, \dttnt{m'}
        }{G  \vdash  \dttnt{n'} \,  \preccurlyeq^*_{ \dttnt{p} }  \, \dttnt{m'}}
      }{G  \vdash  \dttnt{m'} \,  \preccurlyeq^*_{  \bar{  \dttnt{p}  }  }  \, \dttnt{n'}}
      \leqno{\raise 8 pt\hbox{\textbf{Case}}}
\]
  This case then follows by applying the induction hypothesis twice to both 
  $G  \dttsym{;}   \bar{  \dttnt{p_{{\mathrm{1}}}}  }  \, \dttnt{A_{{\mathrm{1}}}}  \mathbin{@}  \dttnt{n_{{\mathrm{1}}}}  \dttsym{,} \, ... \, \dttsym{,}   \bar{  \dttnt{p_{\dttmv{i}}}  }  \, \dttnt{A_{\dttmv{i}}}  \mathbin{@}  \dttnt{n_{\dttmv{i}}}  \vdash   \bar{  \dttnt{p}  }  \, \dttnt{B_{{\mathrm{1}}}}  \mathbin{@}  \dttnt{n'}$
  and
  $G  \dttsym{;}   \bar{  \dttnt{p_{{\mathrm{1}}}}  }  \, \dttnt{A_{{\mathrm{1}}}}  \mathbin{@}  \dttnt{n_{{\mathrm{1}}}}  \dttsym{,} \, ... \, \dttsym{,}   \bar{  \dttnt{p_{\dttmv{i}}}  }  \, \dttnt{A_{\dttmv{i}}}  \mathbin{@}  \dttnt{n_{\dttmv{i}}}  \vdash  \dttnt{p} \, \dttnt{B_{{\mathrm{2}}}}  \mathbin{@}  \dttnt{n'}$
  using the assumptions
  $G  \vdash  \dttnt{n_{{\mathrm{1}}}} \,  \preccurlyeq^*_{ \dttnt{p_{{\mathrm{1}}}} }  \, \dttnt{n'_{{\mathrm{1}}}}$, \ldots, $G  \vdash  \dttnt{n_{\dttmv{i}}} \,  \preccurlyeq^*_{ \dttnt{p_{\dttmv{i}}} }  \, \dttnt{n'_{\dttmv{i}}}$, 
  and the fact that we know $G  \vdash  \dttnt{n'} \,  \preccurlyeq^*_{ \dttnt{p} }  \, \dttnt{n'}$ and $G  \vdash  \dttnt{n'} \,  \preccurlyeq^*_{  \bar{  \dttnt{p}  }  }  \, \dttnt{n'}$.
  
\item[\cW] 
  \[
      \mprset{flushleft}
      \inferrule* [right=\ifrName{axCut}] {
          \bar{  \dttnt{p_{\dttmv{j}}}  }  \, \dttnt{A_{\dttmv{j}}}  \mathbin{@}  \dttnt{n_{\dttmv{j}}}  \in  \dttsym{(}   \bar{  \dttnt{p_{{\mathrm{1}}}}  }  \, \dttnt{A_{{\mathrm{1}}}}  \mathbin{@}  \dttnt{n_{{\mathrm{1}}}}  \dttsym{,} \, ... \, \dttsym{,}   \bar{  \dttnt{p_{\dttmv{i}}}  }  \, \dttnt{A_{\dttmv{i}}}  \mathbin{@}  \dttnt{n_{\dttmv{i}}}  \dttsym{)}  
        \\
          G  \dttsym{;}   \bar{  \dttnt{p_{{\mathrm{1}}}}  }  \, \dttnt{A_{{\mathrm{1}}}}  \mathbin{@}  \dttnt{n_{{\mathrm{1}}}}  \dttsym{,} \, ... \, \dttsym{,}   \bar{  \dttnt{p_{\dttmv{i}}}  }  \, \dttnt{A_{\dttmv{i}}}  \mathbin{@}  \dttnt{n_{\dttmv{i}}}  \dttsym{,}   \bar{  \dttnt{p}  }  \, \dttnt{B}  \mathbin{@}  \dttnt{m}  \vdash  \dttnt{p_{\dttmv{j}}} \, \dttnt{A_{\dttmv{j}}}  \mathbin{@}  \dttnt{n_{\dttmv{j}}}
      }{G  \dttsym{;}   \bar{  \dttnt{p_{{\mathrm{1}}}}  }  \, \dttnt{A_{{\mathrm{1}}}}  \mathbin{@}  \dttnt{n_{{\mathrm{1}}}}  \dttsym{,} \, ... \, \dttsym{,}   \bar{  \dttnt{p_{\dttmv{i}}}  }  \, \dttnt{A_{\dttmv{i}}}  \mathbin{@}  \dttnt{n_{\dttmv{i}}}  \vdash  \dttnt{p} \, \dttnt{B}  \mathbin{@}  \dttnt{m}}
      \leqno{\raise 8 pt\hbox{\textbf{Case}}}
\]
  We know by assumption that $G  \vdash  \dttnt{n_{{\mathrm{1}}}} \,  \preccurlyeq^*_{ \dttnt{p_{{\mathrm{1}}}} }  \, \dttnt{n'_{{\mathrm{1}}}}$, \ldots, $G  \vdash  \dttnt{n_{\dttmv{i}}} \,  \preccurlyeq^*_{ \dttnt{p_{\dttmv{i}}} }  \, \dttnt{n'_{\dttmv{i}}}$, and $G  \vdash  \dttnt{m} \,  \preccurlyeq^*_{ \dttnt{p} }  \, \dttnt{m'}$.
  In particular, we know $G  \vdash  \dttnt{n_{\dttmv{j}}} \,  \preccurlyeq^*_{ \dttnt{p_{\dttmv{j}}} }  \, \dttnt{n'_{\dttmv{j}}}$.  It is also the case that if 
  $  \bar{  \dttnt{p_{\dttmv{j}}}  }  \, \dttnt{A_{\dttmv{j}}}  \mathbin{@}  \dttnt{n_{\dttmv{j}}}  \in  \dttsym{(}   \bar{  \dttnt{p_{{\mathrm{1}}}}  }  \, \dttnt{A_{{\mathrm{1}}}}  \mathbin{@}  \dttnt{n_{{\mathrm{1}}}}  \dttsym{,} \, ... \, \dttsym{,}   \bar{  \dttnt{p_{\dttmv{i}}}  }  \, \dttnt{A_{\dttmv{i}}}  \mathbin{@}  \dttnt{n_{\dttmv{i}}}  \dttsym{)} $, then 
  $  \bar{  \dttnt{p_{\dttmv{j}}}  }  \, \dttnt{A_{\dttmv{j}}}  \mathbin{@}  \dttnt{n'_{\dttmv{j}}}  \in  \dttsym{(}    \bar{  \dttnt{p_{{\mathrm{1}}}}  }    \dttnt{A_{{\mathrm{1}}}}  @  \dttnt{n'}  , \ldots ,   \bar{  \dttnt{p_{\dttmv{i}}}  }    \dttnt{A_{\dttmv{i}}}  @  \dttnt{n'_{\dttmv{i}}}   \dttsym{)} $.  This case then follows by applying the induction
  hypothesis to $G  \dttsym{;}   \bar{  \dttnt{p_{{\mathrm{1}}}}  }  \, \dttnt{A_{{\mathrm{1}}}}  \mathbin{@}  \dttnt{n_{{\mathrm{1}}}}  \dttsym{,} \, ... \, \dttsym{,}   \bar{  \dttnt{p_{\dttmv{i}}}  }  \, \dttnt{A_{\dttmv{i}}}  \mathbin{@}  \dttnt{n_{\dttmv{i}}}  \dttsym{,}   \bar{  \dttnt{p}  }  \, \dttnt{B}  \mathbin{@}  \dttnt{m}  \vdash  \dttnt{p_{\dttmv{j}}} \, \dttnt{A_{\dttmv{j}}}  \mathbin{@}  \dttnt{n_{\dttmv{j}}}$, to obtain,
  $G  \dttsym{;}    \bar{  \dttnt{p_{{\mathrm{1}}}}  }    \dttnt{A_{{\mathrm{1}}}}  @  \dttnt{n'_{{\mathrm{1}}}}  , \ldots ,   \bar{  \dttnt{p_{\dttmv{i}}}  }    \dttnt{A_{\dttmv{i}}}  @  \dttnt{n'_{\dttmv{i}}}   \dttsym{,}   \bar{  \dttnt{p}  }  \, \dttnt{B}  \mathbin{@}  \dttnt{m'_{{\mathrm{1}}}}  \vdash  \dttnt{p_{\dttmv{j}}} \, \dttnt{A_{\dttmv{j}}}  \mathbin{@}  \dttnt{n'_{\dttmv{j}}}$, followed by applying the
  $\ifrName{axCut}$ rule.

\item[\cW] 
  \[
      \mprset{flushleft}
      \inferrule* [right=\ifrName{axCutBar}] {
          \bar{  \dttnt{p_{\dttmv{j}}}  }  \, \dttnt{A_{\dttmv{j}}}  \mathbin{@}  \dttnt{n_{\dttmv{j}}}  \in  \dttsym{(}   \bar{  \dttnt{p_{{\mathrm{1}}}}  }  \, \dttnt{A_{{\mathrm{1}}}}  \mathbin{@}  \dttnt{n_{{\mathrm{1}}}}  \dttsym{,} \, ... \, \dttsym{,}   \bar{  \dttnt{p_{\dttmv{i}}}  }  \, \dttnt{A_{\dttmv{i}}}  \mathbin{@}  \dttnt{n_{\dttmv{i}}}  \dttsym{)}  
        \\
          G  \dttsym{;}   \bar{  \dttnt{p_{{\mathrm{1}}}}  }  \, \dttnt{A_{{\mathrm{1}}}}  \mathbin{@}  \dttnt{n_{{\mathrm{1}}}}  \dttsym{,} \, ... \, \dttsym{,}   \bar{  \dttnt{p_{\dttmv{i}}}  }  \, \dttnt{A_{\dttmv{i}}}  \mathbin{@}  \dttnt{n_{\dttmv{i}}}  \dttsym{,}   \bar{  \dttnt{p}  }  \, \dttnt{B}  \mathbin{@}  \dttnt{m}  \vdash  \dttnt{p_{\dttmv{j}}} \, \dttnt{A_{\dttmv{j}}}  \mathbin{@}  \dttnt{n_{\dttmv{j}}}
      }{G  \dttsym{;}   \bar{  \dttnt{p_{{\mathrm{1}}}}  }  \, \dttnt{A_{{\mathrm{1}}}}  \mathbin{@}  \dttnt{n_{{\mathrm{1}}}}  \dttsym{,} \, ... \, \dttsym{,}   \bar{  \dttnt{p_{\dttmv{i}}}  }  \, \dttnt{A_{\dttmv{i}}}  \mathbin{@}  \dttnt{n_{\dttmv{i}}}  \vdash  \dttnt{p} \, \dttnt{B}  \mathbin{@}  \dttnt{m}}
      \leqno{\raise 8 pt\hbox{\textbf{Case}}}
\]
  Similar to the previous case.
\end{description}

\subsection{Proof of Lemma~\ref{lemma:containment-l-in-dil}: Containment of L in DIL}
\label{subsec:proof_of_containment_L_in_DIL}
This is a proof by induction on the form of the sequent $ \Gamma  \vdash_{ G }  \Delta $.

\begin{description}
\item[\cW] 
  \[
      \mprset{flushleft}
      \inferrule* [right=\ifrName{refl}] {
         \Gamma  \vdash_{  G  \Lsym{,}  \Lsym{(}  \Lmv{n}  \Lsym{,}  \Lmv{n}  \Lsym{)}  }  \Delta 
      }{ \Gamma  \vdash_{ G }  \Delta }
      \leqno{\raise 8 pt\hbox{\textbf{Case}}}      
\]
  We know by the induction hypothesis that every activation of $ \Gamma  \vdash_{  G  \Lsym{,}  \Lsym{(}  \Lmv{n}  \Lsym{,}  \Lmv{n}  \Lsym{)}  }  \Delta $ is derivable.  
  Suppose that $ \mathsf{D}( G  \dttsym{,}  \dttsym{(}  \dttnt{n}  \dttsym{,}  \dttnt{n}  \dttsym{)} )   \dttsym{;}   \mathsf{D}( \Gamma )^{ \dttsym{+} }   \dttsym{,}  \Gamma'  \vdash  \dttsym{+} \, \dttnt{A}  \mathbin{@}  \dttnt{n}$ is an arbitrary activation, where
  $ \mathsf{D}( \Delta )^{ \dttsym{-} }  \equiv  \mathsf{D}( \Delta_{{\mathrm{1}}} )^{ \dttsym{-} }   \dttsym{,}  \dttsym{-} \, \dttnt{A}  \mathbin{@}  \dttnt{n}  \dttsym{,}   \mathsf{D}( \Delta_{{\mathrm{2}}} )^{ \dttsym{-} } $ and $\Gamma' \equiv  \mathsf{D}( \Delta_{{\mathrm{1}}} )^{ \dttsym{-} }   \dttsym{,}   \mathsf{D}( \Delta_{{\mathrm{2}}} )^{ \dttsym{-} } $.
  This is equivalent to $ \mathsf{D}( G )   \dttsym{,}  \dttnt{n} \,  \preccurlyeq_{ \dttsym{+} }  \, \dttnt{n}  \dttsym{;}   \mathsf{D}( \Gamma )^{ \dttsym{+} }   \dttsym{,}  \Gamma'  \vdash  \dttsym{+} \, \dttnt{A}  \mathbin{@}  \dttnt{n}$, and 
  by the admissible rule for reflexivity (Lemma~\ref{lemma:reflexivity}) we have 
  $ \mathsf{D}( G )   \dttsym{;}   \mathsf{D}( \Gamma )^{ \dttsym{+} }   \dttsym{,}  \Gamma'  \vdash  \dttsym{+} \, \dttnt{A}  \mathbin{@}  \dttnt{n}$.

\item[\cW] 
  \[
      \mprset{flushleft}
      \inferrule* [right=\ifrName{Trans}] {
         \Lmv{n_{{\mathrm{1}}}}   G   \Lmv{n_{{\mathrm{2}}}} 
        \\\\
         \Lmv{n_{{\mathrm{2}}}}   G   \Lmv{n_{{\mathrm{3}}}} 
        \\\\
         \Gamma  \vdash_{  G  \Lsym{,}  \Lsym{(}  \Lmv{n_{{\mathrm{1}}}}  \Lsym{,}  \Lmv{n_{{\mathrm{3}}}}  \Lsym{)}  }  \Delta 
      }{ \Gamma  \vdash_{ G }  \Delta }
      \leqno{\raise 8 pt\hbox{\textbf{Case}}}      
\]
  We know by the induction hypothesis that every activation of $ \Gamma  \vdash_{  G  \Lsym{,}  \Lsym{(}  \Lmv{n_{{\mathrm{1}}}}  \Lsym{,}  \Lmv{n_{{\mathrm{3}}}}  \Lsym{)}  }  \Delta $ is derivable.  
  Suppose that $ \mathsf{D}( G  \dttsym{,}  \dttsym{(}  \dttnt{n_{{\mathrm{1}}}}  \dttsym{,}  \dttnt{n_{{\mathrm{3}}}}  \dttsym{)} )   \dttsym{;}   \mathsf{D}( \Gamma )^{ \dttsym{+} }   \dttsym{,}  \Gamma'  \vdash  \dttsym{+} \, \dttnt{A}  \mathbin{@}  \dttnt{n}$ is an arbitrary activation, where
  $ \mathsf{D}( \Delta )^{ \dttsym{-} }  \equiv  \mathsf{D}( \Delta_{{\mathrm{1}}} )^{ \dttsym{-} }   \dttsym{,}  \dttsym{-} \, \dttnt{A}  \mathbin{@}  \dttnt{n}  \dttsym{,}   \mathsf{D}( \Delta_{{\mathrm{2}}} )^{ \dttsym{-} } $ and $\Gamma' \equiv  \mathsf{D}( \Delta_{{\mathrm{1}}} )^{ \dttsym{-} }   \dttsym{,}   \mathsf{D}( \Delta_{{\mathrm{2}}} )^{ \dttsym{-} } $.  This sequent
  is equivalent to $ \mathsf{D}( G )   \dttsym{,}  \dttnt{n_{{\mathrm{1}}}} \,  \preccurlyeq_{ \dttsym{+} }  \, \dttnt{n_{{\mathrm{3}}}}  \dttsym{;}   \mathsf{D}( \Gamma )^{ \dttsym{+} }   \dttsym{,}  \Gamma'  \vdash  \dttsym{+} \, \dttnt{A}  \mathbin{@}  \dttnt{n}$.  Furthermore, it is clear by definition that
  if $ \Lmv{n_{{\mathrm{1}}}}   G   \Lmv{n_{{\mathrm{2}}}} $ and $ \Lmv{n_{{\mathrm{2}}}}   G   \Lmv{n_{{\mathrm{3}}}} $, then $ \dttnt{n_{{\mathrm{1}}}}    \preccurlyeq_{ \dttsym{+} }    \dttnt{n_{{\mathrm{2}}}}  \in   \mathsf{D}( G )  $ and $ \dttnt{n_{{\mathrm{2}}}}    \preccurlyeq_{ \dttsym{+} }    \dttnt{n_{{\mathrm{3}}}}  \in   \mathsf{D}( G )  $.  Thus,
  by the admissible rule for transitivity (Lemma~\ref{lemma:transitivity}) we have $ \mathsf{D}( G )   \dttsym{;}   \mathsf{D}( \Gamma )^{ \dttsym{+} }   \dttsym{,}  \Gamma'  \vdash  \dttsym{+} \, \dttnt{A}  \mathbin{@}  \dttnt{n}$,
  and we obtain our result.  

\item[\cW] 
  \[
      \mprset{flushleft}
      \inferrule* [right=\ifrName{hyp}] {
        \ 
      }{ \Gamma  \Lsym{,}  \Lmv{n}  \Lsym{:}  \Lnt{A}  \vdash_{ G }  \Lmv{n}  \Lsym{:}  \Lnt{A}  \Lsym{,}  \Delta }
      \leqno{\raise 8 pt\hbox{\textbf{Case}}}      
\]
  It suffices to show that every activation of $ \Gamma  \Lsym{,}  \Lmv{n}  \Lsym{:}  \Lnt{A}  \vdash_{ G }  \Lmv{n}  \Lsym{:}  \Lnt{A}  \Lsym{,}  \Delta $ is derivable.  Clearly,
  $ \mathsf{D}( G )   \dttsym{;}   \mathsf{D}( \Gamma )^{ \dttsym{+} }   \dttsym{,}  \dttsym{+} \,  \mathsf{D}( \dttnt{A} )   \mathbin{@}  \dttnt{n}  \dttsym{,}   \mathsf{D}( \Delta )^{ \dttsym{-} }   \vdash  \dttsym{+} \,  \mathsf{D}( \dttnt{A} )   \mathbin{@}  \dttnt{n}$ is a activation of $ \Gamma  \Lsym{,}  \Lmv{n}  \Lsym{:}  \Lnt{A}  \vdash_{ G }  \Lmv{n}  \Lsym{:}  \Lnt{A}  \Lsym{,}  \Delta $.  In addition,
  it is derivable:
  \[
      \mprset{flushleft}
      \inferrule* [right=Exchange] {
        \mprset{flushleft}
        \inferrule* [right=Ax] {
          \mprset{flushleft}
          \inferrule* [right=Refl] {
            \ 
          }{ \mathsf{D}( G )   \vdash  \dttnt{n} \,  \preccurlyeq^*_{ \dttsym{+} }  \, \dttnt{n}}
        }{ \mathsf{D}( G )   \dttsym{;}   \mathsf{D}( \Gamma )^{ \dttsym{+} }   \dttsym{,}   \mathsf{D}( \Delta )^{ \dttsym{-} }   \dttsym{,}  \dttsym{+} \,  \mathsf{D}( \dttnt{A} )   \mathbin{@}  \dttnt{n}  \vdash  \dttsym{+} \,  \mathsf{D}( \dttnt{A} )   \mathbin{@}  \dttnt{n}}
      }{ \mathsf{D}( G )   \dttsym{;}   \mathsf{D}( \Gamma )^{ \dttsym{+} }   \dttsym{,}  \dttsym{+} \,  \mathsf{D}( \dttnt{A} )   \mathbin{@}  \dttnt{n}  \dttsym{,}   \mathsf{D}( \Delta )^{ \dttsym{-} }   \vdash  \dttsym{+} \,  \mathsf{D}( \dttnt{A} )   \mathbin{@}  \dttnt{n}}
\]
  In the previous derivation we make use of the exchange rule, which
  is admissible by Lemma~\ref{lemma:exchange}.

  Now consider any other activation $ \mathsf{D}( G )   \dttsym{;}  \Gamma'  \vdash  \dttsym{+} \,  \mathsf{D}( \dttnt{B} )   \mathbin{@}  \dttnt{n'}$.  It must be the case that 
  $\Gamma' =     \mathsf{D}( \Gamma )^{ \dttsym{+} }   \dttsym{,}  \dttsym{+} \,  \mathsf{D}( \dttnt{A} )   \mathbin{@}  \dttnt{n}   \dttsym{,}   \mathsf{D}( \Delta_{{\mathrm{1}}} )^{ \dttsym{-} }    \dttsym{,}  \dttsym{-} \,  \mathsf{D}( \dttnt{A} )   \mathbin{@}  \dttnt{n}   \dttsym{,}   \mathsf{D}( \Delta_{{\mathrm{2}}} )^{ \dttsym{-} } $ for some $\Delta_{{\mathrm{1}}}$ and $\Delta_{{\mathrm{2}}}$.
  This sequent is then derivable as follows:
  \begin{center}
    \scriptsize
    \begin{math}
      \mprset{flushleft}
      \inferrule* [right={\scriptsize Left-to-Right}] {
        \mprset{flushleft}
        \inferrule* [right={\scriptsize Exchange}] {
          \mprset{flushleft}
          \inferrule* [right={\scriptsize Ax}] {
            \mprset{flushleft}
            \inferrule* [right={\scriptsize Refl}] {
              \ 
            }{ \mathsf{D}( G )   \vdash  \dttnt{n} \,  \preccurlyeq^*_{ \dttsym{+} }  \, \dttnt{n}}
          }{ \mathsf{D}( G )   \dttsym{;}   \mathsf{D}( \Gamma )^{ \dttsym{+} }   \dttsym{,}   \mathsf{D}( \Delta_{{\mathrm{1}}} )^{ \dttsym{-} }   \dttsym{,}   \mathsf{D}( \Delta_{{\mathrm{2}}} )^{ \dttsym{-} }   \dttsym{,}  \dttsym{-} \,  \mathsf{D}( \dttnt{B} )   \mathbin{@}  \dttnt{n'}  \dttsym{,}  \dttsym{+} \,  \mathsf{D}( \dttnt{A} )   \mathbin{@}  \dttnt{n}  \vdash  \dttsym{+} \,  \mathsf{D}( \dttnt{A} )   \mathbin{@}  \dttnt{n} }
        }{ \mathsf{D}( G )   \dttsym{;}      \mathsf{D}( \Gamma )^{ \dttsym{+} }   \dttsym{,}  \dttsym{+} \, \dttnt{A}  \mathbin{@}  \dttnt{n}   \dttsym{,}   \mathsf{D}( \Delta_{{\mathrm{1}}} )^{ \dttsym{-} }    \dttsym{,}   \mathsf{D}( \Delta_{{\mathrm{2}}} )^{ \dttsym{-} }    \dttsym{,}  \dttsym{-} \, \dttnt{B}  \mathbin{@}  \dttnt{n'}  \vdash  \dttsym{+} \,  \mathsf{D}( \dttnt{A} )   \mathbin{@}  \dttnt{n}}
      }{ \mathsf{D}( G )   \dttsym{;}      \mathsf{D}( \Gamma )^{ \dttsym{+} }   \dttsym{,}  \dttsym{+} \,  \mathsf{D}( \dttnt{A} )   \mathbin{@}  \dttnt{n}   \dttsym{,}   \mathsf{D}( \Delta_{{\mathrm{1}}} )^{ \dttsym{-} }    \dttsym{,}  \dttsym{-} \,  \mathsf{D}( \dttnt{A} )   \mathbin{@}  \dttnt{n}   \dttsym{,}   \mathsf{D}( \Delta_{{\mathrm{2}}} )^{ \dttsym{-} }   \vdash  \dttsym{+} \,  \mathsf{D}( \dttnt{B} )   \mathbin{@}  \dttnt{n'}}
    \end{math}
  \end{center}  
  Thus, we obtain our result.

\item[\cW] 
  \[
      \mprset{flushleft}
      \inferrule* [right=\ifrName{monL}] {
         \Lmv{n_{{\mathrm{1}}}}   G   \Lmv{n_{{\mathrm{2}}}} 
        \\\\
         \Gamma  \Lsym{,}  \Lmv{n_{{\mathrm{1}}}}  \Lsym{:}  \Lnt{A}  \Lsym{,}  \Lmv{n_{{\mathrm{2}}}}  \Lsym{:}  \Lnt{A}  \vdash_{ G }  \Delta 
      }{ \Gamma  \Lsym{,}  \Lmv{n_{{\mathrm{1}}}}  \Lsym{:}  \Lnt{A}  \vdash_{ G }  \Delta }
      \leqno{\raise 8 pt\hbox{\textbf{Case}}}      
\]
  Certainly, if $ \Lmv{n_{{\mathrm{1}}}}   G   \Lmv{n_{{\mathrm{2}}}} $, then $\dttnt{n_{{\mathrm{1}}}} \,  \preccurlyeq_{ \dttsym{+} }  \, \dttnt{n_{{\mathrm{2}}}} \in  \mathsf{D}( G ) $.
  We know by the induction hypothesis that all activations of $ \Gamma  \Lsym{,}  \Lmv{n_{{\mathrm{1}}}}  \Lsym{:}  \Lnt{A}  \Lsym{,}  \Lmv{n_{{\mathrm{2}}}}  \Lsym{:}  \Lnt{A}  \vdash_{ G }  \Delta $ are
  derivable.  Suppose $ \mathsf{D}( G )   \dttsym{;}  \Gamma'  \vdash  \dttsym{+} \, \dttnt{B}  \mathbin{@}  \dttnt{n}$ is an arbitrary activation.  Then it must be the case
  that $\Gamma' \equiv  \mathsf{D}( \Gamma )^{ \dttsym{+} }   \dttsym{,}  \dttsym{+} \,  \mathsf{D}( \dttnt{A} )   \mathbin{@}  \dttnt{n_{{\mathrm{1}}}}  \dttsym{,}  \dttsym{+} \,  \mathsf{D}( \dttnt{A} )   \mathbin{@}  \dttnt{n_{{\mathrm{2}}}}  \dttsym{,}   \mathsf{D}( \Delta_{{\mathrm{1}}} )^{ \dttsym{-} }   \dttsym{,}   \mathsf{D}( \Delta_{{\mathrm{2}}} )^{ \dttsym{-} } $, where $ \mathsf{D}( \Delta )^{ \dttsym{-} }  \equiv  \mathsf{D}( \Delta_{{\mathrm{1}}} )^{ \dttsym{-} }   \dttsym{,}  \dttsym{-} \, \dttnt{B}  \mathbin{@}  \dttnt{n}  \dttsym{,}   \mathsf{D}( \Delta_{{\mathrm{2}}} )^{ \dttsym{-} } $.
  Now we apply the monoL admissible rule (Lemma~\ref{corollary:monol}) to obtain 
  $ \mathsf{D}( G )   \dttsym{;}   \mathsf{D}( \Gamma )^{ \dttsym{+} }   \dttsym{,}  \dttsym{+} \,  \mathsf{D}( \dttnt{A} )   \mathbin{@}  \dttnt{n_{{\mathrm{1}}}}  \dttsym{,}   \mathsf{D}( \Delta_{{\mathrm{1}}} )^{ \dttsym{-} }   \dttsym{,}   \mathsf{D}( \Delta_{{\mathrm{2}}} )^{ \dttsym{-} }   \vdash  \dttsym{+} \, \dttnt{B}  \mathbin{@}  \dttnt{n}$, which is an arbitrary activation of 
  $ \Gamma  \Lsym{,}  \Lmv{n_{{\mathrm{1}}}}  \Lsym{:}  \Lnt{A}  \vdash_{ G }  \Delta $.

\item[\cW] 
  \[
      \mprset{flushleft}
      \inferrule* [right=\ifrName{monR}] {
         \Lmv{n_{{\mathrm{1}}}}   G   \Lmv{n_{{\mathrm{2}}}} 
        \\\\
         \Gamma  \vdash_{ G }  \Lmv{n_{{\mathrm{1}}}}  \Lsym{:}  \Lnt{A}  \Lsym{,}  \Lmv{n_{{\mathrm{2}}}}  \Lsym{:}  \Lnt{A}  \Lsym{,}  \Delta 
      }{ \Gamma  \vdash_{ G }  \Lmv{n_{{\mathrm{2}}}}  \Lsym{:}  \Lnt{A}  \Lsym{,}  \Delta }
      \leqno{\raise 8 pt\hbox{\textbf{Case}}}      
\]    
  If $ \Lmv{n_{{\mathrm{1}}}}   G   \Lmv{n_{{\mathrm{2}}}} $, then $\dttnt{n_{{\mathrm{1}}}} \,  \preccurlyeq_{ \dttsym{+} }  \, \dttnt{n_{{\mathrm{2}}}} \in  \mathsf{D}( G ) $.  We know by the induction hypothesis that
  all activations of $ \Gamma  \vdash_{ G }  \Lmv{n_{{\mathrm{1}}}}  \Lsym{:}  \Lnt{A}  \Lsym{,}  \Lmv{n_{{\mathrm{2}}}}  \Lsym{:}  \Lnt{A}  \Lsym{,}  \Delta $ are derivable.  In particular, the activation
  (modulo exchange (Lemma~\ref{lemma:exchange})) $ \mathsf{D}( G )   \dttsym{;}   \mathsf{D}( \Gamma )^{ \dttsym{+} }   \dttsym{,}   \mathsf{D}( \Delta )^{ \dttsym{-} }   \dttsym{,}  \dttsym{-} \,  \mathsf{D}( \dttnt{A} )   \mathbin{@}  \dttnt{n_{{\mathrm{1}}}}  \vdash  \dttsym{+} \,  \mathsf{D}( \dttnt{A} )   \mathbin{@}  \dttnt{n_{{\mathrm{2}}}}$ is
  derivable.  It suffices to show that $ \mathsf{D}( G )   \dttsym{;}   \mathsf{D}( \Gamma )^{ \dttsym{+} }   \dttsym{,}   \mathsf{D}( \Delta )^{ \dttsym{-} }   \vdash  \dttsym{+} \,  \mathsf{D}( \dttnt{A} )   \mathbin{@}  \dttnt{n_{{\mathrm{2}}}}$. This follows from
  the monoR admissible rule (Lemma~\ref{corollary:monor}).  Finally, any other activation of $ \Gamma  \vdash_{ G }  \Lmv{n_{{\mathrm{2}}}}  \Lsym{:}  \Lnt{A}  \Lsym{,}  \Delta $
  can be activated into $ \mathsf{D}( G )   \dttsym{;}   \mathsf{D}( \Gamma )^{ \dttsym{+} }   \dttsym{,}   \mathsf{D}( \Delta )^{ \dttsym{-} }   \vdash  \dttsym{+} \,  \mathsf{D}( \dttnt{A} )   \mathbin{@}  \dttnt{n_{{\mathrm{2}}}}$ (Lemma~\ref{lemma:refocus}).  Thus,
  we obtain our result.
  
\item[\cW] 
  \[
      \mprset{flushleft}
      \inferrule* [right=\ifrName{trueL}] {
         \Gamma  \vdash_{ G }  \Delta 
      }{ \Gamma  \Lsym{,}  \Lmv{n'}  \Lsym{:}   \top   \vdash_{ G }  \Delta }
      \leqno{\raise 8 pt\hbox{\textbf{Case}}}      
\]
  We know by the induction hypothesis that all activations of $ \Gamma  \vdash_{ G }  \Delta $
  are derivable.  Suppose $ \mathsf{D}( G )   \dttsym{;}  \Gamma'  \vdash  \dttsym{+} \,  \mathsf{D}( \dttnt{A} )   \mathbin{@}  \dttnt{n}$ is an arbitrary
  activation of $ \Gamma  \vdash_{ G }  \Delta $.  Then it must be the case that $\Gamma' =  \mathsf{D}( \Gamma )^{ \dttsym{+} }   \dttsym{,}   \mathsf{D}( \Delta_{{\mathrm{1}}} )^{ \dttsym{-} }   \dttsym{,}   \mathsf{D}( \Delta_{{\mathrm{2}}} )^{ \dttsym{-} } $, where
  $ \mathsf{D}( \Delta )^{ \dttsym{-} }  \equiv  \mathsf{D}( \Delta_{{\mathrm{1}}} )^{ \dttsym{-} }   \dttsym{,}  \dttsym{-} \,  \mathsf{D}( \dttnt{A} )   \mathbin{@}  \dttnt{n}  \dttsym{,}   \mathsf{D}( \Delta_{{\mathrm{2}}} )^{ \dttsym{-} } $.  Now by weakening (Lemma~\ref{lemma:weakening}) we know 
  $ \mathsf{D}( G )   \dttsym{;}  \Gamma'  \dttsym{,}  \dttsym{+} \,  \langle  \dttsym{+} \rangle   \mathbin{@}  \dttnt{n'}  \vdash  \dttsym{+} \,  \mathsf{D}( \dttnt{A} )   \mathbin{@}  \dttnt{n}$, and by exchange (Lemma~\ref{lemma:exchange}) \\
  $ \mathsf{D}( G )   \dttsym{;}   \mathsf{D}( \Gamma )^{ \dttsym{+} }   \dttsym{,}  \dttsym{+} \,  \langle  \dttsym{+} \rangle   \mathbin{@}  \dttnt{n'}  \dttsym{,}   \mathsf{D}( \Delta_{{\mathrm{1}}} )^{ \dttsym{-} }   \dttsym{,}   \mathsf{D}( \Delta_{{\mathrm{2}}} )^{ \dttsym{-} }   \vdash  \dttsym{+} \,  \mathsf{D}( \dttnt{A} )   \mathbin{@}  \dttnt{n}$, which is exactly an arbitrary activation of
  $ \Gamma  \Lsym{,}  \Lmv{n'}  \Lsym{:}   \top   \vdash_{ G }  \Delta $.

\item[\cW] 
  \[
      \mprset{flushleft}
      \inferrule* [right=\ifrName{trueR}] {
        \ 
      }{ \Gamma  \vdash_{ G }  \Lmv{n}  \Lsym{:}   \top   \Lsym{,}  \Delta }
      \leqno{\raise 8 pt\hbox{\textbf{Case}}}      
\]
  It suffices to show that every activation of $ \Gamma  \vdash_{ G }  \Lmv{n}  \Lsym{:}   \top   \Lsym{,}  \Delta $ is derivable.
  Consider the activation $ \mathsf{D}( G )   \dttsym{;}   \mathsf{D}( \Gamma )^{ \dttsym{+} }   \dttsym{,}   \mathsf{D}( \Delta )^{ \dttsym{-} }   \vdash  \dttsym{+} \,  \mathsf{D}(  \top  )   \mathbin{@}  \dttnt{n}$.  This is easily derivable
  by applying the \ifrName{unit} rule.  Any other activation of
  $ \Gamma  \vdash_{ G }  \Lmv{n}  \Lsym{:}   \top   \Lsym{,}  \Delta $ is derivable, because $ \mathsf{D}( G )   \dttsym{;}   \mathsf{D}( \Gamma )^{ \dttsym{+} }   \dttsym{,}   \mathsf{D}( \Delta )^{ \dttsym{-} }   \vdash  \dttsym{+} \,  \mathsf{D}(  \top  )   \mathbin{@}  \dttnt{n}$ can be activated by
  Lemma~\ref{lemma:refocus}.

\item[\cW] 
  \[
      \mprset{flushleft}
      \inferrule* [right=\ifrName{falseL}] {
        \ 
      }{ \Gamma  \Lsym{,}  \Lmv{n}  \Lsym{:}   \perp   \vdash_{ G }  \Delta }
      \leqno{\raise 8 pt\hbox{\textbf{Case}}}      
\]
  Suppose $ \mathsf{D}( G )   \dttsym{;}   \mathsf{D}( \Gamma )^{ \dttsym{+} }   \dttsym{,}  \dttsym{+} \,  \mathsf{D}(  \perp  )   \mathbin{@}  \dttnt{n}  \dttsym{,}   \mathsf{D}( \Delta_{{\mathrm{1}}} )^{ \dttsym{-} }   \dttsym{,}   \mathsf{D}( \Delta_{{\mathrm{2}}} )^{ \dttsym{-} }   \vdash  \dttsym{+} \,  \mathsf{D}( \dttnt{A} )   \mathbin{@}  \dttnt{n'}$ is an arbitrary activation of 
  $ \Gamma  \Lsym{,}  \Lmv{n}  \Lsym{:}   \perp   \vdash_{ G }  \Delta $, where $ \mathsf{D}( \Delta )^{ \dttsym{-} }  \equiv  \mathsf{D}( \Delta_{{\mathrm{1}}} )^{ \dttsym{-} }   \dttsym{,}  \dttsym{-} \,  \mathsf{D}( \dttnt{A} )   \mathbin{@}  \dttnt{n'}  \dttsym{,}   \mathsf{D}( \Delta_{{\mathrm{2}}} )^{ \dttsym{-} } $.  
  We can easily see that by definition $ \mathsf{D}( G )   \dttsym{;}   \mathsf{D}( \Gamma )^{ \dttsym{+} }   \dttsym{,}  \dttsym{+} \,  \mathsf{D}(  \perp  )   \mathbin{@}  \dttnt{n}  \dttsym{,}   \mathsf{D}( \Delta_{{\mathrm{1}}} )^{ \dttsym{-} }   \dttsym{,}   \mathsf{D}( \Delta_{{\mathrm{2}}} )^{ \dttsym{-} }   \vdash  \dttsym{+} \,  \mathsf{D}( \dttnt{A} )   \mathbin{@}  \dttnt{n'}$
  is equivalent to $ \mathsf{D}( G )   \dttsym{;}   \mathsf{D}( \Gamma )^{ \dttsym{+} }   \dttsym{,}  \dttsym{+} \,  \langle  \dttsym{-} \rangle   \mathbin{@}  \dttnt{n}  \dttsym{,}   \mathsf{D}( \Delta_{{\mathrm{1}}} )^{ \dttsym{-} }   \dttsym{,}   \mathsf{D}( \Delta_{{\mathrm{2}}} )^{ \dttsym{-} }   \vdash  \dttsym{+} \,  \mathsf{D}( \dttnt{A} )   \mathbin{@}  \dttnt{n'}$. We can derive the latter as follows:
  \begin{center}
    \footnotesize
    \begin{math}
      \mprset{flushleft}
      \inferrule* [right={\footnotesize \ifrName{axCutBar}}] {
         \dttsym{+} \,  \langle  \dttsym{-} \rangle   \mathbin{@}  \dttnt{n}  \in  \Gamma'  \dttsym{,}  \dttsym{-} \,  \mathsf{D}( \dttnt{A} )   \mathbin{@}  \dttnt{n'} 
        \\
        \mprset{flushleft}
        \inferrule* [right=\footnotesize \ifrName{unit}] {
          \ 
        }{ \mathsf{D}( G )   \dttsym{;}  \Gamma'  \dttsym{,}  \dttsym{-} \,  \mathsf{D}( \dttnt{A} )   \mathbin{@}  \dttnt{n'}  \vdash  \dttsym{-} \,  \langle  \dttsym{-} \rangle   \mathbin{@}  \dttnt{n}}
      }{ \mathsf{D}( G )   \dttsym{;}   \mathsf{D}( \Gamma )^{ \dttsym{+} }   \dttsym{,}  \dttsym{+} \,  \langle  \dttsym{-} \rangle   \mathbin{@}  \dttnt{n}  \dttsym{,}   \mathsf{D}( \Delta_{{\mathrm{1}}} )^{ \dttsym{-} }   \dttsym{,}   \mathsf{D}( \Delta_{{\mathrm{2}}} )^{ \dttsym{-} }   \vdash  \dttsym{+} \,  \mathsf{D}( \dttnt{A} )   \mathbin{@}  \dttnt{n'}}
    \end{math}
  \end{center}
  In the previous derivation $\Gamma' \equiv  \mathsf{D}( \Gamma )^{ \dttsym{+} }   \dttsym{,}  \dttsym{+} \,  \langle  \dttsym{-} \rangle   \mathbin{@}  \dttnt{n}  \dttsym{,}   \mathsf{D}( \Delta_{{\mathrm{1}}} )^{ \dttsym{-} }   \dttsym{,}   \mathsf{D}( \Delta_{{\mathrm{2}}} )^{ \dttsym{-} } $.  Thus, any activation of $ \Gamma  \Lsym{,}  \Lmv{n}  \Lsym{:}   \perp   \vdash_{ G }  \Delta $ is  derivable.

\item[\cW] 
  \[
      \mprset{flushleft}
      \inferrule* [right=\ifrName{falseR}] {
         \Gamma  \vdash_{ G }  \Delta 
      }{ \Gamma  \vdash_{ G }  \Lmv{n'}  \Lsym{:}   \perp   \Lsym{,}  \Delta }
      \leqno{\raise 8 pt\hbox{\textbf{Case}}}      
\]
  We know by the induction hypothesis that all activations of $ \Gamma  \vdash_{ G }  \Delta $
  are derivable.  Suppose $ \mathsf{D}( G )   \dttsym{;}  \Gamma'  \vdash  \dttsym{+} \,  \mathsf{D}( \dttnt{A} )   \mathbin{@}  \dttnt{n}$ is an arbitrary activation of 
  $ \Gamma  \vdash_{ G }  \Delta $.  Then it must be the case that $\Gamma' =  \mathsf{D}( \Gamma )^{ \dttsym{+} }   \dttsym{,}   \mathsf{D}( \Delta )^{ \dttsym{-} } $.
  Now by weakening (Lemma~\ref{lemma:weakening}) we know $ \mathsf{D}( G )   \dttsym{;}  \Gamma'  \dttsym{,}  \dttsym{-} \,  \langle  \dttsym{-} \rangle   \mathbin{@}  \dttnt{n'}  \vdash  \dttsym{+} \,  \mathsf{D}( \dttnt{A} )   \mathbin{@}  \dttnt{n}$,
  and by the left-to-right lemma (Lemma~\ref{lemma:refocus}) $ \mathsf{D}( G )   \dttsym{;}  \Gamma'  \dttsym{,}  \dttsym{-} \,  \mathsf{D}( \dttnt{A} )   \mathbin{@}  \dttnt{n}  \vdash  \dttsym{+} \,  \langle  \dttsym{-} \rangle   \mathbin{@}  \dttnt{n'}$, which --
  modulo exchange -- is equivalent to \\
  $ \mathsf{D}( G )   \dttsym{;}   \mathsf{D}( \Gamma )^{ \dttsym{+} }   \dttsym{,}   \mathsf{D}( \Delta )^{ \dttsym{-} }   \vdash  \dttsym{+} \,  \mathsf{D}(  \perp  )   \mathbin{@}  \dttnt{n'}$.  Thus, we obtain
  our result.

\item[\cW] 
  \[
      \mprset{flushleft}
      \inferrule* [right=\ifrName{andL}] {
         \Gamma  \Lsym{,}  \Lmv{n}  \Lsym{:}  \Lnt{T_{{\mathrm{1}}}}  \Lsym{,}  \Lmv{n}  \Lsym{:}  \Lnt{T_{{\mathrm{2}}}}  \vdash_{ G }  \Delta 
      }{ \Gamma  \Lsym{,}  \Lmv{n}  \Lsym{:}   \Lnt{T_{{\mathrm{1}}}}  \land  \Lnt{T_{{\mathrm{2}}}}   \vdash_{ G }  \Delta }
      \leqno{\raise 8 pt\hbox{\textbf{Case}}}      
\]
  We know by the induction hypothesis that all activations of $ \Gamma  \Lsym{,}  \Lmv{n}  \Lsym{:}  \Lnt{T_{{\mathrm{1}}}}  \Lsym{,}  \Lmv{n}  \Lsym{:}  \Lnt{T_{{\mathrm{2}}}}  \vdash_{ G }  \Delta $
  are derivable.  In particular, we know the following:
  \[
     \mathsf{D}( G )   \dttsym{;}   \mathsf{D}( \Gamma )^{ \dttsym{+} }   \dttsym{,}  \dttsym{+} \,  \mathsf{D}( \dttnt{T_{{\mathrm{1}}}} )   \mathbin{@}  \dttnt{n}  \dttsym{,}  \dttsym{+} \,  \mathsf{D}( \dttnt{T_{{\mathrm{2}}}} )   \mathbin{@}  \dttnt{n}  \dttsym{,}   \mathsf{D}( \Delta_{{\mathrm{1}}} )^{ \dttsym{-} }   \dttsym{,}   \mathsf{D}( \Delta_{{\mathrm{2}}} )^{ \dttsym{-} }   \vdash  \dttsym{+} \,  \mathsf{D}( \dttnt{A} )   \mathbin{@}  \dttnt{n'}
  \]
  where
  $ \mathsf{D}( \Delta )^{ \dttsym{-} }  =  \mathsf{D}( \Delta_{{\mathrm{1}}} )^{ \dttsym{-} }   \dttsym{,}  \dttsym{-} \,  \mathsf{D}( \dttnt{A} )   \mathbin{@}  \dttnt{n'}  \dttsym{,}   \mathsf{D}( \Delta_{{\mathrm{2}}} )^{ \dttsym{-} } $. Using exchange we know \\
  $ \mathsf{D}( G )   \dttsym{;}   \mathsf{D}( \Gamma )^{ \dttsym{+} }   \dttsym{,}   \mathsf{D}( \Delta_{{\mathrm{1}}} )^{ \dttsym{-} }   \dttsym{,}   \mathsf{D}( \Delta_{{\mathrm{2}}} )^{ \dttsym{-} }   \dttsym{,}  \dttsym{+} \,  \mathsf{D}( \dttnt{T_{{\mathrm{1}}}} )   \mathbin{@}  \dttnt{n}  \dttsym{,}  \dttsym{+} \,  \mathsf{D}( \dttnt{T_{{\mathrm{2}}}} )   \mathbin{@}  \dttnt{n}  \vdash  \dttsym{+} \,  \mathsf{D}( \dttnt{A} )   \mathbin{@}  \dttnt{n'}$, and
  by the left-to-right lemma $ \mathsf{D}( G )   \dttsym{;}   \mathsf{D}( \Gamma )^{ \dttsym{+} }   \dttsym{,}   \mathsf{D}( \Delta_{{\mathrm{1}}} )^{ \dttsym{-} }   \dttsym{,}   \mathsf{D}( \Delta_{{\mathrm{2}}} )^{ \dttsym{-} }   \dttsym{,}  \dttsym{+} \,  \mathsf{D}( \dttnt{T_{{\mathrm{1}}}} )   \mathbin{@}  \dttnt{n}  \dttsym{,}  \dttsym{-} \,  \mathsf{D}( \dttnt{A} )   \mathbin{@}  \dttnt{n'}  \vdash  \dttsym{-} \,  \mathsf{D}( \dttnt{T_{{\mathrm{2}}}} )   \mathbin{@}  \dttnt{n}$, and finally by one more application
  of exchange \\ $ \mathsf{D}( G )   \dttsym{;}   \mathsf{D}( \Gamma )^{ \dttsym{+} }   \dttsym{,}   \mathsf{D}( \Delta_{{\mathrm{1}}} )^{ \dttsym{-} }   \dttsym{,}   \mathsf{D}( \Delta_{{\mathrm{2}}} )^{ \dttsym{-} }   \dttsym{,}  \dttsym{-} \,  \mathsf{D}( \dttnt{A} )   \mathbin{@}  \dttnt{n'}  \dttsym{,}  \dttsym{+} \,  \mathsf{D}( \dttnt{T_{{\mathrm{1}}}} )   \mathbin{@}  \dttnt{n}  \vdash  \dttsym{-} \,  \mathsf{D}( \dttnt{T_{{\mathrm{2}}}} )   \mathbin{@}  \dttnt{n}$.  At this point we know 
  $ \mathsf{D}( G )   \dttsym{;}   \mathsf{D}( \Gamma )^{ \dttsym{+} }   \dttsym{,}   \mathsf{D}( \Delta_{{\mathrm{1}}} )^{ \dttsym{-} }   \dttsym{,}   \mathsf{D}( \Delta_{{\mathrm{2}}} )^{ \dttsym{-} }   \dttsym{,}  \dttsym{-} \,  \mathsf{D}( \dttnt{A} )   \mathbin{@}  \dttnt{n'}  \vdash  \dttsym{-} \,   \mathsf{D}( \dttnt{T_{{\mathrm{1}}}} )   \ndwedge{ \dttsym{+} }   \mathsf{D}( \dttnt{T_{{\mathrm{2}}}} )    \mathbin{@}  \dttnt{n}$ by using the admissible $\ifrName{andL}$
  rule (Lemma~\ref{lemma:andl}).
  Now using left-to-right we know the following:
  \[
     \mathsf{D}( G )   \dttsym{;}   \mathsf{D}( \Gamma )^{ \dttsym{+} }   \dttsym{,}   \mathsf{D}( \Delta_{{\mathrm{1}}} )^{ \dttsym{-} }   \dttsym{,}   \mathsf{D}( \Delta_{{\mathrm{2}}} )^{ \dttsym{-} }   \dttsym{,}  \dttsym{+} \,   \mathsf{D}( \dttnt{T_{{\mathrm{1}}}} )   \ndwedge{ \dttsym{+} }   \mathsf{D}( \dttnt{T_{{\mathrm{2}}}} )    \mathbin{@}  \dttnt{n}  \vdash  \dttsym{+} \,  \mathsf{D}( \dttnt{A} )   \mathbin{@}  \dttnt{n'}
  \]
  is derivable. Lastly, by exchange 
  $ \mathsf{D}( G )   \dttsym{;}   \mathsf{D}( \Gamma )^{ \dttsym{+} }   \dttsym{,}  \dttsym{+} \,   \mathsf{D}( \dttnt{T_{{\mathrm{1}}}} )   \ndwedge{ \dttsym{+} }   \mathsf{D}( \dttnt{T_{{\mathrm{2}}}} )    \mathbin{@}  \dttnt{n}  \dttsym{,}   \mathsf{D}( \Delta_{{\mathrm{1}}} )^{ \dttsym{-} }   \dttsym{,}   \mathsf{D}( \Delta_{{\mathrm{2}}} )^{ \dttsym{-} }   \vdash  \dttsym{+} \,  \mathsf{D}( \dttnt{A} )   \mathbin{@}  \dttnt{n'}$ is derivable, which is clearly an arbitrary
  activation of $ \Gamma  \Lsym{,}  \Lmv{n}  \Lsym{:}   \Lnt{T_{{\mathrm{1}}}}  \land  \Lnt{T_{{\mathrm{2}}}}   \vdash_{ G }  \Delta $.

\item[\cW] 
  \[
      \mprset{flushleft}
      \inferrule* [right=\ifrName{andR}] {
         \Gamma  \vdash_{ G }  \Lmv{n}  \Lsym{:}  \Lnt{A}  \Lsym{,}  \Delta 
        \\\\
         \Gamma  \vdash_{ G }  \Lmv{n}  \Lsym{:}  \Lnt{B}  \Lsym{,}  \Delta 
      }{ \Gamma  \vdash_{ G }  \Lmv{n}  \Lsym{:}   \Lnt{A}  \land  \Lnt{B}   \Lsym{,}  \Delta }
      \leqno{\raise 8 pt\hbox{\textbf{Case}}}      
\]
  We know by the induction hypothesis that all activations of $ \Gamma
  \vdash_{ G }  \Lmv{n}  \Lsym{:}  \Lnt{A}  \Lsym{,}  \Delta $ as well as
  $ \Gamma  \vdash_{ G }  \Lmv{n}  \Lsym{:}  \Lnt{B}  \Lsym{,}  \Delta $ are derivable.  In particular, $ \mathsf{D}( G )   \dttsym{;}   \mathsf{D}( \Gamma )^{ \dttsym{+} }   \dttsym{,}   \mathsf{D}( \Delta )^{ \dttsym{-} }   \vdash  \dttsym{+} \,  \mathsf{D}( \dttnt{A} )   \mathbin{@}  \dttnt{n}$
  and $ \mathsf{D}( G )   \dttsym{;}   \mathsf{D}( \Gamma )^{ \dttsym{+} }   \dttsym{,}   \mathsf{D}( \Delta )^{ \dttsym{-} }   \vdash  \dttsym{+} \,  \mathsf{D}( \dttnt{B} )   \mathbin{@}  \dttnt{n}$ are derivable. Now by applying the $\ifrName{and}$ rule
  we obtain $ \mathsf{D}( G )   \dttsym{;}   \mathsf{D}( \Gamma )^{ \dttsym{+} }   \dttsym{,}   \mathsf{D}( \Delta )^{ \dttsym{-} }   \vdash  \dttsym{+} \,   \mathsf{D}( \dttnt{A} )   \ndwedge{ \dttsym{+} }   \mathsf{D}( \dttnt{B} )    \mathbin{@}  \dttnt{n}$, which is a particular activation of 
  $ \Gamma  \vdash_{ G }  \Lmv{n}  \Lsym{:}   \Lnt{A}  \land  \Lnt{B}   \Lsym{,}  \Delta $. Finally, consider any other activation, then that sequent implies 
  $ \mathsf{D}( G )   \dttsym{;}   \mathsf{D}( \Gamma )^{ \dttsym{+} }   \dttsym{,}   \mathsf{D}( \Delta )^{ \dttsym{-} }   \vdash  \dttsym{+} \,   \mathsf{D}( \dttnt{A} )   \ndwedge{ \dttsym{+} }   \mathsf{D}( \dttnt{B} )    \mathbin{@}  \dttnt{n}$ is derivable using Lemma~\ref{lemma:refocus}.  
  Thus, we obtain our result.
\newpage

\item[\cW] 
  \[
      \mprset{flushleft}
      \inferrule* [right=\ifrName{disjL}] {
         \Gamma  \Lsym{,}  \Lmv{n}  \Lsym{:}  \Lnt{A}  \vdash_{ G }  \Delta 
        \\\\
         \Gamma  \Lsym{,}  \Lmv{n}  \Lsym{:}  \Lnt{B}  \vdash_{ G }  \Delta 
      }{ \Gamma  \Lsym{,}  \Lmv{n}  \Lsym{:}   \Lnt{A}  \lor  \Lnt{B}   \vdash_{ G }  \Delta }
      \leqno{\raise 8 pt\hbox{\textbf{Case}}}      
\]
  We know by the induction hypothesis that all activations of $ \Gamma  \Lsym{,}  \Lmv{n}  \Lsym{:}  \Lnt{A}  \vdash_{ G }  \Delta $ and
  $ \Gamma  \Lsym{,}  \Lmv{n}  \Lsym{:}  \Lnt{B}  \vdash_{ G }  \Delta $ are derivable.  So suppose
  \[  \mathsf{D}( G )   \dttsym{;}   \mathsf{D}( \Gamma )^{ \dttsym{+} }   \dttsym{,}  \dttsym{+} \,  \mathsf{D}( \dttnt{A} )   \mathbin{@}  \dttnt{n}  \dttsym{,}   \mathsf{D}( \Delta' )^{ \dttsym{-} }   \vdash  \dttsym{+} \,  \mathsf{D}( \dttnt{C} )   \mathbin{@}  \dttnt{n'} \] and
  \[  \mathsf{D}( G )   \dttsym{;}   \mathsf{D}( \Gamma )^{ \dttsym{+} }   \dttsym{,}  \dttsym{+} \,  \mathsf{D}( \dttnt{B} )   \mathbin{@}  \dttnt{n}  \dttsym{,}   \mathsf{D}( \Delta' )^{ \dttsym{-} }   \vdash  \dttsym{+} \,  \mathsf{D}( \dttnt{E} )   \mathbin{@}  \dttnt{n''} \]
  are particular activations, 
  where
  \[  \mathsf{D}( \Delta )^{ \dttsym{-} }  \equiv  \mathsf{D}( \Delta_{{\mathrm{1}}} )^{ \dttsym{-} }   \dttsym{,}  \dttsym{-} \,  \mathsf{D}( \dttnt{C} )   \mathbin{@}  \dttnt{n'}  \dttsym{,}   \mathsf{D}( \Delta_{{\mathrm{2}}} )^{ \dttsym{-} }   \dttsym{,}  \dttsym{-} \,  \mathsf{D}( \dttnt{E} )   \mathbin{@}  \dttnt{n''}  \dttsym{,}   \mathsf{D}( \Delta_{{\mathrm{3}}} )^{ \dttsym{-} } \]
  and
  \[  \mathsf{D}( \Delta' )^{ \dttsym{-} }  \equiv  \mathsf{D}( \Delta_{{\mathrm{1}}} )^{ \dttsym{-} }   \dttsym{,}   \mathsf{D}( \Delta_{{\mathrm{2}}} )^{ \dttsym{-} }   \dttsym{,}   \mathsf{D}( \Delta_{{\mathrm{3}}} )^{ \dttsym{-} } . \]
  By exchange (Lemma~\ref{lemma:exchange}) we know \\
  \[  \mathsf{D}( G )   \dttsym{;}   \mathsf{D}( \Gamma )^{ \dttsym{+} }   \dttsym{,}   \mathsf{D}( \Delta' )^{ \dttsym{-} }   \dttsym{,}  \dttsym{+} \,  \mathsf{D}( \dttnt{A} )   \mathbin{@}  \dttnt{n}  \vdash  \dttsym{+} \,  \mathsf{D}( \dttnt{C} )   \mathbin{@}  \dttnt{n'} \] 
  and
  \[  \mathsf{D}( G )   \dttsym{;}   \mathsf{D}( \Gamma )^{ \dttsym{+} }   \dttsym{,}   \mathsf{D}( \Delta' )^{ \dttsym{-} }   \dttsym{,}  \dttsym{+} \,  \mathsf{D}( \dttnt{B} )   \mathbin{@}  \dttnt{n}  \vdash  \dttsym{+} \,  \mathsf{D}( \dttnt{E} )   \mathbin{@}  \dttnt{n''}. \]
  Now by the left-to-right lemma (Lemma~\ref{lemma:refocus}) we know :
  \[  \mathsf{D}( G )   \dttsym{;}   \mathsf{D}( \Gamma )^{ \dttsym{+} }   \dttsym{,}   \mathsf{D}( \Delta' )^{ \dttsym{-} }   \dttsym{,}  \dttsym{-} \,  \mathsf{D}( \dttnt{C} )   \mathbin{@}  \dttnt{n'}  \vdash  \dttsym{-} \,  \mathsf{D}( \dttnt{A} )   \mathbin{@}  \dttnt{n} \]
  and
  \[  \mathsf{D}( G )   \dttsym{;}   \mathsf{D}( \Gamma )^{ \dttsym{+} }   \dttsym{,}   \mathsf{D}( \Delta' )^{ \dttsym{-} }   \dttsym{,}  \dttsym{-} \,  \mathsf{D}( \dttnt{E} )   \mathbin{@}  \dttnt{n''}  \vdash  \dttsym{-} \,  \mathsf{D}( \dttnt{B} )   \mathbin{@}  \dttnt{n}, \]
  and by applying weakening (and exchange) we know 
  \[  \mathsf{D}( G )   \dttsym{;}   \mathsf{D}( \Gamma )^{ \dttsym{+} }   \dttsym{,}   \mathsf{D}( \Delta' )^{ \dttsym{-} }   \dttsym{,}  \dttsym{-} \,  \mathsf{D}( \dttnt{C} )   \mathbin{@}  \dttnt{n'}  \dttsym{,}  \dttsym{-} \,  \mathsf{D}( \dttnt{E} )   \mathbin{@}  \dttnt{n''}  \vdash  \dttsym{-} \,  \mathsf{D}( \dttnt{A} )   \mathbin{@}  \dttnt{n} \]
  and 
  \[  \mathsf{D}( G )   \dttsym{;}   \mathsf{D}( \Gamma )^{ \dttsym{+} }   \dttsym{,}   \mathsf{D}( \Delta' )^{ \dttsym{-} }   \dttsym{,}  \dttsym{-} \,  \mathsf{D}( \dttnt{C} )   \mathbin{@}  \dttnt{n'}  \dttsym{,}  \dttsym{-} \,  \mathsf{D}( \dttnt{E} )   \mathbin{@}  \dttnt{n''}  \vdash  \dttsym{-} \,  \mathsf{D}( \dttnt{B} )   \mathbin{@}  \dttnt{n}.\]
  At this point we can apply the $\ifrName{and}$ rule to obtain
  \[  \mathsf{D}( G )   \dttsym{;}   \mathsf{D}( \Gamma )^{ \dttsym{+} }   \dttsym{,}   \mathsf{D}( \Delta' )^{ \dttsym{-} }   \dttsym{,}  \dttsym{-} \,  \mathsf{D}( \dttnt{C} )   \mathbin{@}  \dttnt{n'}  \dttsym{,}  \dttsym{-} \,  \mathsf{D}( \dttnt{E} )   \mathbin{@}  \dttnt{n''}  \vdash  \dttsym{-} \,   \mathsf{D}( \dttnt{A} )   \ndwedge{ \dttsym{-} }   \mathsf{D}( \dttnt{B} )    \mathbin{@}  \dttnt{n}, \] to which we can apply
  the left-to-right lemma to and obtain \\
  $ \mathsf{D}( G )   \dttsym{;}   \mathsf{D}( \Gamma )^{ \dttsym{+} }   \dttsym{,}   \mathsf{D}( \Delta' )^{ \dttsym{-} }   \dttsym{,}  \dttsym{-} \,  \mathsf{D}( \dttnt{E} )   \mathbin{@}  \dttnt{n''}  \dttsym{,}  \dttsym{+} \,   \mathsf{D}( \dttnt{A} )   \ndwedge{ \dttsym{-} }   \mathsf{D}( \dttnt{B} )    \mathbin{@}  \dttnt{n}  \vdash  \dttsym{+} \,  \mathsf{D}( \dttnt{C} )   \mathbin{@}  \dttnt{n'}$.  Finally, we can apply 
  exchange again to obtain \[  \mathsf{D}( G )   \dttsym{;}   \mathsf{D}( \Gamma )^{ \dttsym{+} }   \dttsym{,}  \dttsym{+} \,   \mathsf{D}( \dttnt{A} )   \ndwedge{ \dttsym{-} }   \mathsf{D}( \dttnt{B} )    \mathbin{@}  \dttnt{n}  \dttsym{,}   \mathsf{D}( \Delta' )^{ \dttsym{-} }   \dttsym{,}  \dttsym{-} \,  \mathsf{D}( \dttnt{E} )   \mathbin{@}  \dttnt{n''}  \vdash  \dttsym{+} \,  \mathsf{D}( \dttnt{C} )   \mathbin{@}  \dttnt{n'}, \]
  which -- modulo exchange -- is an arbitrary activation of $ \Gamma  \Lsym{,}  \Lmv{n}  \Lsym{:}   \Lnt{A}  \lor  \Lnt{B}   \vdash_{ G }  \Delta $.   Thus, we obtain our result.

\item[\cW] 
  \[
      \mprset{flushleft}
      \inferrule* [right=\ifrName{disjR}] {
         \Gamma  \vdash_{ G }  \Lmv{x}  \Lsym{:}  \Lnt{T_{{\mathrm{1}}}}  \Lsym{,}  \Lmv{x}  \Lsym{:}  \Lnt{T_{{\mathrm{2}}}}  \Lsym{,}  \Delta 
      }{ \Gamma  \vdash_{ G }  \Lmv{x}  \Lsym{:}   \Lnt{T_{{\mathrm{1}}}}  \lor  \Lnt{T_{{\mathrm{2}}}}   \Lsym{,}  \Delta }
      \leqno{\raise 8 pt\hbox{\textbf{Case}}}      
\]
  This case is similar to the case of $\ifrName{andR}$ case, except,
  it makes use of the $\ifrName{andBar}$ rule.

\item[\cW] 
  \[
      \mprset{flushleft}
      \inferrule* [right=\ifrName{impL}] {
         \Lmv{n_{{\mathrm{1}}}}   G   \Lmv{n_{{\mathrm{2}}}} 
        \\\\
         \Gamma  \vdash_{ G }  \Lmv{n_{{\mathrm{2}}}}  \Lsym{:}  \Lnt{T_{{\mathrm{1}}}}  \Lsym{,}  \Delta 
        \\\\
         \Gamma  \Lsym{,}  \Lmv{n_{{\mathrm{2}}}}  \Lsym{:}  \Lnt{T_{{\mathrm{2}}}}  \vdash_{ G }  \Delta 
      }{ \Gamma  \Lsym{,}  \Lmv{n_{{\mathrm{1}}}}  \Lsym{:}   \Lnt{T_{{\mathrm{1}}}}  \supset  \Lnt{T_{{\mathrm{2}}}}   \vdash_{ G }  \Delta }
      \leqno{\raise 8 pt\hbox{\textbf{Case}}}      
\]   
  We know by the induction hypothesis that all activations of $ \Gamma  \vdash_{ G }  \Lmv{n_{{\mathrm{2}}}}  \Lsym{:}  \Lnt{T_{{\mathrm{1}}}}  \Lsym{,}  \Delta $ and 
  $ \Gamma  \Lsym{,}  \Lmv{n_{{\mathrm{2}}}}  \Lsym{:}  \Lnt{T_{{\mathrm{2}}}}  \vdash_{ G }  \Delta $ are derivable.  In particular, we know 
  $ \mathsf{D}( G )   \dttsym{;}   \mathsf{D}( \Gamma )^{ \dttsym{+} }   \dttsym{,}   \mathsf{D}( \Delta )^{ \dttsym{-} }   \vdash  \dttsym{+} \,  \mathsf{D}( \dttnt{T_{{\mathrm{1}}}} )   \mathbin{@}  \dttnt{n_{{\mathrm{2}}}}$ is derivable, and so is $ \mathsf{D}( G )   \dttsym{;}   \mathsf{D}( \Gamma )^{ \dttsym{+} }   \dttsym{,}   \mathsf{D}( \Delta )^{ \dttsym{-} }   \vdash  \dttsym{-} \,  \mathsf{D}( \dttnt{T_{{\mathrm{2}}}} )   \mathbin{@}  \dttnt{n_{{\mathrm{2}}}}$. The latter
  being derivable by applying the induction hypothesis followed by exchange 
  (Lemma~\ref{lemma:exchange}) and the left-to-right lemma (Lemma~\ref{lemma:refocus}). We know $ \Lmv{n_{{\mathrm{1}}}}   G   \Lmv{n_{{\mathrm{2}}}} $ by assumption
  and so by Lemma~\ref{lemma:reach} $ \mathsf{D}( G )   \vdash  \dttnt{n_{{\mathrm{1}}}} \,  \preccurlyeq^*_{ \dttsym{+} }  \, \dttnt{n_{{\mathrm{2}}}}$.  Thus, by applying the \ifrName{impBar} rule we obtain
  $ \mathsf{D}( G )   \dttsym{;}   \mathsf{D}( \Gamma )^{ \dttsym{+} }   \dttsym{,}   \mathsf{D}( \Delta )^{ \dttsym{-} }   \vdash  \dttsym{-} \,   \mathsf{D}( \dttnt{T_{{\mathrm{1}}}} )   \ndto{ \dttsym{+} }   \mathsf{D}( \dttnt{T_{{\mathrm{2}}}} )    \mathbin{@}  \dttnt{n_{{\mathrm{1}}}}$.  At this point we can apply left-to-right to the previous sequent and obtain 
  an activation of $ \Gamma  \Lsym{,}  \Lmv{n_{{\mathrm{1}}}}  \Lsym{:}   \Lnt{T_{{\mathrm{1}}}}  \supset  \Lnt{T_{{\mathrm{2}}}}   \vdash_{ G }  \Delta $.  Any other activations can be used to derive
  $ \mathsf{D}( G )   \dttsym{;}   \mathsf{D}( \Gamma )^{ \dttsym{+} }
  \dttsym{,}   \mathsf{D}( \Delta )^{ \dttsym{-} }   \vdash
  \dttsym{+} \,  \mathsf{D}( \dttnt{T_{{\mathrm{1}}}} )   \mathbin{@}
  \dttnt{n_{{\mathrm{2}}}}$ and $ \mathsf{D}( G )   \dttsym{;}
  \mathsf{D}( \Gamma )^{ \dttsym{+} }   \dttsym{,}   \mathsf{D}(
  \Delta )^{ \dttsym{-} }   \vdash  \dttsym{-} \,  \mathsf{D}(
  \dttnt{T_{{\mathrm{2}}}} )   \mathbin{@}  \dttnt{n_{{\mathrm{2}}}}$,
  and thus, thus we obtain our result.

\item[\cW] 
  \[
      \mprset{flushleft}
      \inferrule* [right=\ifrName{impR}] {
         \Lmv{n_{{\mathrm{2}}}}  \not\in | G |,| \Gamma |,| \Delta | 
        \\\\
         \Gamma  \Lsym{,}  \Lmv{n_{{\mathrm{2}}}}  \Lsym{:}  \Lnt{T_{{\mathrm{1}}}}  \vdash_{   G  \cup  \Lsym{\{}  \Lsym{(}  \Lmv{n_{{\mathrm{1}}}}  \Lsym{,}  \Lmv{n_{{\mathrm{2}}}}  \Lsym{)}  \Lsym{\}}   }  \Lmv{n_{{\mathrm{2}}}}  \Lsym{:}  \Lnt{T_{{\mathrm{2}}}}  \Lsym{,}  \Delta 
      }{ \Gamma  \vdash_{ G }  \Lmv{n_{{\mathrm{1}}}}  \Lsym{:}   \Lnt{T_{{\mathrm{1}}}}  \supset  \Lnt{T_{{\mathrm{2}}}}   \Lsym{,}  \Delta }
      \leqno{\raise 8 pt\hbox{\textbf{Case}}}      
\]
  This case follows the same pattern as the previous cases.  We know by the induction hypothesis
  that all activations of $ \Gamma  \Lsym{,}  \Lmv{n_{{\mathrm{2}}}}  \Lsym{:}  \Lnt{T_{{\mathrm{1}}}}  \vdash_{   G  \cup  \Lsym{\{}  \Lsym{(}  \Lmv{n_{{\mathrm{1}}}}  \Lsym{,}  \Lmv{n_{{\mathrm{2}}}}  \Lsym{)}  \Lsym{\}}   }  \Lmv{n_{{\mathrm{2}}}}  \Lsym{:}  \Lnt{T_{{\mathrm{2}}}}  \Lsym{,}  \Delta $ are derivable.  In particular,
  $ \mathsf{D}( G )   \dttsym{,}  \dttnt{n_{{\mathrm{1}}}} \,  \preccurlyeq_{ \dttsym{+} }  \, \dttnt{n_{{\mathrm{2}}}}  \dttsym{;}   \mathsf{D}( \Gamma )^{ \dttsym{+} }   \dttsym{,}  \dttsym{+} \,  \mathsf{D}( \dttnt{T_{{\mathrm{1}}}} )   \mathbin{@}  \dttnt{n_{{\mathrm{2}}}}  \dttsym{,}   \mathsf{D}( \Delta )^{ \dttsym{-} }   \vdash  \dttsym{+} \,  \mathsf{D}( \dttnt{T_{{\mathrm{2}}}} )   \mathbin{@}  \dttnt{n_{{\mathrm{2}}}}$ is derivable.  By exchange (Lemma~\ref{lemma:exchange})\\
  $ \mathsf{D}( G )   \dttsym{,}  \dttnt{n_{{\mathrm{1}}}} \,  \preccurlyeq_{ \dttsym{+} }  \, \dttnt{n_{{\mathrm{2}}}}  \dttsym{;}   \mathsf{D}( \Gamma )^{ \dttsym{+} }   \dttsym{,}   \mathsf{D}( \Delta )^{ \dttsym{-} }   \dttsym{,}  \dttsym{+} \,  \mathsf{D}( \dttnt{T_{{\mathrm{1}}}} )   \mathbin{@}  \dttnt{n_{{\mathrm{2}}}}  \vdash  \dttsym{+} \,  \mathsf{D}( \dttnt{T_{{\mathrm{2}}}} )   \mathbin{@}  \dttnt{n_{{\mathrm{2}}}}$ is derivable, and by applying the $\ifrName{imp}$ rule
  we obtain $ \mathsf{D}( G )   \dttsym{;}   \mathsf{D}( \Gamma )^{ \dttsym{+} }   \dttsym{,}   \mathsf{D}( \Delta )^{ \dttsym{-} }   \vdash  \dttsym{+} \,   \mathsf{D}( \dttnt{T_{{\mathrm{1}}}} )   \ndto{ \dttsym{+} }   \mathsf{D}( \dttnt{T_{{\mathrm{2}}}} )    \mathbin{@}  \dttnt{n_{{\mathrm{1}}}}$, which is a particular activation of 
  $ \Gamma  \vdash_{ G }  \Lmv{n_{{\mathrm{1}}}}  \Lsym{:}   \Lnt{T_{{\mathrm{1}}}}  \supset  \Lnt{T_{{\mathrm{2}}}}   \Lsym{,}  \Delta $.  Note that in the previous application of $\ifrName{imp}$ we use
  the fact that if $ \Lmv{n_{{\mathrm{2}}}}  \not\in | G |,| \Gamma |,| \Delta | $, then $\dttnt{n_{{\mathrm{2}}}} \, \not\in \, \dttsym{\mbox{$\mid$}}   \mathsf{D}( G )   \dttsym{\mbox{$\mid$}}  \dttsym{,}  \dttsym{\mbox{$\mid$}}   \mathsf{D}( \Gamma )^{ \dttsym{+} }   \dttsym{,}   \mathsf{D}( \Delta )^{ \dttsym{-} }   \dttsym{\mbox{$\mid$}}$. 
  Lastly, any other activation of $ \Gamma  \vdash_{ G }  \Lmv{n_{{\mathrm{1}}}}  \Lsym{:}   \Lnt{T_{{\mathrm{1}}}}  \supset  \Lnt{T_{{\mathrm{2}}}}   \Lsym{,}  \Delta $ implies
  $ \mathsf{D}( G )   \dttsym{;}   \mathsf{D}( \Gamma )^{ \dttsym{+} }   \dttsym{,}   \mathsf{D}( \Delta )^{ \dttsym{-} }   \vdash  \dttsym{+} \,   \mathsf{D}( \dttnt{T_{{\mathrm{1}}}} )   \ndto{ \dttsym{+} }   \mathsf{D}( \dttnt{T_{{\mathrm{2}}}} )    \mathbin{@}  \dttnt{n_{{\mathrm{1}}}}$ is derivable by the left-to-right lemma, and hence is derivable.

\item[\cW] 
  \[
      \mprset{flushleft}
      \inferrule* [right=\ifrName{subL}] {
         \Lmv{n_{{\mathrm{1}}}}  \not\in | G |,| \Gamma |,| \Delta | 
        \\\\
         \Gamma  \Lsym{,}  \Lmv{n_{{\mathrm{1}}}}  \Lsym{:}  \Lnt{T_{{\mathrm{1}}}}  \vdash_{   G  \cup  \Lsym{\{}  \Lsym{(}  \Lmv{n_{{\mathrm{1}}}}  \Lsym{,}  \Lmv{n_{{\mathrm{2}}}}  \Lsym{)}  \Lsym{\}}   }  \Lmv{n_{{\mathrm{1}}}}  \Lsym{:}  \Lnt{T_{{\mathrm{2}}}}  \Lsym{,}  \Delta 
      }{ \Gamma  \Lsym{,}  \Lmv{n_{{\mathrm{2}}}}  \Lsym{:}   \Lnt{T_{{\mathrm{1}}}}  \prec  \Lnt{T_{{\mathrm{2}}}}   \vdash_{ G }  \Delta }
      \leqno{\raise 8 pt\hbox{\textbf{Case}}}      
\]
  We know by the induction hypothesis
  that all activation of \\ $ \Gamma  \Lsym{,}  \Lmv{n_{{\mathrm{1}}}}  \Lsym{:}  \Lnt{T_{{\mathrm{1}}}}  \vdash_{   G  \cup  \Lsym{\{}  \Lsym{(}  \Lmv{n_{{\mathrm{1}}}}  \Lsym{,}  \Lmv{n_{{\mathrm{2}}}}  \Lsym{)}  \Lsym{\}}   }  \Lmv{n_{{\mathrm{1}}}}  \Lsym{:}  \Lnt{T_{{\mathrm{2}}}}  \Lsym{,}  \Delta $ are derivable.  In particular,\\
  $ \mathsf{D}( G )   \dttsym{,}  \dttnt{n_{{\mathrm{1}}}} \,  \preccurlyeq_{ \dttsym{+} }  \, \dttnt{n_{{\mathrm{2}}}}  \dttsym{;}   \mathsf{D}( \Gamma )^{ \dttsym{+} }   \dttsym{,}  \dttsym{+} \,  \mathsf{D}( \dttnt{T_{{\mathrm{1}}}} )   \mathbin{@}  \dttnt{n_{{\mathrm{1}}}}  \dttsym{,}   \mathsf{D}( \Delta )^{ \dttsym{-} }   \vdash  \dttsym{+} \,  \mathsf{D}( \dttnt{T_{{\mathrm{2}}}} )   \mathbin{@}  \dttnt{n_{{\mathrm{1}}}}$ is derivable.  By exchange (Lemma~\ref{lemma:exchange})
  $ \mathsf{D}( G )   \dttsym{,}  \dttnt{n_{{\mathrm{1}}}} \,  \preccurlyeq_{ \dttsym{+} }  \, \dttnt{n_{{\mathrm{2}}}}  \dttsym{;}   \mathsf{D}( \Gamma )^{ \dttsym{+} }   \dttsym{,}   \mathsf{D}( \Delta )^{ \dttsym{-} }   \dttsym{,}  \dttsym{+} \,  \mathsf{D}( \dttnt{T_{{\mathrm{1}}}} )   \mathbin{@}  \dttnt{n_{{\mathrm{1}}}}  \vdash  \dttsym{+} \,  \mathsf{D}( \dttnt{T_{{\mathrm{2}}}} )   \mathbin{@}  \dttnt{n_{{\mathrm{1}}}}$ is derivable.  Now by the left-to-right lemma we know \\
  $ \mathsf{D}( G )   \dttsym{,}  \dttnt{n_{{\mathrm{1}}}} \,  \preccurlyeq_{ \dttsym{+} }  \, \dttnt{n_{{\mathrm{2}}}}  \dttsym{;}   \mathsf{D}( \Gamma )^{ \dttsym{+} }   \dttsym{,}   \mathsf{D}( \Delta )^{ \dttsym{-} }   \dttsym{,}  \dttsym{-} \,  \mathsf{D}( \dttnt{T_{{\mathrm{2}}}} )   \mathbin{@}  \dttnt{n_{{\mathrm{1}}}}  \vdash  \dttsym{-} \,  \mathsf{D}( \dttnt{T_{{\mathrm{1}}}} )   \mathbin{@}  \dttnt{n_{{\mathrm{1}}}}$, and by
  assumption we know $ \Lmv{y}  \not\in | G |,| \Gamma |,| \Delta | $, which implies \\
  $\dttnt{n_{{\mathrm{1}}}} \, \not\in \, \dttsym{\mbox{$\mid$}}   \mathsf{D}( G )   \dttsym{\mbox{$\mid$}}  \dttsym{,}  \dttsym{\mbox{$\mid$}}   \mathsf{D}( \Gamma )^{ \dttsym{+} }   \dttsym{,}   \mathsf{D}( \Delta )^{ \dttsym{-} }   \dttsym{\mbox{$\mid$}}$ is derivable.  Thus, by applying the $\ifrName{imp}$ rule we know  
  $ \mathsf{D}( G )   \dttsym{,}  \dttnt{n_{{\mathrm{1}}}} \,  \preccurlyeq_{ \dttsym{+} }  \, \dttnt{n_{{\mathrm{2}}}}  \dttsym{;}   \mathsf{D}( \Gamma )^{ \dttsym{+} }   \dttsym{,}   \mathsf{D}( \Delta )^{ \dttsym{-} }   \vdash  \dttsym{-} \,   \mathsf{D}( \dttnt{T_{{\mathrm{2}}}} )   \ndto{ \dttsym{-} }   \mathsf{D}( \dttnt{T_{{\mathrm{1}}}} )    \mathbin{@}  \dttnt{n_{{\mathrm{2}}}}$ is derivable.  Clearly, this is a particular activation of
  $ \Gamma  \Lsym{,}  \Lmv{n_{{\mathrm{2}}}}  \Lsym{:}   \Lnt{T_{{\mathrm{1}}}}  \prec  \Lnt{T_{{\mathrm{2}}}}   \vdash_{ G }  \Delta $, and any other activation implies
  $ \mathsf{D}( G )   \dttsym{,}  \dttnt{n_{{\mathrm{1}}}} \,  \preccurlyeq_{ \dttsym{+} }  \, \dttnt{n_{{\mathrm{2}}}}  \dttsym{;}   \mathsf{D}( \Gamma )^{ \dttsym{+} }   \dttsym{,}   \mathsf{D}( \Delta )^{ \dttsym{-} }   \vdash  \dttsym{-} \,   \mathsf{D}( \dttnt{T_{{\mathrm{2}}}} )   \ndto{ \dttsym{-} }   \mathsf{D}( \dttnt{T_{{\mathrm{1}}}} )    \mathbin{@}  \dttnt{n_{{\mathrm{2}}}}$ is derivable by the left-to-right lemma, and hence are derivable.

\item[\cW] 
  \[
      \mprset{flushleft}
      \inferrule* [right=\ifrName{subR}] {
         \Lmv{n_{{\mathrm{1}}}}   G   \Lmv{n_{{\mathrm{2}}}} 
        \\\\
         \Gamma  \vdash_{ G }  \Lmv{n_{{\mathrm{1}}}}  \Lsym{:}  \Lnt{T_{{\mathrm{1}}}}  \Lsym{,}  \Delta 
        \\\\
         \Gamma  \Lsym{,}  \Lmv{n_{{\mathrm{1}}}}  \Lsym{:}  \Lnt{T_{{\mathrm{2}}}}  \vdash_{ G }  \Delta 
      }{ \Gamma  \vdash_{ G }  \Lmv{n_{{\mathrm{2}}}}  \Lsym{:}   \Lnt{T_{{\mathrm{1}}}}  \prec  \Lnt{T_{{\mathrm{2}}}}   \Lsym{,}  \Delta }
      \leqno{\raise 8 pt\hbox{\textbf{Case}}}      
\]
  This case follows in the same way as the case for $\ifrName{impL}$,
  except the particular activation of $ \Gamma  \Lsym{,}  \Lmv{n_{{\mathrm{1}}}}  \Lsym{:}  \Lnt{T_{{\mathrm{2}}}}  \vdash_{ G }  \Delta $ has to have the active formulas such that the
  rule $\ifrName{impBar}$ can be applied.
\end{description}

\subsection{Proof of Lemma~\ref{lemma:containment_of_dil_in_l_part1}}
\label{subsec:proof_of_lemma:containment_of_dil_in_l_part1}

This is a proof by induction on the assumed typing derivation.

\begin{description}
\item[\cW] 
  \[
      \mprset{flushleft}
      \inferrule* [right=\ifrName{ax}] {
        G  \vdash  \dttnt{n} \,  \preccurlyeq^*_{ \dttnt{p} }  \, \dttnt{n'}
      }{G  \dttsym{;}  \Gamma  \dttsym{,}  \dttnt{p} \, \dttnt{A}  \mathbin{@}  \dttnt{n}  \dttsym{,}  \Gamma'  \vdash  \dttnt{p} \, \dttnt{A}  \mathbin{@}  \dttnt{n'}}
      \leqno{\raise 8 pt\hbox{\textbf{Case}}}      
\]
  We only show the case when $\dttnt{p} = \dttsym{+}$, because the case when $\dttnt{p} = \dttsym{-}$ is similar.
  By the definition of the L-translation we must show that the L-sequent
  $  \mathsf{L}( \Gamma )^{ \Lsym{+} }   \Lsym{,}  \Lmv{n}  \Lsym{:}   \mathsf{L}( \Lnt{A} )   \Lsym{,}   \mathsf{L}( \Gamma' )^{ \Lsym{+} }   \vdash_{  \mathsf{L}( G )  }   \mathsf{L}( \Gamma )^{ \Lsym{-} }   \Lsym{,}  \Lmv{n'}  \Lsym{:}   \mathsf{L}( \Lnt{A} )   \Lsym{,}   \mathsf{L}( \Gamma' )^{ \Lsym{-} }  $ is derivable.   
  Since we know that for any $\dttnt{n_{{\mathrm{1}}}}  \dttsym{,}  \dttnt{n_{{\mathrm{2}}}} \, \in \, \dttsym{\mbox{$\mid$}}  \dttnt{n} \,  \preccurlyeq_{ \dttnt{p'} }  \, \dttnt{n}  \dttsym{,}  G  \dttsym{\mbox{$\mid$}}  \dttsym{,}  \dttsym{\mbox{$\mid$}}  \Gamma  \dttsym{\mbox{$\mid$}}$, if $G  \vdash  \dttnt{n_{{\mathrm{1}}}} \,  \preccurlyeq^*_{ \dttnt{p'} }  \, \dttnt{n_{{\mathrm{2}}}}$, then
  $\dttnt{n_{{\mathrm{1}}}} \,  \preccurlyeq_{ \dttnt{p'} }  \, \dttnt{n_{{\mathrm{2}}}} \in G$, it must be the case that $\dttnt{n} \,  \preccurlyeq_{ \dttsym{+} }  \, \dttnt{n'} \in G$, and thus,
  $ \Lmv{n}    \mathsf{L}( G )    \Lmv{n'} $.  At this point we may apply the L inference rule $\dttdrulename{monL}$ using this fact.
  Therefore, we have derived $  \mathsf{L}( \Gamma )^{ \Lsym{+} }   \Lsym{,}  \Lmv{n}  \Lsym{:}   \mathsf{L}( \Lnt{A} )   \Lsym{,}  \Lmv{n'}  \Lsym{:}   \mathsf{L}( \Lnt{A} )   \Lsym{,}   \mathsf{L}( \Gamma' )^{ \Lsym{+} }   \vdash_{  \mathsf{L}( G )  }   \mathsf{L}( \Gamma )^{ \Lsym{-} }   \Lsym{,}  \Lmv{n'}  \Lsym{:}   \mathsf{L}( \Lnt{A} )   \Lsym{,}   \mathsf{L}( \Gamma' )^{ \Lsym{-} }  $, and then 
  we may complete the derivation by applying the L inference rule $\dttdrulename{hyp}$.
  
\item[\cW] 
  \[
      \mprset{flushleft}
      \inferrule* [right=\ifrName{unit}] {
        \ 
      }{G  \dttsym{;}  \Gamma  \vdash  \dttnt{p} \,  \langle  \dttnt{p} \rangle   \mathbin{@}  \dttnt{n}}
      \leqno{\raise 8 pt\hbox{\textbf{Case}}}      
\]
  Suppose $\dttnt{p} = \dttsym{+}$.  Then by the definition of the
  L-translation we must derive 
\[\mathsf{L}( \Gamma )^{ \Lsym{+} }   \vdash_{  \mathsf{L}( G )  }
\Lmv{n}  \Lsym{:}   \top   \Lsym{,}   \mathsf{L}( \Gamma )^{ \Lsym{-}
},
\]
  but this follows by simply applying the L inference rule
  $\dttdrulename{trueR}$.

  \ \\
  \noindent
  Suppose $\dttnt{p} = \dttsym{-}$.  Then by the definition of the
  L-translation we must derive 
\[\mathsf{L}( \Gamma )^{ \Lsym{+} }   \Lsym{,}  \Lmv{n}  \Lsym{:}
\;\perp\;   \vdash_{  \mathsf{L}( G )  }   \mathsf{L}( \Gamma )^{ \Lsym{-}
},
\]
  but this follows by simply applying the L inference rule
  $\dttdrulename{falseL}$.

\item[\cW] 
  \[
      \mprset{flushleft}
      \inferrule* [right=\ifrName{and}] {
         G  \dttsym{;}  \Gamma  \vdash  \dttnt{p} \, \dttnt{A}  \mathbin{@}  \dttnt{n}  \qquad  G  \dttsym{;}  \Gamma  \vdash  \dttnt{p} \, \dttnt{B}  \mathbin{@}  \dttnt{n} 
      }{G  \dttsym{;}  \Gamma  \vdash  \dttnt{p} \, \dttsym{(}   \dttnt{A}  \ndwedge{ \dttnt{p} }  \dttnt{B}   \dttsym{)}  \mathbin{@}  \dttnt{n}}
      \leqno{\raise 8 pt\hbox{\textbf{Case}}}      
\]
  Suppose $\dttnt{p} = \dttsym{+}$.  Then by the induction hypothesis we know the following:
  \[
      \begin{array}{lll}
          \mathsf{L}( \Gamma )^{ \Lsym{+} }   \vdash_{  \mathsf{L}( G )  }  \Lmv{n}  \Lsym{:}  \Lnt{A}  \Lsym{,}   \mathsf{L}( \Gamma )^{ \Lsym{-} }   \\
          \mathsf{L}( \Gamma )^{ \Lsym{+} }   \vdash_{  \mathsf{L}( G )  }  \Lmv{n}  \Lsym{:}  \Lnt{B}  \Lsym{,}   \mathsf{L}( \Gamma )^{ \Lsym{-} }  
      \end{array}      
      \]
  By the definition of the L-translation we must show that
  \[   \mathsf{L}( \Gamma )^{ \Lsym{+} }   \vdash_{  \mathsf{L}( G )  }  \Lmv{n}  \Lsym{:}   \Lnt{A}  \land  \Lnt{B}   \Lsym{,}   \mathsf{L}( \Gamma )^{ \Lsym{-} }  . \]
  This easily follows by applying the L inference rule $\dttdrulename{andR}$.

  \ \\
  \noindent
  Now suppose $\dttnt{p} = \dttsym{-}$.  Then by the induction hypothesis we know the following:
  \[
      \begin{array}{lll}
          \mathsf{L}( \Gamma )^{ \Lsym{+} }   \Lsym{,}  \Lmv{n}  \Lsym{:}  \Lnt{A}  \vdash_{  \mathsf{L}( G )  }   \mathsf{L}( \Gamma )^{ \Lsym{-} }   \\
          \mathsf{L}( \Gamma )^{ \Lsym{+} }   \Lsym{,}  \Lmv{n}  \Lsym{:}  \Lnt{B}  \vdash_{  \mathsf{L}( G )  }   \mathsf{L}( \Gamma )^{ \Lsym{-} }  
      \end{array}
\]
  By the definition of the L-translation we must show that
  \[   \mathsf{L}( \Gamma )^{ \Lsym{+} }   \Lsym{,}  \Lmv{n}  \Lsym{:}   \Lnt{A}  \lor  \Lnt{B}   \vdash_{  \mathsf{L}( G )  }   \mathsf{L}( \Gamma )^{ \Lsym{-} }  . \]
  This easily follows by applying the L inference rule $\dttdrulename{disjL}$.

\item[\cW] 
  \[
      \mprset{flushleft}
      \inferrule* [right=\ifrName{andBar}] {
        G  \dttsym{;}  \Gamma  \vdash  \dttnt{p} \,  \dttnt{A} _{ \dttnt{d} }   \mathbin{@}  \dttnt{n}
      }{G  \dttsym{;}  \Gamma  \vdash  \dttnt{p} \, \dttsym{(}   \dttnt{A_{{\mathrm{1}}}}  \ndwedge{  \bar{  \dttnt{p}  }  }  \dttnt{A_{{\mathrm{2}}}}   \dttsym{)}  \mathbin{@}  \dttnt{n}}
      \leqno{\raise 8 pt\hbox{\textbf{Case}}}      
\]
  Suppose $p = +$ and $d = 1$.  Then by the induction hypothesis we know the following:
  \[
  \begin{array}{lll}
      \mathsf{L}( \Gamma )^{ \Lsym{+} }   \vdash_{  \mathsf{L}( G )  }  \Lmv{n}  \Lsym{:}   \mathsf{L}( \Lnt{A_{{\mathrm{1}}}} )   \Lsym{,}   \mathsf{L}( \Gamma )^{ \Lsym{-} }  
  \end{array}
  \]
  Then by the L admissible inference rule $\Ldrulename{weakR}$ (Lemma~\ref{lemma:right_weakening_in_l})
  we know the following:
  \[
  \begin{array}{lll}
      \mathsf{L}( \Gamma )^{ \Lsym{+} }   \vdash_{  \mathsf{L}( G )  }  \Lmv{n}  \Lsym{:}   \mathsf{L}( \Lnt{A_{{\mathrm{1}}}} )   \Lsym{,}  \Lmv{n}  \Lsym{:}   \mathsf{L}( \Lnt{A_{{\mathrm{2}}}} )   \Lsym{,}   \mathsf{L}( \Gamma )^{ \Lsym{-} }  
  \end{array}
  \]
  Thus, we obtain our result that $  \mathsf{L}( \Gamma )^{ \Lsym{+} }   \vdash_{  \mathsf{L}( G )  }  \Lmv{n}  \Lsym{:}    \mathsf{L}( \Lnt{A_{{\mathrm{1}}}} )   \lor   \mathsf{L}( \Lnt{A_{{\mathrm{2}}}} )    \Lsym{,}   \mathsf{L}( \Gamma )^{ \Lsym{-} }  $ is derivable by applying
  the L inference rule $\Ldrulename{disjR}$.  The case for when $d = 2$ is similar.

  If $\dttnt{p} = \dttsym{-}$ then the result follows similarly to the case
  when $\dttnt{p} = \dttsym{+}$ except that the derivation is a result of
  applying the rule $\Ldrulename{andL}$ after applying the
  admissible L inference rule $\Ldrulename{weakL}$ to the induction
  hypothesis.

\item[\cW] 
  \[
      \mprset{flushleft}
      \inferrule* [right=\ifrName{imp}] {
        \dttnt{n'} \, \not\in \, \dttsym{\mbox{$\mid$}}  G  \dttsym{\mbox{$\mid$}}  \dttsym{,}  \dttsym{\mbox{$\mid$}}  \Gamma  \dttsym{\mbox{$\mid$}}
        \\\\
            \dttsym{(}  G  \dttsym{,}  \dttnt{n} \,  \preccurlyeq_{ \dttnt{p} }  \, \dttnt{n'}  \dttsym{)}  \dttsym{;}  \Gamma  \dttsym{,}  \dttnt{p} \, \dttnt{A}  \mathbin{@}  \dttnt{n'}  \vdash  \dttnt{p} \, \dttnt{B}  \mathbin{@}  \dttnt{n'}
      }{G  \dttsym{;}  \Gamma  \vdash  \dttnt{p} \, \dttsym{(}   \dttnt{A}  \ndto{ \dttnt{p} }  \dttnt{B}   \dttsym{)}  \mathbin{@}  \dttnt{n}}
      \leqno{\raise 8 pt\hbox{\textbf{Case}}}      
\]
  Suppose $\dttnt{p} = \dttsym{+}$.  Then by the induction hypothesis we know the following:
  \[
  \begin{array}{lll}
      \mathsf{L}( \Gamma )^{ \Lsym{+} }   \Lsym{,}  \Lmv{n'}  \Lsym{:}   \mathsf{L}( \Lnt{A} )   \vdash_{ \Lsym{(}    \mathsf{L}( G )   \cup  \Lsym{\{}  \Lsym{(}  \Lmv{n}  \Lsym{,}  \Lmv{n'}  \Lsym{)}  \Lsym{\}}   \Lsym{)} }  \Lmv{n'}  \Lsym{:}   \mathsf{L}( \Lnt{B} )   \Lsym{,}   \mathsf{L}( \Gamma )^{ \Lsym{-} }  
  \end{array}
  \]
  We know by assumption that $\dttnt{n'} \, \not\in \, \dttsym{\mbox{$\mid$}}  G  \dttsym{\mbox{$\mid$}}  \dttsym{,}  \dttsym{\mbox{$\mid$}}  \Gamma  \dttsym{\mbox{$\mid$}}$, and
  hence, $\dttnt{n'} \not\in |  \mathsf{L}( G )  | , |  \mathsf{L}( \Gamma )^{ \Lsym{+} }  |, | \mathsf{L}( \Gamma )^{ \Lsym{-} } |$
  by the definition of the L translation.  Therefore, our result
  follows by simply applying the L inference rule $\Ldrulename{impR}$.

  \ \\
  \noindent
  Suppose $\dttnt{p} = \dttsym{-}$.  This case follows similarly to the case
  when $\dttnt{p} = \dttsym{+}$, but we conclude with the L inference rule
  $\Ldrulename{subL}$.
\item[\cW] 
  \[
      \mprset{flushleft}
      \inferrule* [right=\ifrName{impBar}] {
        G  \vdash  \dttnt{n} \,  \preccurlyeq^*_{  \bar{  \dttnt{p}  }  }  \, \dttnt{n'}
        \\\\
             G  \dttsym{;}  \Gamma  \vdash   \bar{  \dttnt{p}  }  \, \dttnt{A}  \mathbin{@}  \dttnt{n'}  \qquad  G  \dttsym{;}  \Gamma  \vdash  \dttnt{p} \, \dttnt{B}  \mathbin{@}  \dttnt{n'} 
      }{G  \dttsym{;}  \Gamma  \vdash  \dttnt{p} \, \dttsym{(}   \dttnt{A}  \ndto{  \bar{  \dttnt{p}  }  }  \dttnt{B}   \dttsym{)}  \mathbin{@}  \dttnt{n}}
      \leqno{\raise 8 pt\hbox{\textbf{Case}}}      
\]
  Suppose $\dttnt{p} = \dttsym{+}$.  Then by the induction hypothesis we know the following:
  \[
  \begin{array}{rll}
    i. &   \mathsf{L}( \Gamma )^{ \Lsym{+} }   \Lsym{,}  \Lmv{n'}  \Lsym{:}   \mathsf{L}( \Lnt{A} )   \vdash_{  \mathsf{L}( G )  }   \mathsf{L}( \Gamma )^{ \Lsym{-} }  \\
    ii. &   \mathsf{L}( \Gamma )^{ \Lsym{+} }   \vdash_{  \mathsf{L}( G )  }  \Lmv{n'}  \Lsym{:}   \mathsf{L}( \Lnt{B} )   \Lsym{,}   \mathsf{L}( \Gamma )^{ \Lsym{-} }  
  \end{array}
  \]
  Furthermore, we know for any
  $\dttnt{n_{{\mathrm{1}}}}  \dttsym{,}  \dttnt{n_{{\mathrm{2}}}} \, \in \, \dttsym{\mbox{$\mid$}}  \dttnt{n} \,  \preccurlyeq_{ \dttnt{p'} }  \, \dttnt{n}  \dttsym{,}  G  \dttsym{\mbox{$\mid$}}  \dttsym{,}  \dttsym{\mbox{$\mid$}}  \Gamma  \dttsym{\mbox{$\mid$}}$
  if $G  \vdash  \dttnt{n_{{\mathrm{1}}}} \,  \preccurlyeq^*_{ \dttnt{p'} }  \, \dttnt{n_{{\mathrm{2}}}}$, then $\dttnt{n_{{\mathrm{1}}}} \,  \preccurlyeq_{ \dttnt{p'} }  \, \dttnt{n_{{\mathrm{2}}}} \in G$,
  and we know by assumption that $G  \vdash  \dttnt{n} \,  \preccurlyeq^*_{ \dttsym{-} }  \, \dttnt{n'}$, and thus,
  $\dttnt{n} \,  \preccurlyeq_{ \dttsym{-} }  \, \dttnt{n'} \in G$, hence, $ \Lmv{n'}    \mathsf{L}( G )    \Lmv{n} $ by the definition
  of the L-translation.

  \ \\
  \noindent
  It suffices to show that
  $  \mathsf{L}( \Gamma )^{ \Lsym{+} }   \vdash_{  \mathsf{L}( G )  }  \Lmv{n}  \Lsym{:}    \mathsf{L}( \Lnt{B} )   \prec   \mathsf{L}( \Lnt{A} )    \Lsym{,}   \mathsf{L}( \Gamma )^{ \Lsym{-} }  $, but this follows by
  applying the L inference rule $\Ldrulename{subR}$ using i and ii
  from above as well as the fact that we know $ \Lmv{n'}    \mathsf{L}( G )    \Lmv{n} $.

  \ \\
  \noindent
  Now suppose $p = -$.  Similar to the case when $p = +$, but we
  conclude with applying the L inference rule $\Ldrulename{impL}$,
  and the induction hypothesis provides the following:
  \[
  \begin{array}{rll}
    i. &   \mathsf{L}( \Gamma )^{ \Lsym{+} }   \vdash_{  \mathsf{L}( G )  }  \Lmv{n'}  \Lsym{:}   \mathsf{L}( \Lnt{A} )   \Lsym{,}   \mathsf{L}( \Gamma )^{ \Lsym{-} }  \\
    ii. &   \mathsf{L}( \Gamma )^{ \Lsym{+} }   \Lsym{,}  \Lmv{n'}  \Lsym{:}   \mathsf{L}( \Lnt{B} )   \vdash_{  \mathsf{L}( G )  }   \mathsf{L}( \Gamma )^{ \Lsym{-} }  .
  \end{array}
  \]
  
\item[\cW] 
  \[
      \mprset{flushleft}
      \inferrule* [right=\ifrName{axCut}] {
          \dttnt{p} \, \dttnt{B}  \mathbin{@}  \dttnt{n'}  \in  \dttsym{(}  \Gamma  \dttsym{,}   \bar{  \dttnt{p}  }  \, \dttnt{A}  \mathbin{@}  \dttnt{n}  \dttsym{)}   \qquad  G  \dttsym{;}  \Gamma  \dttsym{,}   \bar{  \dttnt{p}  }  \, \dttnt{A}  \mathbin{@}  \dttnt{n}  \vdash   \bar{  \dttnt{p}  }  \, \dttnt{B}  \mathbin{@}  \dttnt{n'} 
      }{G  \dttsym{;}  \Gamma  \vdash  \dttnt{p} \, \dttnt{A}  \mathbin{@}  \dttnt{n}}
      \leqno{\raise 8 pt\hbox{\textbf{Case}}}      
\]
  Suppose $\dttnt{p} = \dttsym{+}$.  Then by the induction hypothesis we know the following:
  \[
    \mathsf{L}( \Gamma )^{ \Lsym{+} }   \Lsym{,}  \Lmv{n'}  \Lsym{:}   \mathsf{L}( \Lnt{B} )   \vdash_{  \mathsf{L}( G )  }  \Lmv{n}  \Lsym{:}   \mathsf{L}( \Lnt{A} )   \Lsym{,}   \mathsf{L}( \Gamma )^{ \Lsym{-} }  
  \]
  Now we know that $ \dttnt{p} \, \dttnt{B}  \mathbin{@}  \dttnt{n'}  \in  \dttsym{(}  \Gamma  \dttsym{,}   \bar{  \dttnt{p}  }  \, \dttnt{A}  \mathbin{@}  \dttnt{n}  \dttsym{)} $, and hence,
  $\Lmv{n'}  \Lsym{:}   \mathsf{L}( \Lnt{B} )  \in  \mathsf{L}( \Gamma )^{ \dttsym{+} } $, which implies we know the following:
  \[
    \mathsf{L}( \Gamma_{{\mathrm{1}}} )^{ \Lsym{+} }   \Lsym{,}  \Lmv{n'}  \Lsym{:}   \mathsf{L}( \Lnt{B} )   \Lsym{,}   \mathsf{L}( \Gamma_{{\mathrm{2}}} )^{ \Lsym{+} }   \Lsym{,}  \Lmv{n'}  \Lsym{:}   \mathsf{L}( \Lnt{B} )   \vdash_{  \mathsf{L}( G )  }  \Lmv{n}  \Lsym{:}   \mathsf{L}( \Lnt{A} )   \Lsym{,}   \mathsf{L}( \Gamma )^{ \Lsym{-} }  
  \]
  Therefore, by applying the admissible L inference rule $\Ldrulename{contrL}$ we know the following:
  \[
    \mathsf{L}( \Gamma_{{\mathrm{1}}} )^{ \Lsym{+} }   \Lsym{,}  \Lmv{n'}  \Lsym{:}   \mathsf{L}( \Lnt{B} )   \Lsym{,}   \mathsf{L}( \Gamma_{{\mathrm{2}}} )^{ \Lsym{+} }   \vdash_{  \mathsf{L}( G )  }  \Lmv{n}  \Lsym{:}   \mathsf{L}( \Lnt{A} )   \Lsym{,}   \mathsf{L}( \Gamma )^{ \Lsym{-} }  
  \]
  This is equivalent to our result:
  \[
    \mathsf{L}( \Gamma )^{ \Lsym{+} }   \vdash_{  \mathsf{L}( G )  }  \Lmv{n}  \Lsym{:}   \mathsf{L}( \Lnt{A} )   \Lsym{,}   \mathsf{L}( \Gamma )^{ \Lsym{-} }  
  \]
  
  \ \\
  \noindent
  Suppose $\dttnt{p} = \dttsym{-}$.  Then by the induction hypothesis we know the following:
  \[
    \mathsf{L}( \Gamma )^{ \Lsym{+} }   \Lsym{,}  \Lmv{n'}  \Lsym{:}   \mathsf{L}( \Lnt{A} )   \vdash_{  \mathsf{L}( G )  }  \Lmv{n}  \Lsym{:}   \mathsf{L}( \Lnt{B} )   \Lsym{,}   \mathsf{L}( \Gamma )^{ \Lsym{-} }  
  \]
  This case now follows similarly to the previous case by exposing
  $\Lmv{n}  \Lsym{:}   \mathsf{L}( \Lnt{B} ) $ in $ \mathsf{L}( \Gamma )^{ \Lsym{-} } $, and then using contraction on the
  right.

\item[\cW] 
  \[
      \mprset{flushleft}
      \inferrule* [right=\ifrName{axCutBar}] {
           \bar{  \dttnt{p}  }  \, \dttnt{B}  \mathbin{@}  \dttnt{n'}  \in  \dttsym{(}  \Gamma  \dttsym{,}   \bar{  \dttnt{p}  }  \, \dttnt{A}  \mathbin{@}  \dttnt{n}  \dttsym{)}   \qquad  G  \dttsym{;}  \Gamma  \dttsym{,}   \bar{  \dttnt{p}  }  \, \dttnt{A}  \mathbin{@}  \dttnt{n}  \vdash  \dttnt{p} \, \dttnt{B}  \mathbin{@}  \dttnt{n'} 
      }{G  \dttsym{;}  \Gamma  \vdash  \dttnt{p} \, \dttnt{A}  \mathbin{@}  \dttnt{n}}
      \leqno{\raise 8 pt\hbox{\textbf{Case}}}      
\]
This case is similar to the previous case except in the case when
$\dttnt{p} = \dttsym{+}$ we use contraction on the right, and then when
$\dttnt{p} = \dttsym{-}$ we use contraction on the left.
\end{description}

\subsection{Proof of Lemma~\ref{lemma:dil-validity_is_l-validity}: DIL-validity is L-validity}
\label{subsec:proof_of_dil-validity_is_l-validity}
Suppose $ \interp{ G  \dttsym{;}  \Gamma  \vdash  \dttnt{p} \, \dttnt{A}  \mathbin{@}  \dttnt{n} }_{N} $ holds for some Kripke model
$(W,R,V)$ and node interpreter $N$ on $|G|$, and $\dttnt{p} = \dttsym{+}$.  It suffices to show that
$  \mathsf{L}( \Gamma )^{ \Lsym{+} }   \vdash_{  \mathsf{L}( G )  }  \Lmv{n}  \Lsym{:}   \mathsf{L}( \Lnt{A} )   \Lsym{,}   \mathsf{L}( \Gamma )^{ \Lsym{-} }  $ is L-valid.
By the definition of the interpretation of DIL-sequents (Definition~\ref{def:validity}) we know that
\[
\text{ if }  \interp{ G }_{N}  \text{ and }  \interp{ \Gamma }_{N} , \text{ then }  \dttnt{p}  \interp{ \dttnt{A} }_{ \dttsym{(}   N\, \dttnt{n}   \dttsym{)} } 
\]
Now to show that $  \mathsf{L}( \Gamma )^{ \Lsym{+} }   \vdash_{  \mathsf{L}( G )  }  \Lmv{n}  \Lsym{:}   \mathsf{L}( \Lnt{A} )   \Lsym{,}   \mathsf{L}( \Gamma )^{ \Lsym{-} }  $ is L-valid we
must show that at least one of the following does not hold:
\begin{center}
  \begin{tabular}{rll}
    i.   & for any $ \Lmv{n_{{\mathrm{1}}}}    \mathsf{L}( G )    \Lmv{n_{{\mathrm{2}}}} $, $ R\, \dttsym{(}   N\, \dttnt{n_{{\mathrm{1}}}}   \dttsym{)} \, \dttsym{(}   N\, \dttnt{n_{{\mathrm{2}}}}   \dttsym{)} $\\
    ii.  & for any $\Lmv{n}  \Lsym{:}   \mathsf{L}( \Lnt{B} )  \in  \mathsf{L}( \Gamma )^{ \Lsym{+} } , \interp{ \mathsf{L}( \Lnt{B} ) }_{ N\, \dttnt{n} }$\\
    iii. & for any $\Lmv{n}  \Lsym{:}   \mathsf{L}( \Lnt{B} )  \in (\Lmv{n}  \Lsym{:}   \mathsf{L}( \Lnt{A} )   \Lsym{,}   \mathsf{L}( \Gamma )^{ \Lsym{-} } ), \lnot \interp{ \mathsf{L}( \Lnt{B} ) }_{ N\, \dttnt{n} }$\\
  \end{tabular}
\end{center}
So if neither $ \interp{ G }_{N} $ or $ \interp{ \Gamma }_{N} $ hold, then neither of i or ii will hold. Thus,
$  \mathsf{L}( \Gamma )^{ \Lsym{+} }   \vdash_{  \mathsf{L}( G )  }  \Lmv{n}  \Lsym{:}   \mathsf{L}( \Lnt{A} )   \Lsym{,}   \mathsf{L}( \Gamma )^{ \Lsym{-} }  $ is L-valid.

\ \\
\noindent
So assume $ \interp{ G }_{N} $ or $ \interp{ \Gamma }_{N} $ hold.  Then both i and ii are satisfied by Lemma~\ref{lemma:DIL-interp-L-interps}.
However, we now know $ \dttsym{+}  \interp{ \dttnt{A} }_{ \dttsym{(}   N\, \dttnt{n}   \dttsym{)} }  =  \interp{ \dttnt{A} }_{ \dttsym{(}   N\, \dttnt{n}   \dttsym{)} } $ holds, and hence by Lemma~\ref{lemma:DIL-interp-L-interps},
iii does not hold.  Therefore, $  \mathsf{L}( \Gamma )^{ \Lsym{+} }   \vdash_{  \mathsf{L}( G )  }  \Lmv{n}  \Lsym{:}   \mathsf{L}( \Lnt{A} )   \Lsym{,}   \mathsf{L}( \Gamma )^{ \Lsym{-} }  $ is L-valid.

\ \\
\noindent
Now suppose $\dttnt{p} = \dttsym{-}$.  It suffices to show that
$  \mathsf{L}( \Gamma )^{ \Lsym{+} }   \Lsym{,}  \Lmv{n}  \Lsym{:}   \mathsf{L}( \Lnt{A} )   \vdash_{  \mathsf{L}( G )  }   \mathsf{L}( \Gamma )^{ \Lsym{-} }  $ is L-valid.  However, notice that we must
show that at least one of the following does not hold:
\begin{center}
  \begin{tabular}{rll}
    i.   & for any $ \Lmv{n_{{\mathrm{1}}}}    \mathsf{L}( G )    \Lmv{n_{{\mathrm{2}}}} $, $ R\, \dttsym{(}   N\, \dttnt{n_{{\mathrm{1}}}}   \dttsym{)} \, \dttsym{(}   N\, \dttnt{n_{{\mathrm{2}}}}   \dttsym{)} $\\
    ii.  & for any $\Lmv{n}  \Lsym{:}   \mathsf{L}( \Lnt{B} )  \in ( \mathsf{L}( \Gamma )^{ \Lsym{+} }   \Lsym{,}  \Lmv{n}  \Lsym{:}   \mathsf{L}( \Lnt{A} ) ), \interp{ \mathsf{L}( \Lnt{B} ) }_{ N\, \dttnt{n} }$\\
    iii. & for any $\Lmv{n}  \Lsym{:}   \mathsf{L}( \Lnt{B} )  \in  \mathsf{L}( \Gamma )^{ \Lsym{-} } , \lnot \interp{ \mathsf{L}( \Lnt{B} ) }_{ N\, \dttnt{n} }$\\
  \end{tabular}
\end{center}
However, notice that ii will allows be false in this case, because if $ \interp{ G }_{N} $ or $ \interp{ \Gamma }_{N} $ hold,
then we know $ \dttsym{-}  \interp{ \dttnt{A} }_{ \dttsym{(}   N\, \dttnt{n}   \dttsym{)} }  = \lnot  \interp{ \dttnt{A} }_{ \dttsym{(}   N\, \dttnt{n}   \dttsym{)} } $, which implies that
$\lnot  \interp{  \mathsf{L}( \dttnt{A} )  }_{ \dttsym{(}   N\, \dttnt{n}   \dttsym{)} } $.  Therefore, $  \mathsf{L}( \Gamma )^{ \Lsym{+} }   \Lsym{,}  \Lmv{n}  \Lsym{:}   \mathsf{L}( \Lnt{A} )   \vdash_{  \mathsf{L}( G )  }   \mathsf{L}( \Gamma )^{ \Lsym{-} }  $ is L-valid.


\section{Proofs from Section~\ref{sec:dualized_type_theory_(dtt)}: Dualized Type Theory}
\label{sec:proofs_from_section_dualized_type_theory}

\subsection{Proof of Lemma~\ref{lemma:disj-elim-adm}}
\label{subsec:proof_of_lemma:disj-elim-adm}
Due to the size of the derivation in question we give several
derivations that combine to form the typing derivation of \\
$G  \dttsym{;}  \Gamma  \vdash   \mathbf{case}\, \dttnt{t} \,\mathbf{of}\, \dttmv{x} . \dttnt{t_{{\mathrm{1}}}} , \dttmv{x} . \dttnt{t_{{\mathrm{2}}}}   \dttsym{:}  \dttnt{p} \, \dttnt{C}  \mathbin{@}  \dttnt{n}$.\ \\

\noindent
The typing derivation begins using cut as follows:

\begin{math}
  $$\mprset{flushleft}
  \inferrule* [right=\ifrName{cut}] {
    D_0
    \\
    D_1
  }{G  \dttsym{;}  \Gamma  \vdash  \nu \, \dttmv{z_{{\mathrm{0}}}}  \dttsym{.}  \dttsym{(}  \nu \, \dttmv{z_{{\mathrm{1}}}}  \dttsym{.}  \dttsym{(}  \nu \, \dttmv{z_{{\mathrm{2}}}}  \dttsym{.}  \dttnt{t}  \mathbin{\Cdot[2]}  \dttsym{(}  \dttmv{z_{{\mathrm{1}}}}  \dttsym{,}  \dttmv{z_{{\mathrm{2}}}}  \dttsym{)}  \dttsym{)}  \mathbin{\Cdot[2]}  \dttsym{(}  \nu \, \dttmv{x}  \dttsym{.}  \dttnt{t_{{\mathrm{2}}}}  \mathbin{\Cdot[2]}  \dttmv{z_{{\mathrm{0}}}}  \dttsym{)}  \dttsym{)}  \mathbin{\Cdot[2]}  \dttsym{(}  \nu \, \dttmv{x}  \dttsym{.}  \dttnt{t_{{\mathrm{1}}}}  \mathbin{\Cdot[2]}  \dttmv{z_{{\mathrm{0}}}}  \dttsym{)}  \dttsym{:}  \dttsym{+} \, \dttnt{C}  \mathbin{@}  \dttnt{n}}
\end{math} \\

\noindent
Then the remainder of the derivation depends on the following sub-derivations:

\begin{math}
  \begin{array}{lcl}
    D_0:\\
    &
    $$\mprset{flushleft}
    \inferrule* [right= \ifrName{cut}] {
      D_3
      \\
      D_4
    }{G  \dttsym{;}  \Gamma  \dttsym{,}  \dttmv{z_{{\mathrm{0}}}}  \dttsym{:}  \dttsym{-} \, \dttnt{C}  \mathbin{@}  \dttnt{n}  \vdash  \nu \, \dttmv{z_{{\mathrm{1}}}}  \dttsym{.}  \dttsym{(}  \nu \, \dttmv{z_{{\mathrm{2}}}}  \dttsym{.}  \dttnt{t}  \mathbin{\Cdot[2]}  \dttsym{(}  \dttmv{z_{{\mathrm{1}}}}  \dttsym{,}  \dttmv{z_{{\mathrm{2}}}}  \dttsym{)}  \dttsym{)}  \mathbin{\Cdot[2]}  \dttsym{(}  \nu \, \dttmv{x}  \dttsym{.}  \dttnt{t_{{\mathrm{2}}}}  \mathbin{\Cdot[2]}  \dttmv{z_{{\mathrm{0}}}}  \dttsym{)}  \dttsym{:}  \dttsym{+} \, \dttnt{A}  \mathbin{@}  \dttnt{n}}
  \end{array}
\end{math} \\

\begin{math}
  \begin{array}{lcc}
    D_1: \\      
    &       
    $$\mprset{flushleft}
    \inferrule* [right= \ifrName{cut}] {
      D_2
      \\
      $$\mprset{flushleft}
      \inferrule* [right= \ifrName{ax}] {
        \,
      }{G  \dttsym{;}  \Gamma  \dttsym{,}  \dttmv{z_{{\mathrm{0}}}}  \dttsym{:}  \dttsym{-} \, \dttnt{C}  \mathbin{@}  \dttnt{n}  \dttsym{,}  \dttmv{x}  \dttsym{:}  \dttsym{+} \, \dttnt{A}  \mathbin{@}  \dttnt{n}  \vdash  \dttmv{z_{{\mathrm{0}}}}  \dttsym{:}  \dttsym{-} \, \dttnt{C}  \mathbin{@}  \dttnt{n}}               
    }{G  \dttsym{;}  \Gamma  \dttsym{,}  \dttmv{z_{{\mathrm{0}}}}  \dttsym{:}  \dttsym{-} \, \dttnt{C}  \mathbin{@}  \dttnt{n}  \vdash  \nu \, \dttmv{x}  \dttsym{.}  \dttnt{t_{{\mathrm{1}}}}  \mathbin{\Cdot[2]}  \dttmv{z_{{\mathrm{0}}}}  \dttsym{:}  \dttsym{-} \, \dttnt{A}  \mathbin{@}  \dttnt{n}}
  \end{array}
\end{math} \\

\begin{math}
  \begin{array}{lcc}
    D_2: \\      
    & 
    $$\mprset{flushleft}
    \inferrule* [right= \ifrName{weakening}] {
      G  \dttsym{;}  \Gamma  \dttsym{,}  \dttmv{x}  \dttsym{:}  \dttsym{+} \, \dttnt{A}  \mathbin{@}  \dttnt{n}  \vdash  \dttnt{t_{{\mathrm{1}}}}  \dttsym{:}  \dttsym{+} \, \dttnt{C}  \mathbin{@}  \dttnt{n}
    }{G  \dttsym{;}  \Gamma  \dttsym{,}  \dttmv{z_{{\mathrm{0}}}}  \dttsym{:}  \dttsym{-} \, \dttnt{C}  \mathbin{@}  \dttnt{n}  \dttsym{,}  \dttmv{x}  \dttsym{:}  \dttsym{+} \, \dttnt{A}  \mathbin{@}  \dttnt{n}  \vdash  \dttnt{t_{{\mathrm{1}}}}  \dttsym{:}  \dttsym{+} \, \dttnt{C}  \mathbin{@}  \dttnt{n}}
  \end{array}
\end{math} \\  

\begin{math}
  \begin{array}{lcl}
    D_4:\\
    &
    $$\mprset{flushleft}
    \inferrule* [right= \ifrName{cut}] {
      D_5
      \\        
        G  \dttsym{;}  \Gamma  \dttsym{,}  \dttmv{z_{{\mathrm{0}}}}  \dttsym{:}  \dttsym{-} \, \dttnt{C}  \mathbin{@}  \dttnt{n}  \dttsym{,}  \dttmv{z_{{\mathrm{1}}}}  \dttsym{:}  \dttsym{-} \, \dttnt{A}  \mathbin{@}  \dttnt{n}  \dttsym{,}  \dttmv{x}  \dttsym{:}  \dttsym{+} \, \dttnt{B}  \mathbin{@}  \dttnt{n}  \vdash  \dttmv{z_{{\mathrm{0}}}}  \dttsym{:}  \dttsym{-} \, \dttnt{C}  \mathbin{@}  \dttnt{n}
    }{G  \dttsym{;}  \Gamma  \dttsym{,}  \dttmv{z_{{\mathrm{0}}}}  \dttsym{:}  \dttsym{-} \, \dttnt{C}  \mathbin{@}  \dttnt{n}  \dttsym{,}  \dttmv{z_{{\mathrm{1}}}}  \dttsym{:}  \dttsym{-} \, \dttnt{A}  \mathbin{@}  \dttnt{n}  \vdash  \nu \, \dttmv{x}  \dttsym{.}  \dttnt{t_{{\mathrm{2}}}}  \mathbin{\Cdot[2]}  \dttmv{z_{{\mathrm{0}}}}  \dttsym{:}  \dttsym{-} \, \dttnt{B}  \mathbin{@}  \dttnt{n}}
  \end{array}
\end{math} \\

\begin{math}
  \begin{array}{lcl}
    D_3:\\
    &
    $$\mprset{flushleft}
    \inferrule* [right= \ifrName{cut}] {
      D_6
      \\
      D_7
    }{G  \dttsym{;}  \Gamma  \dttsym{,}  \dttmv{z_{{\mathrm{0}}}}  \dttsym{:}  \dttsym{-} \, \dttnt{C}  \mathbin{@}  \dttnt{n}  \dttsym{,}  \dttmv{z_{{\mathrm{1}}}}  \dttsym{:}  \dttsym{-} \, \dttnt{A}  \mathbin{@}  \dttnt{n}  \vdash  \nu \, \dttmv{z_{{\mathrm{2}}}}  \dttsym{.}  \dttnt{t}  \mathbin{\Cdot[2]}  \dttsym{(}  \dttmv{z_{{\mathrm{1}}}}  \dttsym{,}  \dttmv{z_{{\mathrm{2}}}}  \dttsym{)}  \dttsym{:}  \dttsym{+} \, \dttnt{B}  \mathbin{@}  \dttnt{n}}
  \end{array}
\end{math} \\

\begin{math}
  \begin{array}{lcl}
    D_5:\\
    &
    $$\mprset{flushleft}
    \inferrule* [right= \ifrName{weakening}] {
      G  \dttsym{;}  \Gamma  \dttsym{,}  \dttmv{x}  \dttsym{:}  \dttsym{+} \, \dttnt{B}  \mathbin{@}  \dttnt{n}  \vdash  \dttnt{t_{{\mathrm{2}}}}  \dttsym{:}  \dttsym{+} \, \dttnt{C}  \mathbin{@}  \dttnt{n}
    }{G  \dttsym{;}  \Gamma  \dttsym{,}  \dttmv{z_{{\mathrm{0}}}}  \dttsym{:}  \dttsym{-} \, \dttnt{C}  \mathbin{@}  \dttnt{n}  \dttsym{,}  \dttmv{z_{{\mathrm{1}}}}  \dttsym{:}  \dttsym{-} \, \dttnt{A}  \mathbin{@}  \dttnt{n}  \dttsym{,}  \dttmv{x}  \dttsym{:}  \dttsym{+} \, \dttnt{B}  \mathbin{@}  \dttnt{n}  \vdash  \dttnt{t_{{\mathrm{2}}}}  \dttsym{:}  \dttsym{+} \, \dttnt{C}  \mathbin{@}  \dttnt{n}}
  \end{array}
\end{math} \\

\begin{math}
  \begin{array}{lcl}
    D_6:\\
    &
    $$\mprset{flushleft}
    \inferrule* [right= \ifrName{weakening}] {
      G  \dttsym{;}  \Gamma  \vdash  \dttnt{t}  \dttsym{:}  \dttsym{+} \, \dttsym{(}   \dttnt{A}  \ndwedge{ \dttsym{-} }  \dttnt{B}   \dttsym{)}  \mathbin{@}  \dttnt{n}
    }{G  \dttsym{;}  \Gamma  \dttsym{,}  \dttmv{z_{{\mathrm{0}}}}  \dttsym{:}  \dttsym{-} \, \dttnt{C}  \mathbin{@}  \dttnt{n}  \dttsym{,}  \dttmv{z_{{\mathrm{1}}}}  \dttsym{:}  \dttsym{-} \, \dttnt{A}  \mathbin{@}  \dttnt{n}  \dttsym{,}  \dttmv{z_{{\mathrm{2}}}}  \dttsym{:}  \dttsym{-} \, \dttnt{B}  \mathbin{@}  \dttnt{n}  \vdash  \dttnt{t}  \dttsym{:}  \dttsym{+} \, \dttsym{(}   \dttnt{A}  \ndwedge{ \dttsym{-} }  \dttnt{B}   \dttsym{)}  \mathbin{@}  \dttnt{n}}
  \end{array}
\end{math} \\

\begin{math}
  \begin{array}{lcl}
    D_7:\\
    &
    $$\mprset{flushleft}
    \inferrule* [right= \ifrName{and}] {
      D_8 
      \\
      D_9
    }{G  \dttsym{;}  \Gamma  \dttsym{,}  \dttmv{z_{{\mathrm{0}}}}  \dttsym{:}  \dttsym{-} \, \dttnt{C}  \mathbin{@}  \dttnt{n}  \dttsym{,}  \dttmv{z_{{\mathrm{1}}}}  \dttsym{:}  \dttsym{-} \, \dttnt{A}  \mathbin{@}  \dttnt{n}  \dttsym{,}  \dttmv{z_{{\mathrm{2}}}}  \dttsym{:}  \dttsym{-} \, \dttnt{B}  \mathbin{@}  \dttnt{n}  \vdash  \dttsym{(}  \dttmv{z_{{\mathrm{1}}}}  \dttsym{,}  \dttmv{z_{{\mathrm{2}}}}  \dttsym{)}  \dttsym{:}  \dttsym{-} \, \dttsym{(}   \dttnt{A}  \ndwedge{ \dttsym{-} }  \dttnt{B}   \dttsym{)}  \mathbin{@}  \dttnt{n}}
  \end{array}
\end{math} \\

\begin{math}
  \begin{array}{lcl}
    D_8:\\
    &
    $$\mprset{flushleft}
    \inferrule* [right= \ifrName{ax}] {
      \,
    }{G  \dttsym{;}  \Gamma  \dttsym{,}  \dttmv{z_{{\mathrm{0}}}}  \dttsym{:}  \dttsym{-} \, \dttnt{C}  \mathbin{@}  \dttnt{n}  \dttsym{,}  \dttmv{z_{{\mathrm{1}}}}  \dttsym{:}  \dttsym{-} \, \dttnt{A}  \mathbin{@}  \dttnt{n}  \dttsym{,}  \dttmv{z_{{\mathrm{2}}}}  \dttsym{:}  \dttsym{-} \, \dttnt{B}  \mathbin{@}  \dttnt{n}  \vdash  \dttmv{z_{{\mathrm{1}}}}  \dttsym{:}  \dttsym{-} \, \dttnt{A}  \mathbin{@}  \dttnt{n}}
  \end{array}
\end{math} \\

\begin{math}
  \begin{array}{lcl}
    D_9:\\
    &
    $$\mprset{flushleft}
    \inferrule* [right= \ifrName{ax}] {
      \,
    }{G  \dttsym{;}  \Gamma  \dttsym{,}  \dttmv{z_{{\mathrm{0}}}}  \dttsym{:}  \dttsym{-} \, \dttnt{C}  \mathbin{@}  \dttnt{n}  \dttsym{,}  \dttmv{z_{{\mathrm{1}}}}  \dttsym{:}  \dttsym{-} \, \dttnt{A}  \mathbin{@}  \dttnt{n}  \dttsym{,}  \dttmv{z_{{\mathrm{2}}}}  \dttsym{:}  \dttsym{-} \, \dttnt{B}  \mathbin{@}  \dttnt{n}  \vdash  \dttmv{z_{{\mathrm{2}}}}  \dttsym{:}  \dttsym{-} \, \dttnt{B}  \mathbin{@}  \dttnt{n}}
  \end{array}
\end{math}

\section{Proofs from Section~\ref{sec:metatheory_of_dtt}: Metatheory of DTT}
\label{sec:proofs_from_section_metatheory_of_dtt}

\subsection{Proof of Lemma~\ref{lemma:renaming_nodes_in_graph}: Node Renaming}
\label{subsec:proof_of_lemma_node_renaming}
This is a proof by induction on the assumed reachability
derivation.  Throughout each case suppose we have nodes $\dttnt{n_{{\mathrm{4}}}}$
and $\dttnt{n_{{\mathrm{5}}}}$.

\begin{description}
\item[\cW] 
  \[
      \mprset{flushleft}
      \inferrule* [right=\ifrName{ax}] {
        \ 
      }{G  \dttsym{,}  \dttnt{n_{{\mathrm{1}}}} \,  \preccurlyeq_{ \dttnt{p} }  \, \dttnt{n_{{\mathrm{3}}}}  \dttsym{,}  G'  \vdash  \dttnt{n_{{\mathrm{1}}}} \,  \preccurlyeq^*_{ \dttnt{p} }  \, \dttnt{n_{{\mathrm{3}}}}}
      \leqno{\raise 8 pt\hbox{\textbf{Case}}}      
    \]
  Trivial.

\item[\cW] 
  \[
      \mprset{flushleft}
      \inferrule* [right=\ifrName{refl}] {
        \ 
      }{G_{{\mathrm{1}}}  \dttsym{,}  G_{{\mathrm{2}}}  \vdash  \dttnt{n} \,  \preccurlyeq^*_{ \dttnt{p} }  \, \dttnt{n}}
      \leqno{\raise 8 pt\hbox{\textbf{Case}}}      
    \]
  Trivial.
  
\item[\cW] 
  \[
      \mprset{flushleft}
      \inferrule* [right=\ifrName{trans}] {
         G_{{\mathrm{1}}}  \dttsym{,}  G_{{\mathrm{2}}}  \vdash  \dttnt{n_{{\mathrm{1}}}} \,  \preccurlyeq^*_{ \dttnt{p} }  \, \dttnt{n'}  \qquad  G_{{\mathrm{1}}}  \dttsym{,}  G_{{\mathrm{2}}}  \vdash  \dttnt{n'} \,  \preccurlyeq^*_{ \dttnt{p} }  \, \dttnt{n_{{\mathrm{3}}}} 
      }{G_{{\mathrm{1}}}  \dttsym{,}  G_{{\mathrm{2}}}  \vdash  \dttnt{n_{{\mathrm{1}}}} \,  \preccurlyeq^*_{ \dttnt{p} }  \, \dttnt{n_{{\mathrm{3}}}}}
      \leqno{\raise 8 pt\hbox{\textbf{Case}}}      
    \]
  By the induction hypothesis we know that for any nodes $\dttnt{n'_{{\mathrm{4}}}}$ and $\dttnt{n'_{{\mathrm{5}}}}$ we have\\
  $\dttsym{[}  \dttnt{n'_{{\mathrm{4}}}}  \dttsym{/}  \dttnt{n'_{{\mathrm{5}}}}  \dttsym{]}  G_{{\mathrm{1}}}  \dttsym{,}  \dttsym{[}  \dttnt{n'_{{\mathrm{4}}}}  \dttsym{/}  \dttnt{n'_{{\mathrm{5}}}}  \dttsym{]}  G_{{\mathrm{2}}}  \vdash  \dttsym{[}  \dttnt{n'_{{\mathrm{4}}}}  \dttsym{/}  \dttnt{n'_{{\mathrm{5}}}}  \dttsym{]}  \dttnt{n_{{\mathrm{1}}}} \,  \preccurlyeq^*_{ \dttnt{p} }  \, \dttsym{[}  \dttnt{n'_{{\mathrm{4}}}}  \dttsym{/}  \dttnt{n'_{{\mathrm{5}}}}  \dttsym{]}  \dttnt{n'}$, and for any nodes $\dttnt{n''_{{\mathrm{4}}}}$ and $\dttnt{n''_{{\mathrm{5}}}}$ 
  we have 
  $\dttsym{[}  \dttnt{n''_{{\mathrm{4}}}}  \dttsym{/}  \dttnt{n''_{{\mathrm{5}}}}  \dttsym{]}  G_{{\mathrm{1}}}  \dttsym{,}  \dttsym{[}  \dttnt{n''_{{\mathrm{4}}}}  \dttsym{/}  \dttnt{n''_{{\mathrm{5}}}}  \dttsym{]}  G_{{\mathrm{2}}}  \vdash  \dttsym{[}  \dttnt{n''_{{\mathrm{4}}}}  \dttsym{/}  \dttnt{n''_{{\mathrm{5}}}}  \dttsym{]}  \dttnt{n'} \,  \preccurlyeq^*_{ \dttnt{p} }  \, \dttsym{[}  \dttnt{n''_{{\mathrm{4}}}}  \dttsym{/}  \dttnt{n''_{{\mathrm{5}}}}  \dttsym{]}  \dttnt{n_{{\mathrm{3}}}}$.  Choose $\dttnt{n_{{\mathrm{4}}}}$ for $\dttnt{n'_{{\mathrm{4}}}}$ and $\dttnt{n''_{{\mathrm{4}}}}$ and
  $\dttnt{n_{{\mathrm{5}}}}$ for $\dttnt{n'_{{\mathrm{5}}}}$ and $\dttnt{n''_{{\mathrm{5}}}}$ to obtain $\dttsym{[}  \dttnt{n_{{\mathrm{4}}}}  \dttsym{/}  \dttnt{n_{{\mathrm{5}}}}  \dttsym{]}  G_{{\mathrm{1}}}  \dttsym{,}  \dttsym{[}  \dttnt{n_{{\mathrm{4}}}}  \dttsym{/}  \dttnt{n_{{\mathrm{5}}}}  \dttsym{]}  G_{{\mathrm{2}}}  \vdash  \dttsym{[}  \dttnt{n_{{\mathrm{4}}}}  \dttsym{/}  \dttnt{n_{{\mathrm{5}}}}  \dttsym{]}  \dttnt{n_{{\mathrm{1}}}} \,  \preccurlyeq^*_{ \dttnt{p} }  \, \dttsym{[}  \dttnt{n_{{\mathrm{4}}}}  \dttsym{/}  \dttnt{n_{{\mathrm{5}}}}  \dttsym{]}  \dttnt{n'}$ and
  $\dttsym{[}  \dttnt{n_{{\mathrm{4}}}}  \dttsym{/}  \dttnt{n_{{\mathrm{5}}}}  \dttsym{]}  G_{{\mathrm{1}}}  \dttsym{,}  \dttsym{[}  \dttnt{n_{{\mathrm{4}}}}  \dttsym{/}  \dttnt{n_{{\mathrm{5}}}}  \dttsym{]}  G_{{\mathrm{2}}}  \vdash  \dttsym{[}  \dttnt{n_{{\mathrm{4}}}}  \dttsym{/}  \dttnt{n_{{\mathrm{5}}}}  \dttsym{]}  \dttnt{n'} \,  \preccurlyeq^*_{ \dttnt{p} }  \, \dttsym{[}  \dttnt{n_{{\mathrm{4}}}}  \dttsym{/}  \dttnt{n_{{\mathrm{5}}}}  \dttsym{]}  \dttnt{n_{{\mathrm{3}}}}$.  Finally, this case follows by reapplying the 
  rule to the previous two facts.

\item[\cW] 
  \[
      \mprset{flushleft}
      \inferrule* [right=\ifrName{flip}] {
        G  \vdash  \dttnt{n'} \,  \preccurlyeq^*_{  \bar{  \dttnt{p}  }  }  \, \dttnt{n}
      }{G  \vdash  \dttnt{n} \,  \preccurlyeq^*_{ \dttnt{p} }  \, \dttnt{n'}}
      \leqno{\raise 8 pt\hbox{\textbf{Case}}}      
    \]
  Similar to the previous case.

\end{description}

\subsection{Proof of Lemma~\ref{lemma:node_substitution_for_reachability}: Node Substitution for Reachability}
\label{subsec:proof_of_lemma_node_substitution_for_reachability}
This is a proof by induction on the form of the assumed reachability
derivation.  Throughout the following cases we assume $G  \dttsym{,}  G'  \vdash  \dttnt{n_{{\mathrm{1}}}} \,  \preccurlyeq^*_{ \dttnt{p_{{\mathrm{1}}}} }  \, \dttnt{n_{{\mathrm{3}}}}$ holds.

\begin{description}
\item[\cW] 
  \[
      \mprset{flushleft}
      \inferrule* [right=\ifrName{ax}] {
        \ 
      }{G_{{\mathrm{1}}}  \dttsym{,}  \dttnt{n_{{\mathrm{4}}}} \,  \preccurlyeq_{ \dttnt{p} }  \, \dttnt{n_{{\mathrm{5}}}}  \dttsym{,}  G_{{\mathrm{2}}}  \vdash  \dttnt{n_{{\mathrm{4}}}} \,  \preccurlyeq^*_{ \dttnt{p} }  \, \dttnt{n_{{\mathrm{5}}}}}
      \leqno{\raise 8 pt\hbox{\textbf{Case}}}      
    \]
  Suppose $G_{{\mathrm{1}}}  \dttsym{,}  \dttnt{n_{{\mathrm{4}}}} \,  \preccurlyeq_{ \dttnt{p} }  \, \dttnt{n_{{\mathrm{5}}}}  \dttsym{,}  G_{{\mathrm{2}}} = G  \dttsym{,}  \dttnt{n_{{\mathrm{1}}}} \,  \preccurlyeq_{ \dttnt{p_{{\mathrm{1}}}} }  \, \dttnt{n_{{\mathrm{2}}}}  \dttsym{,}  G'$.  Then
  either $\dttnt{n_{{\mathrm{1}}}} \,  \preccurlyeq_{ \dttnt{p_{{\mathrm{1}}}} }  \, \dttnt{n_{{\mathrm{2}}}} \in G_{{\mathrm{1}}}$, $\dttnt{n_{{\mathrm{1}}}} \,  \preccurlyeq_{ \dttnt{p_{{\mathrm{1}}}} }  \, \dttnt{n_{{\mathrm{2}}}} \in G_{{\mathrm{2}}}$,
  or $\dttnt{n_{{\mathrm{1}}}} \,  \preccurlyeq_{ \dttnt{p_{{\mathrm{1}}}} }  \, \dttnt{n_{{\mathrm{2}}}} \equiv \dttnt{n_{{\mathrm{4}}}} \,  \preccurlyeq_{ \dttnt{p} }  \, \dttnt{n_{{\mathrm{5}}}}$.  Suppose $\dttnt{n_{{\mathrm{1}}}} \,  \preccurlyeq_{ \dttnt{p_{{\mathrm{1}}}} }  \, \dttnt{n_{{\mathrm{2}}}} \in G_{{\mathrm{1}}}$,
  then $G_{{\mathrm{1}}} = G'_{{\mathrm{1}}}  \dttsym{,}  \dttnt{n_{{\mathrm{1}}}} \,  \preccurlyeq_{ \dttnt{p} }  \, \dttnt{n_{{\mathrm{2}}}}  \dttsym{,}  G''_{{\mathrm{1}}}$.  Then it is easy to see that
  $\dttsym{[}  \dttnt{n_{{\mathrm{3}}}}  \dttsym{/}  \dttnt{n_{{\mathrm{2}}}}  \dttsym{]}  G'_{{\mathrm{1}}}  \dttsym{,}  \dttsym{[}  \dttnt{n_{{\mathrm{3}}}}  \dttsym{/}  \dttnt{n_{{\mathrm{2}}}}  \dttsym{]}  G''_{{\mathrm{1}}}  \dttsym{,}  \dttsym{[}  \dttnt{n_{{\mathrm{3}}}}  \dttsym{/}  \dttnt{n_{{\mathrm{2}}}}  \dttsym{]}  \dttnt{n_{{\mathrm{4}}}} \,  \preccurlyeq_{ \dttnt{p} }  \, \dttsym{[}  \dttnt{n_{{\mathrm{3}}}}  \dttsym{/}  \dttnt{n_{{\mathrm{2}}}}  \dttsym{]}  \dttnt{n_{{\mathrm{5}}}}  \dttsym{,}  \dttsym{[}  \dttnt{n_{{\mathrm{3}}}}  \dttsym{/}  \dttnt{n_{{\mathrm{2}}}}  \dttsym{]}  G_{{\mathrm{2}}}  \vdash  \dttsym{[}  \dttnt{n_{{\mathrm{3}}}}  \dttsym{/}  \dttnt{n_{{\mathrm{2}}}}  \dttsym{]}  \dttnt{n_{{\mathrm{4}}}} \,  \preccurlyeq^*_{ \dttnt{p} }  \, \dttsym{[}  \dttnt{n_{{\mathrm{3}}}}  \dttsym{/}  \dttnt{n_{{\mathrm{2}}}}  \dttsym{]}  \dttnt{n_{{\mathrm{5}}}}$
  is derivable by applying $\dttdrulename{Ax}$.  The case where $\dttnt{n_{{\mathrm{1}}}} \,  \preccurlyeq_{ \dttnt{p_{{\mathrm{1}}}} }  \, \dttnt{n_{{\mathrm{2}}}} \in G_{{\mathrm{2}}}$ is
  similar.

  Now suppose $\dttnt{n_{{\mathrm{1}}}} \,  \preccurlyeq_{ \dttnt{p_{{\mathrm{1}}}} }  \, \dttnt{n_{{\mathrm{2}}}} \equiv \dttnt{n_{{\mathrm{4}}}} \,  \preccurlyeq_{ \dttnt{p} }  \, \dttnt{n_{{\mathrm{5}}}}$.  Then we know by assumption that 
  \[
      \mprset{flushleft}
      \inferrule* [right=\ifrName{ax}] {
        \ 
      }{G_{{\mathrm{1}}}  \dttsym{,}  \dttnt{n_{{\mathrm{1}}}} \,  \preccurlyeq_{ \dttnt{p} }  \, \dttnt{n_{{\mathrm{2}}}}  \dttsym{,}  G_{{\mathrm{2}}}  \vdash  \dttnt{n_{{\mathrm{1}}}} \,  \preccurlyeq^*_{ \dttnt{p} }  \, \dttnt{n_{{\mathrm{2}}}}}
    \]
  Then it suffices to show $\dttsym{[}  \dttnt{n_{{\mathrm{3}}}}  \dttsym{/}  \dttnt{n_{{\mathrm{2}}}}  \dttsym{]}  G_{{\mathrm{1}}}  \dttsym{,}  \dttsym{[}  \dttnt{n_{{\mathrm{3}}}}  \dttsym{/}  \dttnt{n_{{\mathrm{2}}}}  \dttsym{]}  G_{{\mathrm{2}}}  \vdash  \dttsym{[}  \dttnt{n_{{\mathrm{3}}}}  \dttsym{/}  \dttnt{n_{{\mathrm{2}}}}  \dttsym{]}  \dttnt{n_{{\mathrm{1}}}} \,  \preccurlyeq^*_{ \dttnt{p} }  \, \dttsym{[}  \dttnt{n_{{\mathrm{3}}}}  \dttsym{/}  \dttnt{n_{{\mathrm{2}}}}  \dttsym{]}  \dttnt{n_{{\mathrm{2}}}}$, which is equivalent
  to $\dttsym{[}  \dttnt{n_{{\mathrm{3}}}}  \dttsym{/}  \dttnt{n_{{\mathrm{2}}}}  \dttsym{]}  G_{{\mathrm{1}}}  \dttsym{,}  \dttsym{[}  \dttnt{n_{{\mathrm{3}}}}  \dttsym{/}  \dttnt{n_{{\mathrm{2}}}}  \dttsym{]}  G_{{\mathrm{2}}}  \vdash  \dttsym{[}  \dttnt{n_{{\mathrm{3}}}}  \dttsym{/}  \dttnt{n_{{\mathrm{2}}}}  \dttsym{]}  \dttnt{n_{{\mathrm{1}}}} \,  \preccurlyeq^*_{ \dttnt{p} }  \, \dttnt{n_{{\mathrm{3}}}}$.  Now if $\dttnt{n_{{\mathrm{1}}}}$ is equivalent to $\dttnt{n_{{\mathrm{2}}}}$, then
  $\dttsym{[}  \dttnt{n_{{\mathrm{3}}}}  \dttsym{/}  \dttnt{n_{{\mathrm{2}}}}  \dttsym{]}  G_{{\mathrm{1}}}  \dttsym{,}  \dttsym{[}  \dttnt{n_{{\mathrm{3}}}}  \dttsym{/}  \dttnt{n_{{\mathrm{2}}}}  \dttsym{]}  G_{{\mathrm{2}}}  \vdash  \dttsym{[}  \dttnt{n_{{\mathrm{3}}}}  \dttsym{/}  \dttnt{n_{{\mathrm{2}}}}  \dttsym{]}  \dttnt{n_{{\mathrm{1}}}} \,  \preccurlyeq^*_{ \dttnt{p} }  \, \dttnt{n_{{\mathrm{3}}}}$ holds by reflexivity, and if $\dttnt{n_{{\mathrm{1}}}}$ is distinct from $\dttnt{n_{{\mathrm{2}}}}$,
  then $\dttsym{[}  \dttnt{n_{{\mathrm{3}}}}  \dttsym{/}  \dttnt{n_{{\mathrm{2}}}}  \dttsym{]}  G_{{\mathrm{1}}}  \dttsym{,}  \dttsym{[}  \dttnt{n_{{\mathrm{3}}}}  \dttsym{/}  \dttnt{n_{{\mathrm{2}}}}  \dttsym{]}  G_{{\mathrm{2}}}  \vdash  \dttsym{[}  \dttnt{n_{{\mathrm{3}}}}  \dttsym{/}  \dttnt{n_{{\mathrm{2}}}}  \dttsym{]}  \dttnt{n_{{\mathrm{1}}}} \,  \preccurlyeq^*_{ \dttnt{p} }  \, \dttnt{n_{{\mathrm{3}}}}$ is equivalent to 
  $\dttsym{[}  \dttnt{n_{{\mathrm{3}}}}  \dttsym{/}  \dttnt{n_{{\mathrm{2}}}}  \dttsym{]}  G_{{\mathrm{1}}}  \dttsym{,}  \dttsym{[}  \dttnt{n_{{\mathrm{3}}}}  \dttsym{/}  \dttnt{n_{{\mathrm{2}}}}  \dttsym{]}  G_{{\mathrm{2}}}  \vdash  \dttnt{n_{{\mathrm{1}}}} \,  \preccurlyeq^*_{ \dttnt{p} }  \, \dttnt{n_{{\mathrm{3}}}}$.  We know by assumption that $G  \dttsym{,}  G'  \vdash  \dttnt{n_{{\mathrm{1}}}} \,  \preccurlyeq^*_{ \dttnt{p_{{\mathrm{1}}}} }  \, \dttnt{n_{{\mathrm{3}}}}$ holds, which
  is equivalent to $G_{{\mathrm{1}}}  \dttsym{,}  G_{{\mathrm{2}}}  \vdash  \dttnt{n_{{\mathrm{1}}}} \,  \preccurlyeq^*_{ \dttnt{p} }  \, \dttnt{n_{{\mathrm{3}}}}$.  Now if $\dttnt{n_{{\mathrm{3}}}}$ is equal to $\dttnt{n_{{\mathrm{2}}}}$, then 
  $\dttsym{[}  \dttnt{n_{{\mathrm{3}}}}  \dttsym{/}  \dttnt{n_{{\mathrm{2}}}}  \dttsym{]}  G_{{\mathrm{1}}}  \dttsym{,}  \dttsym{[}  \dttnt{n_{{\mathrm{3}}}}  \dttsym{/}  \dttnt{n_{{\mathrm{2}}}}  \dttsym{]}  G_{{\mathrm{2}}}  \vdash  \dttnt{n_{{\mathrm{1}}}} \,  \preccurlyeq^*_{ \dttnt{p} }  \, \dttnt{n_{{\mathrm{3}}}}$ is equivalent to $G_{{\mathrm{1}}}  \dttsym{,}  G_{{\mathrm{2}}}  \vdash  \dttnt{n_{{\mathrm{1}}}} \,  \preccurlyeq^*_{ \dttnt{p} }  \, \dttnt{n_{{\mathrm{3}}}}$.  So suppose
  $\dttnt{n_{{\mathrm{3}}}}$ is distinct from $\dttnt{n_{{\mathrm{2}}}}$, then by Lemma~\ref{lemma:renaming_nodes_in_graph} we know 
  $\dttsym{[}  \dttnt{n_{{\mathrm{3}}}}  \dttsym{/}  \dttnt{n_{{\mathrm{2}}}}  \dttsym{]}  G_{{\mathrm{1}}}  \dttsym{,}  \dttsym{[}  \dttnt{n_{{\mathrm{3}}}}  \dttsym{/}  \dttnt{n_{{\mathrm{2}}}}  \dttsym{]}  G_{{\mathrm{2}}}  \vdash  \dttnt{n_{{\mathrm{1}}}} \,  \preccurlyeq^*_{ \dttnt{p} }  \, \dttnt{n_{{\mathrm{3}}}}$.        
  
\item[\cW] 
  \[
      \mprset{flushleft}
      \inferrule* [right=\ifrName{refl}] {
        \ 
      }{G  \dttsym{,}  \dttnt{n_{{\mathrm{1}}}} \,  \preccurlyeq_{ \dttnt{p_{{\mathrm{1}}}} }  \, \dttnt{n_{{\mathrm{2}}}}  \dttsym{,}  G'  \vdash  \dttnt{n} \,  \preccurlyeq^*_{ \dttnt{p} }  \, \dttnt{n}}
      \leqno{\raise 8 pt\hbox{\textbf{Case}}}      
    \]
  Trivial.

\item[\cW] 
  \[
      \mprset{flushleft}
      \inferrule* [right=\ifrName{trans}] {
         G  \dttsym{,}  \dttnt{n_{{\mathrm{1}}}} \,  \preccurlyeq_{ \dttnt{p_{{\mathrm{1}}}} }  \, \dttnt{n_{{\mathrm{2}}}}  \dttsym{,}  G'  \vdash  \dttnt{n_{{\mathrm{4}}}} \,  \preccurlyeq^*_{ \dttnt{p} }  \, \dttnt{n_{{\mathrm{6}}}}  \qquad  G  \vdash  \dttnt{n_{{\mathrm{6}}}} \,  \preccurlyeq^*_{ \dttnt{p} }  \, \dttnt{n_{{\mathrm{5}}}} 
      }{G  \dttsym{,}  \dttnt{n_{{\mathrm{1}}}} \,  \preccurlyeq_{ \dttnt{p_{{\mathrm{1}}}} }  \, \dttnt{n_{{\mathrm{2}}}}  \dttsym{,}  G'  \vdash  \dttnt{n_{{\mathrm{4}}}} \,  \preccurlyeq^*_{ \dttnt{p} }  \, \dttnt{n_{{\mathrm{5}}}}}
      \leqno{\raise 8 pt\hbox{\textbf{Case}}}      
    \]
  This case by applying the induction to each premise, and then
  reapplying the rule.

\item[\cW] 
  \[
      \mprset{flushleft}
      \inferrule* [right=\ifrName{flip}] {
        G  \dttsym{,}  \dttnt{n_{{\mathrm{1}}}} \,  \preccurlyeq_{ \dttnt{p_{{\mathrm{1}}}} }  \, \dttnt{n_{{\mathrm{2}}}}  \dttsym{,}  G'  \vdash  \dttnt{n_{{\mathrm{5}}}} \,  \preccurlyeq^*_{  \bar{  \dttnt{p}  }  }  \, \dttnt{n_{{\mathrm{4}}}}
      }{G  \dttsym{,}  \dttnt{n_{{\mathrm{1}}}} \,  \preccurlyeq_{ \dttnt{p_{{\mathrm{1}}}} }  \, \dttnt{n_{{\mathrm{2}}}}  \dttsym{,}  G'  \vdash  \dttnt{n_{{\mathrm{4}}}} \,  \preccurlyeq^*_{ \dttnt{p} }  \, \dttnt{n_{{\mathrm{5}}}}}
      \leqno{\raise 8 pt\hbox{\textbf{Case}}}      
    \]
  This case holds by applying the induction hypothesis to the
  premise, and then reapplying the rule.    
\end{description}

\subsection{Proof of Lemma~\ref{lemma:node_substitution_for_typing}: Node Substitution for Typing}
\label{subsec:proof_of_lemma_node_substitution_for_typing}
This is a proof by induction on the form of the assumed typing
derivation.  Throughout each of the following cases we assume
$G  \dttsym{,}  G'  \vdash  \dttnt{n_{{\mathrm{1}}}} \,  \preccurlyeq^*_{ \dttnt{p_{{\mathrm{1}}}} }  \, \dttnt{n_{{\mathrm{4}}}}$ holds.

\begin{description}
\item[\cW] 
  \[
      \mprset{flushleft}
      \inferrule* [right=\ifrName{Ax}] {
        G  \dttsym{,}  \dttnt{n_{{\mathrm{1}}}} \,  \preccurlyeq_{ \dttnt{p_{{\mathrm{1}}}} }  \, \dttnt{n_{{\mathrm{2}}}}  \dttsym{,}  G'  \vdash  \dttnt{n} \,  \preccurlyeq^*_{ \dttnt{p} }  \, \dttnt{n_{{\mathrm{3}}}}
      }{G  \dttsym{,}  \dttnt{n_{{\mathrm{1}}}} \,  \preccurlyeq_{ \dttnt{p_{{\mathrm{1}}}} }  \, \dttnt{n_{{\mathrm{2}}}}  \dttsym{,}  G'  \dttsym{;}  \Gamma_{{\mathrm{1}}}  \dttsym{,}  \dttmv{y}  \dttsym{:}  \dttnt{p_{{\mathrm{2}}}} \, \dttnt{A}  \mathbin{@}  \dttnt{n}  \dttsym{,}  \Gamma_{{\mathrm{2}}}  \vdash  \dttmv{y}  \dttsym{:}  \dttnt{p_{{\mathrm{2}}}} \, \dttnt{A}  \mathbin{@}  \dttnt{n_{{\mathrm{3}}}}}
      \leqno{\raise 8 pt\hbox{\textbf{Case}}}      
    \]
  First, by node substitution for reachability (Lemma~\ref{lemma:node_substitution_for_reachability}) we know\\
  $\dttsym{[}  \dttnt{n_{{\mathrm{4}}}}  \dttsym{/}  \dttnt{n_{{\mathrm{2}}}}  \dttsym{]}  G  \dttsym{,}  \dttsym{[}  \dttnt{n_{{\mathrm{4}}}}  \dttsym{/}  \dttnt{n_{{\mathrm{2}}}}  \dttsym{]}  G'  \vdash  \dttsym{[}  \dttnt{n_{{\mathrm{4}}}}  \dttsym{/}  \dttnt{n_{{\mathrm{2}}}}  \dttsym{]}  \dttnt{n} \,  \preccurlyeq^*_{ \dttnt{p} }  \, \dttsym{[}  \dttnt{n_{{\mathrm{4}}}}  \dttsym{/}  \dttnt{n_{{\mathrm{2}}}}  \dttsym{]}  \dttnt{n_{{\mathrm{3}}}}$.  Thus, by applying the $\dttdrulename{Ax}$ rule we may derive
  $\dttsym{[}  \dttnt{n_{{\mathrm{4}}}}  \dttsym{/}  \dttnt{n_{{\mathrm{2}}}}  \dttsym{]}  G  \dttsym{,}  \dttsym{[}  \dttnt{n_{{\mathrm{4}}}}  \dttsym{/}  \dttnt{n_{{\mathrm{2}}}}  \dttsym{]}  G'  \dttsym{;}  \dttsym{[}  \dttnt{n_{{\mathrm{4}}}}  \dttsym{/}  \dttnt{n_{{\mathrm{2}}}}  \dttsym{]}  \Gamma_{{\mathrm{1}}}  \dttsym{,}  \dttmv{y}  \dttsym{:}  \dttnt{p_{{\mathrm{2}}}} \, \dttnt{A}  \mathbin{@}  \dttsym{[}  \dttnt{n_{{\mathrm{4}}}}  \dttsym{/}  \dttnt{n_{{\mathrm{2}}}}  \dttsym{]}  \dttnt{n}  \dttsym{,}  \dttsym{[}  \dttnt{n_{{\mathrm{4}}}}  \dttsym{/}  \dttnt{n_{{\mathrm{2}}}}  \dttsym{]}  \Gamma_{{\mathrm{2}}}  \vdash  \dttmv{y}  \dttsym{:}  \dttnt{p_{{\mathrm{2}}}} \, \dttnt{A}  \mathbin{@}  \dttsym{[}  \dttnt{n_{{\mathrm{4}}}}  \dttsym{/}  \dttnt{n_{{\mathrm{2}}}}  \dttsym{]}  \dttnt{n_{{\mathrm{3}}}}$.

\item[\cW] 
  \[
      \mprset{flushleft}
      \inferrule* [right=\ifrName{Unit}] {
        \ 
      }{G  \dttsym{,}  \dttnt{n_{{\mathrm{1}}}} \,  \preccurlyeq_{ \dttnt{p_{{\mathrm{1}}}} }  \, \dttnt{n_{{\mathrm{2}}}}  \dttsym{,}  G'  \dttsym{;}  \Gamma  \vdash  \dttkw{triv}  \dttsym{:}  \dttnt{p_{{\mathrm{2}}}} \,  \langle  \dttnt{p_{{\mathrm{2}}}} \rangle   \mathbin{@}  \dttnt{n_{{\mathrm{3}}}}}
      \leqno{\raise 8 pt\hbox{\textbf{Case}}}      
    \]
  Trivial.

\item[\cW] 
  \[
      \mprset{flushleft}
      \inferrule* [right=\ifrName{And}] {
         G  \dttsym{,}  \dttnt{n_{{\mathrm{1}}}} \,  \preccurlyeq_{ \dttnt{p_{{\mathrm{1}}}} }  \, \dttnt{n_{{\mathrm{2}}}}  \dttsym{;}  \Gamma  \vdash  \dttnt{t_{{\mathrm{1}}}}  \dttsym{:}  \dttnt{p_{{\mathrm{2}}}} \, \dttnt{A_{{\mathrm{1}}}}  \mathbin{@}  \dttnt{n_{{\mathrm{3}}}}  \qquad  G  \dttsym{,}  \dttnt{n_{{\mathrm{1}}}} \,  \preccurlyeq_{ \dttnt{p_{{\mathrm{1}}}} }  \, \dttnt{n_{{\mathrm{2}}}}  \dttsym{;}  \Gamma  \vdash  \dttnt{t_{{\mathrm{2}}}}  \dttsym{:}  \dttnt{p_{{\mathrm{2}}}} \, \dttnt{A_{{\mathrm{2}}}}  \mathbin{@}  \dttnt{n_{{\mathrm{3}}}} 
      }{G  \dttsym{,}  \dttnt{n_{{\mathrm{1}}}} \,  \preccurlyeq_{ \dttnt{p_{{\mathrm{1}}}} }  \, \dttnt{n_{{\mathrm{2}}}}  \dttsym{;}  \Gamma  \vdash  \dttsym{(}  \dttnt{t_{{\mathrm{1}}}}  \dttsym{,}  \dttnt{t_{{\mathrm{2}}}}  \dttsym{)}  \dttsym{:}  \dttnt{p_{{\mathrm{2}}}} \, \dttsym{(}   \dttnt{A_{{\mathrm{1}}}}  \ndwedge{ \dttnt{p_{{\mathrm{2}}}} }  \dttnt{A_{{\mathrm{2}}}}   \dttsym{)}  \mathbin{@}  \dttnt{n_{{\mathrm{3}}}}}
      \leqno{\raise 8 pt\hbox{\textbf{Case}}}      
    \]
  This case holds by applying the induction hypothesis to each
  premise, and then reapplying the rule.

\item[\cW] 
  \[
      \mprset{flushleft}
      \inferrule* [right=\ifrName{AndBar}] {
        G  \dttsym{,}  \dttnt{n_{{\mathrm{1}}}} \,  \preccurlyeq_{ \dttnt{p_{{\mathrm{1}}}} }  \, \dttnt{n_{{\mathrm{2}}}}  \dttsym{;}  \Gamma  \vdash  \dttnt{t'}  \dttsym{:}  \dttnt{p_{{\mathrm{2}}}} \,  \dttnt{A} _{ \dttnt{d} }   \mathbin{@}  \dttnt{n_{{\mathrm{3}}}}
      }{G  \dttsym{,}  \dttnt{n_{{\mathrm{1}}}} \,  \preccurlyeq_{ \dttnt{p_{{\mathrm{1}}}} }  \, \dttnt{n_{{\mathrm{2}}}}  \dttsym{;}  \Gamma  \vdash   \mathbf{in}_{ \dttnt{d} }\, \dttnt{t'}   \dttsym{:}  \dttnt{p_{{\mathrm{2}}}} \, \dttsym{(}   \dttnt{A_{{\mathrm{1}}}}  \ndwedge{  \bar{  \dttnt{p_{{\mathrm{2}}}}  }  }  \dttnt{A_{{\mathrm{2}}}}   \dttsym{)}  \mathbin{@}  \dttnt{n_{{\mathrm{3}}}}}
      \leqno{\raise 8 pt\hbox{\textbf{Case}}}      
    \]
  This case holds by applying the induction hypothesis to the
  premise, and then reapplying the rule.

\item[\cW] 
  \[
      \mprset{flushleft}
      \inferrule* [right=\ifrName{Imp}] {
        \dttnt{n'} \, \not\in \, \dttsym{\mbox{$\mid$}}  G  \dttsym{,}  \dttnt{n_{{\mathrm{1}}}} \,  \preccurlyeq_{ \dttnt{p_{{\mathrm{1}}}} }  \, \dttnt{n_{{\mathrm{2}}}}  \dttsym{,}  G'  \dttsym{\mbox{$\mid$}}  \dttsym{,}  \dttsym{\mbox{$\mid$}}  \Gamma  \dttsym{\mbox{$\mid$}}
        \\\\
            \dttsym{(}  G  \dttsym{,}  \dttnt{n_{{\mathrm{1}}}} \,  \preccurlyeq_{ \dttnt{p_{{\mathrm{1}}}} }  \, \dttnt{n_{{\mathrm{2}}}}  \dttsym{,}  G'  \dttsym{,}  \dttnt{n_{{\mathrm{3}}}} \,  \preccurlyeq_{ \dttnt{p_{{\mathrm{2}}}} }  \, \dttnt{n'}  \dttsym{)}  \dttsym{;}  \Gamma  \dttsym{,}  \dttmv{x}  \dttsym{:}  \dttnt{p_{{\mathrm{2}}}} \, \dttnt{A_{{\mathrm{1}}}}  \mathbin{@}  \dttnt{n'}  \vdash  \dttnt{t'}  \dttsym{:}  \dttnt{p_{{\mathrm{2}}}} \, \dttnt{A_{{\mathrm{2}}}}  \mathbin{@}  \dttnt{n'}
      }{G  \dttsym{,}  \dttnt{n_{{\mathrm{1}}}} \,  \preccurlyeq_{ \dttnt{p_{{\mathrm{1}}}} }  \, \dttnt{n_{{\mathrm{2}}}}  \dttsym{,}  G'  \dttsym{;}  \Gamma  \vdash  \lambda  \dttmv{x}  \dttsym{.}  \dttnt{t'}  \dttsym{:}  \dttnt{p_{{\mathrm{2}}}} \, \dttsym{(}   \dttnt{A_{{\mathrm{1}}}}  \ndto{ \dttnt{p_{{\mathrm{2}}}} }  \dttnt{A_{{\mathrm{2}}}}   \dttsym{)}  \mathbin{@}  \dttnt{n_{{\mathrm{3}}}}}
      \leqno{\raise 8 pt\hbox{\textbf{Case}}}      
    \]
  First, if $\dttnt{n'} \, \not\in \, \dttsym{\mbox{$\mid$}}  G  \dttsym{,}  \dttnt{n_{{\mathrm{1}}}} \,  \preccurlyeq_{ \dttnt{p_{{\mathrm{1}}}} }  \, \dttnt{n_{{\mathrm{2}}}}  \dttsym{,}  G'  \dttsym{\mbox{$\mid$}}  \dttsym{,}  \dttsym{\mbox{$\mid$}}  \Gamma  \dttsym{\mbox{$\mid$}}$, then $\dttnt{n'} \, \not\in \, \dttsym{\mbox{$\mid$}}  G  \dttsym{,}  G'  \dttsym{\mbox{$\mid$}}  \dttsym{,}  \dttsym{\mbox{$\mid$}}  \Gamma  \dttsym{\mbox{$\mid$}}$.  Furthermore,
  we know that $\dttsym{[}  \dttnt{n_{{\mathrm{4}}}}  \dttsym{/}  \dttnt{n_{{\mathrm{2}}}}  \dttsym{]}  \dttnt{n'} \, \not\in \, \dttsym{\mbox{$\mid$}}  \dttsym{[}  \dttnt{n_{{\mathrm{4}}}}  \dttsym{/}  \dttnt{n_{{\mathrm{2}}}}  \dttsym{]}  G  \dttsym{,}  \dttsym{[}  \dttnt{n_{{\mathrm{4}}}}  \dttsym{/}  \dttnt{n_{{\mathrm{2}}}}  \dttsym{]}  G'  \dttsym{\mbox{$\mid$}}  \dttsym{,}  \dttsym{\mbox{$\mid$}}  \dttsym{[}  \dttnt{n_{{\mathrm{4}}}}  \dttsym{/}  \dttnt{n_{{\mathrm{2}}}}  \dttsym{]}  \Gamma  \dttsym{\mbox{$\mid$}}$, because we know 
  $\dttnt{n'}$ is distinct from $\dttnt{n_{{\mathrm{2}}}}$ by assumption, and if $\dttnt{n'}$ is equal to $\dttnt{n_{{\mathrm{4}}}}$, then 
  $\dttnt{n'} \, \not\in \, \dttsym{\mbox{$\mid$}}  G  \dttsym{,}  \dttnt{n_{{\mathrm{1}}}} \,  \preccurlyeq_{ \dttnt{p_{{\mathrm{1}}}} }  \, \dttnt{n_{{\mathrm{2}}}}  \dttsym{,}  G'  \dttsym{\mbox{$\mid$}}  \dttsym{,}  \dttsym{\mbox{$\mid$}}  \Gamma  \dttsym{\mbox{$\mid$}}$ implies that $\dttnt{n_{{\mathrm{1}}}}$ must also be $\dttnt{n_{{\mathrm{4}}}}$, because we know by assumption 
  that $G  \dttsym{,}  G'  \vdash  \dttnt{n_{{\mathrm{1}}}} \,  \preccurlyeq^*_{ \dttnt{p_{{\mathrm{1}}}} }  \, \dttnt{n_{{\mathrm{4}}}}$, which could
  only be derived by reflexivity since $\dttnt{n'} \, \not\in \, \dttsym{\mbox{$\mid$}}  G  \dttsym{,}  G'  \dttsym{\mbox{$\mid$}}  \dttsym{,}  \dttsym{\mbox{$\mid$}}  \Gamma  \dttsym{\mbox{$\mid$}}$, but we know by assumption that 
  $\dttnt{n'} \, \not\in \, \dttsym{\mbox{$\mid$}}  G  \dttsym{,}  \dttnt{n_{{\mathrm{1}}}} \,  \preccurlyeq_{ \dttnt{p_{{\mathrm{1}}}} }  \, \dttnt{n_{{\mathrm{2}}}}  \dttsym{,}  G'  \dttsym{\mbox{$\mid$}}  \dttsym{,}  \dttsym{\mbox{$\mid$}}  \Gamma  \dttsym{\mbox{$\mid$}}$, which implies that $\dttnt{n'}$ must be distinct from $\dttnt{n_{{\mathrm{1}}}}$, 
  and hence a contradiction, thus $\dttnt{n'}$ cannot be $\dttnt{n_{{\mathrm{4}}}}$.  Therefore, we know 
  $\dttnt{n'} \, \not\in \, \dttsym{\mbox{$\mid$}}  \dttsym{[}  \dttnt{n_{{\mathrm{4}}}}  \dttsym{/}  \dttnt{n_{{\mathrm{2}}}}  \dttsym{]}  G  \dttsym{,}  \dttsym{[}  \dttnt{n_{{\mathrm{4}}}}  \dttsym{/}  \dttnt{n_{{\mathrm{2}}}}  \dttsym{]}  G'  \dttsym{\mbox{$\mid$}}  \dttsym{,}  \dttsym{\mbox{$\mid$}}  \dttsym{[}  \dttnt{n_{{\mathrm{4}}}}  \dttsym{/}  \dttnt{n_{{\mathrm{2}}}}  \dttsym{]}  \Gamma  \dttsym{\mbox{$\mid$}}$.   

  \ \\
  By the induction hypothesis we know 
  \begin{center}
    \scriptsize
    \begin{math}
      \dttsym{(}  \dttsym{[}  \dttnt{n_{{\mathrm{4}}}}  \dttsym{/}  \dttnt{n_{{\mathrm{2}}}}  \dttsym{]}  G  \dttsym{,}  \dttsym{[}  \dttnt{n_{{\mathrm{4}}}}  \dttsym{/}  \dttnt{n_{{\mathrm{2}}}}  \dttsym{]}  G'  \dttsym{,}  \dttsym{[}  \dttnt{n_{{\mathrm{4}}}}  \dttsym{/}  \dttnt{n_{{\mathrm{2}}}}  \dttsym{]}  \dttnt{n_{{\mathrm{3}}}} \,  \preccurlyeq_{ \dttnt{p_{{\mathrm{2}}}} }  \, \dttsym{[}  \dttnt{n_{{\mathrm{4}}}}  \dttsym{/}  \dttnt{n_{{\mathrm{2}}}}  \dttsym{]}  \dttnt{n'}  \dttsym{)}  \dttsym{;}  \dttsym{[}  \dttnt{n_{{\mathrm{4}}}}  \dttsym{/}  \dttnt{n_{{\mathrm{2}}}}  \dttsym{]}  \Gamma  \dttsym{,}  \dttmv{x}  \dttsym{:}  \dttnt{p_{{\mathrm{2}}}} \, \dttnt{A_{{\mathrm{1}}}}  \mathbin{@}  \dttsym{[}  \dttnt{n_{{\mathrm{4}}}}  \dttsym{/}  \dttnt{n_{{\mathrm{2}}}}  \dttsym{]}  \dttnt{n'}  \vdash  \dttnt{t'}  \dttsym{:}  \dttnt{p_{{\mathrm{2}}}} \, \dttnt{A_{{\mathrm{2}}}}  \mathbin{@}  \dttsym{[}  \dttnt{n_{{\mathrm{4}}}}  \dttsym{/}  \dttnt{n_{{\mathrm{2}}}}  \dttsym{]}  \dttnt{n'}\,,      
    \end{math}
  \end{center}  
  which
  is equivalent to 
  \[ \dttsym{(}  \dttsym{[}  \dttnt{n_{{\mathrm{4}}}}  \dttsym{/}  \dttnt{n_{{\mathrm{2}}}}  \dttsym{]}  G  \dttsym{,}  \dttsym{[}  \dttnt{n_{{\mathrm{4}}}}  \dttsym{/}  \dttnt{n_{{\mathrm{2}}}}  \dttsym{]}  G'  \dttsym{,}  \dttsym{[}  \dttnt{n_{{\mathrm{4}}}}  \dttsym{/}  \dttnt{n_{{\mathrm{2}}}}  \dttsym{]}  \dttnt{n_{{\mathrm{3}}}} \,  \preccurlyeq_{ \dttnt{p_{{\mathrm{2}}}} }  \, \dttnt{n'}  \dttsym{)}  \dttsym{;}  \dttsym{[}  \dttnt{n_{{\mathrm{4}}}}  \dttsym{/}  \dttnt{n_{{\mathrm{2}}}}  \dttsym{]}  \Gamma  \dttsym{,}  \dttmv{x}  \dttsym{:}  \dttnt{p_{{\mathrm{2}}}} \, \dttnt{A_{{\mathrm{1}}}}  \mathbin{@}  \dttnt{n'}  \vdash  \dttnt{t'}  \dttsym{:}  \dttnt{p_{{\mathrm{2}}}} \, \dttnt{A_{{\mathrm{2}}}}  \mathbin{@}  \dttnt{n'}. \] 
  Finally, this case follows by applying the $\dttdrulename{Imp}$ rule using\\
  $\dttnt{n'} \, \not\in \, \dttsym{\mbox{$\mid$}}  \dttsym{[}  \dttnt{n_{{\mathrm{4}}}}  \dttsym{/}  \dttnt{n_{{\mathrm{2}}}}  \dttsym{]}  G  \dttsym{,}  \dttsym{[}  \dttnt{n_{{\mathrm{4}}}}  \dttsym{/}  \dttnt{n_{{\mathrm{2}}}}  \dttsym{]}  G'  \dttsym{\mbox{$\mid$}}  \dttsym{,}  \dttsym{\mbox{$\mid$}}  \dttsym{[}  \dttnt{n_{{\mathrm{4}}}}  \dttsym{/}  \dttnt{n_{{\mathrm{2}}}}  \dttsym{]}  \Gamma  \dttsym{\mbox{$\mid$}}$ and the previous fact.

\item[\cW] 
  \[
      \mprset{flushleft}
      \inferrule* [right={\tiny ImpBar}] {
        G  \dttsym{,}  \dttnt{n_{{\mathrm{1}}}} \,  \preccurlyeq_{ \dttnt{p_{{\mathrm{1}}}} }  \, \dttnt{n_{{\mathrm{2}}}}  \dttsym{,}  G'  \vdash  \dttnt{n_{{\mathrm{3}}}} \,  \preccurlyeq^*_{  \bar{  \dttnt{p_{{\mathrm{2}}}}  }  }  \, \dttnt{n'}
        \\\\
             G  \dttsym{,}  \dttnt{n_{{\mathrm{1}}}} \,  \preccurlyeq_{ \dttnt{p_{{\mathrm{1}}}} }  \, \dttnt{n_{{\mathrm{2}}}}  \dttsym{,}  G'  \dttsym{;}  \Gamma  \vdash  \dttnt{t_{{\mathrm{1}}}}  \dttsym{:}   \bar{  \dttnt{p_{{\mathrm{2}}}}  }  \, \dttnt{A_{{\mathrm{1}}}}  \mathbin{@}  \dttnt{n'}  \qquad  G  \dttsym{,}  \dttnt{n_{{\mathrm{1}}}} \,  \preccurlyeq_{ \dttnt{p_{{\mathrm{1}}}} }  \, \dttnt{n_{{\mathrm{2}}}}  \dttsym{,}  G'  \dttsym{;}  \Gamma  \vdash  \dttnt{t_{{\mathrm{2}}}}  \dttsym{:}  \dttnt{p_{{\mathrm{2}}}} \, \dttnt{A_{{\mathrm{2}}}}  \mathbin{@}  \dttnt{n'} 
      }{G  \dttsym{,}  \dttnt{n_{{\mathrm{1}}}} \,  \preccurlyeq_{ \dttnt{p_{{\mathrm{1}}}} }  \, \dttnt{n_{{\mathrm{2}}}}  \dttsym{,}  G'  \dttsym{;}  \Gamma  \vdash  \langle  \dttnt{t_{{\mathrm{1}}}}  \dttsym{,}  \dttnt{t_{{\mathrm{2}}}}  \rangle  \dttsym{:}  \dttnt{p_{{\mathrm{2}}}} \, \dttsym{(}   \dttnt{A_{{\mathrm{1}}}}  \ndto{  \bar{  \dttnt{p_{{\mathrm{2}}}}  }  }  \dttnt{A_{{\mathrm{2}}}}   \dttsym{)}  \mathbin{@}  \dttnt{n_{{\mathrm{3}}}}}
      \leqno{\raise 8 pt\hbox{\textbf{Case}}}      
    \]
  We now by assumption that $G  \dttsym{,}  G'  \vdash  \dttnt{n_{{\mathrm{1}}}} \,  \preccurlyeq^*_{ \dttnt{p_{{\mathrm{1}}}} }  \, \dttnt{n_{{\mathrm{4}}}}$ holds.  So by node substitution for reachability
  (Lemma~\ref{lemma:node_substitution_for_reachability}) we know $\dttsym{[}  \dttnt{n_{{\mathrm{4}}}}  \dttsym{/}  \dttnt{n_{{\mathrm{2}}}}  \dttsym{]}  G  \dttsym{,}  \dttsym{[}  \dttnt{n_{{\mathrm{4}}}}  \dttsym{/}  \dttnt{n_{{\mathrm{2}}}}  \dttsym{]}  G'  \vdash  \dttsym{[}  \dttnt{n_{{\mathrm{4}}}}  \dttsym{/}  \dttnt{n_{{\mathrm{2}}}}  \dttsym{]}  \dttnt{n_{{\mathrm{3}}}} \,  \preccurlyeq^*_{  \bar{  \dttnt{p_{{\mathrm{2}}}}  }  }  \, \dttsym{[}  \dttnt{n_{{\mathrm{4}}}}  \dttsym{/}  \dttnt{n_{{\mathrm{2}}}}  \dttsym{]}  \dttnt{n'}$.
  Now by the induction hypothesis we know $\dttsym{[}  \dttnt{n_{{\mathrm{4}}}}  \dttsym{/}  \dttnt{n_{{\mathrm{2}}}}  \dttsym{]}  G  \dttsym{,}  \dttsym{[}  \dttnt{n_{{\mathrm{4}}}}  \dttsym{/}  \dttnt{n_{{\mathrm{2}}}}  \dttsym{]}  G'  \dttsym{;}  \dttsym{[}  \dttnt{n_{{\mathrm{4}}}}  \dttsym{/}  \dttnt{n_{{\mathrm{2}}}}  \dttsym{]}  \Gamma  \vdash  \dttnt{t_{{\mathrm{1}}}}  \dttsym{:}   \bar{  \dttnt{p_{{\mathrm{2}}}}  }  \, \dttnt{A_{{\mathrm{1}}}}  \mathbin{@}  \dttsym{[}  \dttnt{n_{{\mathrm{4}}}}  \dttsym{/}  \dttnt{n_{{\mathrm{2}}}}  \dttsym{]}  \dttnt{n'}$  and \\
  $\dttsym{[}  \dttnt{n_{{\mathrm{4}}}}  \dttsym{/}  \dttnt{n_{{\mathrm{2}}}}  \dttsym{]}  G  \dttsym{,}  \dttsym{[}  \dttnt{n_{{\mathrm{4}}}}  \dttsym{/}  \dttnt{n_{{\mathrm{2}}}}  \dttsym{]}  G'  \dttsym{;}  \dttsym{[}  \dttnt{n_{{\mathrm{4}}}}  \dttsym{/}  \dttnt{n_{{\mathrm{2}}}}  \dttsym{]}  \Gamma  \vdash  \dttnt{t_{{\mathrm{2}}}}  \dttsym{:}  \dttnt{p_{{\mathrm{2}}}} \, \dttnt{A_{{\mathrm{2}}}}  \mathbin{@}  \dttsym{[}  \dttnt{n_{{\mathrm{4}}}}  \dttsym{/}  \dttnt{n_{{\mathrm{2}}}}  \dttsym{]}  \dttnt{n'}$.  This case then follows by applying the rule
  $\dttdrulename{ImpBar}$ to the previous three facts.

\item[\cW] 
  \[
      \mprset{flushleft}
      \inferrule* [right=\ifrName{Cut}] {
        G  \dttsym{,}  \dttnt{n_{{\mathrm{1}}}} \,  \preccurlyeq_{ \dttnt{p_{{\mathrm{1}}}} }  \, \dttnt{n_{{\mathrm{2}}}}  \dttsym{,}  G'  \dttsym{;}  \Gamma  \dttsym{,}  \dttmv{y}  \dttsym{:}   \bar{  \dttnt{p_{{\mathrm{2}}}}  }  \, \dttnt{A}  \mathbin{@}  \dttnt{n_{{\mathrm{3}}}}  \vdash  \dttnt{t_{{\mathrm{1}}}}  \dttsym{:}  \dttsym{+} \, \dttnt{C}  \mathbin{@}  \dttnt{n}
        \\\\
            G  \dttsym{,}  \dttnt{n_{{\mathrm{1}}}} \,  \preccurlyeq_{ \dttnt{p_{{\mathrm{1}}}} }  \, \dttnt{n_{{\mathrm{2}}}}  \dttsym{,}  G'  \dttsym{;}  \Gamma  \dttsym{,}  \dttmv{y}  \dttsym{:}   \bar{  \dttnt{p_{{\mathrm{2}}}}  }  \, \dttnt{A}  \mathbin{@}  \dttnt{n_{{\mathrm{3}}}}  \vdash  \dttnt{t_{{\mathrm{2}}}}  \dttsym{:}  \dttsym{-} \, \dttnt{C}  \mathbin{@}  \dttnt{n}
      }{G  \dttsym{,}  \dttnt{n_{{\mathrm{1}}}} \,  \preccurlyeq_{ \dttnt{p_{{\mathrm{1}}}} }  \, \dttnt{n_{{\mathrm{2}}}}  \dttsym{,}  G'  \dttsym{;}  \Gamma  \vdash  \nu \, \dttmv{x}  \dttsym{.}  \dttnt{t_{{\mathrm{1}}}}  \mathbin{\Cdot[2]}  \dttnt{t_{{\mathrm{2}}}}  \dttsym{:}  \dttnt{p_{{\mathrm{2}}}} \, \dttnt{A}  \mathbin{@}  \dttnt{n_{{\mathrm{3}}}}}
      \leqno{\raise 8 pt\hbox{\textbf{Case}}}      
    \]
  This case follows by applying the induction hypothesis to each premise, and then reapplying
  the rule.
\end{description}  

\subsection{Proof of Lemma~\ref{lemma:substitution_for_typing}: Substitution for Typing}
\label{subsec:proof_of_lemma_substitution_for_typing}
This proof holds by a straightforward induction on the second
assumed typing relation. 
\begin{description}
\item[\cW] 
  \[
      \mprset{flushleft}
      \inferrule* [right=\ifrName{Ax}] {
        G  \vdash  \dttnt{n} \,  \preccurlyeq^*_{ \dttnt{p} }  \, \dttnt{n'}
      }{G  \dttsym{;}  \Gamma_{{\mathrm{1}}}  \dttsym{,}  \dttmv{y}  \dttsym{:}  \dttnt{p} \, \dttnt{C}  \mathbin{@}  \dttnt{n}  \dttsym{,}  \Gamma_{{\mathrm{2}}}  \vdash  \dttmv{y}  \dttsym{:}  \dttnt{p} \, \dttnt{C}  \mathbin{@}  \dttnt{n'}}
      \leqno{\raise 8 pt\hbox{\textbf{Case}}}      
    \]
  Trivial.

\item[\cW] 
  \[
      \mprset{flushleft}
      \inferrule* [right=\ifrName{Unit}] {
        \ 
      }{G  \dttsym{;}  \Gamma_{{\mathrm{1}}}  \vdash  \dttkw{triv}  \dttsym{:}  \dttnt{p} \,  \langle  \dttnt{p} \rangle   \mathbin{@}  \dttnt{n}}
      \leqno{\raise 8 pt\hbox{\textbf{Case}}}      
    \]
  Trivial.
\newpage

\item[\cW] 
  \[
      \mprset{flushleft}
      \inferrule* [right=\ifrName{And}] {
         G  \dttsym{;}  \Gamma_{{\mathrm{1}}}  \vdash  \dttnt{t'_{{\mathrm{1}}}}  \dttsym{:}  \dttnt{p} \, \dttnt{A}  \mathbin{@}  \dttnt{n}  \qquad  G  \dttsym{;}  \Gamma_{{\mathrm{1}}}  \vdash  \dttnt{t'_{{\mathrm{2}}}}  \dttsym{:}  \dttnt{p} \, \dttnt{B}  \mathbin{@}  \dttnt{n} 
      }{G  \dttsym{;}  \Gamma_{{\mathrm{1}}}  \vdash  \dttsym{(}  \dttnt{t'_{{\mathrm{1}}}}  \dttsym{,}  \dttnt{t'_{{\mathrm{2}}}}  \dttsym{)}  \dttsym{:}  \dttnt{p} \, \dttsym{(}   \dttnt{C_{{\mathrm{1}}}}  \ndwedge{ \dttnt{p} }  \dttnt{C_{{\mathrm{2}}}}   \dttsym{)}  \mathbin{@}  \dttnt{n}}
      \leqno{\raise 8 pt\hbox{\textbf{Case}}}      
    \]
  Suppose $\Gamma_{{\mathrm{1}}} \equiv \Gamma  \dttsym{,}  \dttmv{x}  \dttsym{:}  \dttnt{p_{{\mathrm{1}}}} \, \dttnt{A}  \mathbin{@}  \dttnt{n_{{\mathrm{1}}}}  \dttsym{,}  \Gamma'$.  Then this case
  follows from applying the induction hypothesis to each premise and
  then reapplying the rule.

\item[\cW] 
  \[
      \mprset{flushleft}
      \inferrule* [right=\ifrName{AndBar}] {
        G  \dttsym{;}  \Gamma_{{\mathrm{1}}}  \vdash  \dttnt{t}  \dttsym{:}  \dttnt{p} \,  \dttnt{C} _{ \dttnt{d} }   \mathbin{@}  \dttnt{n}
      }{G  \dttsym{;}  \Gamma_{{\mathrm{1}}}  \vdash   \mathbf{in}_{ \dttnt{d} }\, \dttnt{t}   \dttsym{:}  \dttnt{p} \, \dttsym{(}   \dttnt{C_{{\mathrm{1}}}}  \ndwedge{  \bar{  \dttnt{p}  }  }  \dttnt{C_{{\mathrm{2}}}}   \dttsym{)}  \mathbin{@}  \dttnt{n}}
      \leqno{\raise 8 pt\hbox{\textbf{Case}}}      
    \]
  Suppose $\Gamma_{{\mathrm{1}}} \equiv \Gamma  \dttsym{,}  \dttmv{x}  \dttsym{:}  \dttnt{p_{{\mathrm{1}}}} \, \dttnt{A}  \mathbin{@}  \dttnt{n_{{\mathrm{1}}}}  \dttsym{,}  \Gamma'$. Then this case
  follows from applying the induction hypothesis to the premise and
  then reapplying the rule.

\item[\cW] 
  \[
      \mprset{flushleft}
      \inferrule* [right=\ifrName{Imp}] {
        \dttnt{n'} \, \not\in \, \dttsym{\mbox{$\mid$}}  G  \dttsym{\mbox{$\mid$}}  \dttsym{,}  \dttsym{\mbox{$\mid$}}  \Gamma_{{\mathrm{1}}}  \dttsym{\mbox{$\mid$}}
        \\\\
            \dttsym{(}  G  \dttsym{,}  \dttnt{n} \,  \preccurlyeq_{ \dttnt{p} }  \, \dttnt{n'}  \dttsym{)}  \dttsym{;}  \Gamma_{{\mathrm{1}}}  \dttsym{,}  \dttmv{x}  \dttsym{:}  \dttnt{p} \, \dttnt{C_{{\mathrm{1}}}}  \mathbin{@}  \dttnt{n'}  \vdash  \dttnt{t}  \dttsym{:}  \dttnt{p} \, \dttnt{C_{{\mathrm{2}}}}  \mathbin{@}  \dttnt{n'}
      }{G  \dttsym{;}  \Gamma_{{\mathrm{1}}}  \vdash  \lambda  \dttmv{x}  \dttsym{.}  \dttnt{t}  \dttsym{:}  \dttnt{p} \, \dttsym{(}   \dttnt{C_{{\mathrm{1}}}}  \ndto{ \dttnt{p} }  \dttnt{C_{{\mathrm{2}}}}   \dttsym{)}  \mathbin{@}  \dttnt{n}}
      \leqno{\raise 8 pt\hbox{\textbf{Case}}}      
    \]
  Similarly to the previous case.

\item[\cW] 
  \[
      \mprset{flushleft}
      \inferrule* [right=\ifrName{ImpBar}] {
        G  \vdash  \dttnt{n} \,  \preccurlyeq^*_{  \bar{  \dttnt{p}  }  }  \, \dttnt{n'}
        \\\\
             G  \dttsym{;}  \Gamma_{{\mathrm{1}}}  \vdash  \dttnt{t'_{{\mathrm{1}}}}  \dttsym{:}   \bar{  \dttnt{p}  }  \, \dttnt{C_{{\mathrm{1}}}}  \mathbin{@}  \dttnt{n'}  \qquad  G  \dttsym{;}  \Gamma_{{\mathrm{1}}}  \vdash  \dttnt{t'_{{\mathrm{2}}}}  \dttsym{:}  \dttnt{p} \, \dttnt{C_{{\mathrm{2}}}}  \mathbin{@}  \dttnt{n'} 
      }{G  \dttsym{;}  \Gamma_{{\mathrm{1}}}  \vdash  \langle  \dttnt{t'_{{\mathrm{1}}}}  \dttsym{,}  \dttnt{t'_{{\mathrm{2}}}}  \rangle  \dttsym{:}  \dttnt{p} \, \dttsym{(}   \dttnt{C_{{\mathrm{1}}}}  \ndto{  \bar{  \dttnt{p}  }  }  \dttnt{C_{{\mathrm{2}}}}   \dttsym{)}  \mathbin{@}  \dttnt{n}}
      \leqno{\raise 8 pt\hbox{\textbf{Case}}}      
    \]
  Suppose $\Gamma_{{\mathrm{1}}} \equiv \Gamma  \dttsym{,}  \dttmv{x}  \dttsym{:}  \dttnt{p_{{\mathrm{1}}}} \, \dttnt{A}  \mathbin{@}  \dttnt{n_{{\mathrm{1}}}}  \dttsym{,}  \Gamma'$.  Then this case
  follows from applying the induction hypothesis to each premise and
  then reapplying the rule.

\item[\cW] 
  \[
      \mprset{flushleft}
      \inferrule* [right=\ifrName{Cut}] {
        G  \dttsym{;}  \Gamma_{{\mathrm{1}}}  \dttsym{,}  \dttmv{y}  \dttsym{:}   \bar{  \dttnt{p}  }  \, \dttnt{C}  \mathbin{@}  \dttnt{n}  \vdash  \dttnt{t'_{{\mathrm{1}}}}  \dttsym{:}  \dttsym{+} \, \dttnt{C'}  \mathbin{@}  \dttnt{n'}
        \\\\
            G  \dttsym{;}  \Gamma_{{\mathrm{1}}}  \dttsym{,}  \dttmv{y}  \dttsym{:}   \bar{  \dttnt{p}  }  \, \dttnt{C}  \mathbin{@}  \dttnt{n}  \vdash  \dttnt{t'_{{\mathrm{2}}}}  \dttsym{:}  \dttsym{-} \, \dttnt{C'}  \mathbin{@}  \dttnt{n'}
      }{G  \dttsym{;}  \Gamma_{{\mathrm{1}}}  \vdash  \nu \, \dttmv{x}  \dttsym{.}  \dttnt{t'_{{\mathrm{1}}}}  \mathbin{\Cdot[2]}  \dttnt{t'_{{\mathrm{2}}}}  \dttsym{:}  \dttnt{p} \, \dttnt{C}  \mathbin{@}  \dttnt{n}}
      \leqno{\raise 8 pt\hbox{\textbf{Case}}}      
    \]
  Similarly to the previous case.
\end{description}  

\subsection{Proof of Lemma~\ref{lemma:type_preservation}: Type Preservation}
\label{subsec:proof_of_lemma:type_preservation}
This is a proof by induction on the form of the assumed typing
derivation.  We only consider non-trivial cases.  All the other
cases either follow directly from assumptions or are similar to the
cases we provide below.

\begin{description}    
\item[\cW]
  \[
      \mprset{flushleft}
      \inferrule* [right=\ifrName{Cut}] {
        G  \dttsym{;}  \Gamma  \dttsym{,}  \dttmv{x}  \dttsym{:}   \bar{  \dttnt{p}  }  \, \dttnt{A}  \mathbin{@}  \dttnt{n}  \vdash  \dttnt{t_{{\mathrm{1}}}}  \dttsym{:}  \dttsym{+} \, \dttnt{B}  \mathbin{@}  \dttnt{n'}
        \\\\
            G  \dttsym{;}  \Gamma  \dttsym{,}  \dttmv{x}  \dttsym{:}   \bar{  \dttnt{p}  }  \, \dttnt{A}  \mathbin{@}  \dttnt{n}  \vdash  \dttnt{t_{{\mathrm{2}}}}  \dttsym{:}  \dttsym{-} \, \dttnt{B}  \mathbin{@}  \dttnt{n'}
      }{G  \dttsym{;}  \Gamma  \vdash  \nu \, \dttmv{x}  \dttsym{.}  \dttnt{t_{{\mathrm{1}}}}  \mathbin{\Cdot[2]}  \dttnt{t_{{\mathrm{2}}}}  \dttsym{:}  \dttnt{p} \, \dttnt{A}  \mathbin{@}  \dttnt{n}}
      \leqno{\raise 8 pt\hbox{\textbf{Case}}}      
  \]
  The interesting cases are the ones where the assumed cut is a
  redex itself, otherwise this case holds by the induction
  hypothesis.  Thus, we case split on the form of this redex. 
  \begin{description}
  \item[\cW] \textbf{Case}
    Suppose $ \nu \, \dttmv{x}  \dttsym{.}  \dttnt{t_{{\mathrm{1}}}}  \mathbin{\Cdot[2]}  \dttnt{t_{{\mathrm{2}}}}  \equiv  \nu \, \dttmv{x}  \dttsym{.}  \lambda  \dttmv{y}  \dttsym{.}  \dttnt{t'_{{\mathrm{1}}}}  \mathbin{\Cdot[2]}  \langle  \dttnt{t'_{{\mathrm{2}}}}  \dttsym{,}  \dttnt{t''_{{\mathrm{2}}}}  \rangle $, thus, $ \dttnt{t_{{\mathrm{1}}}}  \equiv  \lambda  \dttmv{y}  \dttsym{.}  \dttnt{t'_{{\mathrm{1}}}} $ and $ \dttnt{t_{{\mathrm{2}}}}  \equiv  \langle  \dttnt{t'_{{\mathrm{2}}}}  \dttsym{,}  \dttnt{t''_{{\mathrm{2}}}}  \rangle $.  
    This then implies that $ \dttnt{B}  \equiv   \dttnt{B_{{\mathrm{1}}}}  \ndto{ \dttsym{+} }  \dttnt{B_{{\mathrm{2}}}}  $ for some $\dttnt{B_{{\mathrm{1}}}}$ and $\dttnt{B_{{\mathrm{2}}}}$.  Then 
    \[    \dttnt{t}  \equiv  \nu \, \dttmv{x}  \dttsym{.}  \dttnt{t_{{\mathrm{1}}}}  \mathbin{\Cdot[2]}  \dttnt{t_{{\mathrm{2}}}}    \equiv  \nu \, \dttmv{x}  \dttsym{.}  \lambda  \dttmv{y}  \dttsym{.}  \dttnt{t'_{{\mathrm{1}}}}  \mathbin{\Cdot[2]}  \langle  \dttnt{t'_{{\mathrm{2}}}}  \dttsym{,}  \dttnt{t''_{{\mathrm{2}}}}  \rangle  \redto  \nu \, \dttmv{x}  \dttsym{.}  \dttsym{[}  \dttnt{t'_{{\mathrm{2}}}}  \dttsym{/}  \dttmv{y}  \dttsym{]}  \dttnt{t'_{{\mathrm{1}}}}  \mathbin{\Cdot[2]}  \dttnt{t''_{{\mathrm{2}}}}  \equiv  \dttnt{t'} . \]
    Now by inversion we know the following:
    \begin{center}
      \begin{math}
        \begin{array}{lll}
          (1) & G  \dttsym{,}  \dttsym{(}  \dttnt{n'} \,  \preccurlyeq_{ \dttsym{+} }  \, \dttnt{n''}  \dttsym{)}  \dttsym{;}  \Gamma  \dttsym{,}  \dttmv{x}  \dttsym{:}   \bar{  \dttnt{p}  }  \, \dttnt{A}  \mathbin{@}  \dttnt{n}  \dttsym{,}  \dttmv{y}  \dttsym{:}  \dttsym{+} \, \dttnt{B_{{\mathrm{1}}}}  \mathbin{@}  \dttnt{n''}  \vdash  \dttnt{t'_{{\mathrm{1}}}}  \dttsym{:}  \dttsym{+} \, \dttnt{B_{{\mathrm{2}}}}  \mathbin{@}  \dttnt{n''} \\
          & \,\,\,\text{ for some } \dttnt{n''} \, \not\in \, \dttsym{\mbox{$\mid$}}  G  \dttsym{\mbox{$\mid$}}  \dttsym{,}  \dttsym{\mbox{$\mid$}}  \Gamma  \dttsym{,}  \dttmv{x}  \dttsym{:}   \bar{  \dttnt{p}  }  \, \dttnt{A}  \mathbin{@}  \dttnt{n}  \dttsym{\mbox{$\mid$}}\\
          (2) & G  \dttsym{;}  \Gamma  \dttsym{,}  \dttmv{x}  \dttsym{:}   \bar{  \dttnt{p}  }  \, \dttnt{A}  \mathbin{@}  \dttnt{n}  \vdash  \dttnt{t'_{{\mathrm{2}}}}  \dttsym{:}  \dttsym{+} \, \dttnt{B_{{\mathrm{1}}}}  \mathbin{@}  \dttnt{n'''} \\
          (3) & G  \dttsym{;}  \Gamma  \dttsym{,}  \dttmv{x}  \dttsym{:}   \bar{  \dttnt{p}  }  \, \dttnt{A}  \mathbin{@}  \dttnt{n}  \vdash  \dttnt{t''_{{\mathrm{2}}}}  \dttsym{:}  \dttsym{-} \, \dttnt{B_{{\mathrm{2}}}}  \mathbin{@}  \dttnt{n'''}\\
          (4) & G  \vdash  \dttnt{n'} \,  \preccurlyeq^*_{ \dttsym{+} }  \, \dttnt{n'''}\\
        \end{array}
      \end{math}
    \end{center}
  \end{description}
  Using (1) and (4) we may apply node substitution for typing (Lemma~\ref{lemma:node_substitution_for_typing}) to obtain
  \[ (5)\,\dttsym{[}  \dttnt{n'''}  \dttsym{/}  \dttnt{n''}  \dttsym{]}  G  \dttsym{;}  \dttsym{[}  \dttnt{n'''}  \dttsym{/}  \dttnt{n''}  \dttsym{]}  \Gamma  \dttsym{,}  \dttmv{x}  \dttsym{:}   \bar{  \dttnt{p}  }  \, \dttnt{A}  \mathbin{@}  \dttnt{n}  \dttsym{,}  \dttmv{y}  \dttsym{:}  \dttsym{+} \, \dttnt{B_{{\mathrm{1}}}}  \mathbin{@}  \dttnt{n'''}  \vdash  \dttnt{t'_{{\mathrm{1}}}}  \dttsym{:}  \dttsym{+} \, \dttnt{B_{{\mathrm{2}}}}  \mathbin{@}  \dttnt{n'''}. \]
  
  Finally, by applying substitution for typing using (2) and (5) we obtain
  \[ (6)\,\dttsym{[}  \dttnt{n'''}  \dttsym{/}  \dttnt{n''}  \dttsym{]}  G  \dttsym{;}  \dttsym{[}  \dttnt{n'''}  \dttsym{/}  \dttnt{n''}  \dttsym{]}  \Gamma  \dttsym{,}  \dttmv{x}  \dttsym{:}   \bar{  \dttnt{p}  }  \, \dttnt{A}  \mathbin{@}  \dttnt{n}  \vdash  \dttsym{[}  \dttnt{t'_{{\mathrm{2}}}}  \dttsym{/}  \dttmv{y}  \dttsym{]}  \dttnt{t'_{{\mathrm{1}}}}  \dttsym{:}  \dttsym{+} \, \dttnt{B_{{\mathrm{2}}}}  \mathbin{@}  \dttnt{n'''},  \]  and since
  $\dttnt{n''}$ is a fresh in $G$ and $\Gamma$ we know (6) is equivalent to
  \[ (7)\,G  \dttsym{;}  \Gamma  \dttsym{,}  \dttmv{x}  \dttsym{:}   \bar{  \dttnt{p}  }  \, \dttnt{A}  \mathbin{@}  \dttnt{n}  \vdash  \dttsym{[}  \dttnt{t'_{{\mathrm{2}}}}  \dttsym{/}  \dttmv{y}  \dttsym{]}  \dttnt{t'_{{\mathrm{1}}}}  \dttsym{:}  \dttsym{+} \, \dttnt{B_{{\mathrm{2}}}}  \mathbin{@}  \dttnt{n'''}.  \]
  Finally, by applying the $\dttdrulename{Cut}$ rule using (7) and (3) we obtain 
  \[ G  \dttsym{;}  \Gamma  \vdash  \nu \, \dttmv{x}  \dttsym{.}  \dttsym{[}  \dttnt{t'_{{\mathrm{2}}}}  \dttsym{/}  \dttmv{y}  \dttsym{]}  \dttnt{t'_{{\mathrm{1}}}}  \mathbin{\Cdot[2]}  \dttnt{t''_{{\mathrm{2}}}}  \dttsym{:}  \dttnt{p} \, \dttnt{A}  \mathbin{@}  \dttnt{n}. \]

\end{description}

\subsection{Proof of Lemma~\ref{lem:sninterp}: SN Interpretations}
\label{subsec:proof_of_lemma_sn_interpretations}
For purposes of this proof and subsequent ones, define $\delta(t)$ to be
the length of the longest reduction sequence from $t$ to a normal form, for
$t\in\SN$.

The proof of the lemma is by mutual well-founded induction on the
pair $(A,n)$, where $n$ is the number of the proposition in the
statement of the lemma; the well-founded ordering in question is the
lexicographic combination of the structural ordering on types (for $A$) 
and the ordering $1 > 2 > 4 > 3$ (for $n$).

For proposition (1): assume $t\in\interp{A}^+$, and show $t\in\SN$.
Let $x$ be a variable.  By IH(2), $x\in\interp{A}^-$, so by the
definition of $\interp{A}^+$, we have
\[
  \nu \, \dttmv{x}  \dttsym{.}  \dttnt{t}  \mathbin{\Cdot[2]}  \dttmv{x} \in \SN
\]
This implies $t\in\SN$.

For proposition (2): assume $x\in\textit{Vars}$, and show
$x\in\interp{A}^-$.  For the latter, it suffices to assume arbitrary
$y\in\textit{Vars}$ and $t'\in\interp{A}^{+c}$, and show $\nu \, \dttmv{y}  \dttsym{.}  \dttnt{t'}  \mathbin{\Cdot[2]}  \dttmv{x} \in \SN$.  We will prove this by inner induction on
$\delta(t')$, which is defined by IH(4).  By the definition of
$\interp{A}^{+c}$ for the various cases of $A$, we see that $\nu \, \dttmv{y}  \dttsym{.}  \dttnt{t'}  \mathbin{\Cdot[2]}  \dttmv{x}$ cannot be a redex itself, as $t'$ cannot be a cut.  If
$t'$ is a normal form we are done.  If $t\leadsto t''$, then we have
$t''\in\interp{A}^{+c}$ by Lemma~\ref{lem:stepinterp}, and we may
apply the inner induction hypothesis.

For proposition (3): assume $t\in\interp{A}^-$, and show $t\in\SN$.
By the definition of $\interp{A}^-$ and the fact that $\textit{Vars}\subseteq\interp{A}^{+c}$
by definition of $\interp{A}^{+c}$, we have 
\[
  \nu \, \dttmv{y}  \dttsym{.}  \dttmv{y}  \mathbin{\Cdot[2]}  \dttnt{t} \in \SN
\]
This implies $t\in\SN$ as required.

For proposition (4): assume $t\in\interp{A}^{+c}$, and consider the
following cases.  If $t\in\textit{Vars}$ or $A\equiv \langle  \dttsym{+} \rangle $, then
$t$ is normal and the result is immediate.  So suppose $A \equiv
 \dttnt{A_{{\mathrm{1}}}}  \ndto{ \dttsym{+} }  \dttnt{A_{{\mathrm{2}}}} $.  Then $t\equiv \lambda x.t'$ for some $x$ and $t'$
where for all $t''\in\interp{A_1}^+$, $[t''/x]t'\in\interp{A_2}^+$.
By IH(2), the variable $x$ itself is in $\interp{A_1}^+$, so
we know that $t'\equiv[x/x]t'\in\interp{A_2}^+$.  Then by IH(1)
we have $t'\in\SN$, which implies $\lambda x.t'\in\SN$.  If $A\equiv  \dttnt{A_{{\mathrm{1}}}}  \ndto{ \dttsym{-} }  \dttnt{A_{{\mathrm{2}}}} $,
then $t\equiv \langle  \dttnt{t_{{\mathrm{1}}}}  \dttsym{,}  \dttnt{t_{{\mathrm{2}}}}  \rangle$ for some $t_1\in\interp{A_1}^-$ and
$t_2\in\interp{A_2}^+$.  By IH(3) and IH(1), $t_1\in\SN$ and $t_2\in\SN$,
which implies $\langle  \dttnt{t_{{\mathrm{1}}}}  \dttsym{,}  \dttnt{t_{{\mathrm{2}}}}  \rangle\in\SN$.  The cases for $A \equiv  \dttnt{A_{{\mathrm{1}}}}  \ndwedge{ \dttnt{p} }  \dttnt{A_{{\mathrm{2}}}} $
are similar to this one.

\subsection{Proof of Theorem~\ref{thm:sndinterp}: Soundness}
\label{subsec:proof_of_soundness}
The proof is by induction on the derivation of $\Gamma  \vdash_c  \dttnt{t}  \dttsym{:}  \dttnt{p} \, \dttnt{A}$.  We consider
the two possible polarities for the conclusion of the typing judgment separately.

\begin{description}
\item[\cW]

  \[
  \inferrule* [right=\ifrName{ClassAx}] {\ }{\Gamma  \dttsym{,}  \dttmv{x}  \dttsym{:}  \dttnt{p} \, \dttnt{A}  \dttsym{,}  \Gamma'  \vdash_c  \dttmv{x}  \dttsym{:}  \dttnt{p} \, \dttnt{A}}
  \leqno{\raise 8 pt\hbox{\textbf{Case}}}      
  \]

  Since $\sigma\in\interp{\Gamma  \dttsym{,}  \dttmv{x}  \dttsym{:}  \dttnt{p} \, \dttnt{A}  \dttsym{,}  \Gamma'}$, $\sigma(x)\in\interp{A}^p$ as required.

\item[\cW]

  \[
  \inferrule* [right=\ifrName{ClassUnit}] {\ }{\Gamma  \vdash_c  \dttkw{triv}  \dttsym{:}  \dttsym{+} \,  \langle  \dttsym{+} \rangle }
  \leqno{\raise 8 pt\hbox{\textbf{Case}}}      
  \]

  We have $\dttkw{triv}\in\interp{ \langle  \dttsym{+} \rangle }^{+c}$ by definition.

\item[\cW]

  \[
  \inferrule* [right=\ifrName{ClassUnit}] {\ }{\Gamma  \vdash_c  \dttkw{triv}  \dttsym{:}  \dttsym{-} \,  \langle  \dttsym{-} \rangle }
  \leqno{\raise 8 pt\hbox{\textbf{Case}}}      
  \]

  To prove $\dttkw{triv}\in\interp{ \langle  \dttsym{-} \rangle }^{-}$, it suffices to assume
  arbitrary $y\in\textit{Vars}$ and $t\in\interp{ \langle  \dttsym{-} \rangle }^{+c}$, and
  show $\nu \, \dttmv{y}  \dttsym{.}  \dttnt{t}  \mathbin{\Cdot[2]}  \dttkw{triv}\in\SN$.  By definition of
  $\interp{ \langle  \dttsym{-} \rangle }^{+c}$, $t\in\textit{Vars}$, and then $\nu \, \dttmv{y}  \dttsym{.}  \dttnt{t}  \mathbin{\Cdot[2]}  \dttkw{triv}$ is in normal form.

\item[\cW]

  \[
  \inferrule* [right=\ifrName{ClassAnd}] {\Gamma  \vdash_c  \dttnt{t_{{\mathrm{1}}}}  \dttsym{:}  \dttsym{+} \, \dttnt{A} \qquad \Gamma  \vdash_c  \dttnt{t_{{\mathrm{2}}}}  \dttsym{:}  \dttsym{+} \, \dttnt{B}}{\Gamma  \vdash_c  \dttsym{(}  \dttnt{t_{{\mathrm{1}}}}  \dttsym{,}  \dttnt{t_{{\mathrm{2}}}}  \dttsym{)}  \dttsym{:}  \dttsym{+} \,  \dttnt{A}  \ndwedge{ \dttsym{+} }  \dttnt{B} }
  \leqno{\raise 8 pt\hbox{\textbf{Case}}}      
  \]

  By Lemma~\ref{lem:canonpos}, it suffices to show $(\sigma t_1,\sigma
  t_2)\in\interp{ \dttnt{A}  \ndwedge{ \dttsym{+} }  \dttnt{B} }^{+c}$.  This follows directly from the
  definition of $\interp{ \dttnt{A}  \ndwedge{ \dttsym{+} }  \dttnt{B} }^{+c}$, since the IH gives us
  $\sigma t_1\in\interp{A}^+$ and $\sigma t_2\in\interp{B}^+$.

\item[\cW]

  \[
  \inferrule* [right=\ifrName{ClassAnd}] {\Gamma  \vdash_c  \dttnt{t_{{\mathrm{1}}}}  \dttsym{:}  \dttsym{-} \, \dttnt{A_{{\mathrm{1}}}} \qquad \Gamma  \vdash_c  \dttnt{t_{{\mathrm{2}}}}  \dttsym{:}  \dttsym{-} \, \dttnt{A_{{\mathrm{2}}}}}{\Gamma  \vdash_c  \dttsym{(}  \dttnt{t_{{\mathrm{1}}}}  \dttsym{,}  \dttnt{t_{{\mathrm{2}}}}  \dttsym{)}  \dttsym{:}  \dttsym{-} \,  \dttnt{A_{{\mathrm{1}}}}  \ndwedge{ \dttsym{-} }  \dttnt{A_{{\mathrm{2}}}} }
  \leqno{\raise 8 pt\hbox{\textbf{Case}}}      
  \]

  It suffices to assume arbitrary $y\in\textit{Vars}$ and
  $t'\in\interp{ \dttnt{A_{{\mathrm{1}}}}  \ndwedge{ \dttsym{-} }  \dttnt{A_{{\mathrm{2}}}} }^{+c}$, and show $\nu \, \dttmv{y}  \dttsym{.}  \dttnt{t'}  \mathbin{\Cdot[2]}  \dttsym{(}  \sigma \, \dttnt{t_{{\mathrm{1}}}}  \dttsym{,}  \sigma \, \dttnt{t_{{\mathrm{2}}}}  \dttsym{)}\in\SN$.  If $t'\in\textit{Vars}$, then this follows by
  Lemma~\ref{lem:sninterp} from the facts that $\sigma
  t_1\in\interp{A_1}^+$ and $\sigma t_2\in\interp{\dttnt{A_{{\mathrm{2}}}}}^+$, which we have
  by the IH.  So suppose $t'$ is of the form $ \mathbf{in}_{ \dttnt{d} }\, \dttnt{t''} $ for some
  $d$ and some $t''\in\interp{A_d}^+$.  By the definition of $\SN$, it
  suffices to show that all one-step successors $t_a$ of the term in
  question are $\SN$.  The proof of this is by inner induction on
  $\delta(t'') + \delta(\sigma t_1) + \delta(\sigma t_2)$, which exists
  by Lemma~\ref{lem:sninterp}, using also Lemma~\ref{lem:stepinterp}.
  Suppose that we step to $t_a$ by stepping $t''$, $\sigma t_1$, or
  $\sigma t_2$.  Then the result holds by the inner IH.  So consider the
  step
  \[
    \nu \, \dttmv{y}  \dttsym{.}   \mathbf{in}_{ \dttnt{d} }\, \dttnt{t''}   \mathbin{\Cdot[2]}  \dttsym{(}  \sigma \, \dttnt{t_{{\mathrm{1}}}}  \dttsym{,}  \sigma \, \dttnt{t_{{\mathrm{2}}}}  \dttsym{)}  \redto  \nu \, \dttmv{y}  \dttsym{.}  \dttnt{t''}  \mathbin{\Cdot[2]}  \sigma \,   \dttnt{t} _{ \dttnt{d} }  
    \]
    We then have $\nu \, \dttmv{y}  \dttsym{.}  \dttnt{t''}  \mathbin{\Cdot[2]}  \sigma \,   \dttnt{t} _{ \dttnt{d} }  \in\SN$ from the facts
    that $t''\in\interp{A_d}^+$ and $\sigma t_d\in\interp{A_d}^-$, by
    the definition of $\interp{A_d}^+$.

  \item[\cW]

    \[
    \inferrule* [right=\ifrName{ClassAndBar}] {\Gamma  \vdash_c  \dttnt{t}  \dttsym{:}  \dttsym{+} \,  \dttnt{A} _{ \dttnt{d} } }{\Gamma  \vdash_c   \mathbf{in}_{ \dttnt{d} }\, \dttnt{t}   \dttsym{:}  \dttsym{+} \,  \dttnt{A_{{\mathrm{1}}}}  \ndwedge{ \dttsym{-} }  \dttnt{A_{{\mathrm{2}}}} }
    \leqno{\raise 8 pt\hbox{\textbf{Case}}}      
    \]
    By Lemma~\ref{lem:canonpos}, it suffices to prove $ \mathbf{in}_{ \dttnt{d} }\, \sigma \, \dttnt{t} \in\interp{ \dttnt{A_{{\mathrm{1}}}}  \ndwedge{ \dttsym{-} }  \dttnt{A_{{\mathrm{2}}}} }^+$, but by the definition of $\interp{ \dttnt{A_{{\mathrm{1}}}}  \ndwedge{ \dttsym{-} }  \dttnt{A_{{\mathrm{2}}}} }^+$, this follows directly from $\sigma
    t\in\interp{A_d}^+$, which we have by the IH.

  \item[\cW]

    \[
    \inferrule* [right=\ifrName{ClassAndBar}] {\Gamma  \vdash_c  \dttnt{t}  \dttsym{:}  \dttsym{-} \,  \dttnt{A} _{ \dttnt{d} } }{\Gamma  \vdash_c   \mathbf{in}_{ \dttnt{d} }\, \dttnt{t}   \dttsym{:}  \dttsym{-} \,  \dttnt{A_{{\mathrm{1}}}}  \ndwedge{ \dttsym{+} }  \dttnt{A_{{\mathrm{2}}}} }
    \leqno{\raise 8 pt\hbox{\textbf{Case}}}      
    \]
    To prove $ \mathbf{in}_{ \dttnt{d} }\, \sigma \, \dttnt{t} \in\interp{ \dttnt{A_{{\mathrm{1}}}}  \ndwedge{ \dttsym{+} }  \dttnt{A_{{\mathrm{2}}}} }^-$, it suffices
    to assume arbitrary $y\in\textit{Vars}$ and $t'\in\interp{ \dttnt{A_{{\mathrm{1}}}}  \ndwedge{ \dttsym{+} }  \dttnt{A_{{\mathrm{2}}}} }^{+c}$, and show $\nu \, \dttmv{y}  \dttsym{.}  \dttnt{t'}  \mathbin{\Cdot[2]}   \mathbf{in}_{ \dttnt{d} }\, \sigma \, \dttnt{t} \in\SN$.  If
    $t'\in\textit{Vars}$, then this follows from the fact that $\sigma
    t\in\SN$, which we have by Lemma~\ref{lem:sninterp} from $\sigma
    t\in\interp{A_d}^-$ (which the IH gives us).  So suppose $t'$ is of
    the form $(s_1,s_2)$ for some $s_1\in\interp{A_1}^+$ and
    $s_2\in\interp{A_2}^+$.  It suffices to prove that all one-step
    successors of the term in question are in $\SN$, as we did in a
    previous case above.  Lemma~\ref{lem:sninterp} lets us proceed by
    inner induction on $\delta(\sigma t) + \delta(s_1) + \delta(s_2)$,
    using also Lemma~\ref{lem:stepinterp}.  If we step $\sigma t$, $s_1$
    or $s_2$, then the result holds by inner IH.  Otherwise, we have the
    step
    \[
      \nu \, \dttmv{y}  \dttsym{.}  \dttsym{(}  \dttnt{s_{{\mathrm{1}}}}  \dttsym{,}  \dttnt{s_{{\mathrm{2}}}}  \dttsym{)}  \mathbin{\Cdot[2]}   \mathbf{in}_{ \dttnt{d} }\, \sigma \, \dttnt{t}   \redto  \nu \, \dttmv{y}  \dttsym{.}   \dttnt{s} _{ \dttnt{d} }   \mathbin{\Cdot[2]}  \sigma \, \dttnt{t}
      \]
      And this successor is in $\SN$ by the facts that $s_d\in\interp{A_d}^+$
      and $\sigma t\in\interp{A_d}^-$, from the definition of $\interp{A_d}^+$.

    \item[\cW]

      \[
      \inferrule* [right=\ifrName{ClassImp}] {\Gamma  \dttsym{,}  \dttmv{x}  \dttsym{:}  \dttsym{+} \, \dttnt{A}  \vdash_c  \dttnt{t}  \dttsym{:}  \dttsym{+} \, \dttnt{B}}{\Gamma  \vdash_c  \lambda  \dttmv{x}  \dttsym{.}  \dttnt{t}  \dttsym{:}  \dttsym{+} \,  \dttnt{A}  \ndto{ \dttsym{+} }  \dttnt{B} }
      \leqno{\raise 8 pt\hbox{\textbf{Case}}}      
      \]
      By Lemma~\ref{lem:canonpos}, it suffices to assume arbitrary $y\in\textit{Vars}$ and $t'\in\interp{A}^+$,
      and prove $[t'/x](\sigma t)\in\interp{B}^+$.  But this follows immediately from the IH, since
      $[t'/x](\sigma t)\equiv (\sigma[x\mapsto t']) t$ and $\sigma[x\mapsto t]\in\interp{\Gamma  \dttsym{,}  \dttmv{x}  \dttsym{:}  \dttsym{+} \, \dttnt{A}}$.

    \item[\cW]

      \[
      \inferrule* [right=\ifrName{ClassImp}] {\Gamma  \dttsym{,}  \dttmv{x}  \dttsym{:}  \dttsym{-} \, \dttnt{A}  \vdash_c  \dttnt{t}  \dttsym{:}  \dttsym{-} \, \dttnt{B}}{\Gamma  \vdash_c  \lambda  \dttmv{x}  \dttsym{.}  \dttnt{t}  \dttsym{:}  \dttsym{-} \,  \dttnt{A}  \ndto{ \dttsym{-} }  \dttnt{B} }
      \leqno{\raise 8 pt\hbox{\textbf{Case}}}      
      \]
      It suffices to assume arbitrary $y\in\textit{Vars}$ and
      $t'\in\interp{ \dttnt{A}  \ndto{ \dttsym{-} }  \dttnt{B} }^{+c}$, and show $\nu \, \dttmv{y}  \dttsym{.}  \dttnt{t'}  \mathbin{\Cdot[2]}  \lambda  \dttmv{x}  \dttsym{.}  \sigma \, \dttnt{t}\in\SN$.  Let us first observe that $\sigma \, \dttnt{t}\in\SN$, because by
      the IH, for all $\sigma'\in\interp{\Gamma  \dttsym{,}  \dttmv{x}  \dttsym{:}  \dttsym{-} \, \dttnt{A}}$, we have $\sigma'
      t\in\interp{B}^-$, and $\interp{B}^-\subseteq\SN$ by
      Lemma~\ref{lem:sninterp}.  We may instantiate this with
      $\sigma[x\mapsto x]$, since by Lemma~\ref{lem:sninterp},
      $x\in\interp{A}^-$.  Since $\sigma \, \dttnt{t}\in\SN$, we also have $\lambda  \dttmv{x}  \dttsym{.}  \sigma \, \dttnt{t}\in\SN$.  Now let us consider cases for the assumption
      $t'\in\interp{ \dttnt{A}  \ndto{ \dttsym{-} }  \dttnt{B} }^{+c}$.  If $t'\in\textit{Vars}$ then we
      directly have $\nu \, \dttmv{y}  \dttsym{.}  \dttnt{t'}  \mathbin{\Cdot[2]}  \lambda  \dttmv{x}  \dttsym{.}  \sigma \, \dttnt{t}\in\SN$ from $\lambda  \dttmv{x}  \dttsym{.}  \sigma \, \dttnt{t}\in\SN$.  So assume $t'\equiv\langle  \dttnt{t_{{\mathrm{1}}}}  \dttsym{,}  \dttnt{t_{{\mathrm{2}}}}  \rangle$ for some
      $t_1\in\interp{A}^-$ and $t_2\in\interp{B}^+$.  By
      Lemma~\ref{lem:sninterp} again, we may reason by inner induction on
      $\delta(t_1)+\delta(t_2)+\delta(\sigma t)$ to show that all one-step
      successors of $\nu \, \dttmv{y}  \dttsym{.}  \langle  \dttnt{t_{{\mathrm{1}}}}  \dttsym{,}  \dttnt{t_{{\mathrm{2}}}}  \rangle  \mathbin{\Cdot[2]}  \lambda  \dttmv{x}  \dttsym{.}  \sigma \, \dttnt{t}$ are in $\SN$,
      using also Lemma~\ref{lem:stepinterp}.  We can see that $t_1$, $t_2$, and $\sigma t$
      are structurally smaller, and hence, if any one of them steps, then the result follows by the inner IH.  So suppose we have
      the step
      \[
        \nu \, \dttmv{y}  \dttsym{.}  \langle  \dttnt{t_{{\mathrm{1}}}}  \dttsym{,}  \dttnt{t_{{\mathrm{2}}}}  \rangle  \mathbin{\Cdot[2]}  \lambda  \dttmv{x}  \dttsym{.}  \sigma \, \dttnt{t}  \redto  \nu \, \dttmv{y}  \dttsym{.}  \dttnt{t_{{\mathrm{2}}}}  \mathbin{\Cdot[2]}  \dttsym{[}  \dttnt{t_{{\mathrm{1}}}}  \dttsym{/}  \dttmv{x}  \dttsym{]}  \dttsym{(}  \sigma \, \dttnt{t}  \dttsym{)}
        \]
        Since $t_1\in\interp{A}^-$, the substitution $\sigma[x\mapsto t_1]$ is
        in $\interp{\Gamma  \dttsym{,}  \dttmv{x}  \dttsym{:}  \dttsym{-} \, \dttnt{A}}$.  So we may apply the IH to obtain $[t_1
          / x ] (\sigma t) \equiv \sigma[x\mapsto t_1]\in\interp{B}^-$.  Then
        since $t_2\in\interp{B}^+$, we have $\nu \, \dttmv{y}  \dttsym{.}  \dttnt{t_{{\mathrm{2}}}}  \mathbin{\Cdot[2]}  \dttsym{[}  \dttnt{t_{{\mathrm{1}}}}  \dttsym{/}  \dttmv{x}  \dttsym{]}  \dttsym{(}  \sigma \, \dttnt{t}  \dttsym{)}$ by definition of $\interp{B}^+$.
\enlargethispage{\baselineskip}
      \item[\cW]

        \[
        \inferrule* [right=\ifrName{ClassImpBar}] {\Gamma  \vdash_c  \dttnt{t_{{\mathrm{1}}}}  \dttsym{:}  \dttsym{-} \, \dttnt{A}  \qquad \Gamma  \vdash_c  \dttnt{t_{{\mathrm{2}}}}  \dttsym{:}  \dttsym{+} \, \dttnt{B}}{\Gamma  \vdash_c  \langle  \dttnt{t_{{\mathrm{1}}}}  \dttsym{,}  \dttnt{t_{{\mathrm{2}}}}  \rangle  \dttsym{:}  \dttsym{+} \, \dttsym{(}   \dttnt{A}  \ndto{ \dttsym{-} }  \dttnt{B}   \dttsym{)}}
        \leqno{\raise 8 pt\hbox{\textbf{Case}}}      
        \]
        By Lemma~\ref{lem:canonpos}, as in previous cases of positive typing,
        it suffices to prove $\langle  \sigma \, \dttnt{t_{{\mathrm{1}}}}  \dttsym{,}  \sigma \, \dttnt{t_{{\mathrm{2}}}}  \rangle\in\interp{ \dttnt{A}  \ndto{ \dttsym{-} }  \dttnt{B} }^{+c}$.  By the definition of $\interp{ \dttnt{A}  \ndto{ \dttsym{-} }  \dttnt{B} }^{+c}$, this
        follows directly from $\sigma t_1\in\interp{A}^-$ and $\sigma
        t_2\in\interp{B}^+$, which we have by the IH.

      \item[\cW]

        \[
        \inferrule* [right=\ifrName{ClassImpBar}] {\Gamma  \vdash_c  \dttnt{t_{{\mathrm{1}}}}  \dttsym{:}  \dttsym{+} \, \dttnt{A}  \qquad \Gamma  \vdash_c  \dttnt{t_{{\mathrm{2}}}}  \dttsym{:}  \dttsym{-} \, \dttnt{B}}{\Gamma  \vdash_c  \langle  \dttnt{t_{{\mathrm{1}}}}  \dttsym{,}  \dttnt{t_{{\mathrm{2}}}}  \rangle  \dttsym{:}  \dttsym{-} \, \dttsym{(}   \dttnt{A}  \ndto{ \dttsym{+} }  \dttnt{B}   \dttsym{)}}
        \leqno{\raise 8 pt\hbox{\textbf{Case}}}      
        \]
        It suffices to assume arbitrary $y\in\textit{Vars}$ and
        $t'\in\interp{ \dttnt{A}  \ndto{ \dttsym{+} }  \dttnt{B} }^{+c}$, and show $\nu \, \dttmv{y}  \dttsym{.}  \dttnt{t'}  \mathbin{\Cdot[2]}  \langle  \sigma \, \dttnt{t_{{\mathrm{1}}}}  \dttsym{,}  \sigma \, \dttnt{t_{{\mathrm{2}}}}  \rangle\in\SN$.  By the IH, we have $\sigma t_1\in\interp{A}^+$
        and $\sigma t_2\in\interp{B}^-$, and hence $\sigma t_1\in\SN$ and
        $\sigma t_2\in\SN$ by Lemma~\ref{lem:sninterp}.  If
        $t'\in\textit{Vars}$, then these facts are sufficient to show the term
        in question is in $\SN$.  So suppose $t'\equiv \lambda x.t_3$, for
        some $x\in\textit{Vars}$ and $t''$ such that for all
        $t_4\in\interp{A}^+$, $[t_4/x]t_3\in\interp{B}^+$.  By similar
        reasoning as in a previous case, we have $t_3\in\SN$.  So we may
        proceed by inner induction on $\delta(t_1)+\delta(t_2)+\delta(t_3)$ to
        show that all one-step successors of $\nu \, \dttmv{y}  \dttsym{.}  \lambda  \dttmv{x}  \dttsym{.}  \dttnt{t_{{\mathrm{3}}}}  \mathbin{\Cdot[2]}  \langle  \sigma \, \dttnt{t_{{\mathrm{1}}}}  \dttsym{,}  \sigma \, \dttnt{t_{{\mathrm{2}}}}  \rangle$ are in $\SN$, using also Lemma~\ref{lem:stepinterp}.  We can see that $t_3$, $\sigma t_1$, and $\sigma t_2$ are structurally smaller, and
        hence, if anyone of them steps, then the result follows by the inner IH.  So consider this step:
        \[
          \nu \, \dttmv{y}  \dttsym{.}  \lambda  \dttmv{x}  \dttsym{.}  \dttnt{t_{{\mathrm{3}}}}  \mathbin{\Cdot[2]}  \langle  \sigma \, \dttnt{t_{{\mathrm{1}}}}  \dttsym{,}  \sigma \, \dttnt{t_{{\mathrm{2}}}}  \rangle  \redto  \nu \, \dttmv{y}  \dttsym{.}  \dttsym{[}  \sigma \, \dttnt{t_{{\mathrm{1}}}}  \dttsym{/}  \dttmv{x}  \dttsym{]}  \dttnt{t_{{\mathrm{3}}}}  \mathbin{\Cdot[2]}  \sigma \, \dttnt{t_{{\mathrm{2}}}}
          \]
          Since we have that $\sigma t_1\in\interp{A}^+$, the assumption about substitution
          instances of $t_3$ gives us that $[\sigma t_1/x]t_3\in\interp{B}^+$, which is
          then sufficient to conclude $\nu \, \dttmv{y}  \dttsym{.}  \dttsym{[}  \sigma \, \dttnt{t_{{\mathrm{1}}}}  \dttsym{/}  \dttmv{x}  \dttsym{]}  \dttnt{t_{{\mathrm{3}}}}  \mathbin{\Cdot[2]}  \sigma \, \dttnt{t_{{\mathrm{2}}}}\in\SN$
          by the definition of $\interp{B}^+$.

        \item[\cW]

          \[
          \inferrule* [right=\ifrName{ClassCut}] {\Gamma  \dttsym{,}  \dttmv{x}  \dttsym{:}  \dttsym{-} \, \dttnt{A}  \vdash_c  \dttnt{t_{{\mathrm{1}}}}  \dttsym{:}  \dttsym{+} \, \dttnt{B}  \qquad \Gamma  \dttsym{,}  \dttmv{x}  \dttsym{:}  \dttsym{-} \, \dttnt{A}  \vdash_c  \dttnt{t_{{\mathrm{2}}}}  \dttsym{:}  \dttsym{-} \, \dttnt{B}}{\Gamma  \vdash_c  \nu \, \dttmv{x}  \dttsym{.}  \dttnt{t_{{\mathrm{1}}}}  \mathbin{\Cdot[2]}  \dttnt{t_{{\mathrm{2}}}}  \dttsym{:}  \dttsym{+} \, \dttnt{A}}
          \leqno{\raise 8 pt\hbox{\textbf{Case}}}      
          \]
          It suffices to assume arbitrary $y\in\textit{Vars}$ and
          $t'\in\interp{A}^-$, and show $\nu \, \dttmv{y}  \dttsym{.}  \dttsym{(}  \nu \, \dttmv{x}  \dttsym{.}  \sigma \, \dttnt{t_{{\mathrm{1}}}}  \mathbin{\Cdot[2]}  \sigma \, \dttnt{t_{{\mathrm{2}}}}  \dttsym{)}  \mathbin{\Cdot[2]}  \dttnt{t'}\in\SN$.  By the IH and part 2 of Lemma~\ref{lem:sninterp}, we
          know that $\sigma t_1\in\interp{B}^+$ and $\sigma t_2\in\interp{B}^-$.
          By Lemma~\ref{lem:sninterp} again, we have $t'\in\SN$, $\sigma
          t_1\in\SN$, and $\sigma t_2\in\SN$.  So we may reason by induction on
          $\delta(t')+\delta(\sigma t_1)+\delta(\sigma t_2)$ to show that all
          one-step successors of $\nu \, \dttmv{y}  \dttsym{.}  \dttsym{(}  \nu \, \dttmv{x}  \dttsym{.}  \sigma \, \dttnt{t_{{\mathrm{1}}}}  \mathbin{\Cdot[2]}  \sigma \, \dttnt{t_{{\mathrm{2}}}}  \dttsym{)}  \mathbin{\Cdot[2]}  \dttnt{t'}$
          are in $\SN$, using also Lemma~\ref{lem:stepinterp}.  We can see that $t'$,
          $\sigma t_1$, and $\sigma t_2$ are structurally smaller, and hence, if any one of them steps, then the result follows by
          the inner IH.  The only possible other reduction is by the
          \dttdrulename{RBetaL} reduction rule (Figure~\ref{fig:dtt-red}).  And
          then, since $t'\in\interp{A}^-$, we may apply the IH to conclude that
          $[t'/x](\sigma t_1)\in\interp{B}^+$ and $[t'/x](\sigma t_2)\in\interp{B}^-$.
          By the definition of $\in\interp{B}^+$, this suffices to prove
          $\nu \, \dttmv{y}  \dttsym{.}  \dttsym{[}  \dttnt{t'}  \dttsym{/}  \dttmv{x}  \dttsym{]}  \sigma \, \dttnt{t_{{\mathrm{1}}}}  \mathbin{\Cdot[2]}  \dttsym{[}  \dttnt{t'}  \dttsym{/}  \dttmv{x}  \dttsym{]}  \sigma \, \dttnt{t_{{\mathrm{2}}}}\in\SN$, as required.

        \item[\cW]

          \[
          \inferrule* [right=\ifrName{ClassCut}] {\Gamma  \dttsym{,}  \dttmv{x}  \dttsym{:}  \dttsym{-} \, \dttnt{A}  \vdash_c  \dttnt{t_{{\mathrm{1}}}}  \dttsym{:}  \dttsym{+} \, \dttnt{B}  \qquad \Gamma  \dttsym{,}  \dttmv{x}  \dttsym{:}  \dttsym{-} \, \dttnt{A}  \vdash_c  \dttnt{t_{{\mathrm{2}}}}  \dttsym{:}  \dttsym{-} \, \dttnt{B}}{\Gamma  \vdash_c  \nu \, \dttmv{x}  \dttsym{.}  \dttnt{t_{{\mathrm{1}}}}  \mathbin{\Cdot[2]}  \dttnt{t_{{\mathrm{2}}}}  \dttsym{:}  \dttsym{-} \, \dttnt{A}}
          \leqno{\raise 8 pt\hbox{\textbf{Case}}}      
          \]
          It suffices to consider arbitrary $y\in\textit{Vars}$ and
          $t'\in\interp{A}^{+c}$, and show $\nu \, \dttmv{y}  \dttsym{.}  \dttnt{t'}  \mathbin{\Cdot[2]}  \dttsym{(}  \nu \, \dttmv{x}  \dttsym{.}  \sigma \, \dttnt{t_{{\mathrm{1}}}}  \mathbin{\Cdot[2]}  \sigma \, \dttnt{t_{{\mathrm{2}}}}  \dttsym{)}\in\SN$.  By the IH and part 2 of Lemma~\ref{lem:sninterp},
          we have $\sigma t_1\in\interp{B}^+$ and $\sigma t_2\in\interp{B}^-$,
          which implies $\sigma t_1\in\SN$ and $\sigma t_2\in\SN$ by
          Lemma~\ref{lem:sninterp} again.  We proceed by inner induction on
          $\delta(t')+\delta(\sigma t_1)+\delta(\sigma t_2)$, using
          Lemma~\ref{lem:stepinterp}, to show that all one-step successors of
          $\nu \, \dttmv{y}  \dttsym{.}  \dttnt{t'}  \mathbin{\Cdot[2]}  \dttsym{(}  \nu \, \dttmv{x}  \dttsym{.}  \sigma \, \dttnt{t_{{\mathrm{1}}}}  \mathbin{\Cdot[2]}  \sigma \, \dttnt{t_{{\mathrm{2}}}}  \dttsym{)}$ are in $\SN$.  We can see that 
          $t'$, $\sigma t_1$, and $\sigma t_2$ are structurally smaller, and hence, if any one of them steps, then the result holds
          by inner IH.  The only other reduction possible is by
          \dttdrulename{RBetaR}, since $t'$ cannot be a cut term by the
          definition of $\interp{A}^{+c}$.  In this case, the IH gives us
          $[t'/x]\sigma t_1\in\interp{B}^+$ and $[t'/x]\sigma
          t_2\in\interp{B}^-$, and we then have $\nu \, \dttmv{y}  \dttsym{.}  \dttsym{[}  \dttnt{t'}  \dttsym{/}  \dttmv{x}  \dttsym{]}  \sigma \, \dttnt{t_{{\mathrm{1}}}}  \mathbin{\Cdot[2]}  \dttsym{[}  \dttnt{t'}  \dttsym{/}  \dttmv{x}  \dttsym{]}  \sigma \, \dttnt{t_{{\mathrm{2}}}}\in\SN$ by the definition of $\interp{B}^+$.
\end{description}

    


\end{document}